\documentclass[fleqn,final]{unmpandathesis}
\pdfoutput=1
\usepackage{mystuff}
\usepackage[square,sort&compress]{natbib}
\setcitestyle{square,aysep={},yysep={;}}
\usepackage[margin=10pt,labelfont=bf]{caption} 

\begin{document}
\DeclareGraphicsExtensions{.pdf}
\frontmatter

\title{Parameter Estimation, Model Reduction and Quantum Filtering}

\author{Bradley A. Chase}

\degreesubject{Ph.D., Physics}

\degree{Doctor of Philosophy \\ Physics}

\documenttype{Dissertation}

\previousdegrees{B.S., Physics, Rice University, 2006 \\
                 B.A., Computer Science, Rice University, 2006}

\date{December, \thisyear}

\maketitle

\makecopyright
%
\begin{acknowledgments}
    \renewcommand{\baselinestretch}{1.45}\selectfont
   As difficult as distilling three years of research into a cohesive dissertation might seem, it is likewise impossible to adequately express my gratitude and appreciation for the many teachers, mentors, friends and family who have played a part in my personal and academic journey.  I will highlight a few.  
   
   First and foremost, I thank my advisor JM Geremia.  JM has a remarkable combination of stamina, breadth of knowledge and creativity that is extraordinarily catalyzing and inspiring.  I am truly thankful for his introducing me to the topics of quantum filtering and control, especially in the context of experimental work ongoing in our lab.  I am most grateful for his understanding that research is a marathon, but one which is worth running as hard as you can.  I also thank Andrew Landahl for his guidance and also for his limitless enthusiasm and optimism for any research problem, big or small, easy or hard.  I also owe a great deal of thanks to Ramon van Handel.  Although we have never met, he laid the tracks which my research followed and this thesis certainly reflects that influence.  I firmly believe the clarity and insight which he has provided the fields of quantum filtering and control has and will continue to profoundly shape the form and progress of research within those disciplines.  I also thank Carl Caves and Ivan Deutsch for bringing together a strong group of quantum information researchers, fostering our growth and instilling a strong sense of community during our weekly group meetings.  I particularly appreciate Carl's ability to explain why the needle is important for understanding the haystack and Ivan's ability to explain why the haystack is important for understanding the needle.
   
   Within our research group, I heartily thank Rob Cook for allowing me to pester him with questions, for carefully questioning my assumptions and for his Bob Dylan voice; Ben Baragiola for his blanket skepticism, for Digging, and for his vocal percussion skills on coffee days; Heather Partner for willing to go to the board, for notifying me of product recalls and for her mastery of binding; Tom Jones for his love of the game, for convincing me that Mathematica is very useful for some things and for owning headphones; Brigette Black for her taste in television comedies and for her wordsmithery; Thomas Loyd for his tablet powers and for his exquisite pronunciation of Stieltjes; and everyone for Otter Pops.  I thank the residents of Room 30---Anil Shaji, Collin Trail, Steve Flammia, Animesh Datta, Seth Merkel, Carlos Riofrio, Jonas Anderson, Brian Mischuck, Iris Reichenbach, Sergio Boixo, Alexandre Tacla, Aaron Denney, Vaibhav Madhok, Leigh Norris, and Pat Rice---for discussing research ideas or simply listening to me prattle on about quantum stochastic differential equations.  I also sincerely appreciate all of my colleagues' thoroughness and openness in presenting their own research during the information group meetings.  I have learned a great deal from you all.
   
   In their roles as teachers and mentors, I thank Sudhakar Prasad, Ivan Deutsch, Carl Caves,  Krzysztof W\'{o}dkiewicz, and David Dunlap from UNM; Han Pu, Randy Hulet, Paul Stevenson and B. Paul Padley from Rice University; and Bijli Myers and Chap Percival from Pine View School.  
   
   I thank Larry Herskowitz for his pull-up jumper, Anthony Salvagno for his ability to finish at the rim and both of them for enjoying the pounder.
   
   I thank my sisters for leading by example, particularly in their academic and professional excellence.  I thank my parents for instilling in me a strong work ethic and an appreciation of scholarship.
   
   To Denise, thank you for your strength and support, for your encouragement and optimism, and for being the light at the end of the tunnel.  
\end{acknowledgments}

\maketitleabstract
\begin{abstract}
   This thesis explores the topics of parameter estimation and model reduction in the context of quantum filtering.  The last is a mathematically rigorous formulation of continuous quantum measurement, in which a stream of auxiliary quantum systems is used to infer the state of a target quantum system.  Fundamental quantum uncertainties appear as noise which corrupts the probe observations and therefore must be filtered in order to extract information about the target system.  This is analogous to the classical filtering problem in which techniques of inference are used to process noisy observations of a system in order to estimate its state.  Given the clear similarities between the two filtering problems, I devote the beginning of this thesis to a review of classical and quantum probability theory, stochastic calculus and filtering.  This allows for a mathematically rigorous and technically adroit presentation of the quantum filtering problem and solution.  

   Given this foundation, I next consider the related problem of quantum parameter estimation, in which one seeks to infer the strength of a parameter that drives the evolution of a probe quantum system.  By embedding this problem in the state estimation problem solved by the quantum filter, I present the optimal Bayesian estimator for a parameter when given continuous measurements of the probe system to which it couples.  For cases when the probe takes on a finite number of values, I review a set of sufficient conditions for asymptotic convergence of the estimator.  For a continuous-valued parameter, I present a computational method called quantum particle filtering for practical estimation of the parameter.  Using these methods, I then study the particular problem of atomic magnetometry and review an experimental method for potentially reducing the uncertainty in the estimate of the magnetic field beyond the standard quantum limit.  The technique involves double-passing a probe laser field through the atomic system, giving rise to effective non-linearities which enhance the effect of Larmor precession allowing for improved magnetic field estimation.
   
   I then turn to the topic of model reduction, which is the search for a reduced computational model of a dynamical system.  This is a particularly important task for quantum mechanical systems, whose state grows exponentially in the number of subsystems.  In the quantum filtering setting, I study the use of model reduction in developing a feedback controller for continuous-time quantum error correction.  By studying the propagation of errors in a noisy quantum memory, I present a computation model which scales polynomially, rather than exponentially, in the number of physical qubits of the system.  Although inexact, a feedback controller using this model performs almost indistinguishably from one using the full model.  I finally review an exact but polynomial model of collective qubit systems undergoing arbitrary symmetric dynamics which allows for the efficient simulation of spontaneous-emission and related open quantum system phenomenon.  
\clearpage 
\end{abstract}
\tableofcontents
\listoffigures


\mainmatter

\chapter{Introduction}
A striking feature of quantum mechanics is its inherent uncertainty.  Even when given a complete description of a system, quantum mechanics generally prescribes probabilities for measurement outcomes when a corresponding classical theory prescribes certainties.  Given that quantum mechanics is a fundamental theory, one might suspect that quantum uncertainty significantly restricts our ability to accomplish physical tasks.  Yet rather remarkably, quantum information theory shows that this is not always the case.  In fact, there are many tasks for which a quantum system significantly outperforms its classical counterpart, most notably quantum algorithms for factoring \citep{Shor:1994a} and searching \citep{Grover:1996a}, quantum protocols for communication \citep{Bennett:1993a,Bennett:1984a} and quantum techniques for precision measurement \citep{Xiao:1987a}.

Yet the need to cope with uncertainty is not unique to quantum systems.  Indeed noise is nearly ubiquitous in any real world situation, when it is impractical or impossible to exactly describe the physics of the environment surrounding the system of interest or even the details of the system itself.  Such uncertainty gives rise to a stochastic, rather than deterministic, description of a system and of the corresponding measurement process.  Again, it is perhaps startling that in the face of uncertainty, one can still perform tasks remarkably well, although we experience such performances whenever we fly on an airplane, turn on a computer or purchase the correct birthday gift for a loved one.

Over the past century, the fields of stochastic control and estimation theory have made great strides in formalizing techniques for overcoming the presence of noise.  One such technique is \emph{filtering}, which is a method for estimating the state of a stochastic system by appropriately processing noisy measurements of that system \citep{Liptser:1977}.  Another technique is \emph{feedback}, in which one seeks to control a stochastic system to achieve a particular goal \citep{Zhou:1996a}.  Not surprisingly, the two are intimately related, in that deciding a feedback policy often first requires filtering the noisy measurements to determine exactly what the system is doing under all that noise.  In building a mathematical apparatus for handling noise, these theories are broadly applicable across a variety of engineering and scientific disciplines.  As our technological capability to manipulate and measure distinctly quantum systems matures along with the host of quantum information processing tasks we seek to perform, it is clear that control and estimation techniques will play an important role in the quantum engineering realm.  

Certainly, the goals of quantum control and estimation are no different than those for classical systems; primarily, the capability to build robust and stable systems which accomplish a desired task.  However, the engineering difficulties are more fundamental in the quantum case---one must isolate a quantum system from its environment in order to preserve quantum coherence and manipulate intrinsic quantum uncertainties, yet the isolation cannot be so severe as to preclude useful external interactions required for controlling and measuring the system.  Dealing with these inimical demands at an abstract level is well appreciated in quantum information theory, particularly in the areas of quantum error-correction \citep{Gottesman:1997a} and quantum fault-tolerance \citep{Aharonov:1996a}.  Less abstractly, a plethora of robust methods have been developed for spin control and nuclear magnetic resonance applications \citep{Vandersypen:2004a}, where both fundamental quantum uncertainty and technical noise play important roles.  Some of the methods of classical control and estimation theory appear implicitly in both of these quantum engineering approaches, but given the success and relative maturity of the classical methods, it seems prudent to make the analogy more explicit, mining the vast library of known classical techniques which can then be suitably modified to reflect the constraints imposed by the laws of quantum mechanics.  Even more simply, putting quantum control and estimation theories in the language of their classical progenitors provides an elegant and technically convenient way to decompose and study a quantum engineering problem.

Such a reformulation is exceptionally useful in the domain of quantum optics, where statistical properties of laser light map rather directly onto the classical stochastic formalism.  Belavkin was one of the first to flesh out this mapping, developing a rigorous theory of quantum filtering and control \citep{Belavkin:1979a,Belavkin:1987a,Belavkin:1999a} in terms of the axiomatic probability theory and optimal control formalism used when dealing with classical stochastic systems, suitably adapted to the quantum domain by \citet{Hudson:1984a}.  As experimental prowess and potential applications grew, Belavkin's filtering techniques were independently discovered in a more heuristic approach  called quantum trajectory theory \citep{Carmichael:1991a}.  Initially, quantum trajectories were seen as a computational tool for simulating the dynamics of open quantum systems, averaging over many stochastic quantum jump evolutions to simulate quantum master equation dynamics.  This soon evolved into a theory of continuous quantum measurement and feedback \citep{Wiseman:1994a}, which in conjunction with a renaissance of the earlier filtering work \citep{Bouten:2007b} and a closing gap between theoretical possibilities and experimental realities \citep{Mabuchi:1999a}, suggest quantum control theory and quantum optics are a useful pair for exploring quantum control applications, including precision metrology \citep{Armen:2002a,Geremia:2003a} and quantum error correction \citep{Ahn:2002a}.

It is within this propitious environment that my own research in quantum filtering and control has developed, largely along two main threads\footnote{A third research thread that doesn't fit within the quantum filtering and control umbrella is work I did with Andrew Landahl on the computational universality of quantum walks in one spatial dimension \citep{Chase:2008b}.}---quantum parameter estimation and quantum model reduction.  The former is essentially a filtering problem, in which one seeks to estimate a parameter that modulates the dynamical evolution of a probe quantum system.  With knowledge of the dynamics, suitable measurements of the probe system can be used to determine the parameter of interest.  However, the inherent quantum fluctuations in the probe measurement appear as noise which corrupts the probe signal, requiring statistical inference or filtering to best estimate the parameter.  In Chapter \ref{chapter:quantum_parameter_estimation}, I review my work on developing a general filter for quantum parameter estimation via continuous quantum measurement \citep{Chase:2009a}.  By embedding the parameter estimation problem in the state estimation problem of quantum filtering, I develop the optimal Bayesian parameter filter and discuss conditions for its convergence to the true parameter value.  I also discuss an approximate computational method called quantum particle filtering suitable for practical quantum parameter estimation.  In Chapter \ref{chapter:magnetometry}, I review the application of these techniques for a proposed experimental demonstration of precision magnetometry \citep{Chase:2009b,Chase:2009c,Chase:2009d}.  By double-passing an optical field through an atomic system, one hopes to create effective nonlinear interactions which offer improved sensitivity to the strength of an external magnetic field.  Determining the magnetic field strength from measurements of the scattered optical field is precisely the filtering problem discussed above.  Although a careful derivation of approximate quantum Kalman filters using the method of projection filtering shows no improvement, numerical simulations of the exact dynamics and quantum particle filters suggest an improvement does exist.  By studying this example, I hope to demonstrate that the quantum filtering formalism provides an elegant framework for studying parameter estimation problems.  

The second topic of model reduction deals with developing a computationally reduced description of a quantum dynamical system, whose most general description grows exponentially with the number of subsystems involved.  In practice, one is oftentimes only interested in the dynamics of a restricted set of observables or a restricted set of initial states, both of which may not require calculating the exact dynamics.  I review such a case in Chapter \ref{chapter:error_correction}, presenting a reduced model of error propagation in a continuously measured quantum memory subject to noise \citep{Chase:2008a}.  The model is then used by a classical feedback controller to perform continuous error correction, with almost no change in performance relative to an exact model.  The reduced description scales only polynomially in the number of physical qubits, an improvement over the exponential scaling of the exact model.  Similar reductions will generally be useful for feedback controllers and filters which are usually processed on a classical computer.  In Chapter \ref{chapter:collective}, I present similar but unrelated model reduction research, describing a polynomial but exact model of collective dynamical processes of ensembles of qubits \citep{Chase:2008c}.  This allows for efficient numerical simulation of a broad range of collective qubit systems, particularly those involving spontaneous emission.  

But before delving into my own research, I begin in Chapters \ref{chapter:classical} and \ref{chapter:quantum} by reviewing some essential elements of classical and quantum probability theory, stochastic calculus and filtering.  The goal is to provide a ``user's guide'' to the existing body of mathematical physics literature, occasionally delving into the mathematical details, but focusing more on the tools needed for quantum control and filtering problems.  There are several reasons for such an exposition.  Firstly, it has been my experience that these methods are underappreciated in the quantum optics community, perhaps out of apprehension towards the mathematical rigor and language involved when more familiar quantum information approaches seem to suffice.  However, climbing the seemingly steep initial learning curve quickly provides rewards in the form of an elegant and oftentimes superior approach for studying quantum continuous measurement and control problems.  Secondly, there are technical reasons for preferring the rigorous results, especially due to mathematical issues inherent with continuous stochastic processes which include singular white noise terms.  I believe these issues can be appreciated without fully detailing the mathematical technicalities involved, which I certainly do not claim to master.  Lastly and perhaps most importantly, I earnestly believe that taking the rigorous approach is the key to opening the vast library of existing classical control and estimation tools which will allow for significant and rapid progress in the field of quantum engineering.

\chapter{Classical Probability and Filtering}
\label{chapter:classical}

Many scientists are familiar with the basic elements of probability theory---distributions, expectations, random variables---and are quite comfortable performing calculations using these elements.  Given a fair, six-sided die whose faces are labeled 1-6, we are comfortable stating that the probability of rolling any particular face is $\frac{1}{6}$.  The probability of rolling a face with an even number is also easily calculated as
\begin{equation}
    P(\text{roll is even}) = \sum_{i \in \{2,4,6\}} P(\text{roll face } i) = 3\times \frac{1}{6} = \frac{1}{2} .
\end{equation}
This suggests the general rule ``The probability of obtaining some set of mutually exclusive outcomes is the sum of the probabilities of each of the outcomes'', or mathematically, for a collection of $n$ disjoint sets $A_1,\ldots,A_n$, the rule is
\begin{equation} \label{chap:probability:eq:additivity_basic}
    P(A_1 \cup \ldots \cup A_n) = \sum_i P(A_i) .
\end{equation}

We are also familiar with the uniform distribution on the interval $[0,1]$.  Given a random variable $X$ with such a distribution, we know that $P(0 \leq X \leq 1) = 1$, i.e. that the random variable will take on some value on that interval.  More generally, for $0 \leq a,b \leq 1$, we have 
\begin{equation}
    P(a \leq X \leq b) = b - a ,
\end{equation}
which also correctly calculates the probability of a point, i.e.~$P(X=a) = P(a \leq X \leq a) = 0$.  Given our understanding of intervals, we then have
\begin{equation}
    P( \frac{1}{4} \leq X \leq \frac{3}{4}) = P( \frac{1}{4} \leq X \leq \frac{1}{2})
                             + P( \frac{1}{2} < X \leq \frac{3}{4}) = \frac{1}{2}
\end{equation}
which is reminiscent of our general rule in Eq.~\eqref{chap:probability:eq:additivity_basic}.  We might expect this rule to extend to an uncountably infinite number of disjoint sets, which for the entire uniform distribution implies
\begin{equation}
	\begin{split}
	    P(0 \leq X \leq 1) &= \sum_{x \in [0,1]} P(X = x)\\
	    1 &= 0.
	\end{split}
\end{equation}

Clearly, something just went wrong with trying to extend our rule to an uncountably infinite number of disjoint sets.  This turns out to be the case in many situations, where simply applying the intuitive discrete probability rules in the continuous case gives ridiculous answers.  Oftentimes, its not even clear how to formulate questions using our intuitive rules, such as for the uniformly distributed variable $X$ above, what is the probability that it takes on a rational value?  

Given our ultimate interest in describing continuous random variables, especially uncountably infinite collections of random variables indexed by the continuous label time, it is important that we use a probability theory that deals with these complications carefully.  Indeed, the filtering problem is to consider the system$(X_t)$/ observations $(Y_t)$ pair
\begin{align}
    \frac{dX_t}{dt} &= f(t,X_t) + g(t,X_t) \times \text{``noise''}\\
    \frac{dY_t}{dt} &= f(t,X_t) + g(t,X_t) \times \text{``noise''}\\    
\end{align} 
and perform inference about the state of the system based on the measurements.  Since the ``noise'' terms are stochastic, both the system and observations are precisely the uncountable collections of random variables we need to consider.  In fact, much care will be taken to define the ``noise'' terms in a mathematically sensible manner so that the filtering problem can be posed in a sensible fashion.  

All of these details require a carefully laid mathematical foundation in terms of axiomatic probability theory, formalized by \citet{Kolmogorov:1956a}, which unifies features of discrete and continuous probability in terms of measure theory.  It also allows us to generalize probabilities to other spaces, including functional spaces which will be needed to describe stochastic processes.  The first section in this chapter will overview some of the important properties of axiomatic probability theory, followed by a section on stochastic processes and white noise and a closing section devoted to solving the filtering problem.  The presentation of topics in this section primarily follows \citep{vanHandel:2007a}, with added insight from \citep{Oksendal:1946a,Williams:1991a,Geremia:2008a}.
\section{Classical Probability Theory}
The basic ingredient of probability theory is the \emph{sample space} $\Omega$ which describes the set of all possible \emph{outcomes} in the probabilistic system under consideration.  In the die example above, this would simply be $\Omega = \{1,2,3,4,5,6\}$, where the individual outcomes $\omega \in \Omega$ label the different faces of the die.  While we could ask questions about individual outcomes, we are really more interested in related objects called \emph{events}, which are the yes or no questions one could ask about the system.  Such events are represented by subsets $A \subset \Omega$ where the elements $\omega \in A$ are those corresponding to a yes answer of the related question.  For our example, the question ``Did I roll an even number?'' is represented by the subset $\{2,4,6\} \subset \Omega$ and the basic question ``Did I roll a 2?'' is represented by the subset $\{2\} \subset \Omega$.  The collection of such subsets, corresponding to the collection of relevant yes/no questions, is itself put into a set $\mathcal{F}$ which is called the $\sigma$-algebra over $\Omega$.
\begin{definition}\label{def:sigma_algebra}
A $\sigma$-\emph{algebra} $\mathcal{F}$ over $\Omega$ is a collection of subsets of $\Omega$ which satisfies
\begin{enumerate}
    \item $\Omega \in \mathcal{F}$
    \item If the set $A \in \mathcal{F}$, then the complement $A^c \in \mathcal{F}$
    \item Countable $\bigcup_n A_n \in \mathcal{F}$ if each $A_n \in \mathcal{F}$
\end{enumerate}
\end{definition}
The first two requirements are not terribly surprising.  Certainly the question, ``Did anything happen?''~must be valid.  Indeed, the most trivial $\sigma$-algebra valid for any $\Omega$ is $\mathcal{F}=\{\varnothing,\Omega\}$.  Similarly, if a particular binary question is acceptable, $A \in \mathcal{F}$, we should tautologically be able to ask whether ``not'' of that question occurred, e.g. ``Did I not roll a 2?''.  This implies the complement $A^c \in \mathcal{F}$.  The remaining and more technical requirement relates to our general rule from the beginning of the chapter.  The intuitive idea is that for two events $A,B \in \mathcal{F}$, we should be able to combine them to make the question ``Did $A$ or $B$ happen?'' ($A \cup B$) or the question ``Did $A$ and $B$ happen?'' ($A \cap B$).  The restriction to countable ``or'' compositions\footnote{Note that the composition of ``and'' questions comes from having the complement of sets in the $\sigma$-algebra.} prevents the pathological case we had above for elements on the real line and by taking it as an axiom, we can entirely avoid it.

The pair $\{\Omega,\mathcal{F}\}$ is a mathematical object called a \emph{measurable space} and elements in $\mathcal{F}$ are called \emph{measurable sets}.  Such an object is defined precisely to sidestep the issues with uncountably infinite compositions.  From the name, we anticipate that a measurable space is something we can define a measure on, which is just a convenient way to talk about sizes of collections of elements in $\Omega$.  For a probability theory, we will want a specific measure $\mathbb{P}$ which assigns probabilities to events in sensible way.  But the trick is to define the measure on sets in $\mathcal{F}$ and not directly on elements in $\Omega$, thereby only defining the measure on sets which are sensible without having to worry how those sets are composed from elements in $\Omega$.  In other words, we need not worry about decomposing an event which should have a non-zero probability, e.g an interval, in terms of the its uncountably infinite constituents, which have zero probability, e.g. points.  This is encapsulated in the following definition.

\begin{definition}\label{def:probability_measure}
A \emph{probability measure} is a map $\mathbb{P} : \mathcal{F} \mapsto [0,1]$ which satisfies
\begin{enumerate}
    \item For a countable collection $\{A_n : A_n \in \mathcal{F}, A_n \cap A_m = \varnothing \text{ for } n\neq m \}$, $\mathbb{P}(\bigcup_n A_n) = \sum_n \mathbb{P}(A_n)$
    \item $\mathbb{P}(\varnothing) = 0$, $\mathbb{P}(\Omega) = 1$
\end{enumerate}
\end{definition}  
The first part of the definition is precisely our general rule, but \emph{restricted} to countable collections.  The second part is just to set the baseline meanings which we expect for any probability theory; that the probability of nothing happening is zero and the probability of anything happening is one.  

The tuple $(\Omega,\mathcal{F},\mathbb{P})$ is called a \emph{probability space} and formalizes the intuitive rules we desire such that they apply for both discrete and continuous spaces.  In essence, the measure $\mathbb{P}$ is the workhorse, in that in encapsulates every probabilistic statement we make regarding the theory.  As such, $\mathbb{P}$ is often referred to as the \emph{state} of a random system and the probabilities it assigns to events are based on a physical model, counting, betting odds or whatever perspective lets you sleep at night.  

As a final introductory note, one might read that events $A$ for which $\mathbb{P}(A) = 1$ are said to occur ``almost surely'', abbreviated a.s.  This statement reflects the fact that sets of measure zero may contribute to an event, even though they individually have zero probability.

\subsection{Generated $\sigma$-algebras and the Borel $\sigma$-algebra}
For discrete spaces, the power set of $\Omega$ is an obvious choice for the $\sigma$-algebra, but it turns out (again) to be more complicated for continuous spaces\footnote{For technical reasons beyond me, it turns out one can actually have too many sets in $\mathcal{F}$ on which one can define a consistent $\mathbb{P}$.  \citet{Banach:1929a} actually showed that no probability measure exists on the power set of $\mathbb{R}$ such that the probability of any single point is zero.}.  For such spaces (and for later purposes), it is convenient to have a method for \emph{generating} a valid $\mathcal{F}$ from a collection of events we know we are interested in.  Consider a potentially uncountable collection of subsets $\mathcal{F}_0 = \{A_i \in \Omega\}$ which is not necessarily a $\sigma$-algebra.  In order to generate a $\sigma$-algebra from $\mathcal{F}_0$, we consider all $\sigma$-algebras which have $\mathcal{F}_0$ as a sub-collection.  Taking the intersection of these $\sigma$-algebras also results in a $\sigma$-algebra and is the smallest $\sigma$-algebra which contains all elements in $\mathcal{F}_0$.  The result of this operation, written $\mathcal{F} = \sigma\{A_i\}$ is called the $\sigma$-algebra \emph{generated} by $\mathcal{F}_0$.  

\begin{example}[Example 1.1.8 in \citet{vanHandel:2007a}]
	As a concrete example, consider the six-sided die for which we generate a $\sigma$-algebra from the questions ``Did we throw a one?'' and ``Did we throw a four?''
	\begin{equation}
	    \sigma\{\{1\},\{4\}\} = \{\varnothing,\{1\},\{4\},\{1\}^c,\{4\}^c, \{1,4\},\{1,4\}^c,\Omega\} .
	\end{equation}
	We see that a consistent $\sigma$-algebra implies that answering the two basic questions also allows us to answer questions such as ``Did we throw a one or a four?'' and ``Did we not throw a one?''.  Really, the generated $\sigma$-algebra reflects all the yes/no questions we can logically answer from observing its input set of events, which here is knowledge of rolling a one or a four.
\end{example}

An important $\sigma$-algebra for continuous spaces is the Borel $\sigma$-algebra (on the reals), defined as
\begin{definition}\label{def:borel}
    The \emph{Borel $\sigma$-algebra} (on $\mathbb{R}$), written $\mathcal{B}$, is the $\sigma$-algebra generated from the set of all open intervals on $\mathbb{R}$. Note that this is a generated set, since the complement of open intervals is a closed interval, which is clearly not contained within the set of open intervals.
\end{definition}
\subsection{Random Variables}
Although a probability space is all we need to start discussing a probabilistic system, we are ultimately interested in more glamorous inquiries than simple yes/no questions.  As physicists, we are particularly interested in describing observations or measurements we might make of the system, at which point we need to relate the labels on the measuring device to properties of the system.  In order to make this mapping precise, we first introduce the following definitions.

\begin{definition}\label{def:measurable_function}
    Let $(\Omega,\mathcal{F})$ and $(S,\mathcal{S})$ be measurable spaces.  The function $X(\omega):\Omega \mapsto S$ is an \emph{($\mathcal{F}$-)measurable function} if $X^{-1}(S) \equiv \{ \omega \in \Omega : X(\omega) \in S\} \in \mathcal{F}$ for every $S \in \mathcal{S}$. 
\end{definition}
\begin{definition}\label{def:random_variable}
An \emph{(S-valued) random variable} is an ($\mathcal{F}$)-measurable function $X(\omega):\Omega \mapsto S$ from the probability space $(\Omega,\mathcal{F},\mathbb{P})$ to the measurable space $(S,\mathcal{S})$. We will often consider real-valued random variables, which map elements in the sample space to $(\mathbb{R},\mathcal{B})$ and which we will call simply \emph{random variables}.
\end{definition}

The notion of measurability is what really allows us to define probabilities on random variables.  In fact, if the random variable is $\mathcal{F}$-measurable, that means that all yes/no questions needed to determine its value are contained within $\mathcal{F}$, so that we need only invert the map $X$ to determine the associated probability.  That is, the probability for a random variable $X$ to take on some value $A \in \mathcal{B}$ is written
\begin{equation}
    \mathbb{P}(X \in A) = \mathbb{P}(X^{-1}(A)) = \mathbb{P}(\{\omega \in \Omega : X(\omega) \in A \}),
\end{equation} 
where the first two forms are shorthand for the explicit form on the right.  But it is conceivable that $\mathcal{F}$ contains more yes/no questions than are actually needed for a particular random variable $X$.  As such, we can consider the $\sigma$-algebra generated by a random variable
\begin{equation}
    \mathcal{F}_X = \sigma\{X\} = \{ X^{-1}(A) : A \in \mathcal{B} \}
\end{equation}
This is actually a convenient way to generate a $\sigma$-algebra for $\Omega$ when we have a collection of random variables we are interested in; simply take the smallest $\sigma$-algebra which contains all those generated by each random variable in the set. 

Abstractly, $\mathcal{F}_X$ encodes the information that we learn by measuring $X$.  Such a notion will be important when we consider conditioning and inference, when it will be useful to relate $\sigma$-algebras generated from different random variables.
\begin{definition}\label{defn:random_var_measurable}
    For two random variables $X$ and $Y$ defined on the probability space $(\Omega, \mathcal{F}, \mathbb{P})$, we say that $Y$ is \emph{$\mathcal{F}_X$ measurable} (or simply \emph{$X$-measurable}) if $\mathcal{F}_Y \subseteq \mathcal{F}_X$ or equivalently, there exists a measurable function $\phi : \mathbb{R} \mapsto \mathbb{R}$ such that $Y = \phi(X)$. 
\end{definition}

\begin{example}
Consider the probability space for throwing two coins, given by $\Omega = \{HH,TT,HT,TH\}$ with $\mathcal{F}$ and $\mathbb{P}$ defined but unimportant for this example.  Further, define a boolean random variable $X$ by 
\begin{equation}
    X(HH) = X(TT) = 1 \qquad X(HT) = X(TH) = 0,
\end{equation}
which is the parity of the two tosses. It is straightforward to see that $\mathcal{F}_X = \{\varnothing, \{HH,TT\},\{HT,TH\},\Omega\}$.  Also consider the random variable $Y$ defined by
\begin{equation}
    Y(HH) = Y(HT) = 1 \qquad Y(TH) = Y(TT) = 0,
\end{equation}  
which corresponds to the outcome of the first toss and which has a generated $\sigma$-algebra $\mathcal{F}_Y = \{\varnothing,\{HH,HT\} ,\{TH,TT\},\Omega\}$.  We immediately see that $Y$ is not $X$ measurable as well as the opposite, though measurability need not be symmetric.  This is completely sensible, as learning the outcome of the first toss is not enough information to determine the parity of the two tosses together.  
\end{example}

\noindent Related to measurability is the notion of \emph{independence}:
\begin{definition}\label{def:independence}
Two random variables $X,Y$ defined on a probability space $(\Omega,\mathcal{F},\mathbb{P})$ are \emph{independent} if
$    \mathbb{P}(A \cap B) = \mathbb{P}(A)\mathbb{P}(B) \text{ for all } 
             A \in \mathcal{F}_X, B \in \mathcal{F}_Y$.
\end{definition}
In contrast to measurability, in which one variable can be determined exactly by knowing the value of the other, independent variables share no information.  That is, knowing the value of $X$ tells you \emph{absoutely nothing} about the value of $Y$.  Note that independence is a property of the probability measure $\mathbb{P}$, whereas measurability only depends on the structure of the $\sigma$-algebras generated by the random variables.  Additionally, just because a random variable is not measurable with respect to another, the two are not necessarily independent.  This is generally the case we will be interested in for filtering, when we learn partial information about related random variables when given the value of a particular one.
 
Note that every random variable induces a probability measure on the reals given by
\begin{equation}
    \mu_X(B) = \mathbb{P}(X^{-1}(B)), \quad B \in \mathcal{B}
\end{equation}
We call $\mu_X$ the distribution of the random variable $X$.  A particularly important and familiar random variable is a \emph{Gaussian random variable} $X : \Omega \mapsto \mathbb{R}$ with mean $\mu$ and variance $\sigma^2$ has the distribution
    \begin{equation}
        \mu_X( B ) = \int_B\frac{1}{\sigma\sqrt{2\pi}}\exp(-\frac{(x-\mu)^2}{2\sigma^2}) dx .
    \end{equation} 

\begin{definition}\label{def:indicator_function}
	A very useful random variable is the \emph{indicator function} $\chi_A : \Omega \mapsto [0,1]$, defined for $A \in \mathcal{F}$ to be
	\begin{equation}
	    \chi_A(\omega) = \begin{cases}
	        1, & \omega \in A\\
	        0, & \omega \not \in A
	    \end{cases} .
	\end{equation}
\end{definition}
\noindent We can use indicator functions to rewrite a general random variable $X$ over the sets $S_i$ on which it is constant, which provide a partition of $\Omega$, e.g. $X(\omega \in S_i) = x_i$ where $\bigcup_i S_i = \Omega$.  We then have
\begin{equation} \label{probability:eq:indicator_partition}
    X(\omega) = \sum_i x_i \chi_{S_i}(\omega) .
\end{equation}
As we shall see in the following section, indicator functions are useful as they allow us to work exclusively with expectations, rather than directly with the probability measure, since for some $A \in \mathcal{F}$, $\mathbb{P}(A) = \mathbb{E}(\chi_A)$.  This will allow us to gloss over conditional probability and focus instead on conditional expectations, which are more relevant for filtering.
\subsection{Expectation}
The notion of expectation is another topic most are familiar with from previous work with probability.  Conceptually, it corresponds to the average value of a random variable one would expect in the limit of repeating many trials of the underlying probability experiment.  For \emph{simple random variables} $X$ which take on a finite number of values $x_1,x_2,\ldots, x_n$, the expectation above reduces to the familiar form
\begin{equation}
    \E{X} = \sum_{i=1} x_i \mathbb{P}(X = x_i)
\end{equation}
where the expectation is well-defined so long as the possible values are finite.  For a continuous-valued random variable $X$, we define it be a nondecreasing sequence $X_n$ which converges to $X$ and set $\E{X} = \lim_{n\to\infty}\E{X_n}$.  One can show \citep{Williams:1991a} that such a procedure uniquely converges to the following definition.
\begin{definition}\label{def:expectation}
    Let $(\Omega,\mathcal{F},\mathbb{P})$ be a probability space with a random variable $X$.  The expectation of $X$ with respect to the measure $\mathbb{P}$ is 
    \begin{equation}
        \E{X} = \int_\Omega X(\omega) \mathbb{P}(d\omega)
    \end{equation}
where the integral is interpreted in the Lesbesgue sense.
\end{definition}
The fact that we extend to the continuous case via the integral above should come as no surprise as that is how we extend sums to the familiar Riemann integral in calculus.  But given that probability theories are defined on more general measurable spaces, we use a more general integral---the Lebesgue integral, which allows us to integrate measurable functions (via $\mathbb{P}$), unlike the Riemann integral which only allows us to integrate continuous functions.

When studying stochastic processes, we will find it very useful to refer to the following classification of random variables in terms of their expectation.
\begin{definition}\label{def:integrable}
    For a random variable $X$ and $p \geq 1$, define $\norm{X}_p = (\mathbb{E}(\abs{X}^p))^{1/p}$.  A random variable is \emph{$p$-integrable} if $\norm{X}_p < \infty$.  For $p = 2$, such a random variable is \emph{square-integrable}.  A random variable satisfying $\abs{X} \leq K$ for some $K \in \mathbb{R}$ is called \emph{bounded} and $\norm{X}_\infty$ is the smallest $K$ which bounds $X$.
\end{definition}

Using this definition, we can introduce the spaces $\mathcal{L}^p(\Omega,\mathcal{F},\mathbb{P}) = \{X : \norm{X}_p < \infty \}$ which are common spaces in functional analysis.  Of particular use is the space $\mathcal{L}^2$ which for $\Omega = \mathbb{R}$ is almost\footnote{$\norm{\cdot}_2$ is not quite a norm because $\norm{X}_2 = 0$ only implies that $X = 0$ under the measure $\mathbb{P}$, not that the function is identically $X(\omega) = 0$ for all $\omega$.} the familiar space of square-integrable functions.  As such, we will often make use of the implied inner product
\begin{equation} \label{probability:eq:inner_product}
    \langle X, Y \rangle = \E{XY} = \int_\Omega X(\omega)Y(\omega)\mathbb{P}(d\omega)
\end{equation}
which will allow for an intuitively pleasing interpretation of the conditional expectation as an orthogonal projection.
\subsection{Conditioning}
Given all the above groundwork, we are now ready to tackle the important task of \emph{conditioning}.  As hinted at above, we will focus on conditional expectation, since the rules of conditional probability are readily recovered using indicator functions.  To get a feel for things, and to appreciate the need for the more technical machinery to come, we begin with a straightforward definition for discrete spaces.
\begin{definition}\label{def:discrete_conditional_expectation}
For a probability space $(\Omega,\mathcal{F},\mathbb{P})$, consider the discrete random variables $X$ and $Y$. Suppose $Y$ yields a finite partition of $\Omega$ (as in Eq.~\eqref{probability:eq:indicator_partition}) in terms of sets $A_k$ for $k = 1, \ldots, n$.  Then the conditional expectation of $X$ given $Y$ is
\begin{equation} \label{probability:eq:discrete_conditional}
    \E{X | Y} = \sum_{k = 1}^n \frac{\E{X\chi_{A_k}}}{\mathbb{P}(A_k)} \chi_{A_k}
\end{equation}
where $\frac{\E{X\chi_{A_k}}}{\mathbb{P}(A_k)}$ is arbitrary if $\mathbb{P}(A_k) = 0$. 
\end{definition}

How do we interpret this definition?  Firstly, we see that the conditional expectation is simply \emph{another (discrete) random variable}, expanded in terms of the indicator functions $\chi_{A_k}$ or written in ``the basis'' of $Y$.  Also, note that the actual values $Y$ takes on are irrelevant; we are only interested in them so far as they allow us to identify the different sets $A_k$.  The term $\frac{\E{X\chi_{A_k}}}{\mathbb{P}(A_k)}$ averages $X$ only over the events which correspond to $A_k$, dividing by $\mathbb{P}(A_k)$ to renormalize for this subset.  Note that there is an arbitrariness when $\mathbb{P}(A_k) = 0$, since that event does not happen (a.s.).  The averaging is done for each partitioning set $A_k$, so that once handed a particular value $y$ of $Y$, the conditional expectation returns the value of $X$ averaged over the appropriate partition for $y$.  As we will soon make precise, $\E{X | Y}$ can be interpreted as the random variable which returns an estimate of $X$ when given the value of $Y$. 

For the usual reasons, this simple definition needs work to be extended to the continuous case.  Suppose $Y$ were actually a real-valued random variable.  It may not generate a finite partition of the continuous sample space $\Omega$, which may have uncountably infinite elements.  More importantly, since $\mathbb{P}(Y = y) = 0$ for any point $y$ on the real line, the arbitrary case above actually turns into a nightmare; if we were to take the partitions $A_k$ to be points, then the entire conditional expectation is arbitrary!  A healthy dose of measure theory shows that one can define the conditional expectation in terms of a sequence of approximating discrete versions (which proves existence and uniqueness), but the technicalities are not particularly enlightening for us (see \citet{vanHandel:2007a}).  But it is important to know what definition ultimately works, so we will instead simply use the following (Kolmogorov) axiomatic definition.

\begin{definition}\label{def:conditional_expectation_kolmogorov}
Let $X$ be a random variable on $(\Omega,\mathcal{F},\mathbb{P})$ and $\mathcal{G}$ be any $\sigma$-algebra on the sample space $\Omega$.  The \emph{conditional expectation} $\E{X | \mathcal{G}}$ is the unique $\mathcal{G}$-measurable random variable which satisfies $\E{\chi_AX} = \E{\chi_A \E{X | \mathcal{G}}}$ for all $A \in \mathcal{G}$.  
\end{definition}
\noindent Rather than conditioning directly on a particular $\sigma$-algebra, we often instead condition on one generated by another random variable, as was done in the discrete case above.  As a short-hand, we write $\E{X | Y}$ to indicate the more precise form of $\E{X | \mathcal{F}_Y}$, where $\mathcal{F}_Y$ is the $\sigma$-algebra generated by the random variable $Y$.  

From the perspective of statistical inference, the following theorem shows that we can interpret the conditional expectation $\E{X|\mathcal{G}}$ as the best estimate of $X$, in a least-squares sense, given the information in $\mathcal{G}$.
\begin{theorem}[Proposition 2.3.3 in \citep{vanHandel:2007a}]\label{thm:conditional_least_squares}
Given $X$ and $\mathcal{G}$ as in Def.~\ref{def:conditional_expectation_kolmogorov}, $\E{X|\mathcal{G}}$ is the unique $\mathcal{G}$-measurable random variable that satisfies 
\begin{equation}
	\E{(X - \mathbb{E}[X | \mathcal{G}})^2] = \min_{Y\in \mathcal{L}^2(\mathcal{G})}\E{(X-Y)^2} ,
\end{equation}
where $\mathcal{L}^2(\mathcal{G}) = \{ Y \in \mathcal{L}^2 : \mathcal{F}_y \in \mathcal{G}\}$.  We therefore call $\E{X | \mathcal{G}}$ the least-mean-square estimate of $X$ given $\mathcal{G}$.
\end{theorem}
\noindent We can actually interpret this statement as the orthogonal projection of $X$ onto the linear subspace $\mathcal{L}^2(\mathcal{H}) \subset \mathcal{L}^2$ with respect to the inner product in Eq.~\eqref{probability:eq:inner_product}.
\begin{proof}
For all $Y \in \mathcal{L}^2(\mathcal{G})$, we can write
\begin{equation}
    \E{(X - Y)^2} = \E{(X - \mathbb{E}[X | \mathcal{G}}
                     + \E{X | \mathcal{G} }  - Y)^2]
\end{equation}
where $\Delta = \E{X | \mathcal{G} }  - Y$ is $\mathcal{G}$-measurable, by definition of the conditional expectation and $Y$.  Rewriting, then
\begin{align}
    \E{(X - Y)^2} &= \E{(X -\mathbb{E}[X | \mathcal{G}} + \Delta)^2]\\
                           &= \E{(X -\mathbb{E}[X | \mathcal{G}})^2]
                                     + 2\E{\Delta(X -\mathbb{E}[X | \mathcal{G}})]
                                     + \E{\Delta^2}
\end{align}
But by the Kolmogorov definition of conditional expectation (Def.~\ref{def:conditional_expectation_kolmogorov}), we have
\begin{equation}
    \E{\Delta\mathbb{E}[X | \mathcal{G}}] = \E{\Delta X}
\end{equation}
so that the middle term above is identically zero, leaving
\begin{equation}
    \E{(X-Y)^2} = \E{(X -\mathbb{E}[X | \mathcal{G}})^2] +  \E{\Delta^2}
\end{equation}
Since $\E{\Delta^2} \geq 0$, the equation is minimized when $\Delta = 0$, which is precisely the least-squares property.  This coincides with the geometric intepretation, since if $\Delta \in \mathcal{L}^2(\mathcal{G})$ and $\E{X | \mathcal{G}}$ is orthogonal projection of $X$ onto $\mathcal{L}^2(\mathcal{G})$, then $ X - \E{X | \mathcal{G}} \perp \mathcal{L}^2{\mathcal{G}}$ and therefore $\langle  X - \E{X | \mathcal{G}} , \Delta \rangle = \E{( X - \mathbb{E}[X | \mathcal{G}} )\Delta] = 0$.

If the conditional expectation were not unique, then there would exist some other $\mathcal{G}$-measurable random variable $Y'$ that also minimizes $\E{(X-Y)^2}$ over all $Y$.  As demonstrated above, this would mean $\E{(X-Y')^2} = \E{(X - \mathbb{E}[X | \mathcal{G}})^2]$.  But we could equally well write
\begin{equation}
    \E{(X-Y')^2} = \E{(X - \mathbb{E}[X | \mathcal{G}})^2] 
                 + \E{(\mathbb{E}[X | \mathcal{G}} - Y')^2]
\end{equation}
where the cross term again disappears due to orthogonality.  If $Y'$ is truly a minimum, we must have $\E{(\mathbb{E}[X | \mathcal{G}} - Y')^2] = 0$ or really $Y' = \E{X | \mathcal{G}}$ (a.s).
\end{proof}

\subsection{Radon-Nikodym Theorem}
Another definition of conditional expectation is in terms of the Radon-Nikodym theorem.  Although the Kolmogorov definition is perfectly adequate for our purposes, studying this alternate definition will introduce concepts that are essential in developing the filtering equations and will be revisited when studying the stability of the quantum parameter estimation filter in Chapter \ref{chapter:quantum_parameter_estimation}.
\begin{definition}\label{def:absolutely_continuous}
Let $(\Omega,\mathcal{F},\mathbb{P})$ be a probability space.  A probability measure $\mathbb{Q}$ is \emph{absolutely continuous} with respect to $\mathbb{P}$, written $\mathbb{Q} \ll \mathbb{P}$ if $\mathbb{Q}(A) = 0$ for all events $A \in \mathcal{F}$ where $\mathbb{P}(A) = 0$.
\end{definition}
\noindent Absolute continuity is an important concept when we are interested in changing probability measures, which is essentially a change of variables technique to allow for easier calculations (much like the change of variables technique in calculus).  The above definition tells us \emph{when} such a change of variables is even possible.  

The basic technique of transformation is as follows.  Let $f(\omega)$ be a nonnegative random variable on $(\Omega,\mathcal{F},\mathbb{P})$ satisfying $\mathbbm{E}[f] = 1$.  For any $A \in \mathcal{F}$, we define the new measure $\mathbb{Q}$ as
\begin{equation}
    \mathbb{Q}(A) = \mathbb{E}_{\mathbb{P}}[\chi_A f] = \int_A f(\omega) \mathbb{P}(d\omega)
\end{equation}
where $\mathbb{Q}$ satisfies the requirements of a probability measure, i.e. $\mathbb{Q}(\varnothing) = 0$, $\mathbb{Q}(\Omega) = 1$ since $\mathbbm{E}[f] = 1$, and the countable disjoint sets decomposition follows directly from the definition of conditional expectation and the measure $\mathbb{P}$.  We can then relate the expectations under either measure for some other random variable $g(\omega)$ as
\begin{equation}
    \mathbb{E}_\mathbb{Q}[g] = \int_\Omega g(\omega) \mathbb{Q}(d\omega)
                 = \int_\Omega g(\omega)f(\omega)\mathbb{P}(d\omega)
                 = \mathbb{E}_\mathbb{P}[gf] .
\end{equation}  
The function $f$ above is called the \emph{density} of the measure $\mathbb{Q}$ with respect to the measure $\mathbb{P}$ and is written $d\mathbb{Q}/d\mathbb{P}$. 

If we think for a little, we immediately see that independent of a choice of $f$, events which have probability measure zero under $\mathbb{P}$ must also have probability measure zero under $\mathbb{Q}$---there is no $f$ such that $f  \mathbb{P}(d\omega)$ can be non-zero if $\mathbb{P}(d\omega) = 0$.  This observation is formalized in the following theorem, for which we omit the proof.
\begin{theorem}[\textbf{Radon-Nikodym}]\label{thm:radon_nikodym}
Consider the measures $\mathbb{P},\mathbb{Q}$ on the measurable space $(\Omega,\mathcal{F})$ such that $\mathbb{Q} \ll \mathbb{P}$, then there exists a unique $\mathcal{F}$-measurable random variable $f$ with $\mathbb{E}_{\mathbb{P}}[f] = 1$ such that $\mathbb{E}_{\mathbb{P}}[\chi_A f] = \mathbb{Q}[A]$ for all $A \in \mathcal{F}$.  We therefore call $f$ the \emph{density or Radon-Nikodym derivative}, $d\mathbb{Q}/d\mathbb{P}$.
\end{theorem}
\noindent Although the theorem simply formalizes our intuition, the important part is that if it exists, the Radon-Nikodym derivative is unique.  In particular, if we were to follow the technical route, and define the conditional expectation as a sequence of finite approximations, we would find that it converges to
\begin{equation}
    \E{X | \mathcal{F}} = \frac{d\mathbb{Q}\rvert_\mathcal{F}}{d\mathbb{P}\rvert_\mathcal{F}}
    , \qquad \mathbb{Q}(A) = \mathbb{P}(\chi_A X)
\end{equation}
where $\mathbb{Q}\rvert_{\mathcal{F}}$ indicates the measure is restricted to the $\sigma$-algebra $\mathcal{F}$.  Since Theorem \ref{thm:radon_nikodym} shows that this derivative is unique, so too is the conditional expectation and we need not worry about the ambiguities leftover from extending the discrete definition.
\subsection{Summary}
Before moving on to stochastic processes, let's highlight what we have learned so far.  Foremost is that dealing with continuous probability spaces is not a trivial extension of the intuitive rules most are familiar with.  Fortunately, by defining probability spaces, random variables and expectations using measure theory, we can overcome most of the technical issues.  As such, the basic definition of a \emph{probability space} is in terms of the measurable space $(\Omega,\mathcal{F})$ and the measure $\mathbb{P}$.  The \emph{$\sigma$-algebra} $\mathcal{F}$ is used to encode the yes/no questions one could ask about the outcomes $\omega$ in the sample space $\Omega$.  \emph{Random variables} are one step up from the $\sigma$-algebra, and provide a mapping of outcomes in $\Omega$ to some other measurable space, on which a measure is induced via $\mathbb{P}$.  Essentially, random variables allow us to work with quantities of interest which are not simple yes/no questions regarding the original sample space $\Omega$.  \emph{Measurability} tells us when one random variable's value is determined entirely by the value of another; \emph{independence} tells us when random variables values are completely unrelated.

From there, we introduced the concept of \emph{expectation}, which is the average value of a random variable expected after repeated sampling from the given probability model.  Expectation induces an ``almost''-inner product on the space of random variables.  This picture provides a nice interpretation of \emph{conditional expectation}, in which we create a new random variable $\E{X | Y}$ which returns the average of $X$ when given the value of $Y$. This is equivalent to a least-squares projection, in terms of the aforementioned inner product, of $X$ onto the space of $\mathcal{F}_Y$-measurable random variables.  These basic ingredients will be important as we move on to consider more complex probability concepts.

\section{Classical Stochastic Processes}
As we steadily move towards discussion of dynamic stochastic systems and the processing of stochastic signals, we will use the following definitions to imbue our previous probability constructions with a notion of \emph{time}.
\begin{definition}\label{def:stochastic_process}
    A \emph{stochastic process} is a map
    \begin{equation}
        X_t(\omega) : \mathbb{R}_{+} \times \omega \mapsto \mathbb{R}
    \end{equation}
    where the argument $t$ is interpreted as time.
\end{definition}
\noindent We see this is nothing more than a family of random variables labeled by the increasing and positive index $t$.  For a given $\omega_i \in \Omega$, the function $X_t(\omega_i)$ traces out a \emph{trajectory} in time.  As time passes, so does our ability to answer yes-no questions about other events in $\mathcal{F}$ and we should be able to partition the $\sigma$-algebra into questions which may or not be answerable given the information we have now.  For the die example, in which the stochastic process is repeated rolls, we can answer the question ``Was each roll a one up to time $t_1$?'' at time $t_1$ and certainly no sooner.  Additionally, once we are able to answer this question, we should be able to do so for eternity.  That is, there is no way we can ``unlearn'' information about events.  Such a  \emph{filtration} of $\mathcal{F}$ is formalized in the following definition. 
\begin{definition}\label{def:filtration}
    The elementary space $(\Omega,\mathcal{F},\mathbb{P})$ admits a \emph{filtration} in terms of an increasing sequence of $\sigma$-algebras, labeled $\mathcal{F}_t \subset \mathcal{F}$ where $\mathcal{F}_s \subset \mathcal{F}_t$ for $s < t$. 
\end{definition}
\noindent  Note that many filtrations exists on a probability space, though we are often interested in one generated by a particular stochastic process, which may be written
    \begin{equation}
        \mathcal{F}_t^{X} = \sigma\{X_s^{-1}(\mathcal{B}), s \leq t \}
    \end{equation}

Given this extra structure for the $\sigma$-algebra, measurability may also be defined relative to the passage of time.
\begin{definition}\label{def:adapted}
    Consider the probability space $(\Omega,\mathcal{F},\mathbb{P})$ with filtration $\mathcal{F}_t$.  The stochastic process $X_t$ is called $\mathcal{F}_t$-adapted if $X_s$ is $\mathcal{F}_t$-measurable for every $s \leq t$.
\end{definition}
\noindent Adapted processes encompass most of the stochastic processes that we will consider in filtering theory.  Intuitively, these processes are ones that do not look into the future, in that at time $t$, the values of the entire stochastic history up to that time, $\{X_{s \leq t}\}$, are completely determined by the yes-no questions answerable at time $t$, represented by $\mathcal{F}_t$.

An important class of stochastic processes are those whose future values are best estimated by its current value.  Such processes are called \emph{martingales}.
\begin{definition}\label{def:martingale}
    A stochastic process $X_t$ is an $\mathcal{F}_t$-martingale if it is $\mathcal{F}_t$-adapted, has bounded expectation ($\E{\abs{X_t}} < \infty$ for all $t$) and satisfies $\E{X_t | \mathcal{F}_s} = X_s$ for all $s \leq t$. 
\end{definition}
\noindent Martingales are perhaps best appreciated in terms of their etymological roots in gambling theory.  If we let the stochastic process $X_s$ represent our winnings at time $s$, then $\E{X_t | \mathcal{F}_s}$ represents our expected future winnings at time $t$, given our knowledge of events up to time $s \leq t$.  If the game is fair, which is in our best interests, but still worth playing, which is certainly in the best interests of the casino, this expectation should be $X_s$.  That is, on average, we expect to come out even when playing the game.  It turns out that this simple property has far reaching implications and is a powerful tool for proving other properties of stochastic processes.

\subsection{White Noise and the Wiener Process}
For a given stochastic process $X_t$, it will be desirable to formulate an equation of motion which describes its time-evolution, in which a noise term traces out an individual trajectory or realization appropriate for a given probability measure.  As an equation, we thus desire
\begin{equation}\label{probability:eq:noisy_process}
    \frac{dX_t}{dt} = a(t,X_t) + b(t,X_t) \times \text{``noise''}
\end{equation}
where I have intentionally been imprecise in representing the noise term.  As we shall see, the noise term is mathematically difficult to handle in general and even in the particular case when noise is \emph{white}\footnote{Note that the ``color'' of the noise has to do with the correlation properties of a stochastic process; it says nothing about the \emph{distribution} of the noise itself.  For the most part, we will consider Gaussian white noise processes, which are delta-correlated in time with Gaussian distributed increments.}---delta-correlated in time with a flat power spectrum.  White noise is common in engineering and physics due its simple properties and fairly broad applications, from modeling random walks to financial derivative prices.  Before formally developing a sensible equation of motion for processes driven by white noise, lets first consider an example which highlights the difficulties involved in simply formulating white noise as a stochastic process.
\begin{example}[From Introduction in \citet{vanHandel:2007a}]
    Consider a discrete-time, noisy channel, in which at time-step $n$, the message $a_n$ is transmitted, e.g. $x_n = a_n + \xi_n$.  The noise $\xi_n$ can be assumed to be independent and identically distributed (i.i.d.) at different times, as the noisy channel quickly loses traces of its previous state.  Moreover, if $\xi_n$ is really the sum of many independent effects, the central limit theorem suggests that it should be Gaussian distributed.  We therefore can take $\{\xi_n\}$ to be discrete time Gaussian white noise with some mean and variance.
      
    Extending to a continuous-time model, our intuition tells us to replace the discrete label $n$ with the continuous label $t$.  Assuming zero-mean and unit-variance for the noise, we then have $\E{\xi_t} = 0$ and $\E{\xi_s\xi_t} = 0$ if $s \neq t$ and $\E{\xi_t^2} = 1$.  Now suppose we transmit a message $a_0$ as $x_t = a_0 + \xi_t$.  Consider the time averaged process over a small interval $[0,\epsilon]$
    \begin{equation}
        X_\epsilon = \frac{1}{\epsilon}\int_0^{\epsilon} x_t dt
                   = a_0 + \Xi_\epsilon
    \end{equation}
    where $\Xi_\epsilon = \frac{1}{\epsilon}\int_0^{\epsilon}\xi_t dt$.  Clearly $\E{\Xi_\epsilon} = 0$, but more interestingly\footnote{Note that this goes to zero since $t = s$ on a set of measure zero, so the expectation factors and goes to zero.  Of course, one really expects to get a delta function here, but as we soon see, that has a different mathematical meaning than these real-valued random variables.}
    \begin{equation}
        \E{\Xi_\epsilon^2} = \frac{1}{\epsilon^2}\int_0^\epsilon
                \int_0^\epsilon \E{\xi_s\xi_t}ds dt = 0,
    \end{equation}
    Thus to decode the message, we simply time average $x_t$ for an arbitrarily short amount of time.  
    
    This is most likely not the model we envisioned; we expect some effort is needed to recover the corrupted message.  Rather than working directly with $\xi_t$, we could instead focus on the time-averaged process.  Clearly $\Xi_1$ is a zero mean, Gaussian random variable with unit variance.  If we want to retain the independence of noise at different times, this suggests that $\int_0^{1/2}\xi_t dt$ and  $\int_{1/2}^1 \xi_t dt$ are also independent, mean-zero random variables, but now with variance $1/2$.  Generalizing, we can then introduce the \emph{Wiener process}
    \begin{equation}
        W_t = \int_0^t \xi_s ds .
    \end{equation}
We will formalize this slightly in a bit, but the idea is that Gaussian white noise is the time derivative of the Wiener process, $dW_t/dt$.  However, we will find that it is non-differentiable almost everywhere.  Indeed, given that $\E{W_sW_t} = \min{(s,t)}$ due to the independence of different increments, we have
\begin{equation}
    \E{\xi_s\xi_t} = \frac{d}{dt}\frac{d}{ds}\E{W_sW_t}
                           = \frac{d}{dt}\frac{d}{ds}\frac{s+t-\abs{t-s}}{2}
                           = \frac{d}{dt}\frac{1 + \sign(t-s)}{2}
                           = \delta(t-s)
\end{equation}
which is the Dirac delta ``function'', a manifestly non-differentiable object. Moreover, given that $\delta(t)$ is really a distribution and not a true function, we immediately see the difficulties in defining it as a stochastic process.  Working through the details, one would find that our current probability framework does not allow for a stochastic process with the desired properties of white noise.  Yet as we previously mentioned, delta-correlation is the most common definition used for white noise in engineering in physics.  Fortunately, we will be able to use the Wiener process, which does have a rigorous mathematical definition, to formally handle a process like that in Eq.~\ref{probability:eq:noisy_process}. 
\end{example}
  
We loosely want to think of the Wiener process as the $N\to\infty$ limit of the random walk
\begin{equation} \label{classical:eq:wiener_random_walk}
    x_t(N) = \sum_{n=1}^{Nt} \frac{\xi_n}{\sqrt{N}}
\end{equation}  
where $\xi_n$ are the i.i.d random variables with zero mean and unit variance.  The idea is that in the infinite limit, the central limit theorem tells us that any sum of i.i.d random variables is Gaussian distributed.  Unfortunately, that theorem does not apply for the entire stochastic process $\{x_t(N) : t \in \mathbb{R}^+\}$, which has uncountably infinite elements.  But it is good enough for any finite number of elements from this collection, which allows us to define the Wiener process as follows.
\begin{definition}\label{def:weiner_process}
    A \emph{Wiener process} $W_t$ is a stochastic process with continuous trajectories and which for any set of times $t_1 < t_2 < \ldots < t_n, n < \infty$, the increments $W_{t_1}, W_{t_2} - W_{t_1}, \ldots, W_{t_n} - W_{t_{n-1}}$ are independent Gaussian random variables with zero mean and respective variances $t_1, t_2 - t_1, \ldots, t_n - t_{n-1}$.  
\end{definition}
\noindent It turns out that proving the existence of such a process is more involved than worth detailing for our purposes (see \citet[3.2]{vanHandel:2007a}).  One can also show that the Wiener process is unique in the sense that any two processes $W_t,V_t$ which satisfy the above definition give rise to the same probability law, e.g. $\E{f(W_t)} = \E{f(V_t)}$.  It can also be shown that with unit probability, the sample paths of a Wiener process are continuous everywhere but differentiable nowhere.  

Given the above definition, there are some basic properties we can now consider.  First is that the Wiener process $W_t$ introduces a natural filtration $\mathcal{F}_t^W = \sigma\{ W_s : s \leq t\}$.  Relatedly, given an arbitrary filtration $\mathcal{F}_t$, we say $W_t$ is an \emph{$\mathcal{F}_t$-Wiener process} if it is adapted and $W_t - W_s$ is independent of $\mathcal{F}_s$ for any $t > s$.  Two further properties are considered in the following lemmas.
\begin{lemma}\label{lem:wiener_martingale}
    An $\mathcal{F}_t$-Wiener process is a $\mathcal{F}_t$-martingale.
    \begin{proof}
        We want to show that $\E{W_t | \mathcal{F}_s} = W_s$ for $t \geq s$. Clearly $W_s = \E{W_s | \mathcal{F}_s}$ since $W_s$ is $\mathcal{F}_s$-adapted, which allows us to rewrite the condition as $\E{W_t - W_s | \mathcal{F}_s} = 0$.  But we just stated that $W_t - W_s$ is independent of $\mathcal{F}_s$ and therefore has a zero conditional expectation.
    \end{proof}
\end{lemma}
\begin{definition}\label{def:markov}
    An \emph{$\mathcal{F}_t$-Markov process} is an $\mathcal{F}_t$-adapted process $X_t$ such that $\E{f(X_t) | \mathcal{F}_s} = \E{f(X_t) | X_s}$ for all $t \geq s$ and bounded/measurable $f$.
\end{definition}
\begin{lemma}\label{lem:wiener_markov}
    An $\mathcal{F}_t$-Wiener process is an $\mathcal{F}_t$-Markov process.
\end{lemma}
\noindent Physicists are very familiar with Markov processes which describe a statistical process with no memory.  In the formal definition, this is manifest in that the expectation of any future function of the process depends only on the value of the process now.  This is the same as saying the future statistical properties of the process are completely determined by its current value.  It is certainly reassuring that Brownian motion, represented by the Wiener process, satisfies this property. 

\subsection{The It\^{o} Integral}
In our steady march towards a mathematical model of dynamic stochastic processes, we are now ready to consider defining stochastic integrals of Gaussian white noise, e.g. $\int_0^t f_s \xi_s ds$.  Of course, given the discontinuity and non-differentiability of $\xi_s$, we instead hope to define an integral over the Wiener process, e.g. $\int_0^t f_s dW_s$, which is at least continuous.  An obvious approach would be in terms of the \emph{Stieltjes integral}, which is an appropriate generalization of the Riemann integral to non-differentiable integrators.  For our purposes, this means we define a sequence of refining partitions $\pi_n$ of the time interval of integration $[0,t]$ so that we may write
\begin{equation}
    \int_0^t f(s)dW_s = \lim_{\pi_n} \sum_{t_i \in \pi_n} f(t_i) (W_{t_{i+1}} - W_{t_i})
\end{equation}
where the $t_i$ make up the partition $\pi$ of $[0,t]$.  It is certainly not clear that this limit converges and does so independently of the choice of partitions $\pi_n$.  This is especially worrisome given the non-differentiability of $W_t$.  Perhaps as anticipated, a rigorous consideration shows that this stochastic integral formulation does not converge uniquely and depends sensitively on the choice of approximating sequence---there are actually examples where the sequence can be chosen so that the integral converges to any desirable function!  

The source of the troubles comes from the fact that the Wiener process has \emph{infinite total variation} over any interval.  Total variation is the total distance your finger would have to travel tracing out the contour of the Wiener process over the given interval.  This is infinite for \emph{any} interval.  As a description of a physical process, this is clearly absurd!  A particle undergoing Brownian motion would surely require an infinite amount of energy to travel an infinite distance.  Of course, a Wiener process is an idealization of a true physical model, but this seemingly undesirable property is an important consequence of the properties of white noise that we do want to model (delta-correlated, martingale, Markov).  Consider that even if the total displacement $\abs{f(t)- f(s)}, t > s$ is small, the function can still oscillate very rapidly within that interval to get a large total variation; the non-differentiable Wiener process therefore oscillates extraordinarily rapidly over any such interval.  Loosely speaking, the whole trouble boils down to the fact that no matter how fine the partition $\pi$, you don't get any better handle on the Wiener increments; the infinite variation means you will \emph{never} get a level of detail independent of the choice of partition.

Fortunately, one can show that even though the total variation of a Wiener process is infinite, the quadratic variation is finite, e.g. for the interval $[0,1]$ and any sequence of partitions 
    \begin{equation}
        \lim_{n\to\infty}\sum_{t_i \in \pi_n}(W_{t_{i+1}} - W_{t_i})^2 \mapsto 1
    \end{equation}
Thus rather than having the stochastic integral converge almost surely (a.s.), we can instead consider convergence in $\mathcal{L}^2$.  More exactly, for some random variable $X$ and a sequence $\{X_n\}$, we say that $X_n \to X$ a.s. if $\mathbb{P}(\{\omega \in \Omega : X_n(\omega)\to X\}) = 1$.  We say that $X_n \to X$ in $\mathcal{L}^2$ if $\norm{X_n - X}_2 \to 0$ as $n\to\infty$.  There are several types of convergence for sequence of random variables which are related in sometimes unintuitive ways.  See \citet{vanHandel:2007a} for more discussion.  

Taking the $\mathcal{L}^2$ approach, consider the simple, square-integrable $\mathcal{F}_t^W$-adapted process $X_t^n$.  The first two properties suggest there are a series of $N < \infty$ non-random jump times $t_i$ (though this could be relaxed) such that $X_{t_i}^n$ is a constant $\mathcal{F}_{t_i}$-measurable random variable in $\mathcal{L}^2$.  That is, since the stochastic process is ``simple'', there are a finite number of times where it jumps to different values.  For simplicity, we assume these times are the same for all $\omega$.  The idea is to leverage the fact that more general $X_t$ processes are limits of simple processes $X_t^{n}$ and if we can define the integral consistently for the latter, the former will inherit the definition.

It is fairly straightforward to define a consistent integral for $X_t^n$:
\begin{equation}
    I(X_{.}^n) = \int_0^T X_t^n dW_t = \sum_{i=0}^N X_{t_i}^n (W_{t_{i+1}} - W_{t_i}) .
\end{equation}
We now want to show that a sequence of such integrals will converge in $\mathcal{L}^2$ to a particular integral for some $X_t$, independent of the approximations $X_t^n$.  To do so, we make use of the following isometry.
\begin{lemma}[\textbf{It\^{o} Isometry}]\label{lem:ito_isometry}
    Let $X_t^n$ be the simple, square-integrable, $\mathcal{F}_t^W$-adapted process discussed above. Then,
    \begin{equation}
	    \E{\left(\int_0^T X_t^n dW_t\right)^2} = \E{\int_0^T (X^n_t)^2 dt} .
    \end{equation}
    \begin{proof}
        \begin{align}
            \E{\left(\int_0^T X_t^n dW_t\right)^2} &= 
                        \sum_{i,j} \E{X^n_{t_i}X^n_{t_j}(W_{t_{i+1}} - W_{t_i})
                                    (W_{t_{j+1}} - W_{t_j})}
        \end{align}
    Now assume $i \neq j$.  From the Wiener process definition \eqref{def:weiner_process} and properties, we know that disjoint increments are independent for disjoint time intervals and moreover, $W_t - W_s$ is independent of $\mathcal{F}_s$ for any $t > s$.  Without a loss of generality, assume $t_i > t_j$.  Then $(W_{t_{i+1}} - W_{t_i})$ is independent of $X^n_{t_i}$, which is $\mathcal{F}_{t_i}$-adapted, and independent of $X^n_{t_j}$, which is $\mathcal{F}_{t_j}$-adapted.  Since it is also over a different interval than $(W_{t_{j+1}} - W_{t_j})$, we can factor its expectation completely and calculate it to be zero by definition.  This leaves terms for which $i = j$, in which case we have
    \begin{multline}
        \E{\left(\int_0^T X_t^n dW_t\right)^2} = 
                    \sum_{i} \E{{(X^n_{t_i})}^2}\E{(W_{t_{i+1}} - W_{t_i})^2}\\
                    = \sum_i \E{{(X^n_{t_i})}^2}(t_{i+1} - t_i)
                    = \E{\int_0^T {(X_t^n)}^2 dt}
    \end{multline}
    Note that the fact that $X_t^n \in \mathcal{L}^2$ is necessary for convergence to the final integral.
    \end{proof}
\end{lemma}

Recall that an isometry is a distance preserving map between two metric spaces. The property under consideration is an isometry if we consider the process $X_t^n$ as a measurable map on $[0,T] \times \Omega$, which admits a natural product measure $\mu_T \times \mathbb{P}$, where $\mu_T$ is the Lebesgue measure on $[0,T]$ which is simply $T$ times the length of the interval.  Using this definition, we see that the It\^{o} Isometry can be written as
\begin{equation}
    \norm{I(X^n_{.})}_{2,\mathbb{P}} = \norm{X_{.}^n}_{2,\mu_T\times\mathbb{P}}
\end{equation}
where the left-hand term is the $\mathcal{L}^2$ norm on $\Omega$ and the right-hand term is the $\mathcal{L}^2$-norm on $[0,T] \times \Omega$.  This isometry preserves the $\mathcal{L}^2$-distance for $\mathcal{F}_t^W$-adapted simple integrands as
\begin{equation}
    \norm{I(X^n_.) - I(Y_.^n)}_{2,\mathbb{P}} \mapsto \norm{X^n_. - Y^n_.}_{2,\mu_T \times \mathbb{P}}
\end{equation}
But the beauty is that one can show\footnote{Essentially, one shows that the approximating sequence is a Cauchy sequence in $\mathcal{L}^2$, after which convergence is easy.} for some $X_. \in \mathcal{L}^2(\mu_T \times \mathbb{P})$ that there exists some sequence of simple integrands such that
\begin{equation}
    \lim_{n\to\infty }\norm{X_.^n - X_.}^2_{2,\mu_T \times \mathbb{P}} 
        = \E{\int_0^T(X_t^n - X_t)^2 dt} \to 0
\end{equation}
then $I(X.)$ can be defined uniquely as the limit in $\mathcal{L}^2{(\mathbb{P})}$ of the simple integrals $I(X.^n)$!  This turns out to be true for \emph{any} $\mathcal{F}_t^W$-adapted process and gives rise to the following definition of the stochastic integral.
\begin{definition}\label{def:ito_integral}
    Consider the $\mathcal{F}_t^W$-adapted stochastic process $X_t$ in $\mathcal{L}^2(\mu_T\times\mathbb{P})$.  The \emph{It\^{o} integral} 
    \begin{equation}
           I(X_{.}) = \int_0^T X_t dW_t
    \end{equation}
    is defined as the unique limit in $\mathcal{L}^2(\mathbb{P})$ of simple integrals $I(X.^n)$.
\end{definition}
\noindent One can show that the It\^{o} integral has continuous sample paths, is an $\mathcal{F}_t^W$-martingale and satisfies\footnote{These properties actually depend on localizing the It\^{o} integral, which amounts to defining it on arbitrarily long time intervals.  Extending to the infinite interval is difficult.}
\begin{equation}
    \E{\int_0^T X_t dW_t} = \int_0^T \E{X_t dW_t} = \int_0^T \E{X_t}\E{dW_t} = 0 ,
\end{equation}
where the fact that $X_t$ is $\mathcal{F}_t^W$-adapted means it is independent of $dW_t$ (it is a \emph{non-anticipative function}) and the expectation may be factored.

In short, the It\^{o} integral is defined uniquely as a converging sequence of approximations in $\mathcal{L}^2$, where we leverage the fact that simple stochastic processes converge uniquely to show that the It\^{o} integral also converges.  The fact that it is limited to $\mathcal{F}_t^W$-adapted random variables is not a significant restriction for us, especially given the resulting useful properties we recover, including having zero expectation and being a martingale.  Indeed, one approach towards the filtering problem is based on the following relationship between arbitrary martingales and It\^{o} integrals.
\begin{lemma}[Martingale Representation]\label{lem:martingale_representation}
    Consider the $\mathcal{F}_t^W$-martingale $M_t \in \mathcal{L}^2(\mathbb{P})$.  Then there exists a unique $\mathcal{F}_t^W$-adapted process $H_t$ such that
    \begin{equation}
        M_t = M_0 + \int_0^t H_s dW_s 
    \end{equation}
\end{lemma}
\noindent This lemma is extremely useful in that if we can show that some stochastic process is a martingale with respect to the Wiener filtration $\mathcal{F}_t^W$, we know it can be expressed as an It\^{o} integral.  As we shall soon find in the following section, this is equivalent to showing that the process admits a stochastic differential equation analogous to the desired trajectory in \eqref{probability:eq:noisy_process}.
\subsection{Stochastic Differential Equations}
We are now finally in a position to consider dynamical processes driven by white noise.  The basic idea is that the time evolution of complicated stochastic processes can be expressed simply in terms of the basic Wiener process, whose statistics and properties are well known to us, and an appropriate deterministic term.  This is often the route taken in statistical physics, where trajectories are written as \emph{Langevin equations}.  Unfortunately, the ordinary differential equation picture we had in mind in Eq.~\ref{probability:eq:noisy_process} is not useful, as there is no way to express Gaussian white noise directly as a sensible mathematical object.  However, our success in defining the It\^{o} integral suggests that we can deal sensibly with the integral of the noise process, written in terms of Wiener increments, which gives a form
\begin{equation}
    X_t = x_0 + \int_0^t a(s,X_s)ds + \int_0^t b(s,X_s)dW_s .
\end{equation}
But given our predilection for differential equations, we often express the above integral as a \emph{stochastic differential equation (SDE)}, written
\begin{equation} 
    dX_t = a(t,X_t)dt + b(t,X_t)dW_t
\end{equation}
where differentials are used to remind us that this is not a true derivative equation, as $dW_t/dt$ is not a well-defined mathematical object.  The SDE form is really no more than a notational convenience for referring to the more accurate integral form.  This convenience is most obvious when considering functions of such stochastic processes, for which the normal chain rule of calculus no longer holds.  This is seen in the following example and theorem.
\begin{example}[Based on {\citep[Chap. 4]{vanHandel:2007a}}]
    During our introductory calculus course, we are quickly inculcated with algebraic rules for evaluating derivatives and integrals of a variety of functional forms.  One familiar rule is for powers and reads
    \begin{equation} \label{probability:eq:normal_power_rule_integral}
        \int_0^T X_t dX_t = \int_{X_0}^{X_T} u du = \left.\frac{u^2}{2}\right\rvert_{X_0}^{X_T} .
    \end{equation}
Does this hold if $X_t = W_t$?  We can check by explicitly calculating the integral.  Given that the It\^{o} integral is defined in terms of a convergent sequence in $\mathcal{L}^2$, we take the approximating simple versions of $W_t$ to be $W_t$ taken at jump times given by dyadic rationals.  We will not show it, but such an approximation does converge to $W_t$ appropriately and is therefore a valid expansion of the stochastic integral.  Writing this out, we have
\begin{align}
    \int_0^T W_t dW_t &= \mathcal{L}^2\lim_{n\to\infty} \sum_{k=0}^{2^n-1}
                W_{k2^{-n}T}(W_{(k+1)2^{-n}T} - W_{k2^{-n}T})\\
                &= \mathcal{L}^2\lim_{n\to\infty} \frac{1}{2}
                    \left[ W_T^2 
                    - \sum_{k=0}^{2^n-1}(W_{(k+1)2^{-n}T}
                    - W_{k2^{-n}T})^2\right]
\end{align}  
where we have simply rearranged terms in the sum.  We note that the second term converges in $\mathcal{L}^2$ to the total quadratic variation, so that
\begin{equation}
    \int_0^T W_t dW_t = \frac{1}{2}\left[W_T^2 - T\right] .
\end{equation}
But this is not the same as the familiar calculus rule in Eq.~\eqref{probability:eq:normal_power_rule_integral}, which indicates (noting $W_0 = 0$), 
\begin{equation}
    \int_0^T W_t dW_t = \frac{1}{2}W_T^2 .
\end{equation}
Clearly, the It\^{o} integral is a more complicated beast.  Fortunately, the following theorem shows that only a slightly modified chain rule is needed. 
\end{example}
\begin{theorem}[\textbf{It\^{o} Rule, one dimension}]\label{thm:ito_rule}
    Consider the stochastic process $X_t$ with stochastic differential equation
    \begin{equation}
        dX_t = a(t,X_t)dt + b(t,X_t)dW_t
    \end{equation}
    Now consider a function $f(t,X_t)$ that is differentiable with respect to its first argument and twice differentiable with respect to its second.  This function then satisfies the stochastic differential equation
    \begin{align}
        df(t,X_t) &= \partialD{f(t,X_t)}{t}dt + \partialD{f(t,X_t)}{X_t} dX_t
                  + \frac{1}{2}\partialD{{^2}f(t,X_t)}{X_t^2}dX_t^2\\
                  &= \left[\partialD{f(t,X_t)}{t} + a(t,X_t)\partialD{f(t,X_t)}{X_t}
                      + \frac{1}{2}\partialD{{^2}f(t,X_t)}{X_t^2} b^2(t,X_t)\right]dt\nonumber\\
                  &   + b(t,X_t)\partialD{f(t,X_t)}{X_t} dW_t
    \end{align}
where higher order differentials were evaluated according to $dtdW_t = dt^2 = 0$ and $dW_t^2 = dt$. 
\end{theorem}
\begin{lemma}[\textbf{It\^{o} product rule}]\label{lem:ito_product_rule}
    The It\^{o} product rule for stochastic processes $X_t,Y_t$ is
    \begin{equation}
        d(X_tY_t) = dX_tY_t + X_tdY_t + dX_tdY_t .
    \end{equation}
\end{lemma}
\begin{lemma}[\textbf{It\^{o} Rule, multidimensional}]\label{lem:multi_dim_ito_rule}
    Consider the $n$-dimensional stochastic process $X_t : \mathbb{R}^+ \times \Omega \mapsto \mathbb{R}^n$ written
    \begin{equation}
        dX_t = a(t,X_t)dt + \sum_{j=1}^{m} b^j(t,X_t)dW_t^{(j)}
             = a(t,X_t)dt + b(t,X_t)dW_t
    \end{equation}
    where $a(t,X_t), b^j(t,X_t) : \mathbb{R}^+ \times \mathbb{R}^n \mapsto \mathbb{R}^n$ and each $W_t^j$ is an independent Wiener process.  If we collect these into the $m$-dimensional Wiener process $W_t = (W_t^{1}, \ldots, W_t^m)$ and introduce $b(t,X_t) : \mathbb{R}^{+} \times \mathbb{R}^n \mapsto \mathbb{R}^n \times \mathbb{R}^m$, we may use the more compact form on the right.
    
    Further consider the transformed process $Y_t = g(t,X_t) : \mathbb{R}^+ \times \mathbb{R}^n \mapsto \mathbb{R}^p$, where $p$ is not necessarily equal to $n$.  Then $Y_t$ satisfies the SDE
    \begin{equation}
        dY_t^k = \partialD{g^k(t,X_t)}{t}dt + \sum_{i} \partialD{g^k(t,X_t)}{X_t^i} dX_t^i
             + \frac{1}{2}\sum_{ij} \partialD{{^2}g^k(t,X_t)}{X_t^iX_t^j}dX_t^i dX_t^j
    \end{equation}
    where the superscript indicates the $i,j,k$-th entry in the vector and second order differentials are evaluated using $dtdW_t^{j} = dt^2 = 0$ and $dW_t^{i}dW_t^j = \delta_{ij}dt$.  
    
    For the simple case when $p = 1$, we may use the definition of $X_t$ to conveniently write this as
    \begin{subequations} \label{probability:eq:ito_rule_generator_form}
    	\begin{align}
	        dY_t &= \mathscr{L} g(t,X_t) dt +  \grad( g(t,X_t))^T b(t,X_t)dW_t\\
	        \mathscr{L} &= \partialD{}{t} + \grad( g(t,X_t))^T a(t,X_t) 
	                        + \frac{1}{2} \sum_{i,j=1}^n\sum_{k=1}^m b^{ik}(t,X_t)b^{jk}(t,X_t)
	                            \partialD{{^2}g(t,X_t)}{X_t^i \partial X_t^j}
	    \end{align}
    \end{subequations}
    
\end{lemma}
The It\^{o} rule is really no more than a Taylor expansion followed by a careful consideration of the $\mathcal{L}^2$-convergence of the resulting terms.  Not surprisingly, all terms which are a product of $dt$ and any other differential tend to zero.  However, one also finds that $dW_t^2$ converges to $dt$ in $\mathcal{L}^2$, which is effectively a restatement of the It\^{o} Isometry in Lemma \ref{lem:ito_isometry}.  At a heuristic level, many people often express $dW_t$ as $\sqrt{dt}\xi_t$, where $\xi_t$ is a mean-zero, Gaussian random variable with unit variance.  Then it is clear that any consistent chain rule which retains terms to first order in $dt$ must also retain the term for $dW_t^2$.

The upside is that we have an integral which retains statistically pleasing properties; mean-zero stochastic term driven by white noise which is also a martingale.  At the same time, we also have an algebraic formalism for transforming SDE representations of more complicated stochastic processes, at the small cost of having to add an extra term to the usual chain rule.  

\subsection{Wong-Zakai Theorem and Stratonovich Integrals}
Even though we have made significant progress, one might still be concerned that the It\^{o} formalism is simply a mathematical construct that has no connection to any real-world stochastic process.  Should we really be so blithe in throwing away the usual chain rule?  Given the arbitrariness of the Stieltjes stochastic integral, what was the justification for choosing the It\^{o} construction?  If the use of white noise is an approximation to begin with, how faithfully does the It\^{o} SDE capture it?  All of these questions are related and are well-appreciated in the study of stochastic processes.  

To make the issue more precise, consider the standard ordinary differential equation driven by a fluctuating, but not white, noise term $\xi_t^n$:
\begin{equation}
    \frac{d}{dt}X_t^n = a(t,X_t^n) + b(t,X_t^n)\xi_t^n
\end{equation}
We assume $\xi_t^n$ is a sensible noise process whose sample paths are piecewise continuous.  We are interested in the case that this approximates a true Gaussian white noise process in the sense that
\begin{equation}
    \lim_{n\to\infty}\sup_{t}\norm{W_t - W_t^n} \to 0 \text{ a.s }
\end{equation}
where $W_t^n = \int_0^t \xi_s^n ds$.  That is, in some limit, the time integral of $\xi_t^n$ uniformly approximates the Wiener process.  As the process becomes more and more singular, the question is how to interpret the resulting stochastic differential equation.  The following theorem, due to \citet{Wong:1965a}, tells us what to do.
\begin{theorem}[\textbf{Wong-Zakai Theorem}]\label{thm:wong_zakai}
    Given the ordinary differential equation of the form
    \begin{equation}
        \frac{d}{dt}X^n_t = a(t,X_t^n) + b(t,X_t^n)\xi_t^n
    \end{equation}
    where $\xi_t^n$ converges uniformly to Gaussian white noise as $n\to\infty$, the solution $X^n_t$ converges as as $n\to\infty$ to
    \begin{equation}
        dX_t = a(t,X_t) dt + b(t,X_t^n) \circ dW_t
    \end{equation}
   where the stochastic term is interpreted in the Stratonovich sense.
\end{theorem}
\begin{definition}\label{def:stratonovich}
    The \emph{Stratonovich integral} 
    \begin{equation}
           \int_0^T X_t \circ dW_t
    \end{equation}
    is defined as the unique limit in $\mathcal{L}^2(\mathbb{P})$ of the simple integrals 
    \begin{equation}
        \int_0^T X_t^n \circ dW_t = \lim_{\pi_n} \sum_{t_i \in \pi_n} 
                        \frac{1}{2}( X_{t_{i+1}}^n + X_{t_i}^n) (W_{t_{i+1}} - W_{t_i}) .
    \end{equation}
    The Stratonovich integral obeys the standard calculus chain rules, but has non-trivial expectation and is not a martingale.
\end{definition}

Gadzooks! Wong and Zakai tell us that any physical process, which naturally obeys the normal rules of calculus, results in a Stratonovich integral in a white noise limit.  This is not a complete surprise, as the Stratonovich integral obeys the normal chain rule and taking a limit of processes which also obey the chain rule shouldn't break that property.  But remember that the formulation of the It\^{o} integral was a \emph{choice} of how to overcome the lack of an unambiguous convergence of stochastic integrals.  The Stratonovich form is just a \emph{different} choice in defining a stochastic integral.  For deterministic integrals, any choice of increments converges to the same Riemann integral, so we didn't have to worry about which formulation is used.  For stochastic integrals, the Wong-Zakai theorem tells us how to \emph{interpret} an SDE which arises from taking a physical limit; after that, we are free to choose which form to use.  If the two forms are not related, then the It\^{o} definition would be useless for studying physical systems driven by approximate white noise.  Fortunately, it turns out the the two formulations are simply related.
\begin{lemma}\label{lem:ito_to_stratonovich}
    The solution of the multi-dimensional It\^{o} SDE
    \begin{equation}
        dX_t = a(t,X_t)dt + b(t,X_t)dW_t 
    \end{equation} 
    is also solution of a corresponding Stratonivich SDE, written
    \begin{equation}
        dX_t = \bar{a}(t,X_t)dt + b(t,X_t)\circ dW_t ,
    \end{equation}
    with
    \begin{equation} \label{probability:eq:ito_to_stratonovich}
        \bar{a}^{j}(t,X_t) = a^{j}(t,X_t) - \frac{1}{2} \sum_{k=1}^n b^{k}(t,X_t) \partialD{b^j(t,X_t)}{X_t^k}
    \end{equation}
    where the superscripts denote the $j$-th or $k$-th entry in the corresponding vector.
\end{lemma}

We see then that it is straightforward to convert between the two forms, only needing to account for the \emph{It\^{o} drift term}.  This term accounts for the loss of the non-anticipative property for the Stratonovich Wiener increment.  That is, the stochastic process multiplying the noise increment no longer occurs at an independent time interval, which effectively couples the noise at different times and is why we lose the nice statistical properties.  Nonetheless, after using the Wong-Zakai theorem to derive a Stratonovich SDE from a physical model, we simply convert to the equivalent It\^{o} form to make our calculations easier.  This duality will prove useful in Chapter~\ref{chapter:quantum_parameter_estimation} when we study the techniques of projection filtering, which \emph{require} a valid chain rule consistent with differential manifolds and is one of the few circumstances when the Stratonovich form will be preferred.
\subsection{Summary}
The goal of the second part of this chapter was to introduce time into our theory of probability.  This allowed us to consider stochastic processes, which are random variables that are a function of time.  Our hope of writing a stochastic process driven by white noise was hampered at first, as we learned that white noise has no sensible mathematical representation as a stochastic process.  Fortunately, we were able to work with the integral of white noise in terms of the Wiener process, which in turn allowed us to define more general stochastic processes as It\^{o} integrals against the Wiener process.  This gave rise to stochastic differential equations, which are dynamical equations for the evolution of stochastic trajectories involving both deterministic and stochastic terms.  Due to the subtleties of the It\^{o} integral, we found that SDEs obey a modified chain rule which requires retaining terms to second order in Wiener increments.  We also found that the physical limit of increasingly better approximations of white noise converges to a Stratonovich, rather than an It\^{o}, SDE.  Fortunately, we found that a given stochastic process has an equivalent representation in either form, so that the statistically superior properties of the It\^{o} integral may be used in analysis.
\section{Classical Filtering Theory}
Using the techniques we have developed thus far, we are finally ready to tackle the filtering problem.  We consider an $n$-dimensional, unobserved stochastic process $X_t$, governed by the SDE
\begin{equation} \label{classical_filtering:eq:system}
    dX_t = a(t,X_t)dt + b(t,X_t)dW_t  \qquad \text{ ``system''}
\end{equation}
and a related $m$-dimensional observed stochastic process $Y_t$, governed by the SDE
\begin{equation} \label{classical_filtering:eq:observations}
    dY_t = c(t,X_t)dt + d(t)dV_t \qquad \text{ ``observations/measurements''}
\end{equation}
where $dW_t,dV_t$ are two independent Wiener processes of $k$ and $p$ dimensions, respectively.  Note that we have already imposed a particular structure on the stochastic processes under consideration; they are driven by white noise and admit an SDE description\footnote{Meaning $a,b,c,d$ are bounded, $d^{-1}$ exists and is bounded and $X_t,Y_t$ have a unique $\mathcal{F}_t$-adapted solution; some of these restrictions may be lifted with suitable care.  Note that we could easily extend the SDE formalism to include Poisson noise processes in addition to Gaussian noise processes.}.  Given the broad applicability of Gaussian white noise in physics and related disciplines, limiting ourselves to this class of processes is not a significant restriction, especially given the analytic results we will be able to derive.

Returning to the problem at hand, Eqs.~\eqref{classical_filtering:eq:system} and \eqref{classical_filtering:eq:observations} are known in control theory as the \emph{system-observations pair} and formalize the structure of the inference problem.  That is, the unobserved system $X_t$ undergoes a stochastic time-evolution.  We are interested in some property of the system, but only have access to the observations $Y_t$.  Unfortunately, $Y_t$ is not $\mathcal{F}^{X}_t$-measurable, since it involves the independent noise process $dV_t$ and we therefore do not \emph{know} $X_t$ after measuring $Y_t$.  Fortunately, $Y_t$ carries some information about the system, albeit of a set structure and corrupted by the extra noise.  Using the techniques of inference we have developed, we can still construct an estimate of the system conditioned on the observations.
\begin{definition}\label{def:filtering_problem}
    Given a system-observations pair as above, the \emph{filtering problem} is to calculate the least-squares best-estimate of the current state of the system given the observations record.  Mathematically, we write this as
    \begin{equation}
        \pi_t[X_t] = \E{X_t | \mathcal{F}_t^Y}
    \end{equation} 
    where $\mathcal{F}_t^Y$ is the filtration generated by the observations process up to time $t$.
\end{definition}

Actually, there is a more general class of inference problems one could consider, written
\begin{equation}
    \pi_t[f_s] = \E{f(X_s) | \mathcal{F}_t^Y}
\end{equation}
where one estimates some arbitrary function of the state at an arbitrary time.  If $s = t$ and $f(X) = X$, this is simply the filtering problem already discussed. For $s = 0$ and $f(X) = X$, this is the \emph{smoothing problem}, for which $\pi_t[X_0]$ is an estimate of the \emph{initial state}.  For $s > t$ and $f(X) = X$, this is the \emph{predictor problem}, for which $\pi_t[X_{s>t}]$ is an estimate of a \emph{future state}.  Choosing $s$ to be an intermediary time or $f$ to be a more complicated function correspond to other valid inference problems.

Nonetheless, the most relevant problem for our purposes is the filtering problem.  The rest of this section is devoted to developing a recursive formula for $\pi_t[f(X_t)]$, written in shorthand as $\pi_t[f]$, so that for each differential observation increment $dY_t$, we can readily update the filtered estimate
\begin{equation}
    d\pi_t[f(X_t)] = q(t,X_t)dt + r(t,X_t)dY_t 
\end{equation}
for some functions $q$ and $r$ which we will need to determine.  We will take $f$ to be a square-integrable real-valued function, so that to reconstruct the multi-dimensional $X_t$, we would need a set of estimates $\pi_t[f^i]$, with functions $f^i(X_t) = X_t^i$.  Making $f$ one-dimensional will greatly simplify the notation without losing any essential details.

Our general approach is the reference probability method, which we will also use to develop the quantum filter.  The basic idea is rather simple; if $X_t$ and $Y_t$ were independent, then the conditional expectation of $X_t$ amounts to a simple averaging.  If we can find a measure under which the two processes are independent, then it will be trivial to evaluate the conditional expectation under this measure.  Of course, if $X_t$ and $Y_t$ were actually independent, the filtering problem would be pointless since we would never learn anything about the state from the observations.  So we must also find a way to relate the calculation under the new measure back to the original calculation under the old measure.  The first two parts of this section focus on developing these two relations, finding a measure under which the processes are independent and another for relating conditional expectations under different measures.
\subsection{Girsanov's Theorem}
In many areas of mathematics, a change of variables often simplifies a seemingly difficult problem.  In the domain of probability, a similar approach is to change the underlying probability measure, which may simplify the statistics of a random variable.  We have already considered such a change using the Radon-Nikodym theorem (Thm.~\ref{thm:radon_nikodym}).  Being able to make such a transformation is particularly useful for stochastic processes driven by Gaussian white noise, whose deterministic terms obfuscate many of the nice statistical properties of a pure It\^{o} integral over the Wiener process.  The following theorem shows how to construct a new measure under which such a statistically complicated stochastic process becomes a Wiener process.
\begin{theorem}[\textbf{Girsanov}]\label{thm:Girsanov}
    Let $W_t$ be an $n$-dimensional, $\mathcal{F}_t$-Wiener process on $(\Omega,\mathcal{F},\mathbb{P})$ with filtration $\mathcal{F}_t$.  Also consider the $n$-dimensional stochastic process $X_t$ governed by the SDE
    \begin{equation} \label{Girsanov:eq:girsanov_form}
        dX_t = F_t dt + dW_t \qquad t \in [0,T_f]
    \end{equation}    
    Assuming $F_t$ is It\^{o} integrable, define
    \begin{equation}
        \Lambda = \exp\left[ -\int_0^{T_f} F_s^{T} dW_s - \frac{1}{2}\int_0^{T_f} \normsq{F_s} ds\right] .
    \end{equation}
    
    If $\mathbb{E}_{\mathbb{P}}[\Lambda] = 1$, then $X_t$ is an $\mathcal{F}_t$-Wiener process under $\mathbb{Q}(A) = \mathbb{E}_{\mathbb{P}}(\Lambda \chi_A)$.
    \begin{proof}
        For simplicity, we will proof this result for a one-dimensional process.  For a more general proof, see Theorem 4.5.3 in \citet{vanHandel:2007a} the first half of which is essentially reproduced here.  Recall from Definition \ref{def:weiner_process}, a stochastic process is characterized by continuous sample paths and independent, Gaussian distributed increments with zero mean and variance equal to the interval length.  Given that $X_t$ is written as an SDE, it has continuous sample paths by construction.  In order to show the increment properties, we consider a given interval $X_t - X_s$ with $t > s$.  If under the new measure $X_t - X_s$ has the appropriate distribution independent of any $\mathcal{F}_s$-measurable random variable, we satisfy both requirements.  We verify this using the method of generating or characteristic functions.  That is, for $X_t$ as defined above and $Z$ an arbitrary $\mathcal{F}_s$-measurable random variable, we want
        \begin{equation}
            \mathbb{E}_{\mathbb{Q}}[e^{\alpha (X_t - X_s) + \beta Z }]
                = e^{-\alpha^2\frac{(t-s)}{2}}\mathbb{E}_{\mathbb{Q}}[e^{\beta Z}]
        \end{equation}
where $\alpha,\beta \in \mathbb{R}$ are the generating parameters and $e^{-\alpha^2\frac{(t-s)}{2}}$ is the characteristic function of a mean zero, variance $t-s$, Gaussian random variable. 

    Using the definitions above and introducing the $\mathcal{F}_t$-adapted process
    \begin{equation}
        \Lambda_t = \exp\left[ -\int_0^t F_s dW_s - \frac{1}{2}\int_0^t F_s^2 ds\right] .
    \end{equation}
    we find explicitly that
    \begin{align}
        \mathbb{E}_{\mathbb{Q}}[e^{\alpha (X_t - X_s) + \beta Z }] 
                    &= \mathbb{E}_{\mathbb{P}}[\Lambda_{T_f} e^{\alpha (X_t - X_s) + \beta Z }] \\
                    &= \mathbb{E}_{\mathbb{P}}[\mathbb{E}_{\mathbb{P}}[\Lambda_{T_f} | \mathcal{F}_t]
                        e^{\alpha (X_t - X_s) + \beta Z }] \\
                    &= \mathbb{E}_{\mathbb{P}}[\Lambda_t e^{\alpha (X_t - X_s) + \beta Z }] \\
                    &= \mathbb{E}_{\mathbb{P}}[\Lambda_s
                        e^{\int_{s}^t(\alpha F_r -\frac{1}{2} F_r^2) dr 
                        + \int_s^{t}(\alpha - F_r) dW_r
                        + \beta Z}]\\
                    &= e^{-\alpha^2\frac{(t-s)}{2}}
                    \mathbb{E}_{\mathbb{P}}[\Lambda_s
                        e^{-\frac{1}{2}\int_{s}^t(\alpha - F_r)^2 dr 
                        + \int_s^{t}(\alpha - F_r) dW_r
                        + \beta Z}]
    \end{align}
    where in reaching the last line we have completed the square and pulled out one of the deterministic terms.  The manipulations in the first three lines are simply an application of the definition of conditional expectation (Definition \ref{def:conditional_expectation_kolmogorov}), where all terms save $\Lambda_{T_f}$ are $\mathcal{F}_t$-measurable, so that we may replace $\Lambda_{T_f}$ with $\Lambda_t$ under the overall expectation.  Similarly, since $\Lambda_s e^{\beta Z}$ is $\mathcal{F}_s$-measurable but the remaining exponential terms are not, we again apply conditional expectation to write
    \begin{equation}
        \mathbb{E}_{\mathbb{Q}}[e^{\alpha (X_t - X_s) + \beta Z }]  
                = e^{-\alpha^2\frac{(t-s)}{2}}
                \mathbb{E}_{\mathbb{P}}[\Lambda_s e^{\beta Z}
                    \mathbb{E}_{\mathbb{P}}[e^{-\frac{1}{2}\int_{s}^t(\alpha - F_r)^2 dr 
                    + \int_s^{t}(\alpha - F_r) dW_r} | \mathcal{F}_s]]
    \end{equation}
    
    Focusing on the last conditional expectation term, set $\theta_t = (\alpha - F_t)$ and define
    \begin{equation}
        dR_t = -\frac{1}{2}\theta_t^2 dt + \theta_t dW_t
    \end{equation}
    If we can show that $e^{R_t}$ is a martingale, then the conditional
    expectation under considertation is simply
    \begin{equation}
        \mathbb{E}_{\mathbb{P}}[e^{-\frac{1}{2}\int_{s}^t(\alpha - F_r)^2 dr 
        + \int_s^{t}(\alpha - F_r) dW_r} | \mathcal{F}_s]
        = \mathbb{E}_{\mathbb{P}}[e^{R_t-R_s} | \mathcal{F}_s]
        = e^{R_s - R_s} 
        = 1
    \end{equation}
    Using It\^{o}'s rule, we find
    \begin{align}
        d(e^{R_t}) &= e^{R_t}dR_t + \frac{1}{2}e^{R_t} (dR_t)^2\\
                   &= e^{R_t}\left[ -\frac{1}{2}\theta_t^2 dt + \theta_t dW_t
                    + \frac{1}{2} \theta_t^2 dt\right]\\
                   &= \theta_t e^{R_t} dW_t
    \end{align}
    But since this is precisely an It\^{o} integral driven by Gaussian white noise, we know from the Martinagle Representation Lemma \ref{lem:martingale_representation} that it is indeed a martingale.  Notice also that $\Lambda_t$ is of the same form, since the minus sign on the $dW_t$ coefficient still squares to cancel the deterministic term via the It\^{o} correction.  As such, we can drop the conditional expectation as desired and use the martingale property of $\Lambda_t$ to write
    \begin{align}
        \mathbb{E}_{\mathbb{Q}}[e^{\alpha (X_t - X_s) + \beta Z }]  
            &= e^{-\alpha^2\frac{(t-s)}{2}}
            \mathbb{E}_{\mathbb{P}}[\Lambda_s e^{\beta Z}]\\
            &= e^{-\alpha^2\frac{(t-s)}{2}}
            \mathbb{E}_{\mathbb{P}}[\mathbb{E}_{\mathbb{P}}[\Lambda_{T_f} e^{\beta Z} | \mathcal{F}_s]]\\
            &= e^{-\alpha^2\frac{(t-s)}{2}}\mathbb{E}_{\mathbb{Q}}[e^{\beta Z}]
    \end{align}
    where in reaching the last step we have used the conditional expectation property that $\mathbb{E}[\mathbb{E}[X | \mathcal{F}]] = \mathbb{E}[X]$ to recognize the definition of $\mathbb{E}_{\mathbb{Q}}$ as desired.
    \end{proof}
\end{theorem}

Girsanov's theorem allows us to find a measure under which stochastic processes like the observations process in Eq.~\eqref{classical_filtering:eq:observations} are Wiener processes.  If we can find a measure such that $Y_t$ is independent of $X_t$ and is equivalent to a Wiener process, we might then be able to evaluate the conditional expectation easily.  The following section addresses that task.
\subsection{Bayes Formula}
Although the Radon-Nikodym theorem (Thm.~\ref{thm:radon_nikodym}) relates expectations under related measures, we have yet to develop a method for relating \emph{conditional} expectations under different probability measures.  The following formula, reminiscent of the familiar Bayes rule for conditional probabilities, provides a means for doing so.
\begin{theorem}[\textbf{Bayes formula}]\label{thm:bayes_formula}
    Let $(\Omega,\mathcal{F},\mathbb{P})$ be a probability space with another measure $\mathbb{Q}$ such that $\mathbb{P} \ll \mathbb{Q}$.  Then for some $\mathcal{G} \subset \mathcal{F}$ and random variable $X$ such that $\mathbb{E}_{\mathbb{P}}[\abs{X}] < \infty$, the following \emph{Bayes formula} relates conditional expectations as follows:
    \begin{equation}
        \mathbb{E}_{\mathbb{P}}[ X | \mathcal{G}]
                = \frac{\mathbb{E}_{\mathbb{Q}}[X \frac{d\mathbb{P}}{d\mathbb{Q}} | \mathcal{G}]}
                        {\mathbb{E}_{\mathbb{Q}}[\frac{d\mathbb{P}}{d\mathbb{Q}} | \mathcal{G}]}
    \end{equation}
    where $\frac{d\mathbb{P}}{d\mathbb{Q}}$ is the Radon-Nikodym derivative. 
    \begin{proof}
        Again, we follow the exposition of Lemma 7.1.3 in \citet{vanHandel:2007a}.  Let $S \in \mathcal{G}$.  Since both sides satisfy the Kolmogorov definition of conditional probability, we can use the arbitrary $\mathcal{G}$-measurable random variable $I_S$ to show that both sides satisfy the conditional expectation property.  Starting from the numerator on the right, we have
        \begin{equation}
            \mathbb{E}_{\mathbb{Q}}[I_S \mathbb{E}_{\mathbb{Q}}[X\frac{d\mathbb{P}}{d\mathbb{Q}}|\mathcal{G}]]
                = \mathbb{E}_{\mathbb{Q}}[I_S X\frac{d\mathbb{P}}{d\mathbb{Q}}]
                = \mathbb{E}_{\mathbb{P}}[I_S X]
                = \mathbb{E}_{\mathbb{P}}[I_S X]
        \end{equation}
        where we have used the properties of conditional expectation and the Radon-Nikodym relation.  Using the conditional expectation property again and running the above in reverse, we find
        \begin{equation}
            \mathbb{E}_{\mathbb{P}}[I_S X] 
                = \mathbb{E}_{\mathbb{P}}[I_S \mathbb{E}_{\mathbb{P}}[X|\mathcal{G}]]
                = \mathbb{E}_{\mathbb{Q}}[I_S \frac{d\mathbb{P}}{d\mathbb{Q}}
                        \mathbb{E}_{\mathbb{P}}[X|\mathcal{G}]]
                = \mathbb{E}_{\mathbb{Q}}[I_S 
                        \mathbb{E}_{\mathbb{Q}}[\frac{d\mathbb{P}}{d\mathbb{Q}} | \mathcal{G}]
                        \mathbb{E}_{\mathbb{P}}[X|\mathcal{G}]] .
        \end{equation}
        But since this is true for for any $S$, it must hold without the outer expectations and $I_S$, so that
        \begin{equation}
           \mathbb{E}_{\mathbb{Q}}[X\frac{d\mathbb{P}}{d\mathbb{Q}}|\mathcal{G}]
           =  \mathbb{E}_{\mathbb{Q}}[\frac{d\mathbb{P}}{d\mathbb{Q}} | \mathcal{G}]
           \mathbb{E}_{\mathbb{P}}[X|\mathcal{G}]
        \end{equation}
    If we divide by $\mathbb{E}_{\mathbb{Q}}[\frac{d\mathbb{P}}{d\mathbb{Q}} | \mathcal{G}]$ we recover the Bayes formula.  
    \end{proof}
\end{theorem}
With this result and the Girsanov theorem, we are now ready to solve the filtering problem.
\subsection{Non-Linear Filtering Equations}
With the Girsanov theorem and Bayes formula in hand, we can now proceed to find a formula for $\pi_t[f] = \mathbb{E}_{\mathbb{P}}[f(X_t) | \mathcal{F}_t^Y]$.  Our first step is to find a new measure $\mathbb{Q}$ under which $X_t$ and $\mathcal{F}_t^Y$ are independent.  Since $X_0$ is already independent of $W_t,V_t$, our task is really to show that $dW_t,d\bar{Y}_t$ are two independent $\mathcal{F}_t^Y$-Wiener processes under $\mathbb{Q}$, where we have set
\begin{equation}
    d\bar{Y}_t = d^{-1}(t)c(t,X_t)dt + V_t = d^{-1}(t)dY_t
\end{equation}
Noting that this is precisely the Girsanov form in Eq.~\eqref{Girsanov:eq:girsanov_form}, introduce 
\begin{equation}
    \Lambda_t = \exp\left[ -\int_0^t [d^{-1}(s)c(s,X_s)]^T d\bar{Y}_t 
                    - \frac{1}{2}\int_0^t \normsq{d^{-1}(s)c(s,X_s)}ds\right] 
\end{equation}
so that the new measure $\mathbb{Q}_{T_f}$ is defined by the density $d\mathbb{P}/d\mathbb{Q}_{T_f} = \Lambda_{T_f}$.  From the Girsanov theorem, we know that $\bar{Y}_t$ is a Wiener process independent of $W_t$ and $X_0$, since for the Girsanov form in Eq.~\eqref{Girsanov:eq:girsanov_form}, the process is independent of the stochastic coefficient $F_t$ under the new measure.  Thus, under $\mathbb{Q}$, $X_t$ and $\bar{Y}_t$ are independent and we use Bayes formula to rewrite the conditional expectation as
\begin{equation} \label{filtering:eq:kallianpur_striebel}
    \pi_t[f] = \frac{\mathbb{E}_{\mathbb{Q}_t}[f(X_t)\Lambda_t | \mathcal{F}_t^{Y}]}
                         {\mathbb{E}_{\mathbb{Q}_t}[\Lambda_t | \mathcal{F}_t^{Y}]}
                  = \frac{\sigma_t(f)}{\sigma_t(1)}
\end{equation}
where we have introduced the unnormalized estimate $\sigma_t$ in the obvious way.  Eq. \eqref{filtering:eq:kallianpur_striebel} is known as the \emph{Kallianpur-Striebel formula}. 

We now focus an deriving an SDE for the unnormalized form.  We begin by using the It\^{o} rule to calculate
\begin{align}
    d\Lambda_t &= \Lambda_t [d^{-1}(s)c(s,X_s)]^T d\bar{Y}_t
\end{align}
and using the multi-dimensional It\^{o} rule in Eq.~\eqref{probability:eq:ito_rule_generator_form}
\begin{align}
    df(X_t) &= \mathscr{L}_t f(X_t) dt + [\grad f(X_t)]^Tb(t,X_t)dW_t .
\end{align}
From the It\^{o} product rule in Lemma \ref{lem:ito_product_rule}, we find 
\begin{multline}
    f(X_t)\Lambda_t = f(X_0) +
                    \int_0^t \Lambda_s\mathscr{L}_s f(X_s) ds 
                  + \int_0^t\Lambda_s[\grad f(X_s)]^Tb(s,X_s)dW_s\\
                 +\int_0^t  f(X_s)  \Lambda_s [d^{-1}(s)c(s,X_s)]^T d\bar{Y}_s 
\end{multline}
where I have used the integral, rather than the SDE form and noted $\Lambda_0 = 1$.  In order to recover the $\sigma_t(f)$ form, we need to calculate $\mathbb{E}_{\mathbb{Q}}[\cdot | \mathcal{F}_t^Y]$ of both sides of the above equation.  Given that the integrals are essentially sums, the expectations may be brought inside and applied directly to the integrands.  But by construction, $dW_s$ is independent of $\mathcal{F}_t^Y$ under the measure $\mathbb{Q}$; after all, that is why we picked $\mathbb{Q}$.  As such, the conditional part is dropped, leaving $\mathbb{E}_{\mathbb{Q}}[\Lambda_s[\grad f(X_s)]^Tb(s,X_s) dW_s] = 0$, since $dW_s$ is an standard Wiener process under $\mathbb{Q}$.  Additionally, by properties of conditional expectation, $\mathcal{F}_t^Y \mapsto \mathcal{F}_s^Y$ under the integral, since for the adapted processes under consideration, $\mathcal{F}_t^Y$ provides no extra information for conditioning than what is already in $\mathcal{F}_s^Y$.  Lastly, since $d\bar{Y}_s$ is $\mathcal{F}_s^Y$-measurable under $\mathbb{Q}$, it may also be pulled out of the conditional expectation.  This leaves
\begin{multline}
    \mathbb{E}_{\mathbb{Q}}[f(X_t)\Lambda_t] = \mathbb{E}_{\mathbb{Q}}[f(X_0) | \mathcal{F}_t^Y] +
                    \int_0^t \mathbb{E}_{\mathbb{Q}}[\Lambda_s\mathscr{L}_s f(X_s)
                            | \mathcal{F}_s^Y] ds \\
                 +\int_0^t  \mathbb{E}_{\mathbb{Q}} [f(X_s)  \Lambda_s [d^{-1}(s)c(s,X_s)]^T
                            | \mathcal{F}_s^Y ] d\bar{Y}_s ,
\end{multline}
from which we identify the \emph{Zakai equation}
\begin{equation}
    d\sigma_t(f) = \sigma_t(\mathscr{L}_t f) dt
                 + \sigma_t(d^{-1}(s)c(s,X_s) f)^T d\bar{Y}_s .
\end{equation}

In order to recover the SDE for the full filter, we note that 
\begin{equation}
	d\sigma_t(1) = \sigma_t(d^{-1}(s)c(s,X_s))^T d\bar{Y}_s 
\end{equation}
and use the It\^{o} rule to calculate
\begin{align} \label{classical:eq:unnormalized_filter_ratio_sde}
    d\left[\frac{\sigma_t(f)}{\sigma_t(1)}\right] &= \frac{d\sigma_t(f)}{\sigma_t(1)}
                                       - \frac{\sigma_t(f)d\sigma_t(1)}{\sigma_t(1)^2}
                                       - \frac{1}{2} \frac{d\sigma_t(1)d\sigma_t(f)}{\sigma_t(1)^2}
                             + \frac{\sigma_t(f)d\sigma_t(1)d\sigma_t(f)}{\sigma_t(1)^3}
\end{align}
Plugging in for these terms, noting that $\sigma_t(f)/\sigma_t(1) = \pi_t[f]$ and rearranging the result leads one to the \emph{Kushner-Stratonovich} equation given in the following theorem. 
\begin{theorem}[\textbf{Kushner-Stratonovich}]\label{classical:thm:kushner_stratonovich}
    The solution to the filtering problem satisfies the SDE
	\begin{multline} \label{filtering:eq:kushner_stratonovich}
	    d\pi_t[f] = \pi_t[\mathscr{L}_tf] dt \\
	              + \left(\pi_t[d(t)^{-1}c(t,X_t) f]
	                   -\pi_t[f]\pi_t[d(t)^{-1}c(t,X_t)]\right)^T
	                \left(d\bar{Y}_t - \pi_t[d(t)^{-1}c(t,X_t)]\right) 
	\end{multline}
	with $\pi_0(f) = \mathbb{E}_{\mathbb{P}}[f(X_0)]$. 
\end{theorem}
\noindent This is precisely a recursive equation of the form we desired, in which the estimate of $f(X_t)$ is updated in place with each measurement increment $d\bar{Y}_t = d(t)^{-1}Y_t$.  

Before exploring the details of this equation, let us first reflect on the path we have taken in deriving it.  For a seemingly simple form, what was really the point of changing measures and constructing the $d\bar{Y}_t$ process?  As was stated as motivation, by constructing the measure $\mathbb{Q}$ under which $X_t$ and $\mathcal{F}_t^Y$ were independent, the conditional expectation with respect to that measure becomes relatively trivial.  Indeed, that is what we found in calculating the Zakai equation for $\sigma_t(f)$.  Due to the nature of $\mathbb{Q}$, we were able to completely drop terms involving $dW_t$.  By the definition of conditional expectation, $\mathbb{E}[f(X_t) | \mathcal{F}_t^Y]$ is precisely an orthogonal projection onto the space $\mathcal{F}_t^Y$; since $dW_t$ is independent of $\mathcal{F}_t^Y$, it is dropped in the orthogonal projection.  But a more important feature of working under the new measure was that the process $d\bar{Y}_t$ could be pulled out of the conditional expectation since it is manifestly $\mathcal{F}_t^Y$-measurable under $\mathbb{Q}$.  As a result, the integral over $d\bar{Y}_t$ is essentially just the averaging we sought from the beginning and is the essential property that allows us to express the filter as a SDE over the process $d\bar{Y}_t$.  The rest of the work was merely applying Bayes formula to relate the Zakai equation for $\sigma_t(f)$ back to $\pi_t[f]$.

It is worth recognizing the following important process in the Kushner-Straton\-ovich equation.
\begin{definition}\label{def:innovations_process}
    The \emph{innovations process}, written
    \begin{equation}
        \bar{V}_t =  \bar{Y}_t - \int_0^s\pi_s[d(s)^{-1}c(s,X_s)]ds
    \end{equation}
    is an $\mathcal{F}_t^Y$-Wiener process and satisfies the SDE
    \begin{equation}
        d\bar{V}_t =  (d(t)^{-1}c(t,X_t)] - \pi_t[d(t)^{-1}c(t,X_t)])dt + dV_t .
    \end{equation}
    The proof that it is a Wiener process is essentially identical to the generating function approach used to proof Girsanov's theorem and is found in Proposition 7.2.9 in \citet{vanHandel:2007a}.  Another approach is to show $d\bar{V}_t$ is a martingale that satisfies the It\^{o} product $d\bar{V}_t^2 = dt$, which by L\'{e}vy's theorem\footnote{Essentially L\'{e}vy's theorem tells us that if a given process $M_t$ and the related one $M_t^2 - t$ are martingales, then $M_t$ is a Wiener process.  See \citep{Williams:1991a} for more discussion.} means it is a Wiener process.
\end{definition}
Structurally, the form of the innovations process gives considerable insight into its properties.  If we were to \emph{know} $X_t$, the innovations process would be identically the Wiener process $dV_t$, which is the noise corrupting the measurement that serves no purpose save to make our lives more difficult.   Looking at the SDE form for $d\bar{V}_t$, we also see that it contains  $dV_t$ in addition to the difference of the estimate and true process value.  But by definition, that piece satisfies
\begin{equation} \label{filtering:eq:innovations:expectation}
    \mathbb{E}_{\mathbb{P}}[(d(t)^{-1}c(t,X_t)] - \pi_t[d(t)^{-1}c(t,X_t)]) | \mathcal{F}_s^Y] = 0 
    \qquad t \geq s
\end{equation}
so that the difference must be orthogonal to $\mathcal{F}_t^Y$\footnote{It might seem weird that all the pieces used to construct $d\bar{V}_t$ come from $Y_t$, yet this difference term is nonetheless not $\mathcal{F}_t^Y$-measurable.  But note that we don't have access to this piece by itself, we get $V_t$ along with.  The innovations process smartly pulls out the information coming solely from $X_t$, as best as it can in the presence of $V_t$.}.  This is what gives the innovations process its name, in that the difference $(d(t)^{-1}c(t,X_t)] - \pi_t[d(t)^{-1}c(t,X_t)])$ contains only the ``new'' or ``innovative'' information that would cause us to update our estimate.  In a more heuristic view, the innovations process tries to make the measurements look as much as possible like the corrupting process $V_t$, so that the filter averages that white noise away to zero.  Anything that makes $\bar{V}_t$ look different than $V_t$ is then useful information about the process of interest.  The added benefit that $\bar{V}_t$ is still a Wiener process, thanks in part to the property in Eq.~\eqref{filtering:eq:innovations:expectation}, means we can leverage all of the It\^{o} properties we like when studying the filter.

Of course, the lingering important question is whether one can use the filter in practice.  Looking at Eq.~\eqref{filtering:eq:kushner_stratonovich}, we see that calculating $\pi_t[f]$ requires calculation of terms such as $\pi_t[\mathscr{L}_tf]$ and $\pi_t[d(t)^{-1}c(t,X_t) f]$.  Plugging those terms back into the Kushner-Stratonovich equation will undoubtedly require calculation of iterated forms such as $\pi_t[\mathscr{L}_t^2f]$ and beyond, until a closed set of equations is reached.  In general, we would expect to need an infinite number of equations to close the loop for the real-valued process $X_t$.  Another perspective, which will prove useful for the quantum filter, is to work with an adjoint form of the filter, in which we introduce a random density $p_t(X)$ which satisfies
\begin{equation}
    \pi_t[f] = \mathbb{E}_{\mathbb{P}}[f(X_t) | \mathcal{F}_t^Y] = \int f(x)p_t(x) dx .
\end{equation}
Integrating the Kushner-Stratonovich equation by parts gives the nonlinear, stochastic partial integro-differential equation
\begin{equation} \label{filtering:eq:kusher_stratonovich_density_form}
    dp_t(x) = \mathscr{L}_t^{*} p_t(x) dt
            + p_t(x)\left[ d^{-1}(t)(c(t,x) - \pi_t[c(t,x)])\right]^T d\bar{V}_t 
\end{equation}
where
\begin{equation}
    \mathscr{L}_t^{*}p(x) = - \sum_{i=1}^n \partialD{}{x^i}\left(a^{i}(t,x) p(x)\right)
                          + \frac{1}{2} \sum_{i,j = 1}^n
                            \sum_{k = 1}^m \partialD{{^2}}{x^i\partial x^j}\left(
                                b^{ik}(t,x) b^{jk}(t,x)p(x)\right)
\end{equation}

This form is generally not any more useful the the Kushner-Stratonovich equation, but is a duality similar to the Schr\"{o}dinger and Heisenberg pictures in quantum mechanics.  A similar PDE can be developed for the Zakai equation ($\sigma_t(f)$), which is at least a linear equation that admits more straightforward numerical approximations. Of course, there is one well-known continuous distribution which requires only a few characteristic parameters---the Gaussian distribution.  In the following section, we consider systems whose conditional state is well-described by a Gaussian distribution and therefore admits a simple and tractable filter with wide applicability. 
\subsection{Kalman-Bucy Filter}
Perhaps the simplest systems-observation pair we can consider is one governed by the pair of linear stochastic differential equations
\begin{align}
    dX_t = A_t X_t dt + B_t dW_t\\
    dY_t = C_t X_t dt + D_t dV_t
\end{align}
where $X_t, Y_t$ are $n,m$-dimensional, real-valued stochastic processes, $W_t,V_t$ are independent, $k,p$-dimensional Wiener process and $A_t,B_t,C_t,D_t$ are real-valued, non-random matrices of dimension $ n \times n$, $k \times n$, $m \times m$ and $p \times m$ respectively.  In physics and engineering, many problems are well-described or well-approximated by a linear description and are often appealing due to their relative analytical simplicity.  As we will find in the following theorem, the filter for these simple systems is also simple, making linear stochastic models attractive for practical filtering and control applications. 
\begin{theorem}[\textbf{Kalman-Bucy Filter}]\label{thm:kalman_bucy}
    The solution to the linear stochastic filtering problem, written $\pi_t[X] = \mathbb{E}_{\mathbb{P}}[X_t | \mathcal{F}_t^Y]$, with $\pi_0[X]$ Gaussian distributed, satisfies the SDE
    \begin{equation} \label{filter:eq:kalman_bucy:state}
        d\pi_t[X] = A_t \pi_t[X] dt + P_t(D_t^{-1}C_t)^Td\bar{V}_t
    \end{equation}
    with innovations process $d\bar{V}_t = D_t^{-1}(dY_t - C_t \pi_t[X] dt)$ and deterministic covariance matrix $P_t = \E{(X_t - \pi_t[X])(X_t - \pi_t[X])^T}$ satisfying the Riccati equation
    \begin{equation} \label{filter:eq:kalman_bucy:variance}
        \frac{dP_t}{dt} = A_t P_t + P_t A_t^{T} - P_t C_t^{T}(D_tD_t^{T})^{-1}C_tP_t
                            + B_tB_t^T .
    \end{equation}
    \begin{proof}[Proof by citation and vigorous handwaving]
        For an excellent and detailed derivation of the Kalman-Bucy filter, consult \citet[Chap. 6]{Oksendal:1946a} or the original papers \citep{Kalman:1960a,Kalman:1961a}.  Another approach is to simply use the linear forms of $X_t$ and $Y_t$ in our results from the previous section, although there are technical reasons we should hesitate, primarily that the change of measure $\Lambda_t$ is generally not square-integrable.  Nonetheless, such subtleties can be handled and we would end up with the right answer.  The details of the procedure are not enlightening, so I only review the strategy, which is to consider the density form of the Zakai equation, analogous to \eqref{filtering:eq:kusher_stratonovich_density_form} and written
        \begin{equation}
        \sigma_t(f) = \int f(x)q_t(x)dx \qquad
         dq_t(x) = \mathscr{L}_t^{*}q_t(x)dt + q_t(x)(d(t)^{-1}c(t,x))^T d\bar{Y}_t    
        \end{equation}
        Plugging in the definitions for the linear system, we have
        \begin{align} \label{filtering:eq:zakai_density_form}
            dq_t(x) &=  \left[\frac{1}{2} \sum_{i,j=1}^n 
                            (B_tB_t^T)^{ij} \partialD{{^2}q_t(x)}{x^i\partial x^j}
                    - \sum_{i=1}^n\partialD{}{x^i} (A_t x)^i q_t(x) \right] dt\nonumber\\
                   &+  q_t(x)(D_t^{-1}C_t x)^T d\bar{V}_t                                      
        \end{align}
        We would then want to check that a density of the form
        \begin{equation}
            q_t(x) = N_t \exp\left(-\frac{1}{2}(x - \pi_t[X])^TP_t^{-1}(x - \pi_t[X])\right) ,
        \end{equation}
        where $N_t$ is a non-random normalization function, is a solution to Eq.~\eqref{filtering:eq:zakai_density_form}.  The check involves several applications of the It\^{o} rules followed by a comparison of terms.  The interested reader should feel free to check this for themself; the rest of us will have to take my word for it.
    \end{proof}
\end{theorem}

Unlike the non-linear filter, which estimates some function of the state, $\pi_t[f(X_t)]$, the Kalman-Bucy filter estimates the potentially multi-dimensional state itself, $\pi_t[X_t]$.  The form in Eq.~\eqref{filter:eq:kalman_bucy:state} has two important pieces.  A deterministic term propagates the state according to the dynamics induced by the linear map $A_t$.  Since this is a non-random term for the true state dynamics, we should not be surprised that the filter's estimate is simply the same dynamics applied to the estimated state.  The second term, which is proportional to the innovations process $d\bar{V}_t$, is responsible for conditioning and depends on the \emph{deterministic} covariance matrix $P_t$\footnote{The matrix $P_t(D_t^{-1}C_t)^T$ which multiplies $d\bar{V}_t$ is called the Kalman gain matrix by control theorists.}.  Remarkably, just from the structure of the linear system-observation pair, the appropriate weighting of the input signal is completely determined.  In another sense, our uncertainty in the estimate, given by the entries in $P_t$, is also completely determined by the structure of the linear system-observation pair---nothing in the observation causes us to change our certainty in the estimate.  This is a direct consequence of the Gaussianity of the stochastic processes and the linearity of the system.  Due to the nice transformation properties of Gaussians, we may trace the effect of the noise and initial state uncertainty through the dynamics and therefore know precisely how our uncertainty in $\pi_t[X]$ changes, weighting any updates due to the innovations process by that uncertainty.  Perhaps reassuringly, when the uncertainties in $P_t$ are large, we weight $d\bar{V}_t$ more heavily and when we are relatively sure of the estimate, the entries in $P_t$ are smaller and we weight the innovations less.  As an added practical benefit, the time evolution of the covariance matrix $P_t$ needs to be solved only once, using methods in Appendix \ref{appendix:riccati}, and the solution may be reused for each application of the filter.  The Kalman-Bucy filter is therefore a very practical tool for estimating the state of an $n$-dimensional linear system, requiring stochastic integration of the $n$-dimensional estimate $\pi_t[X]$ and standard integration of the distinct $\frac{n(n+1)}{2}$ elements in the symmetric covariance matrix $P_t$.

\begin{example}[Parameter estimation]\label{filtering:example:parameter_estimation}
    As an example use of the Kalman filter, consider the task of estimating the forcing parameter of a particle undergoing Brownian motion.  The general techniques used will serve as a useful basis for the research presented in Chapter \ref{chapter:quantum_parameter_estimation}.  We begin by letting $x_t$ represent the position of the particle and introduce the SDE
    \begin{equation} \label{filtering:example:parameter_estimation:system}
        dx_t = \xi dt + dW_t,
    \end{equation}
where $\xi$ is the forcing term we need to estimate.  Continuous measurements of the particle are given by the SDE
\begin{equation}
        dy_t =  x_t + dV_t .
\end{equation}
While we could go through the effort to calculate $\E{\xi | \mathcal{F}_t^Y}$ from first principles, a more clever approach is to leverage the fact that $\xi$ is a linear parameter in the dynamics and is thus amenable to the Kalman filter approach.  That is, we define the augmented system $X_t = [ x_t, \xi ]^T$, which gives rise to the linear systems-observations pair
\begin{align}
    dX_t &= AX_t + B dW_t\\
    dY_t &= CX_t + D dV_t
\end{align}
where
\begin{equation}
    A = \begin{pmatrix}
        0 & 1\\
        0  & 0
    \end{pmatrix}\qquad
    B = \begin{pmatrix}
        1 \\ 0
    \end{pmatrix}\qquad
    C = \begin{pmatrix}
        1 & 0
    \end{pmatrix}\qquad
    D = 1 .
\end{equation}
\begin{figure}[b]
    \centering
        \includegraphics[scale=0.75]{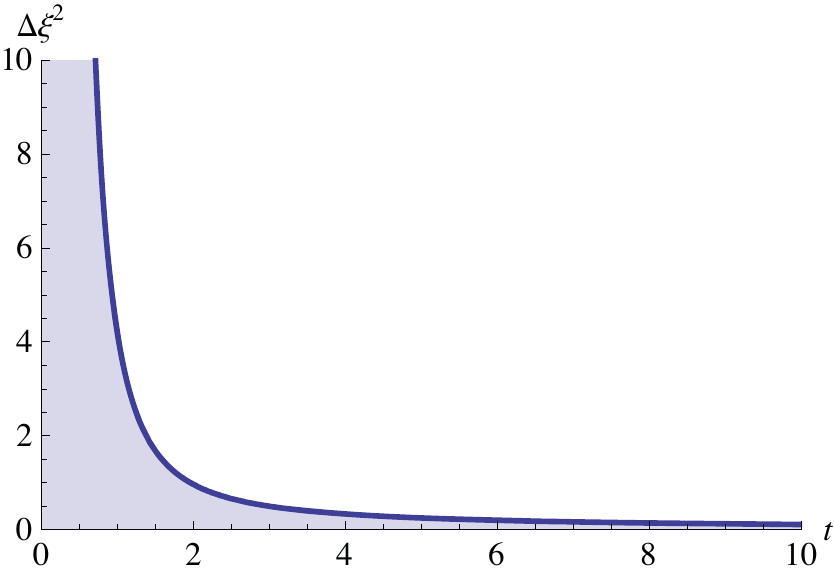}
    \caption[Plot of uncertainty in $\xi$ parameter for the Kalman parameter estimation in Example \ref{filtering:example:parameter_estimation}]{Plot of uncertainty in $\xi$ parameter for the Kalman parameter estimation in Example \ref{filtering:example:parameter_estimation} with $\Delta\xi_0^2 = 10^5$.}
    \label{fig:filtering:example:parameter_estimation:parameter_uncertainty}
\end{figure}
The covariance matrix
\begin{equation}
    P_t = \begin{pmatrix}
        \Delta x_t^2 & \Delta x_t \xi\\
        \Delta x_t \xi & \Delta \xi^2
    \end{pmatrix}
\end{equation}
admits an analytic solution using the techniques in Appendix \ref{appendix:riccati}.  Setting the initial $P_0 = \left(\begin{smallmatrix}
    0 & 0\\
    0 & \Delta \xi_0^2
\end{smallmatrix}\right)$, we find 
\begin{equation}
    P_t =  
     \begin{pmatrix}
     	\frac{1}{\coth{t}-\frac{\Delta \xi_0^2}{1+t \Delta \xi_0^2}} & 
     	 \frac{\Delta \xi_0^2}{\coth{t}-\Delta \xi_0^2+t \coth{t} \Delta \xi_0^2} \\
	     \frac{\Delta \xi_0^2}{\coth{t}-\Delta \xi_0^2+t \coth{t} \Delta \xi_0^2} &
	      \frac{\Delta \xi_0^2}{1+t \Delta \xi_0^2-\Delta \xi_0^2 \tanh{t}}
     \end{pmatrix}
\end{equation}
and the $\Delta\xi^2$ entry is plotted in Figure \ref{fig:filtering:example:parameter_estimation:parameter_uncertainty} for  $\Delta\xi_0^2 = 10^5$.  Ideally, we would want to take $\Delta\xi_0^2 \to \infty$ to reflect a complete uncertainty in $\xi$.  Doing so gives
\begin{equation}
    \lim_{\Delta \xi_0^2 \to \infty} P_t = \begin{pmatrix}
        \frac{t}{t\coth{t} - 1 } & \frac{1}{t\coth{t} - 1}\\
        \frac{1}{t\coth{t} - 1}  & \frac{1}{t - \tanh{t}}
    \end{pmatrix}
\end{equation}
which does not reduce to $P_0$ for $t = 0$.  This is because the infinite uncertainty in $\xi$ immediately washes out the certainty we had in $x_0$, since at the first time step, we have no clue what $\xi$ and $dW_t$ will do to the particle.  As such, knowing the initial position of the particle provides essentially no help in estimating the future position and forcing parameter when we have complete initial uncertainty in the parameter.

In order to test the filter, we use the numerical integration techniques in Appendix \ref{appendix:numerical_methods_for_stochastic_differential_equations} to integrate the dynamics of Eq.~\ref{filtering:example:parameter_estimation:system} for a \emph{known} value of $\xi$, say $\xi = 1$ .  Using this system, the measurement record for $dY_t$ is generated and fed into the filtering equation, which constructs the innovations process and provides an estimate of the parameter $\xi$ and the state $x_t$.  Figure \ref{fig:filtering:example:parameter_estimation:filter_performance} shows the performance of the filter for a single run with step-size $\Delta t = 10^{-3}$ and initial parameter uncertainty $\Delta \xi_0^2 = 10^5$.  The top plot shows the noisy measurement process $dY_t$, which is the only signal one gets experimentally.  The middle plot shows the true state $x_t$ and filtered state $\pi_t[x]$.  We see that after large initial fluctuations, the filter does a good job of latching on to the true particle position.  Similarly, the bottom plot shows large initial fluctuations in the estimate $\pi_t[\xi]$, as the filter has difficulty distinguishing forcing changes in the position due to $\xi$ versus changes due to the noise term $W_t$.  However, after this initial period, the Kalman filter quickly latches on to the true value $\xi = 1$ as was suggested by the deterministic uncertainty plotted in Fig.~\ref{fig:filtering:example:parameter_estimation:parameter_uncertainty}.  
\end{example}
\begin{figure}[bt]
    \centering
        \includegraphics[scale=1]{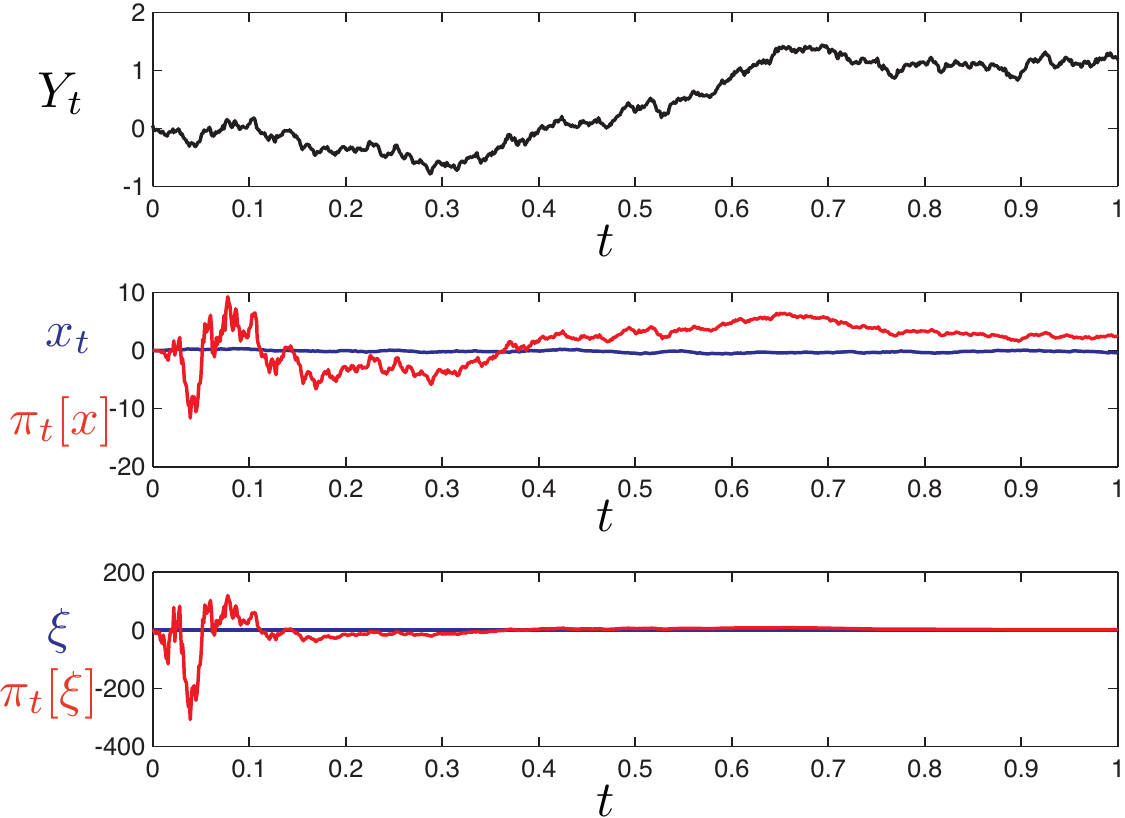}
    \caption[Plot of Kalman filter performance for parameter estimation in Example \ref{filtering:example:parameter_estimation}]{Plot of Kalman filter performance for parameter estimation in Example \ref{filtering:example:parameter_estimation}.  Top plot shows observations record $Y_t$.  Middle plot shows true state $x_t$ in blue and estimated state $\pi_t[x]$ in red.  Bottom plot shows true parameter value $\xi = 1$ in blue and estimate value $\pi_t[\xi]$ in red.  }
    \label{fig:filtering:example:parameter_estimation:filter_performance}
\end{figure}
\section{Summary}
Given such a whirlwind of a chapter, what are the take away points?  In a broad sense, I hope the exhausted reader is now convinced that analysis of continuous-time stochastic processes requires the use of rigorous mathematics, including axiomatic probability theory, measure theory and stochastic calculus.  But more importantly, I hope the reader is further convinced that one need not be an expert in these techniques to appreciate their necessity and to use the resulting formalism gained by such prudence.  Indeed, Example \ref{filtering:example:parameter_estimation} was meant to show how easy it is to apply these techniques to solve a ``real world'' inference problem.  Similarly, all the rigamarole that went into defining Gaussian white noise relative to the Wiener process and constructing stochastic processes in terms of the It\^{o} integral can be safely placed on the shelf; mindless applications of the It\^{o} rule and straightforward composition of stochastic differential equations are all we need to apply our techniques in practice.  I also hope the reader appreciates the power one gains by developing a clear mathematical framework, particularly with regard to filtering and, although not mentioned here, the filter's use for optimal control of stochastic systems \citep{Liptser:1977,Zhou:1996a}. 
\chapter{Quantum Probability and Filtering}
\label{chapter:quantum}

Most modern formulations of quantum mechanics present the theory in terms of the following postulates, here adapted from \citep{Nielsen:2000a}.
\begin{itemize}
    \item The \emph{state} of a \emph{pure} quantum system is completely described by a normalized vector $\ket{\psi}$ in a complex Hilbert space $\mathcal{H}$.  A statistical ensemble of pure states $\ket{\psi_j}$, with probabilities $p_j$, is called a \emph{mixed state} and is written as the \emph{density matrix} $\rho = \sum_j p_j\ketbra{\psi_j}{\psi_j}$.
    \item The \emph{time evolution} of a quantum system is described by a unitary operator $U_t$ and acts as
    $\ket{\psi_t} = U_t\ket{\psi_0}$ for pure states and $\rho_t = U_t\rho_0U_t^{\dag}$ for mixed states.
    \item Physical \emph{observations} are described by self-adjoint, linear operators on $\mathcal{H}$ with eigenvalues $\lambda_j$ and eigenprojectors $P_j$.  The probability of measuring outcome $\lambda_j$ is given by the \emph{Born rule}---$\bra{\psi}P_j\ket{\psi}$ for pure states and $\Tr{\rho P_j}$ for mixed states.
    \item Given a particular measurement outcome $j$, the conditioned state is determined via the \emph{projection postulate}, 
    \begin{equation} \label{quantum:eq:born_rule}
       \begin{split}
       	\psi' &= \frac{P_j\ket{\psi}}{\sqrt{\bra{\psi}P_j\ket{\psi}}} \qquad\text{ for pure states, }\\
	    \rho' &= \frac{P_j\rho P_j}{\Tr{\rho P_j}} \qquad\text{ for mixed states.}
       \end{split}
    \end{equation}
    \item The state of a \emph{composite} quantum system is described by the tensor product of the constituent systems, $\ket{\psi^{(1)}}\otimes \ket{\psi^{(2)}}\otimes \cdots$ for pure states and $\rho^{(1)}\otimes\rho^{(2)}\otimes\cdots$ for mixed states.
\end{itemize}
Nascent in these postulates are rudimentary features of probability theory.  Measurement outcomes are described by probabilities, which are assigned via the quantum state, much as the probability measure $\mathbb{P}$ assigns probabilities to elements in the $\sigma$-algebra, or by extension, to the potential values of random variables.  Similarly, the conditioning provided by the projection postulate is analogous to conditional expectation in probability theory.  As we turn towards solving the quantum filtering problem, in which we perform inference on the state of a quantum system conditioned on continuous measurements of that system, it would be natural to leverage the techniques we developed in solving the classical filtering problem.  But the exposition in the last chapter should have convinced you that care must be taken in developing a mathematically well-posed probability theory, filtering problem and solution.

\begin{figure}[b]
    \centering
        \includegraphics[scale=1]{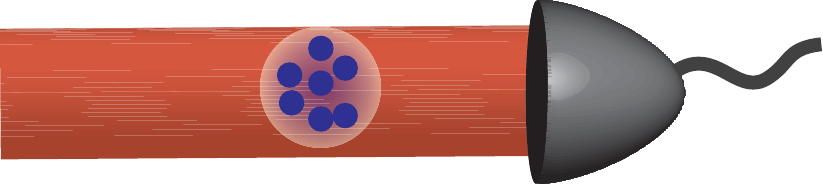}
    \caption[]{Schematic of continuous measurement in quantum optics, in which light scattered
    by a cloud of atoms is continuously measured by a photodetector. }
    \label{fig:quantum:filtering_setup}
\end{figure}

As such, the first section of this chapter reviews quantum probability theory, stressing its differences with the classical theory developed in Chapter \ref{chapter:classical}.  This will make the inchoate features noted above more precise and allows us to interpret the projection postulate as a consequence of conditional expectation rather than as a postulate.  In so doing, we will also find how the distinctly quantum possibility of non-commuting observables limits our ability to condition, which in turn will help formulate the quantum filtering problem.  The second section focuses on quantum stochastic processes, particularly the quantum analog of the Wiener process which we will relate to quadratures of the quantized electric field when in a vacuum or coherent state.  With those tools in hand, we will then solve the quantum filtering problem of quantum optics, depicted in Fig.~\ref{fig:quantum:filtering_setup}, where an optical field is scattered by a cloud of atoms.  Continuous measurements on the light correspond to an observations process which may be filtered to learn about the atomic system.  The exposition in this chapter closely follows \citep{Bouten:2007b}, with added perspective from \cite{Barchielli:2003,Geremia:2008a,Kummerer:1998a,vanHandel:2005a,Accardi:2002a}. 

\paragraph{A word on notation} 
\label{quantum:par:about_notation}
I will be cavalier about placing ``hats'' on operators in this section, as context tends to make that clear and I find $O$ more visually pleasing than $\hat{O}$.  On occasions where confusion may ensue, I will use them. 

\section{Quantum Probability Theory}
Quantum probability theory is the non-commutative generalization of Kolmogorov's axiomatic probability theory.  Just as in the classical case, subsuming discrete and continuous theories within a general measure-theoretic framework will provide an abstraction capable of carefully dealing with the filtering problem.  But unlike the case of classical probability theory, we do not start with an obvious ``intuitive'' theory of discrete quantum probability.  Consequently, we begin this section by studying finite-dimensional quantum systems, where we can focus on the essential ingredients of quantum probability.  After that, we can extend our definitions to infinite-dimensional systems by dealing with the subtleties of functional analysis and measure theory.  
\subsection{Quantum Probability for Discrete Systems}
Let us fix $\mathcal{H} = \mathbb{C}^n$, an $n$-dimensional, complex vector space.  Observables in this space are self-adjoint linear operators ${A} = {A}^{\dag}$, which may be represented as $n\times n$ complex matrices.  From the spectral theorem, we know that a given observable ${A}$ can be diagonalized as
\begin{equation}
    {A} = \sum_i a_i {P}_{a_i},
\end{equation}
where $a_i \in \mathbb{R}$ satisfies the eigenvalue relation
\begin{equation}
    {A}\ket{a_i} = a_i\ket{a_i}
\end{equation}
for the eigenvector $\ket{a_i}$ and associated projector ${P}_{a_i} = \ketbra{a_i}{a_i}$.  From the postulates of quantum mechanics, we know that the probability of observing a particular outcome $a_i$ when in the state $\rho$ is $\Tr{P_{a_i}\rho}$.  Clearly, ${A}$ is a lot like a random variable, in that it relates a particular value $a_i$ to a particular event, $P_i$.  Indeed, the spectral decomposition is essentially identical to the decomposition of random variables in terms of indicator functions we considered in Eq.~\ref{probability:eq:indicator_partition}.  We therefore see that the set of projectors $\{P_{a_i}\}$ is much like $\mathcal{F}^X$, the set of events generated by some random variable $X$.  Similarly, the linear map $\mathbf{P}(P_{a_i}) = \Tr{P_{a_i}\rho}$ is the measure or state which assigns probabilities to those events.  It is important to note that this relation is clearest in the Heisenberg picture, where the state remains fixed and the observables change in time.  This is in analogy to stochastic processes, which change in time relative to a fixed probability measure.

Things get a bit more complicated if we want to describe joint probabilities for two different events.  Classically, we simply have sets $F_1,F_2 \in \mathcal{F}$, so the joint probability for the two events is $\mathbb{P}(F_1 \cap F_2) = \mathbb{E}[\chi_{F_1}\chi_{F_2}]$.  In quantum mechanics, we consider projectors ${P}_{a_i},{P}_{b_i}$ for two different observables ${A},{B}$.  We then hope that the joint probability of observing outcome $a_i$ and $b_i$ is $\mathbf{P}[{P}_{a_i}{P}_{b_i}] = \Tr{{P_{a_i}}{P_{b_i}}\rho}$.  Yet, ${A}$ and ${B}$ will not commute in general, so that the joint probability calculation depends on the \emph{order} of the projectors involved.  But this is entirely contrary to what we mean by a joint probability, which is equivalent to the yes/no question ``Did outcome $a_i$ and outcome $b_i$ occur?''.  Surely this must be the same as the question ``Did outcome $b_i$ and outcome $a_i$ occur?''.  However, we simply cannot pose this question unambiguously in quantum mechanics.  This is no surprise really, as in a given experiment, we cannot ascribe underlying values to all observables consistently; i.e. there is no (local) hidden variable model for the system.  More concretely, if given a quantum spin, there is no sensible way to describe the event that the $x$ and $y$ projections take on specific values simultaneously\footnote{Note that we are talking about projective measurements on a single system, not generalized measurements which might allow for imprecise, but simultaneous, measurements of non-commuting observables.  Such measurements will fit within the quantum probability formalism by explicitly accounting for the auxiliary systems needed to perform them.}.

The ``incompatibility'' of non-commuting quantum events is really the only departure from classical probability theory.  In essence, it states that for a single experimental realization, we may only speak sensibly about a set of commuting observables or events; all other non-commuting events are incompatible with the experiment under consideration and it makes no sense to discuss their probabilities.  Thus, our first step in constructing a quantum probability space is to fix our a set of commuting observables in a mathematically well-defined structure.
\begin{definition}\label{def:star_algebra}
    A \emph{*-algebra} $\mathscr{A}$ is a set of operators closed under arbitrary complex-linear combinations, products and adjoints of its members and contains the identity operator.  A \emph{commutative *-algebra} is a *-algebra whose elements all commute.
\end{definition}
\noindent As was the case classically, it will often be useful to consider generating such a set from a particular observable ${A}$.
\begin{definition}\label{def:generated_star_algebra}
    Given an operator ${A}$, the set $\mathscr{A} = \{{X} : {X} = f({A}), f: \mathbb{R}\mapsto \mathbb{C}\}$ is the smallest commutative *-algebra generated by ${A}$. 
\end{definition}
The generated *-algebra captures the structure of compatible observations, in that given the spectral decomposition of the observable of interest, ${A}$, we may directly calculate any observable $f({A}) \in \mathscr{A}$ as
\begin{equation} \label{quantum:eq:spectral_decomposition_discrete}
    A = \sum_{i} a_i {P}_{a_i} \implies f({A}) = \sum_i f(a_i)P_{a_i} .
\end{equation}
Thus, if we measure outcome $a_i$ we immediately \emph{know} the outcome for any compatible observation, specifically $f(a_i)$, up to any degeneracies in the eigenspectrum.  It is therefore the eigenspace, represented by the label $i$, which truly characterizes compatible observables, where the actual value $a_i$ is just there to give us the correct units.  As we will soon see, this is enough to develop most of a corresponding classical probability space.  The only remaining ingredient is to formalize the measure for the space, as given in the following defintion.
\begin{definition}\label{def:state_quantum}
    A \emph{state} on a *-algebra is the linear map $\mathbf{P} : \mathscr{A} \mapsto \mathbb{C}$ which is positive, ${A} \geq 0 \implies \mathbf{P}({A}) \geq 0$ and normalized $\mathbf{P}(I) = 1$.  Note that one can always write this as $\mathbf{P}({A}) = \Tr{{A}\rho}$ for some density matrix $\rho$.
\end{definition}
\noindent  We now have all the ingredients necessary to map a given commutative $*$-algebra and state into a corresponding classical probability space.
\begin{theorem}[\textbf{Spectral Theorem, Finite Dimensions, (Adapted from Theorem 2.4 in \citep{Bouten:2007b})}]\label{thm:discrete_spectral_theorem}
    Let $\mathscr{A}$ be a commutative *-algebra on a finite-dimensional Hilbert space and let $\mathbf{P}$ be a state on $\mathscr{A}$.  Then there exists a probability space $(\Omega,\mathcal{F},\mathbb{P})$ and a linear, bijective map $\iota$ from elements of $\mathscr{A}$ to measurable functions on $\Omega$ such that $\iota({A}{B}) = \iota({A})\iota({B})$ and $\iota({A}^{\dag}) = \iota({A})^*$ and the probability measure is determined by $\mathbf{P}({A}) = \mathbb{E}_{\mathbb{P}}[\iota({A})]$.
\begin{proof}
    We will simply construct the probability space by hand, taking care to formalize the intuitive relations between projectors and events discussed above.  To begin, given that $\mathscr{A}$ is commutative, we may simultaneously diagonalize each $n \times n$ matrix ${A} \in \mathscr{A}$; for convenience, suppose that each ${A}$ is already diagonal with entries $A_{ii}$.  Then set $\Omega = \{1,\ldots,n\}$, so that $\omega \in \Omega$ serve as labels for the different eigenspaces.  Define the map $\iota({A}) : \Omega \mapsto  \mathbb{C}$ by $\iota({A})(i) = A_{ii}$.  Thus the map $\iota$ takes operators in $\mathscr{A}$ to random variables on the dummy sample space $\Omega$.  Each random variable $\iota({A})$ just takes on the appropriate eigenvalue of ${A}$ when given the eigenspace label $\omega \in \Omega$.  We then generate the $\sigma$-algebra as $\mathcal{F} = \{ \iota({A}) : {A} \in \mathscr{A}\}$ and define the probability measure via $\mathbb{P}(F) = \mathbf{P}(\iota^{-1}(\chi_F))$ for $F \in \mathcal{F}$.
\end{proof}
\end{theorem}
Thus, a commutative *-algebra and quantum state are \emph{equivalent} to a classical probability space.  Once restricted to a commuting set of observables, there is nothing particularly quantum left to worry about.  Of course, we will want to consider a variety of experimental realizations, in which on each trial we might study different observables which do not commute.  This generalization suggests the following definition of a finite-dimensional quantum probability space.
\begin{definition}\label{def:quantum_probability_space_finite}
    A \emph{finite-dimensional quantum probability space} is the pair $(\mathscr{N},\mathbf{P})$, where $\mathscr{N}$ is a $*$-algebra of operators on a finite-dimensional Hilbert space and $\mathbf{P}$ is a state on $\mathscr{N}$.
\end{definition}
\noindent Note that unlike a classical probability space, there is no sample space in the quantum setting; the corresponding classical space simply inherits an $\Omega$ passively through the eigenspace labels.  For the $n$-dimensional space $\mathcal{H}$, we tend to take $\mathscr{N}$ to be the set of all bounded operators on that space, written $\mathscr{B}(\mathcal{H})$.  For a given experimental setup, one selects the commutative sub-*-algebra $\mathscr{A} \subset \mathscr{N}$ relevant for the observations we intend to make.  Using Theorem \ref{thm:discrete_spectral_theorem}, one can then construct the corresponding classical probability space and calculate a variety of statistics using techniques from the previous chapter.  
\begin{example}[Example 2.6 in \citep{Bouten:2007b}] \label{quantum:example:basic_qubit_space}
As a concrete example, consider a single spin-1/2 particle or qubit, which has Hilbert space $\mathcal{H} = \mathbb{C}^2$.  The $*$-algebra of operators may be expanded as
     \begin{equation}
    	 \mathscr{N} = \{\alpha_0 I + \alpha_1 \sigma_x + \alpha_2 \sigma_y + \alpha_3 \sigma_z : \alpha_i \in \mathbb{C}\}
     \end{equation}
     where the Pauli matrices are given by
    \begin{equation}
        I = \begin{pmatrix}
            1 & 0 \\
            0 & 1
        \end{pmatrix} \quad
        \sigma_x = \begin{pmatrix}
            0 & 1 \\
            1 & 0
        \end{pmatrix} \quad
        \sigma_y = \begin{pmatrix}
            0 & -i \\
            i & 0
        \end{pmatrix} \quad
        \sigma_z = \begin{pmatrix}
            1 & 0 \\
            0 & -1
        \end{pmatrix} .
    \end{equation}
    To round out the quantum probability space, we consider the pure qubit state pointing up along $x$, written $\ket{+x} = \frac{1}{\sqrt{2}}\left(\begin{smallmatrix} 1\\1 \end{smallmatrix}\right)$ in the standard basis, so that the quantum probability state is $\mathbf{P}(A) = \bra{+x}A\ket{+x}$.  This completes the quantum probability space $(\mathscr{N},\mathbf{P})$.
    
    In order to apply the spectral theorem, we select the commutative sub-algebra $\mathscr{A}$ generated by the observable $\sigma_z$.  Admittedly, there aren't really many other interesting observables in this commutative algebra, but we can still work through the quantum probability formalism.  Since $\sigma_z$ is already diagonal as written, we read off the two-eigenvalues $\pm z = \pm 1$ and projectors
    \begin{equation}
        P_{+z} = \begin{pmatrix}
            1 & 0\\ 0 & 0
        \end{pmatrix}
        \qquad
        P_{-z} = \begin{pmatrix}
            0 & 0 \\ 0 & -1
        \end{pmatrix}
    \end{equation}
    Applying the spectral theorem, we introduce $\Omega = \{1,2\}$ and $\mathcal{F}=\{\varnothing, \{1\},\{2\}, \Omega\}$.  Since observables in $\mathscr{A}$ are of the form $\alpha P_{+z} + \beta P_{-z}$ for $\alpha,\beta \in \mathbb{C}$, we simply need to know how $\iota$ acts on the projectors.  This is simply $\iota(P_{+}) = \chi_{\{1\}}$ and $\iota(P_-) = \chi_{\{2\}}$.  We then see that, for example,  $\mathbb{P}(\{1\}) = \mathbf{P}(\iota^{-1}(\chi_{\{1\}}) = \bra{+x}P_{+z}\ket{+x} = 1/2$ as expected.
\end{example}

The quantum probability formalism will also allow us to calculate conditional expectations, in which we determine the expected value of a future measurement outcome given a current measurement outcome.  Clearly such an expectation only makes sense when the two measurements are compatible, otherwise there would never be an experiment in which we could even in theory attempt to assign observed values to each measurement simultaneously.  Yet, this may appear troubling at first.  For example, consider a spin-1/2 particle, on which we seek to condition a $\sigma_y$ measurement given a $\sigma_z$ measurement.  Although these observables do not commute, it appears completely sensible to calculate a future expected $\sigma_y$ measurement given a $\sigma_z$ outcome.  Indeed, we know it to be precisely zero, since the quantum state is in one of the two $\sigma_z$ eigenstates after the $\sigma_z$ measurement and both eigenstates have zero $\sigma_y$ expectation.  We clearly have a consistent way to describe observed values for these non-commuting observables, so how do we reconcile this with the limitations imposed by quantum probability theory? 

It is actually straightforward if we carefully consider what conditional expectation means in this context.  Classically for two events $A,B$, the conditional probability of $B$ given $A$ is the probability that $B$ is true given that $A$ is also true in the same realization.  For the spin under consideration, a naive statement of conditional expectation corresponds to the current expected $y$-projection value of the spin given that it also currently has a particular $z$-projection value.  We know that this is not sensible from fundamental quantum uncertainty, as the spin cannot have perfectly defined $\sigma_z$ and $\sigma_y$ values at the same time.  However, it is more likely that we meant to consider the conditional expectation which corresponds to the expected $\sigma_y$ measured value conditioned on a previous $\sigma_z$ measurement.  But this means that the expected $\sigma_z$ value is actually written down somewhere and in order to sensibly talk about performing \emph{both} measurements, we really need to include this other physical system which was used to measure the spin indirectly.  This corresponds to including a physical model of the measurement apparatus or probe system used to perform the indirect $\sigma_z$ measurement in our quantum probability model.  After all, in an experiment there is some physical process by which we learn the direction of the spin, perhaps by coupling the position of the particle to its spin state via a Stern-Gerlach device, after which the position tells us about the spin state.  By including such extra quantum degrees of freedom explicitly, we can then pose the measurement of $\sigma_z$ as an indirect measurement on an auxiliary space, which will then commute with direct $\sigma_y$ measurements on the spin\footnote{Perhaps this seems like only sidestepping the issue, as one can always question why one measurement is considered direct whereas the other is considered indirect.  Moreover, how do we measure the position of the spin after it goes through the Stern-Gerlach device?  Isn't that just another measurement that also requires a physical measurement model?  I agree that the so-called Heisenberg chain of measurements is unsettling, but the issues are more philosophical than practical.  At some point, perhaps all the way to the neurons in our brain, we will assume that a projective measurement happens.  For the sake of being able to consider conditional expectation and inference within the quantum probability setting, it will be sufficient to consider projective measurements only one level away, on the probe system, which could include the entire universe save the primary quantum system if that is more comforting.  } .  

Continuing along then, we see that conditional expectation can be posed sensibly if we include the measurement model within the quantum probability space.  We therefore define the conditional expectation by first selecting the commutative sub-algebra $\mathscr{A} \subset \mathscr{N}$ which represents the measurement we will condition upon.  Then there is some other set $\mathscr{A}' = \{{B} \in \mathscr{N} : {A}{B} = {B}{A}\quad \forall\quad {A} \in \mathscr{A}\}$ called the \emph{commutant} which represents the set of observables which can be simultaneously diagonalized with any ${A} \in \mathscr{A}$.  For some ${B} \in \mathscr{A}'$, the conditional expectation is then inherited from the corresponding classical probability space as $\mathbf{P}(B|\mathscr{A}) = \iota^{-1}(\mathbb{E}_{\mathbb{P}}(\iota({B})|\sigma\{\iota(\mathscr{A})\}))$.  It is important to note that elements in $\mathscr{A}'$ need not commutate with each other, just as they need not be in $\mathscr{A}$ directly.  Physically, the elements in $\mathscr{A}$ are the commutative set of observables on the probe system and elements in $\mathscr{A}'$ are the observables on the initial quantum system, which trivially commute with members in $\mathscr{A}$ but not necessarily each other.  The example at the end of this section should help clarify these different *-algebras.

Although this is enough to perform calculations, one would hope that the abstract mapping between quantum and classical in Theorem \ref{thm:discrete_spectral_theorem} would allow us to calculate the conditional expectation without explicitly working through the $\iota$ mapping.  This turns out to be possible, especially in light of the least-squares projection interpretation of conditional expectation.  The finite-dimensional *-algebra is actually a finite dimensional linear linear space with the Hilbert-Schmidt inner product\footnote{Again, it is actually not quite enough to be a norm, as $\normsq{{A}} = \langle {A},{A}\rangle$ may be zero even if ${A}$ is not the zero operator.} $\langle {A}, {B} \rangle = \mathbf{P}({A}^{\dag}{B})$.  The conditional expectation is then precisely the orthogonal projection from $\mathscr{A}'$ onto the linear subspace $\mathscr{A}$.  We can expand this projection easily in terms of an orthogonal basis for $\mathscr{A}$, which from the spectral theorem is simply the set of eigenprojectors of $\mathscr{A}$.  We then have
\begin{equation} \label{quantum:eq:discrete_conditional}
    \mathbf{P}(B | \mathscr{A}) = \sum_{i;\mathbf{P}({P}_{a_i}) \neq 0}
                            \frac{{P}_{a_i}}{\norm{{P}_{a_i}}_{\mathbf{P}}}
                            \left\langle  \frac{{P}_{a_i}}{\norm{{P}_{a_i}}_{\mathbf{P}}}
                            , {B}\right\rangle
                            = \sum_{i;\mathbf{P}({P}_{a_i}) \neq 0}
                            \frac{\mathbf{P}({P}_{a_i}{B})}
                                 {\mathbf{P}({P}_{a_i})} {P_{a_i}}
\end{equation}
which looks exactly like our explicit formula for discrete conditional expectations in Eq.~\eqref{probability:eq:discrete_conditional}.  Similar to what we saw in that equation, the conditional expectation is an \emph{operator} on $\mathscr{A}$ and we see that the weighting factors in that basis, given by $  \mathbf{P}({P}_{a_i}{B})/\mathbf{P}({P}_{a_i})$, are the expected values of ${B}$ restricted to that eigenspace.  Note that if ${B} \not\in \mathscr{A}'$, the inner product would depend on the order of its arguments and would in general give a complex coefficient in the sum even if ${B}$ were an observable. 

Before attempting to extend these definitions to infinite-dimensional spaces, we close this section with a physical example which will hopefully clarify the above definitions. 
\begin{example}[Based on Example 2.9 in \citep{Bouten:2007b}]
    We work with the qubit system introduced in Example \ref{quantum:example:basic_qubit_space}, but here consider conditioning a $\sigma_z$ measurement on an initial $\sigma_x$ measurement.  As we just found in developing the conditional expectation, since $[\sigma_z,\sigma_x] \neq 0$, we need to introduce an auxiliary probe system in order to discuss conditioning the measurement.  As such, we introduce another qubit system, with quantum probability space $( \mathcal{N}_p, \mathbf{P}_p )$, so that the joint space is $(\mathcal{N}_s \otimes \mathcal{N}_p, \mathbf{P}_s \otimes \mathbf{P}_p)$, where the subscripts stand for system and probe.  Our measurement procedure should work for any system state (afterall, the point of measuring is to learn something we don't know), so it is described by the arbitrary density matrix $\rho_s$.  Conversely, the probe must start in a known fiducial state, here $\ket{+z}$, so that any changes in its state reflect information about the system, thus $\mathbf{P}_p({A}) = \Tr{{A} P_{+z}}$.
    
    Now suppose we are only capable of performing $\sigma_z$ measurements.  Therefore, in order to perform the indirect $\sigma_x$ system measurement using the probe qubit, we must find a unitary $U$ such that measuring $U^{\dag}(I \otimes \sigma_z)U$ gives the same statistics as measuring $\sigma_x \otimes I$ would on the system prior to the interaction.  Also note that the future direct $\sigma_z$ measurement on the spin will then commute with this indirect measurement, i.e. $[U^{\dag}(I \otimes \sigma_z)U, U^{\dag}(\sigma_z \otimes I)U] = 0$, so that $U^{\dag}(\sigma_z \otimes I)U$ is in the commutant of $U^{\dag}(I \otimes \sigma_z)U$ and the conditional expectation is well-defined. 
    
    Following a general procedure in Example 2.9 in \citep{Bouten:2007b}, we construct the unitary
    \begin{equation}
        U = P_{+x} \otimes I + P_{-x} \otimes \sigma_x,
    \end{equation}
    where 
    \begin{equation}
        P_{\pm x} = \ketbra{\pm x}{\pm x} = \frac{1}{2}\begin{pmatrix}
            1 & 1\\
            1 & \pm 1
        \end{pmatrix}
        \qquad
        \sigma_x = \ketbra{+z}{-z} + \ketbra{-z}{+z} .
    \end{equation}
    We now verify explicitly that measuring $\pm z$ on the probe qubit occurs with the same probabilities as measuring $\pm x$ on the initial system qubit.  The probability of measuring $+z$ is given by
    \begin{align}
        \mathbf{P}_s\otimes\mathbf{P}_p(U^{\dag}(I \otimes P_{+z})U)
               &= \mathbf{P}_s\otimes\mathbf{P}_p( P_{+x} \otimes P_{+z} + P_{-x} \otimes P_{-z})\\
               &= \mathbf{P}_s(P_{+x})\mathbf{P}_p(P_{+z}) + \mathbf{P}_s(P_{-x})
                        \underbrace{\mathbf{P}_p(P_{-z})}_{=0}\\
               &= \mathbf{P}_s(P_{+x}) \label{quantum:example_discret:ref}
    \end{align}
    where the particular initial probe state $\ket{+z}$ implies $\mathbf{P}_p(P_{-z}) = \Tr{P_{-z}P_{+z}} = 0$. Similarly, the probability for measuring $-z$ is given by
    \begin{align}
        \mathbf{P}_s\otimes\mathbf{P}_p(U^{\dag}(I \otimes P_{-z})U)
               &= \mathbf{P}_s\otimes\mathbf{P}_p( P_{+x} \otimes P_{-z} + P_{-x} \otimes P_{+z})\\
               &= \mathbf{P}_s(P_{-x})
    \end{align}
    so that the probabilities correspond as desired. 
    
    Given $U$, we may now consider the conditional expectation.  We set $\mathscr{A}$ as the commutative *-algebra generated by the probe measurement $U^{\dag}(I \otimes \sigma_z) U$ so that $U^{\dag}(\sigma_z \otimes I)U \in \mathscr{A}'$ as desired.  From Eq.~\ref{quantum:eq:discrete_conditional}, we find
    \begin{align}
        &\mathbf{P}_s\otimes\mathbf{P}_p(U^{\dag}( \sigma_z \otimes I)U | \mathscr{A})\\
         &= \sum_{a = \pm z} \frac{\mathbf{P}_s\otimes\mathbf{P}_p
                        (U^{\dag}(\sigma_z \otimes I )UU^{\dag}(I \otimes P_{a})U)}
                        {\mathbf{P}_s\otimes\mathbf{P}_p(U^{\dag}(I \otimes P_{a})U)}
                        U^{\dag}(I \otimes P_a)U\\
        &= \sum_{a = \pm z} \frac{\mathbf{P}_s\otimes\mathbf{P}_p
                       (U^{\dag}(\sigma_z \otimes P_{a})U)}
                       {\mathbf{P}_s\otimes\mathbf{P}_p(U^{\dag}(I \otimes P_{a})U)}
                       U^{\dag}(I \otimes P_a)U                                             
    \end{align}
    Without a loss of generality, lets consider one of the conditional probability terms in this sum, say for $a = +z$, the $\mathbf{P}_s\otimes\mathbf{P}_p
                   (U^{\dag}(\sigma_z \otimes P_{+z})U)/
                   \mathbf{P}_s\otimes\mathbf{P}_p(U^{\dag}(I \otimes P_{+z})U)$ factor.  We know from Eq.~\eqref{quantum:example_discret:ref} that the denominator is simply the probability for the system qubit to be measured in $+x$,i.e.~$\mathbf{P}_s(P_{+x})$.  Focusing on the numerator, we find
    \begin{multline}
        \mathbf{P}_s\otimes\mathbf{P}_p(U^{\dag}(\sigma_z \otimes P_{+z})U) =
         \mathbf{P}_s\otimes\mathbf{P}_p(P_{+x}\sigma_z P_{+x} \otimes P_{+z})
        +\mathbf{P}_s\otimes\mathbf{P}_p(P_{-x}\sigma_z P_{+x} \otimes \sigma_x P_{+z})\\
        +\mathbf{P}_s\otimes\mathbf{P}_p(P_{+x}\sigma_z P_{-x} \otimes P_{+z} \sigma_x )
        + \mathbf{P}_s\otimes\mathbf{P}_p(P_{-x}\sigma_z P_{-x} \otimes P_{-z} )
    \end{multline}
But since $\mathbf{P}_p(P_{-z}) = \mathbf{P}_p(\sigma_x P_{+z}) = \mathbf{P}_p(P_{+z}\sigma_x) = 0$ and $\mathbf{P}_p(P_{+z}) = 1$, only the first term survives.  A similar calculation holds for the $a = -z$ term in the sum, so that the conditional expectation is
\begin{equation}
    \begin{split}
   	 \mathbf{P}_s\otimes\mathbf{P}_p(U^{\dag}( \sigma_z \otimes I)U | \mathscr{A}) &= 
	       \frac{\mathbf{P}_s(P_{+x}\sigma_z P_{+x})}{\mathbf{P}_s(P_{+x})} U^{\dag}(I \otimes P_{+z})U\\
	  & + \frac{\mathbf{P}_s(P_{-x}\sigma_z P_{-x})}{\mathbf{P}_s(P_{-x})} U^{\dag}(I \otimes P_{-z})U .
    \end{split}   
\end{equation}
Recalling that $\mathbf{P}_s({A}) = \Tr{{A}\rho_s}$, we introduce the conditioned density matrices $\rho_{\pm x}=  P_{\pm x} \rho_s P_{\pm x} / \Tr{P_{\pm x}\rho_s}$, so that we may further simplify our expression to
\begin{equation}
       	 \mathbf{P}_s\otimes\mathbf{P}_p(U^{\dag}( \sigma_z \otimes I)U | \mathscr{A}) = 
       	    \Tr{\rho_{+x} \sigma_z} U^{\dag}(I \otimes P_{+z})U
       	  + \Tr{\rho_{-x} \sigma_z} U^{\dag}(I \otimes P_{-z})U
\end{equation}
We see that the conditional expectation is a diagonal observable in $\mathcal{A}$, where the eigenvalues associated with each outcome of the probe measurement are precisely the conditional probabilities one finds using the Born rule!   That is, once the probe measurement determines whether outcome $U^{\dag}(I \otimes P_{+z})U$ or $U^{\dag}(I \otimes P_{-z})U$ occurs, this conditional observable immediately reduces to the corresponding expected value of $\sigma_z$ for the conditioned qubit system state.  What is perhaps remarkable, is that the Born rule is then a consequence of conditional expectation, which is \emph{not} a axiomatic definition, but a derived one following the Radon-Nikodym approach and using the least-squares criterion.  This is in contrast to the quantum case, where the Born rule is assumed axiomatically.
\end{example}
\subsection{Quantum Probability Spaces}
The task of developing a general quantum probability theory which describes both finite and infinite dimensional spaces is fraught with the same difficulties we faced in developing a general classical probability theory, but now the infinities can confound us in two ways---issues related to simply describing infinite dimensional quantum systems and issues related to describing infinite dimensional probability spaces.  For the former case, this means the relatively straightforward linear algebraic tools in the previous section must be promoted to more sophisticated functional analysis tools.  For the latter, we again will use methods of measure theory.

We begin by considering a complex Hilbert space $\mathcal{H}$, which may be finite or infinite dimensional.  We further consider $\mathscr{B}(\mathcal{H})$, the set of bounded, linear operators on $\mathcal{H}$.  By restricting consideration to bounded operators for the time being, we can avoid some details which are better handled after introducing the quantum probability space.  As is familiar for quantum systems, the Hilbert space adjoint of an operator $A \in \mathscr{B}(\mathcal{H})$ is written $A^{\dag}$ and is defined by $\bra{\psi}(A\ket{\phi}) = (\bra{\psi}A^{\dag})\ket{\phi}$ for all $\ket{\psi},\ket{\phi} \in \mathcal{H}$.  Given that $\mathscr{B}(\mathcal{H})$ is already a Hilbert space (a complex vector space with norm given by the trace inner product) with operator multiplication, it is an algebra.  Adding in the adjoint operation via $\dag$ makes $\mathscr{B}(\mathcal{H})$ a *-algebra by Definition \ref{def:star_algebra}.  

One would hope that a *-algebra defines a suitable set of operators for a quantum probability space, but this is not true for infinite-dimensional systems.  In particular, we are faced with issues of convergence of a sequence of such operators, which is important for defining quantum probability operations as a limit of sequences of simple operators.  The problem is that there are multiple types of convergence which induce different topologies on $\mathscr{B}(\mathcal{H})$.  Consider a sequence of operators $\{T_n\}$ on $\mathcal{H}$.  By stating that $T_n$ converges to $T$, we could mean that $\norm{T_n -T} \mapsto 0$, where the norm is induced via the Hilbert-Schmidt, trace inner-product norm on the $\mathscr{B}(\mathcal{H})$ Hilbert space.  We could instead mean that $T_n \ket{\psi} \mapsto T\ket{\psi}$ for any $\ket{\psi} \in \mathcal{H}$ or that $\mu(T_n \ket{\psi}) \mapsto \mu(T \ket{\psi})$ for all linear functions $f : \mathcal{H} \mapsto \mathbb{C}$.  A plethora of different topologies defined relative to different convergences exists for sequences\footnote{For topological spaces, we really consider the generalization of sequences called ``nets'', which is a function from a directed set to the topological space.  Sequences are essentially nets where the directed set is the natural numbers.  Generalizing to nets allows one to consider convergence in topological spaces which are are not ``first-countable'', lacking a countable neighborhood basis for elements in the space. I'm already way out of my league on this one, so I defer to textbooks on topology for the real details.} in $\mathscr{B}(\mathcal{H})$.   The following definition classifies the particular topology useful for defining a quantum probability space.
\begin{definition}\label{def:normal_topology}
    Consider a positive linear functional $ g : \mathscr{B}(\mathcal{H}) \mapsto \mathbb{C}$.  It is called \emph{normal} if $g(\sup_a A_a) = \sup_a g(A_a)$ for any upper bounded increasing net $(A_a)$ of positive $A_a \in \mathscr{B}(\mathcal{H})$.  The \emph{normal topology} on $\mathscr{B}(\mathcal{H})$ is defined by the family of seminorms $\{A \mapsto \abs{g(A)} : g \text{ normal } \}$.
\end{definition}
Given this topology, we may define the algebra suitable for quantum probability spaces.
\begin{definition}\label{def:von_neumann_algebra}
    A \emph{von Neumann algebra}\footnote{There are other equivalent ways to define a von Neumann algebra, often in terms of the weak and strong operator topologies, see \citep{Redei:2007a}.} $\mathscr{N}$ is a *-subalgebra of $\mathscr{B}(\mathcal{H})$ which is closed in the normal topology.  A state $\mathbf{P}$ on $\mathscr{N}$ is the restriction of a normal state on $\mathscr{B}(\mathcal{H})$ to $\mathscr{N}$. 
\end{definition}
Of course, it might be tedious to study the topology of some group of operators whenever we are interested in defining a von Neumann algebra.  Fortunately, the following theorem will enable us to generate a von Neumann algebra from a relevant set of operators.
\begin{theorem}[\textbf{Double Commutant Theorem (Theorem 3.8 in \citep{Bouten:2007b})}]\label{thm:double_commutant}
    Let $\mathscr{S} \subset \mathscr{B}(\mathcal{H})$ be a self-adjoint set ( if $S \in \mathscr{S}$ then $S^{\dag} \in \mathscr{S}$).  Then $\mathscr{A} = \mathscr{S}''$ is the smallest von Neumann sub-algebra in $\mathscr{B}(\mathcal{H})$ which contains $\mathscr{S}$.
\end{theorem}
Therefore in order to \emph{generate} a von Neumann algebra, we look at the set of operators which commute with what commutes with the operators we started with, i.e. for $\mathscr{S} \subset \mathscr{B}(\mathcal{H})$ the generate von Neumann algebra is $(\mathscr{S} \cup \mathscr{S}^{\dag})''$.
\begin{definition}\label{def:generated_von_neumann}
    $\mathscr{A} = \vN{A_1,\ldots,A_n}$ is the smallest von Neumann algebra \emph{generated} by the observables $A_1,\ldots,A_n$.
\end{definition}
With these definitions, one can now define a spectral theorem appropriate for infinite dimensional systems.
\begin{theorem}[\textbf{Spectral Theorem (Theorem 3.3 in \citep{Bouten:2007b})}]\label{thm:general_spectral_theorem}
    Let $\mathscr{C}$ be a commutative von Neumann algebra.  Then there exists a measure space $(\Omega,\mathcal{F},\mu)$ and a *-isomorphism $\iota$ (up to $\mu$-a.s) which maps from $\mathscr{C}$ to $L^{\infty}(\Omega,\mathcal{F},\mu)$, the algebra of bounded functions on the measure space.  A probability measure $\mathbb{P}$, absolutely continuous with respect to $\mu$, is defined via the normal state $\mathbf{P}$ on $\mathscr{C}$ as $\mathbf{C} = \mathbb{E}_{\mathbb{P}}[\iota(C)]$ for all $C \in \mathscr{C}$.\footnote{The reason for using $\mu$ rather than $\mathbb{P}$ is that there will be $P \in \mathscr{C}$ such that $\mathbf{P}(P) = 0$, which renders $\iota$ not invertible on those null sets.  That is also why the ultimate probability measure $\mathbb{P}$ is absolutely continuous with respect to $\mu$.}
\end{theorem}
The technical reasons for moving to von Neumann algebras and the normal topology are not particularly enlightening for us.  In fact, throughout the rest of this thesis, we will rarely worry about the distinction between $*$-algebras and von Neumann algebras.  Nonetheless, there are reasons why these choices were made and I encourage the interested reader to consult Section 3.1 in \citep{Bouten:2007b} for more discussion.  The basic idea for choosing a von Neumann algebra is similar to the reason why one cannot generally use the power set of $\Omega$ in defining the $\sigma$-algebra for a classical probability space---it is ``too big".  By restricting to the normal topology, we guarantee that the von Neumann algebra is generated by its projections.  Similarly, the restriction to normal states ensures monotone convergence of a sequence of observables which is related to the countable additivity requirement we have for classical probability measures.  

\begin{definition}\label{def:quantum_probability_space}
    A \emph{quantum probability space} is the pair $(\mathscr{N},\mathbf{P})$ where $\mathscr{N}$ is a von-Neumann algebra and $\mathbf{P}$ is a normal state on $\mathscr{N}$.
\end{definition}
This is essentially identical to Definition \ref{def:quantum_probability_space_finite}, only with $*$-algebras generalized to von Neumann algebras and states generalized to normal states.  As such, we would use it in the same way, selecting a commutative von Neumann subalgebra $\mathscr{A} \subset \mathscr{N}$ which corresponds to the observables we plan to measure in a given experimental realization.  The statistics for those observables may then be calculated using the spectral theorem (Thm.~\ref{thm:general_spectral_theorem}).  The essential point is that {\bf \emph {a commutative quantum probability space is identical to a classical probability space}}.
\subsection{Quantum Random Variables}
Recall that in the discrete setting, quantum random variables were simply self-adjoint operators, whose spectral decomposition in terms of projectors was analogous to the decomposition of discrete classical random variables in terms of indicator functions of events.  Generalizing this decomposition to the continuous setting proceeds analogously.  We consider the quantum probability space $(\mathscr{N},\mathbf{P})$ and select a particular self-adjoint $A \in \mathscr{N}$ which generates the commutative von Neumann algebra $\mathscr{A} = \vN{A} \subset \mathscr{N}$.  From the spectral theorem (Thm.~\ref{thm:general_spectral_theorem}), we know that there exists a classical probability space $(\Omega,\mathcal{F},\mathbb{P})$ and isomorphism $\iota$ that maps $A$ to a random variable on $\Omega$ which we write as $a : \Omega \mapsto \mathbb{R}$.  Since this is a continuous, real-valued random variable, we know that we can use the Borel algebra $\mathcal{B}$ to decompose $a$ into its events.  That is for some Borel set $B \in \mathcal{B}$, the event $ a \in B$ corresponds to the set $\{w \in \Omega : a(\omega) \in B = a^{-1}(B) \in \mathcal{F}\}$.  To map this back to the quantum space, we invert $\iota$.  The projector that corresponds to this event---``$A$ takes on a value in $B$"---is then written $P_A(B) = \iota^{-1}(\chi_{a \in B})$.  The map $P_A$ is known as the \emph{spectral measure} in functional analysis and allows us to decompose $A$ as
\begin{equation} \label{quantum:eq:spectral_decomposition_infinite}
    A = \int_{\mathbb{R}} \lambda P_A(d\lambda) .
\end{equation}
This is exactly the generalization of the finite-dimensional spectral decomposition in Eq.~\eqref{quantum:eq:spectral_decomposition_discrete}, where $\lambda$ plays the role of the eigenvalue and $P_A(d\lambda)$ plays the role of eigenprojectors.  Again, we have the interpretation that any $f(A)$ can be trivially evaluated using this decomposition once we know which event, or equivalently which eigenspace, occurred. 

Aside from the functional analysis machinery, bounded observables in the general case are treated in exactly the same way as finite-dimensional quantum observables.  Unfortunately, many observables of interest in quantum mechanics are not described by bounded operators, most notably position and momentum.  Although rigorous methods of dealing with such observables exist, I will only sketch a technique discussed in \citep{Bouten:2007b}.  Our von Neumann algebra $\mathscr{N} \subset \mathscr{B}(\mathcal{H})$ contains only bounded operators and we need to somehow relate an unbounded operator $A$ to this algebra.  To do so, define the operator $T_A = (A + iI)^{-1}$.  Since $A$ is self-adjoint, it has a real spectrum, so we know that $T_A$ is invertible and has bounded inverse.  If $T_A \in \mathscr{N}$, we say $A$ is \emph{affiliated} to $\mathscr{N}$.  This is analogous to the classical notion of measurability, in that $A$ is not strictly in $\mathscr{N}$, but its value may be determined if we know the yes-no outcomes of events in $\mathscr{N}$.  Since $A$ is a self-adjoint, linear operator it is trivially affiliated to $\mathscr{B}(\mathcal{H})$; if it is also bounded, then it is affiliated to $\mathscr{N}$ if and only if $A \in \mathscr{A}$.

In order to close the loop, we want to represent $A$ as a classical random variable using the spectral theorem, which was only developed for bounded functions.  We note that the von Neuman algebra generated by $A$ is trivially $\vN{A} = \vN{T_A}$, since the identity operator doesn't change anything.  Moreover, $T_A$ commutes with its adjoint, so $\vN{T_A}$ is commutative and bounded; we may therefore apply the spectral theorem, packaging $A$ in $T_A$, applying $\iota$ and then mapping back.  That is, the classical (unbounded) random variable corresponding to $A$ is $\iota(A) = \iota(T_A)^{-1} - i$.  From this, we can define the spectral measure $P_A$ using Eq.~\ref{quantum:eq:spectral_decomposition_infinite} and proceed without further worry.  Given that this technique exists, we will not worry too much about unbounded operators and their domains throughout the rest of this thesis.

Let's now consider two examples which will clarify the above definitions and which will prove useful when considering quantum white noise processes.
\begin{example}[Example 3.9 in \citep{Bouten:2007b}]\label{quantum:example:l2_position}
    Let $\mathcal{H} = L^{2}(\mathbb{R})$, the vector space of square-normalizable functions and let $\mathscr{N} = \mathscr{B}(\mathcal{H})$.  This is the Hilbert space for a continuous, one-dimensional quantum system, e.g.~a particle on a line.  We define the vector $\ket{\psi} \in \mathcal{H}$ in the position basis as
    \begin{equation}
        \psi(x) = \frac{1}{(2\pi\sigma^2)^{1/4}}e^{-\frac{(x-\mu)^2}{4\sigma^2}}.
    \end{equation}
This pure state defines the quantum probability state $\mathbf{P}(X) = \bra{\psi}X\ket{\psi}$, so that we now have a complete quantum probability space. 

From standard quantum mechanics, we are familiar with two (unbounded) observables on this space, position
\begin{equation}
    \hat{x}\psi(x) = x\psi(x)
\end{equation}
and momentum
\begin{equation}
    \hat{p}\psi(x) = -i\hbar\frac{d}{dx}\psi(x) .
\end{equation}
Using our quantum probability machinery, we can consider what classical random variables these represent under the given state.  Clearly $\hat{x}$ is diagonal (affiliated to $L^{\infty}(\mathbb{R}) \subset \mathscr{N}$) and therefore the state $\ket{\psi}$ tells us it is a Gaussian random variable with mean $\mu$ and variance $\sigma^2$.  Alternatively, we could consider the characteristic function of $\hat{x}$, written $x(k) = \mathbf{P}(e^{ik\hat{x}})$.  Calculating explicitly
\begin{equation}
    x(k) = \bra{\psi}e^{ik\hat{x}}\ket{\psi} 
         = \frac{1}{\sqrt{2\pi\sigma^2}}\int_{-\infty}^{\infty} dx
            e^{ikx}e^{-(x-\mu)^2/2\sigma^2}
         = e^{ik\mu - k^2\sigma^2/2}
\end{equation}
which we recognize as the characteristic function of Gaussian random variable with mean $\mu$ and variance $\sigma^2$.  Similarly, we can definite the characteristic function of $\hat{p}$ as $p(k) = \mathbf{P}(e^{ik\hat{p}})$ and recalling that $\hat{p}$ is the generator of displacements in position
\begin{equation}
    p(k) = \bra{\psi}e^{ik\hat{p}}\ket{\psi}
         = \int_{-\infty}^{\infty}dx \psi(x)\psi(x + \hbar k)
         = e^{-\hbar^2 k^2/8\sigma^2}
\end{equation}
which is also the characteristic function of a Gaussian, but with mean zero and variance $\hbar^2/4\sigma^2$.  Note that $\Delta \hat{x} \Delta \hat{p} = \hbar/2$ as expected for the minimum uncertainty state $\ket{\psi}$. 
\end{example}

The final example in this section considers the Hilbert space of the harmonic oscillator, which given its fundamental role in quantizing the electromagnetic field, serves as an important step towards quantum white noise processes which are prevalent in quantum optics.  The following theorem will play in important part in characterizing operators on this space.
\begin{theorem}[\textbf{Stone's Theorem (Theorem 3.10 in \citep{Bouten:2007b})}]\label{thm:stone}
    Let $\mathscr{N}$ be a von Neumann algebra and let $\{U_t\}_{t\in\mathbb{R}}$ be a strongly continuous group of unitary operators.  Then there is a unique self-adjoint $A$ affiliated to $\mathscr{N}$ called the \emph{Stone generator} such that $U_t = e^{itA}$.
\end{theorem}
\noindent This theorem is often implicitly used in quantum mechanics when analyzing continuous symmetries, such as when identifying the Hamiltonian as the generator of time displacements.
\begin{example}[Adapted from Example 3.11 in \citep{Bouten:2007b}]
    \label{quantum:example:harmonic_oscillator}
    Let $\mathscr{N} = \mathscr{B}(\mathcal{H})$ with $\mathcal{H} = \ell^2(\mathbb{N})$, the set of square normalizable functions on the integers.  On this space, define the orthonormal number basis $\ket{n}$ with $n = 0, 1, \ldots$ where $\braket{n}{k} = \delta_{nk}$.  We further define the unnormalized \emph{exponential state} for $\alpha \in \mathbb{C}$ as
    \begin{equation}
        \ket{e(\alpha)} = \sum_n \frac{\alpha^n}{\sqrt{n!}}\ket{n}
    \end{equation}
which are the unnormalized form of the coherent states $\ket{\alpha} = \ket{e(\alpha)}e^{-\modsq{\alpha}/2}$.  As we know from quantum mechanics, the exponential vectors provide an overcomplete basis for $\mathcal{H}$, which mathematically means their linear span $\mathcal{D}$ is dense in $\mathcal{H}$.  As the last ingredient for our quantum probability space, we define the quantum probability states $\mathbf{P}_{\alpha}(X) = \bra{\alpha}X\ket{\alpha}$.

Lets consider some observables on this space.  The most straightforward is the diagonal operator $\hat{n}$ which acts on number states as $\hat{n}\ket{n} = n\ket{n}$ and although unbounded is affiliated to $\ell^{\infty}(\mathbb{N}) \subset \mathscr{N}$.  The spectral measure of $\hat{n}$ is simply $P_{\hat{n}}(B)\ket{\psi} = \chi_B(k)\ket{\psi}$ which occurs with probability
\begin{equation}
    \mathbf{P}_{\alpha}(P_{\hat{n}}(B)) = \bra{e(\alpha)}P_{\hat{n}}(B)\ket{e(\alpha)}e^{-\modsq{\alpha}}
                     = \sum_{k \in B} \frac{e^{-\modsq{\alpha}}(\modsq{\alpha})^k}{k!} .
\end{equation}
Thus, the event ``$\hat{n}$ takes on a value in $B$" occurs with the probability written above, suggesting that for coherent states, $\hat{n}$ is a Poisson-distributed random variable with intensity $\modsq{\alpha}$. 

That is basically it for diagonal random variables in the number basis, but given our familiarity with the quantum harmonic oscillator, we suspect position and momentum observables are hiding somewhere as well.  In light of Stone's theorem above and our prior knowledge that position and momentum are generators of displacements, we will construct them from a unitary representation of the translation group.  Generally, we have a two-dimensional translation, which we implement with the unitary Weyl (or displacment) operator
\begin{equation}
    W_\gamma \ket{e(\alpha)} = \ket{e(\alpha + \gamma)}e^{-\gamma^{*}\alpha - \modsq{\gamma}/2}
\end{equation}
where $\gamma \in \mathbb{C}$ determines the displacement in the complex plane.  This is analogous to the standard displacement operator in quantum optics used to transform coherent states.  One can verify the unitarity of $W_\alpha$ directly and further note that the Weyl operators form a group under multiplication since $W_\alpha W_\beta = W_{\alpha + \beta}e^{i\imag{\beta^*\alpha}}$.  Note that we have defined the action of $W_\gamma$ on the exponential vectors, from which their linear span may be used to extend the action to all of $\mathcal{H}$. 

In order to apply Stone's theorem, we need to turn this into a one parameter unitary group.  As such, fix a particular $\beta \in \mathbb{C}$ and consider the one parameter group $\{W_{t\beta}\}_{t\in\mathbb{R}}$.  It is continuous since $W_{t\beta}\ket{e(\alpha)} \mapsto \ket{e(\alpha)}$ as $t \mapsto 0$, so by Thm.~\ref{thm:stone}, there exists a self-adjoint $B_\beta$ such that $W_{t\beta} = e^{it B_\beta}$.  We can then ask for the distribution of this generator under the coherent state in terms of the characteristic function.  Letting $b_\beta(k) = W_{k\beta} = e^{ikB_{\beta}}$ be the characteristic function of $B_\beta$, we find
\begin{equation}
    B_\beta(k) = \mathbf{P}_\alpha(W_{k\beta}) 
               = \braket{e(\alpha)}{e(\alpha + k\beta)}e^{-k\beta^{*}\alpha -k^2\modsq{\beta}/2 - \modsq{\alpha}}
               = e^{2ik\imag(\alpha^{*}\beta) - k^2\modsq{\beta}/2}
\end{equation}
which means that for coherent states, $B_\beta$ is a Gaussian random variable with mean $2\imag(\alpha^{*}\beta)$ and variance $\modsq{\beta}$.

We can also find an explicit representation of $B_\beta$ acting on the exponential vectors.  Given the Stone representation of $W_{t\beta}$, this is simply
\begin{equation}
    B_\beta \ket{e(\alpha)} = \left.\frac{1}{i}\frac{d}{dt}W_{t\beta}\ket{e(\alpha)}\right|_{t=0}
        = i\beta^{*}\alpha \ket{e(\alpha)} - \left. i\frac{d}{dt}\ket{e(\alpha + t\beta)}\right|_{t=0} .
\end{equation}
In order to recover the familiar harmonic oscillator operators, we need to explore particular $\beta$ values.  Given that $\hat{x}$ generates displacements in momentum, which is the imaginary axis in the complex plane, we set $\hat{x} = B_{i}$.  Similarly, since $\hat{p}$ generates displacements in position, we set $\hat{p} = B_{-1}$.  Given these two operators, we can introduce the lowering operator $\hat{a} = (\hat{x} + i \hat{p})/2$, which from the representation of $B_\beta$ above means
\begin{align}
    \hat{a}\ket{e(\alpha)} &= \frac{1}{2}\left[ \alpha \ket{e(\alpha)} - \left. i\frac{d}{dt}\ket{e(\alpha + ti)} \right|_{t=0}
    + \alpha \ket{e(\alpha)}  
    + \left. \frac{d}{dt}\ket{e(\alpha - t)}\right|_{t=0}\right]\\
    &= \alpha\ket{e(\alpha)} +
       \frac{1}{2}\sum_n\frac{d}{dt} \left. 
       \left[-i(\alpha + ti)^n  + (\alpha -t)^n\right]\right|_{t=0}\frac{\ket{n}}{\sqrt{n!}}\\
   &= \alpha\ket{e(\alpha)} +
      \frac{1}{2}\sum_n\left.
      \left[n(\alpha + ti)^{n-1}  -n(\alpha -t)^{n-1}\right]\right|_{t=0}\frac{\ket{n}}{\sqrt{n!}}\\
  &= \alpha\ket{e(\alpha)} +
     \frac{1}{2}\sum_n\left.
     \left[n(\alpha)^{n-1}  -n(\alpha)^{n-1}\right]\right|_{t=0}\frac{\ket{n}}{\sqrt{n!}}\\
  &= \alpha\ket{e(\alpha)}
\end{align}
Thus, the lowering operator acts as expected on exponential (and by extension coherent) states and one can easily show that $\hat{a}\ket{n+1} = \sqrt{n+1}\ket{n}$ as expected.  One can also verify that the raising operator, which is the adjoint $\hat{a}^{\dag}$, acts as $\hat{a}^{\dag}\ket{n} = \sqrt{n+1}\ket{n}$ so that $\hat{n} = \hat{a}^{\dag}\hat{a}$.  

On the one hand, this example shows that all of the familiar observables and operators of the harmonic oscillator may be posed in the quantum probability framework.  On the other hand, if one were to instead focus on a classical probability model for these observables, it is seems unusual that both Poisson and Gaussian random variables emerge from the same state $\mathbf{P}_\alpha$ and moreover, there is a continuous map between the two via $\hat{x},\hat{p} \mapsto (\hat{x}-i\hat{p})(\hat{x}+i\hat{p})/4 = \hat{n}$.  One could never continuously transform two continuous random variables into a discrete random variable in classical probability theory.  The reason we can do so here is that $\hat{x},\hat{p}$ and $\hat{n}$ do not commute, indicating that we could never realize them in the same measurement and therefore need not worry about applying the spectral theorem to all simultaneously.  
\end{example} 
\subsection{Quantum Conditional Expectation}
As was the case for finite dimensional quantum systems, all the heavy lifting needed to construct a quantum conditional expectation is handled by the spectral theorem (Thm.~\ref{thm:general_spectral_theorem}), which relates a commutative von Neumann algebra to a classical probability space.  Additionally, all the details about including an explicit probe model for conditioning are no different than was the case in finite dimensions.  I forego recounting those details and instead focus on some subtleties of the quantum conditional expectation that have heretofore been overlooked.  For completeness, I first restate the quantum conditional expectation in terms of the quantum probability model.
\begin{definition}\label{def:conditional_expectation}
    Consider the quantum probability space $(\mathscr{N},\mathbf{P})$ and let $\mathscr{A} \subset \mathscr{N}$ be a commutative von Neumann subalgebra.  Then the map $\mathbf{P}( \cdot | \mathscr{A}) : \mathscr{A}' \mapsto \mathscr{A}$ is (a version of) \emph{the conditional expectation} if $\mathbf{P}(\mathbf{P}(B | \mathscr{A})A) = \mathbf{P}(BA)$ for all $A \in \mathscr{A}$ and $B \in \mathscr{A}'$.
\end{definition}

Firstly, what does ``a version of'' mean in this context?  As is often the case for infinite dimensional systems, there is a freedom of definition for operators which have measure zero under the state $\mathbf{P}$.  Thus, the uniqueness of conditional expectation in the quantum probability setting means that any two version of $\mathbf{P}(B | \mathscr{A})$, call them $P$ and $Q$, satisfy $\norm{P - Q}_{\mathbf{P}} = 0$ where $\normsq{X}_{\mathbf{P}} = \mathbf{P}(X^{\dag}X)$.  If $P$ and $Q$ happen to differ on a part of Hilbert space where the state $\mathbf{P}$ has no support, then they would be different operators, but not in any important way relative to the conditional expectation.

Secondly, we only defined the spectral theorem for bounded, self-adjoint operators.  For such operators, the conditional expectation is explicitly calculable as $\mathbf{P}(B|\mathscr{A}) = \iota^{-1}(\mathbb{E}_{\mathbb{P}}(\iota(B) | \sigma\{\iota(\mathscr{A}\}))$.  Although we have discussed how to extend such a definition to unbounded operators, it is not clear how to find an explicit form for the conditional expectation when the operators are not self-adjoint.  After all, such operators do not generally have a spectral decomposition, so the simple mapping through $\iota$ does not exist.  But we can trivially decompose an operator in terms of its self-adjoint parts.  That is, $B \in \mathscr{A}'$ may be written $B = B_1 + i B_2$, where $B_1 = (B+B^{\dag})/2$ and $B_2 = i(B^{\dag} - B)/2$.  Since $B_1$,$B_2$ are self-adjoint and since the conditional expectation is linear, we may use $\iota$ on $B_1$ and $B_2$ such that $\mathbf{P}(B | \mathscr{A}) \equiv  \mathbf{P}( B_1 | \mathscr{A}) + i\mathbf{P}(B_2 | \mathscr{A})$.  

\subsection{Quantum Bayes formula}
As we saw in Chapter \ref{chapter:classical}, using the explicit formula for conditional expectation is not always convenient when working in infinite-dimensional spaces.  This is also true for infinite dimensional quantum spaces, as the simple formula in Eq.~\eqref{quantum:eq:discrete_conditional} is often unwieldy, especially in the filtering problem, where conditional expectation calculations are often easier under a different measure.  We therefore will often use the following quantum Bayes formula when performing inference.
\begin{theorem}[\textbf{Quantum Bayes Formula (Lemma 3.18 in \citep{Bouten:2007b})}]\label{thm:quantum_bayes_formula}
    Let $\mathscr{A}$ be a commutative von Neumann algebra and let $\mathscr{A}'$ be equipped with a normal state $\mathbf{P}$.  Choose the reference operator $V \in \mathscr{A}'$ such that $V^{\dag}V > 0$ and $\mathbf{P}(V^{\dag}V) = 1$.  Then we define a new state on $\mathscr{A}'$ by $\mathbf{Q}(A) = \mathbf{P}(V^{\dag}AV)$ so that
    \begin{equation}
        \mathbf{Q}(A | \mathscr{A}) = \frac{\mathbf{P}(V^{\dag} AV | \mathscr{A})}
                        {\mathbf{P}(V^{\dag}V | \mathscr{A})} \qquad
                        \forall A \in \mathscr{A}' .
    \end{equation}
\begin{proof}
    Let $K$ be an arbitrary element of $\mathscr{A}$.  Then for all $A \in \mathscr{A}'$, we have
    \begin{align}
        \mathbf{P}(\mathbf{P}(V^{\dag}AV | \mathscr{A}) K)
                &= \mathbf{P}(V^{\dag}AKV)\\
                &= \mathbf{Q}(AK)\\
                &= \mathbf{Q}(\mathbf{Q}(A | \mathscr{A})K)\\
                &= \mathbf{P}(V^{\dag}V \mathbf{Q}(A | \mathscr{A}) K)\\
                &= \mathbf{P}(\mathbf{P}(V^{\dag}V \mathbf{Q}(A | \mathscr{A}) K) | \mathscr{A})\\
                &= \mathbf{P}(\mathbf{P}(V^{\dag}V | \mathscr{A})\mathbf{Q}(X | \mathscr{A})K)
    \end{align}
    Since $K$ was general, this must be true for the other operators under the outermost $\mathbf{P}$, so that we read off Bayes formula by moving $\mathbf{P}(V^{\dag}V | \mathscr{A})$ to the other side of the equation.  In the manipulations above, we used the fact that $K$ commutes with all operators involved, the definition of the conditional expectation and the ``module property'' that $\mathbf{P}(AB | \mathscr{C}) = B \mathbf{P}(A | \mathscr{C})$ if $B \in \mathscr{C}$.
\end{proof}
\end{theorem}
\noindent  The definition (and proof) are very similar to the classical Bayes formula in Theorem \ref{thm:bayes_formula}, but we do have an added interpretation in the quantum setting.  Although $V$ will not always be unitary, the transformation to the state $\mathbf{Q}$ is reminiscent of moving into an interaction picture, which is a common tool in standard quantum mechanics for simplifying calculations.  We see such a change in the following example that is similar to the reference probability approach we will use in deriving the quantum filter. 
\begin{example}[Example 3.19 in \citep{Bouten:2007b}]\label{quantum:example:stern_gerlach}
    Consider modeling a Stern-Gerlach (SG) experiment in which we measure the spin state of an atom using its spatial degree of freedom.  Following our previous examples, we define the spin degree of freedom for a spin-1/2 particle by the von Neumann algebra $\mathscr{N}_\mu = \mathscr{B}(\mathbb{C}^2)$ spanned by the Pauli operators and we define the position degree of freedom along the $z$ axis by $\mathscr{N}_q = \mathscr{B}(\ell^2(\mathbb{N}))$ with position operator $\hat{q}$ and momentum operator $\hat{p}$.  Note that we are using the harmonic oscillator definitions from Example \ref{quantum:example:harmonic_oscillator} rather than the $\mathcal{L}^2(\mathbb{R})$ definition from Example \ref{quantum:example:l2_position}, which are equivalent up to a change in units.  Thus, the overall von Neumann algebra is $\mathscr{N} = \mathscr{N}_\mu \otimes \mathscr{N}_q$.  We will assume that the initial states of the two degrees of freedom are uncorrelated so that we may write $\mathbf{P} = \mathbf{P}_\mu \otimes \mathbf{P}_0$, with $\mathbf{P}_\mu(X) = \bra{\psi_0}X\ket{\psi_0}$ for an arbitrary spin-state $\ket{\psi_0}$ and $\mathbf{P}_0 = \bra{0}X\ket{0}$.  We choose the vacuum state as the initial position state, indicating the atom is initially at rest in a minimum uncertainty state.  
    
    Our simple model of the SG device corresponds to appling a magnetic field gradient that is linearly related to the spin along $z$ and the position along $z$.  This will cause a displacement of the momentum relative to the spin state, so that measuring the momentum will provide an indirect measurement of $\sigma_z$.  For simplicity, we will ignore the free Hamiltonian of the system which would transform the shift in momentum into a translation in position.  That is, we assume we can measure momentum directly, so that it acts as a probe for the internal spin state of the atom.  The unitary which describes the action of the magnetic field gradient is 
    \begin{equation}
        U = \exp\left( i \kappa \sigma_z \otimes \hat{q} \right)
          = P_{z,+1} \otimes e^{i\kappa\hat{q}} 
          + P_{z,-1} \otimes e^{-i\kappa\hat{q}}
          = P_{z,+1} \otimes W_{i\kappa}
          + P_{z,-1} \otimes W_{-i\kappa}
    \end{equation}
    where $\kappa$ represents the time integrated gradient in appropriate units.  Since we intend to measure the momentum, we begin by considering the statistics of that measurement in terms of the characteristic function for $U^\dag(I \otimes \hat{p})U$,
    \begin{align}
        \mathbf{P}(e^{ik U^{\dag}(I \otimes \hat{p})U})
            &= \mathbf{P}(U^{\dag}(I \otimes W_{-k})U)\\
            &= \mathbf{P}_\mu(P_{z,+1})\mathbf{P}_z(W_{-i\kappa}W_{-k}W_{i\kappa})
             + \mathbf{P}_\mu(P_{z,-1})\mathbf{P}_z(W_{i\kappa}W_{-k}W_{-i\kappa})\\
            &= \mathbf{P}_\mu(P_{z,+1}) e^{2i\kappa k - k^2/2}
             +  \mathbf{P}_\mu(P_{z,-1}) e^{-2i\kappa k - k^2/2}
    \end{align}
    where we have used the group property of the Weyl operator and calculations from Example \ref{quantum:example:harmonic_oscillator}.  The characteristic function tells us that the atom's momentum distribution after the interaction is a sum of two Gaussians, each with unit variance but with means $\pm 2\kappa$ weighted by the probability of having spin up or down given by $\mathbf{P}_\mu(P_{z,\pm 1})$.  Note that this distribution does not perfectly resolve the spin states.  If our policy was to assign the spin state according to the sign of the observed momentum, there is some probability to assign the wrong spin state since the tails of the Gaussians overlap as is seen in Fig.~\ref{quantum:fig:stern_gerlach_probability}.  This probability becomes smaller as the field gradient $\kappa$ becomes larger.
    \begin{figure}[h]
        \centering
            \includegraphics[scale=1]{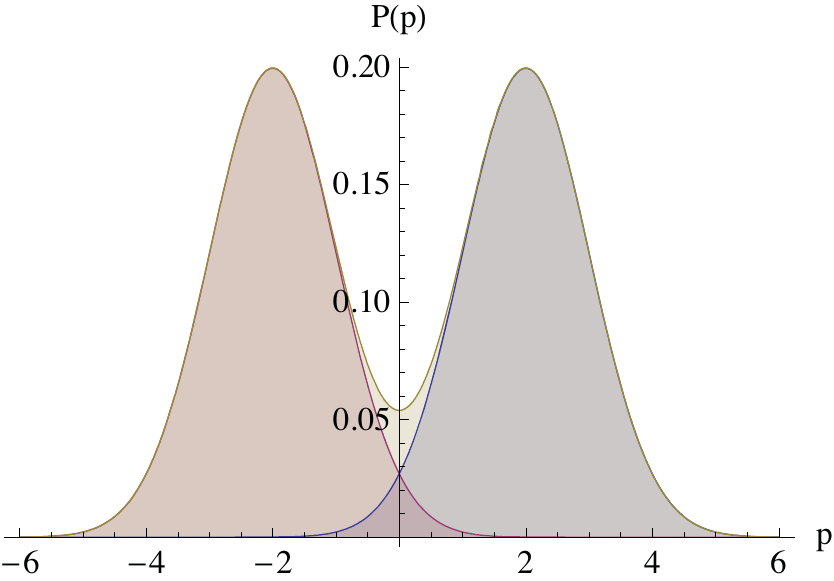}
        \caption[Probability distribution for momentum measurement in Example \ref{quantum:example:stern_gerlach}]{Probability distribution for momentum measurement $U^{\dag}(I\otimes\hat{p})U$ with $\kappa = 1$ in arbitrary units and initial spin state $\ket{+x}$.}
        \label{quantum:fig:stern_gerlach_probability}
    \end{figure}
    
    Of course, the purpose of using the position degree of freedom as a probe for the spin degree is so that we may talk about the conditional expectation of spin observables, such as $\sigma_x$, given the indirect $\sigma_z$ measurement.  The two clearly commute since $[U^{\dag}(I\otimes \hat{p}) U, U^{\dag}(\sigma_x \otimes I)U] = 0$.  We therefore set $\mathscr{A} = \vN{U^{\dag}(I\otimes \hat{p}) U}$, for which $U^{\dag}(\sigma_x \otimes I)U \in \mathscr{A}'$ so that $\mathbf{P}(U^{\dag}(\sigma_x \otimes I)U | \mathscr{A})$ is well defined.  Following the development of the quantum Bayes formula, we note that for a unitary $U$ and state $\mathbf{Q}(X) = \mathbf{P}(U^\dag X U)$, the definition of the conditional expectation shows that $\mathbf{P}(U^{\dag}X U | U^{\dag}\mathscr{C}U) = U^{\dag}\mathbf{Q}(X | \mathscr{C}) U$.  For our problem, this means
    \begin{equation}
        \mathbf{P}(U^{\dag}(\sigma_x \otimes I)U | \mathscr{A}) 
            = U^{\dag}\mathbf{Q}(\sigma_x \otimes I | \vN{I \otimes \hat{p}}) U ,
    \end{equation}
which is analogous to performing the conditional expectation calculation in the Schr\"{o}dinger picture, then using $U$ to transform back to the Heisenberg picture.  Note that $\mathscr{A} = U^{\dag}\vN{I\otimes\hat{p}}U$.  Given that $U$ entangles the two subsystems of the atom, we expect the conditional expectation calculation would be easier under the state $\mathbf{Q}$ and a simple application of quantum Bayes rule would allow us to evaluate the desired conditional expectation.  Unfortunately, $U$ does not commute with $I \otimes \hat{p}$ ($U \not\in \mathscr{A}'$), so the Bayes rule will not work in this form.

However, given that we are working in the vacuum state, we can perform some tricks to construct a related operator $V$ which does allow us to apply the Bayes rule.  Specifically, the part of $U$ that gives us trouble is $e^{i\kappa \hat{q}}$, which clearly does not commute with $\hat{p}$.  But given that $a\ket{0} =0$, we find
\begin{multline}
    e^{i\kappa\hat{q}}\ket{0} = e^{i\kappa(a + a^{\dag})}\ket{0}
        = e^{-\kappa^2/2}e^{i\kappa a^{\dag}}e^{i\kappa a}\ket{0}
        = e^{-\kappa^2/2}e^{i\kappa a^{\dag}}\ket{0}\\
        = e^{-\kappa^2/2}e^{i\kappa a^{\dag}}e^{-i\kappa a}\ket{0}
        = e^{-\kappa^2}e^{\kappa \hat{p}}\ket{0}
\end{multline}
so that
\begin{equation}
    \mathbf{P}_0(e^{-i\kappa\hat{q}}Xe^{i\kappa\hat{q}})
     = e^{-2\kappa^2}\mathbf{P}_0(e^{\kappa\hat{p}}Xe^{\kappa\hat{p}}) .
\end{equation}
Thus in the vacuum, we can replace expressions involving $\hat{q}$ with expressions involving $\hat{p}$ without changing the results of any calculations.  We can therefore replace $U$ by 
\begin{equation}
    V = e^{-\kappa^2}e^{\kappa \sigma_z \otimes\hat{p}}
      = e^{-\kappa^2}(P_{z,+1}\otimes e^{\kappa\hat{p}} + P_{z,-1}\otimes e^{-\kappa \hat{p}})
\end{equation}
so that $\mathbf{Q}(X) = \mathbf{P}(U^{\dag}XU) = \mathbf{P}(V^{\dag}XV)$.  Although $V$ is not unitary, it does commute with $I \otimes \hat{p}$ so that we can apply the quantum Bayes formula to find
\begin{equation}
    \mathbf{P}(U^{\dag}(\sigma_x \otimes I)U | \mathscr{A})
         = \frac{U^{\dag}\mathbf{P}(V^{\dag}(\sigma_x \otimes I)V | \vN{I\otimes\hat{p}})U}
                    {U^{\dag}\mathbf{P}(V^{\dag}V | \vN{I \otimes\hat{p}})U} .
\end{equation}
Now note that
\begin{multline}
    V^{\dag}(\sigma_x \otimes I)V = P_{z,+1}\sigma_x P_{z,+1} \otimes e^{2\kappa\hat{p}}
                 + P_{z,-1}\sigma_x P_{z,-1} \otimes e^{-2\kappa\hat{p}}\\
                 + P_{z,+1}\sigma_x P_{z,-1} \otimes I
                 + P_{z,-1}\sigma_x P_{z,+1} \otimes I
\end{multline}
and
\begin{equation}
    V^{\dag}V = P_{z,+1} \otimes e^{2\kappa \hat{p}}
              + P_{z,-1} \otimes e^{-2\kappa \hat{p}}
\end{equation}
Since $I \otimes \hat{p}$ is independent\footnote{Recall that $\mathbf{P}(B | \mathscr{A}) = \mathbf{P}(B)I$ if $B$ is independent of $\mathscr{A}$.} of any spin operator under $\mathbf{P}$, we use the module property to pull it through the conditional expectation and find
\begin{multline}
    \mathbf{P}(V^{\dag}(\sigma_x \otimes I)V | \vN{I\otimes\hat{p}})
        = \mathbf{P}_\mu(P_{z,+1}\sigma_x P_{z,+1})e^{2\kappa \hat{p}} \\
        +\mathbf{P}_\mu(P_{z,-1}\sigma_x P_{z,-1})e^{-2\kappa \hat{p}}
        + 2 \real{}\mathbf{P}_\mu(P_{z,-1}\sigma_x P_{z,+1})
\end{multline}
and 
\begin{equation}
    \mathbf{P}(V^{\dag}V | \vN{I\otimes\hat{p}}) = \mathbf{P}_\mu(P_{z,+1})e^{2\kappa \hat{p}}
        + \mathbf{P}(P_{z,-1}) e^{-2\kappa \hat{p}} .
\end{equation}
Wrapping these in $U^\dag$ and $U$ gives the overall conditional expectation as
\begin{equation}
	\begin{split}
	    &\mathbf{P}(U^{\dag}(\sigma_x \otimes I)U | \mathscr{A}) \\
	   =&\left(\mathbf{P}_\mu(P_{z,+1}\sigma_x P_{z,+1})e^{2\kappa U^{\dag}(I\otimes \hat{p})U} 
	        +\mathbf{P}_\mu(P_{z,-1}\sigma_x P_{z,-1})e^{-2\kappa U^{\dag}(I\otimes \hat{p})U}\right. \\
	        &\Bigl.	+ 2 \real{}\mathbf{P}_\mu(P_{z,-1}\sigma_x P_{z,+1})\Bigr) \left/\left(
	        \mathbf{P}_\mu(P_{z,+1})e^{2\kappa U^{\dag}( I \otimes \hat{p})U}
	            + \mathbf{P}(P_{z,-1}) e^{-2\kappa U^{\dag}( I \otimes \hat{p}) U}\right)\right.
	\end{split}
\end{equation}
Although the result looks a bit unwieldy, it is actually rather straightforward.  Once we perform the SG measurement of $\sigma_z$, we will have determined the value of $U^{\dag}(I\otimes \hat{p})U$, which is then plugged into the above expression to immediately evaluate the conditional expectation of $U^{\dag}(\sigma_x\otimes I)U$.  As we saw in Fig.~\ref{quantum:fig:stern_gerlach_probability}, this is not quite the same result the projection postulate would give for a projective measurement of $\sigma_z$, reflecting the physical nature that the SG device does not perfectly discriminate the $\sigma_z$ outcome for any finite $\kappa$.  However, as $\kappa \to\infty$, we do recover the projection postulate.  

So what was the point of this example?  After all, the Born rule provides a perhaps more transparent way to calculate the same probabilities.  While this is true, the fact is that we rarely have a truly projective measurement available; the quantum probability formalism allows us to handle these generalized measurement with ease.  As a result, conditioning is again a consequence of conditional expectation---no extra postulates are needed.  Most importantly, these techniques are highly reminiscent of those used in developing the classical filtering equation and will therefore be essential when we develop the quantum filter.
\end{example}
\subsection{Summary}
The purpose of this section was to lay the groundwork for performing inference in the quantum setting.  By developing a quantum probability formalism, we found that a commutative set of observables is \emph{identical} to a classical probability theory, indicating we can easily leverage all the techniques we developed in Chapter \ref{chapter:classical}.  As was the case classically, care must be taken for infinite dimensional spaces, but the resulting tools are not substantively different.  In developing the quantum conditional expectation, we found that inference is only possible between commutative observables.  This requires us to include a model for the probe quantum system within the quantum probability space, after which the familiar Born postulate for conditioning quantum systems simply pops out of conditional expectation calculations.  Finally, we developed a quantum Bayes rule for relating conditional expectation calculations under different states.   
\section{Quantum Stochastic Processes}
Heartened by our success in developing a quantum probability theory, we now consider developing a quantum analog of the classical stochastic processes discussed in Chapter \ref{chapter:classical}.  Given the broad applicability of classical white noise processes in describing classical stochastic systems, we hope that an analogous quantum white noise and stochastic differential equation formalism will allow us to cast the the quantum optics filtering problem in similar language, after which we may extend the classical filtering solution to the quantum case. 

  There are then two separate issues to address in this section.  First, we need to develop a mathematical description of quantum white noise processes in terms of a quantum probability model and similarly devise a quantum stochastic calculus for manipulating such processes.  The second task is to connect this mathematical model to a physical one, in which quantum white noise arises naturally from a suitable limit of a standard physical model.  Although both issues admit rigorous solutions, I will primarily focus on the details salient for solving the quantum optics filtering problem as was depicted in Fig.~\ref{fig:quantum:filtering_setup}.  There are several approaches for developing the fundamental quantum noise processes, including starting with the classical processes and then extending to a quantum probability model or starting from a quantum model and developing the quantum processes directly.  We will follow the latter approach as done in \citet{Bouten:2004a,Barchielli:2003}.   This offers a more straightforward route to quantum Brownian motion than the development in \citet{Bouten:2007b}, which focuses on developing both Poisson and Gaussian quantum noise processes by generating a quantum probability space from a classical probability space.  The interested reader should consult \citep{Parthasarathy:1992a} and \citep{Accardi:2002a} for a thorough and rigorous presentation of the topics discussed within this section.
\subsection{Symmetric Fock Space}
Given our ultimate goal of describing experiments in quantum optics, the Hilbert space for our quantum probability space is naturally that for the quantum electromagnetic field.  In this section, we focus on the development of this space in terms of a single polarized mode of the field which may then be extended to describe the full quantized field over many spatial modes.   The Hilbert space for a single photon in this single mode is
\begin{equation}
    \mathcal{H} = \mathbb{C}^2 \otimes \mathcal{L}^2(\mathbb{R}) \cong \mathcal{L}^2(\mathbb{R}; \mathbb{C}^2) . 
\end{equation}
$\mathcal{H}$ is simply the space of $\mathcal{L}^2$ integrable functions in time that return elements in $\mathbb{C}^2$.  Thus an element in $\mathcal{H}$ is a function $f_t \in \mathcal{L}^2(\mathbb{R};\mathbb{C}^2)$, which for every time $t$, tells us the polarization state of a single photon.  If we fix an orthonormal polarization basis $e_1,e_2$ in $\mathbb{C}^2$, then we can decompose these functions as $f_t^{(1)}e_1 + f_t^{(2)}e_2$, so that the inner product is
\begin{equation}
\Ip{f}{g} = \int dt f^{\dag} g = \int dt f_t^{(1)*}g_t^{(1)} + f_t^{(2)*}g_t^{(2)}
\end{equation}
Embedding time directly into the Hilbert space is perhaps unusual, since time usually appears as a parameter via unitary evolution.  Later when tying this formalism to a physical model, we will see that the explicit inclusion of time in $\mathcal{H}$ is analogous to an interaction picture representation, where the states have a explicit time dependence due to a free field evolution.  For now at least, we take this approach so that the resulting quantum stochastic processes, which are merely families of operators on $\mathcal{H}$ indexed by time, are defined in analogy to classical stochastic processes.

States of the electromagnetic field mode involve potentially many photons, which are best considered in the second quantized picture, where the Hilbert space is the \emph{symmetric Fock space}
\begin{equation}
    \mathcal{F} = \mathbb{C} \oplus \bigoplus_{n=1}^{\infty} \mathcal{H}^{\otimes_s n} .
\end{equation}
Each tensor sum term corresponds to a sector with a fixed number of photons, e.g. the zero photon sector, the one photon sector, \ldots ; within a given sector, we use the symmetrized tensor product $\otimes_s$, which ensures that only symmetric states of the constituent photons are possible (which must be true for bosons).  Following our approach in Example \ref{quantum:example:harmonic_oscillator}, we define the \emph{exponential vectors} for $f \in \mathcal{H}$ as
\begin{equation}
    \ket{e(f)} = 1 \oplus \bigoplus_{n=1}^{\infty} \frac{1}{\sqrt{n!}} f^{\otimes n} ,
\end{equation}
which when normalized are the \emph{coherent vectors} $\ket{\psi(f)} = \exp(-\frac{1}{2}\normsq{f})\ket{e(f)}$.  Note that $\braket{e(f)}{e(g)} = \exp\Ip{f}{g}$.  These states are dense in $\mathcal{F}$ so that we may define the action of operators on them and extended to all other states.  The coherent vectors are analogous to coherent states of the harmonic oscillator, which we recall had number state amplitudes related to powers of the complex number $\alpha$.  For the coherent vectors $\ket{\psi(f)}$, this generalizes to having the same single particle state $f$ for each photon in the different sectors, where this state is specified over all time $t$.  An important state for our purposes is the vacuum vector $\ket{\Phi} = \ket{\psi(0)} = \ket{e(0)} = 1 \oplus 0 \oplus \ldots$, which defines the \emph{vacuum state} $\mathbf{P}_\phi(A) = \bra{\Phi}A\ket{\Phi}$ for $A \in \mathscr{B}(\mathcal{F})$.

The quantum probability space is then defined by the von Neumann algebra $\mathscr{N} = \mathscr{B}(\mathcal{F})$, the set of bounded operators on the symmetric Fock space, with vacuum state $\mathbf{P}_\phi$.  Before studying specific operators in this space, we note that it admits a natural decomposition analogous to those for classical filtrations called a \emph{continuous tensor product structure}
\begin{equation}
    \mathcal{F} = \mathcal{F}_{s]} \otimes \mathcal{F}_{[s,t]} \otimes \mathcal{F}_{[t}
\end{equation}
for $ 0 < s < t$.  This continuous decomposition also holds for the von Neumann algebra 
\begin{equation}
    \mathscr{N}  = \mathscr{N}_{s]} \otimes
                  \mathscr{N}_{[s,t]} \otimes
                  \mathscr{N}_{[t}
                = \mathscr{B}(\mathcal{F}_{s]}) \otimes 
                    \mathscr{B}(\mathcal{F}_{[s,t]}) \otimes 
                    \mathscr{B}(\mathcal{F}_{[t})
\end{equation}
and the exponential vectors
\begin{equation}
    \ket{e(f)} = \ket{e(f_{s]})} \otimes \ket{e(f_{[s,t]})} \otimes \ket{e(f_{[t})} .
\end{equation}
Thus the operator process $\{X_t\}$ affiliated to $\mathscr{N}$ is \emph{adapted} if $X_t$ is affiliated to $\mathscr{N}_{t]}$ for every $t$, which is equivalent to it having the form $X_t \otimes I$ on $\mathcal{N}_{t]}\otimes\mathcal{N}_{[t}$.  

\subsection{Quantum White Noise}
Given the close analogy of the symmetric Fock space to the harmonic oscillator space considered in Example \ref{quantum:example:harmonic_oscillator}, we expect to find a Gaussian operator process by studying Weyl transformations of the exponential vectors.  Taking this analogy seriously, pick a $g \in \mathcal{H}$ and define the \emph{Weyl operator} $W(g) \in \mathscr{B}(\mathcal{F})$ as
\begin{equation}
    W(g)\ket{e(f)} = e^{-\frac{1}{2}\normsq{g} - \Ip{f}{g}}\ket{e(f+g)} .
\end{equation}
Recall that the harmonic oscillator Weyl operator performed a translation in $\mathbb{C}$ by some complex number $\gamma$; the Weyl operator here extends this to a translation in $\mathcal{L}^2(\mathbb{R} \otimes \mathbb{C}^2)$ by the single photon state $g$.  Note that the Weyl operators form a group via the relation
\begin{equation}
    W(f)W(g) = e^{-i \imag{\Ip{f}{g}}}W(f+g) \qquad f,g \in \mathcal{H} .
\end{equation}
From the continuous tensor product structure, we see that $W(g\chi_{[0,t]})$ is an adapted operator process.  

In order to apply Stone's theorem to find the generators of these translations, we pick a particular $g \in \mathcal{H}$ and study the one parameter group $\{W(tg)\}_{t\in\mathbb{R}^+}$.  Stone's theorem then tells us that there exists a self-adjoint $B(g) \in \mathscr{N}$ such that
\begin{equation}
    W(tg) = e^{it B(g)}
\end{equation}
The operators $B(g)$ are known as field operators, which we will later tie to the more familiar electromagnetic field operators.  For now, we continue in the tradition of Example \ref{quantum:example:harmonic_oscillator} and consider the statistics of these operators under the vacuum state.  Their characteristic function is
\begin{equation} \label{quantum:eq:fock:weyl_characteristic}
    b_g(k) = \mathbb{P}_{\phi}(W(kg)) = \braket{e(0)}{e(kg)} e^{-\frac{k^2}{2}\normsq{g}}
                 = e^{-\frac{k^2}{2}\normsq{g}}
\end{equation}
which indicates that $B(g)$ is a mean zero Gaussian random variable with variance $\normsq{g}$.  Note that if we were to use an arbitrary coherent state, rather than the vacuum state, the field operators would still be Gaussian with the same variance, but with non-zero mean.

Now in order to identify this with a classical stochastic process via the spectral theorem, we need to consider a commutative operator process.  Specifically, consider the operator process $\{B_t^{\phi,q} = B(e^{i\phi}e_q\chi_{[0,t]})\}$ for a fixed $\phi$ and polarization $q$.  By construction, this is an adapted process and from the continuous tensor product structure, we further know that $B(e^{i\phi}e_q \chi_{[s,t]})$ is affiliated to $\mathscr{N}_{[s,t]}$.  Therefore increments for independent intervals commute and since $B_0^{\phi,q} = I$, we know that any pair $B_t^{\phi,q}$ and $B_s^{\phi,q}$ commute.  Thus $\vN{\{B_t^{\phi,q}\}_{t\in\mathbb{R}}}$ is a commutative von Neumann algebra and from the spectral theorem, is equivalent to a classical stochastic process.  But we also know from Eq.~\eqref{quantum:eq:fock:weyl_characteristic} that the increment $B_t^{\phi,q} - B_s^{\phi,q}$ ($t > s$) is a mean zero Gaussian random variable with variance $t-s$ when in the vacuum state.  The continuous tensor product structure further implies that time independent increments are statistically independent, so that Definition \ref{def:weiner_process} tells us $\iota(B_t^{\phi,q})$ is \emph{identically} a classical Wiener process!  Indeed, by varying $\phi$, we see that an entire family of Wiener processes may be constructed.  However, they do not generally commute with each other, so that only one may be identified in a given realization.  

There is a particular set of these \emph{quantum Wiener processes} which we now identify.  Let $Q_t^q = B(ie_q\chi_{[0,t]})$, $P_t^q = B(-e_q\chi_{[0,t]})$ and $A_t^q = (Q_t^q + iP_t^q)/2$.  These are analogous to the position, momentum and annihilation operators introduced in Example \ref{quantum:example:harmonic_oscillator} and correspond to different quadratures of each polarization mode of the quantum electromagnetic field.  These allow us to introduce the fundamental noises
\begin{align}
    A^q_t\ket{e(f)} &= \left(\int_0^t f_q(t)dt\right) \ket{e(f)}\\
    \bra{e}(g) A_t^{q,\dag}\ket{e(f)} &= \left(\int_0^t g_q^{*}(s)ds\right)\braket{e(g)}{e(f)}\\
    \bra{e(g)} \Lambda_t^{qr} \ket{e(f)} &= \left(\int_0^t g_q^{*}(s)f_r(s)ds\right)\braket{e(g)}{e(f)}
\end{align}
The creation $A^q_t$ and annihilation $A_t^{q,\dag}$ processes are precisely quantum Wiener processes we just studied and are formally related to the familiar interaction picture Bose field operators for the single mode via
\begin{equation} \label{quantum:eq:field_mode_wiener_relation}
    A^q_t = \int_0^t a_q(s) ds \qquad A^{q,\dag}_t = \int_0^t a_q^{\dag} ds
\end{equation}
where $[a_q(s),a_r(s)^{\dag}] = \delta_{qr}\delta(t-s)$ and all other commutators are zero.  Given the delta-time correlation, we see that these canonical field operators are singular objects, in analogy to the usual delta distribution definition of classical white noise.  The remaining gauge or scattering process $\Lambda_t^{qr}$ may also be represented in terms of the usual Bose fields as
\begin{equation}
    \Lambda_t^{qr} = \int_0^t a_q^{\dag}(s) a_r(s) ds .
\end{equation}
As detailed in \citep{Bouten:2007b}, the diagonal entries $\Lambda_t^{qq}$ correspond to counting quanta in a given polarization mode and can be related to classical Poisson processes when the field is in a coherent state, recovering the other stochastic process expected in generalizing Example \ref{quantum:example:harmonic_oscillator}.
\subsection{Quantum Stochastic Calculus} \label{quantum:subsect:quantum_stochastic_calculus}
Since we are ultimately interested in describing quantum stochastic processes driven by the fundamental noises, e.g. systems with a formal\footnote{Meaning a careful interpretation of what the time-derivative of a quantum Wiener process means.} Hamiltonian $H(t) = H_0 + H_1 \dot{Q}_t+ H_2\dot{P}_t + H_3 \dot{\Lambda}_t$, our next task is to develop an appropriate quantum stochastic integral and calculus, keeping in mind the mathematical issues we had to handle in the classical case.  Note that in order to more clearly focus on the essentials, I have restricted consideration to a single polarization mode by dropping the polarization index on the fundamental noises; it should be clear how to generalize the following to account for multiple polarization modes.  Suppose we were only interested in integrals with respect to a single quadrature, say $Q_t$, where the integrands are adapted quantum stochastic processes that commute with $Q_t$.  Given this commutative set, we inherit the It\^{o} integral and calculus definitions through the spectral theorem, so that all the mathematical subtleties are handled by our classical construction in Chapter \ref{chapter:classical}.  Of course, we are really interested in processes which are driven by all three fundamental noises, which do not commute with each other and therefore do not admit a simultaneous classical probability mapping.  Following \citet{Bouten:2007b}, I will attempt to sketch the development of quantum stochastic calculus following \citet{Parthasarathy:1992a,Hudson:1984a}, noting many of the technical issues involved, but neglecting to delve into the details.

Recalling the physical picture we have in mind (Fig.~\ref{fig:quantum:filtering_setup}), we see that there are really two physical systems to consider---the optical field and the atoms.  Letting $(\mathscr{N}_f,\mathbf{P}_\phi)$ be the quantum probability space for the electromagnetic field in the vacuum state that was developed in the last section, we similarly need to define the quantum probability space $(\mathscr{N}_s, \mathbf{P}_s)$ for the atomic system.  Generalizing slightly, we set $\mathscr{N}_s = \mathscr{B}(\mathcal{H}_s)$ and $\mathbf{P}_s(A) = \Tr{A\rho}$, where $\mathcal{H}_s$ is some finite dimensional Hilbert space and $\rho$ is a corresponding density matrix on that space.  Thus the overall space, $(\mathscr{N}_s \otimes \mathscr{N}_f,\mathbf{P}_\phi\otimes \mathbf{P}_s)$ allows us to study how the quantum noises couple to the atomic system and how the two jointly evolve.  This will be made more precise later in this section and for the time being, we focus on the mathematical problem of defining integrals of the form $\int_0^t L_s dM_s$ where $M_t$ is one of the fundamental noises and $L_t$ is an adapted process, which here means it is affiliated to $\mathscr{N}_s \otimes {\mathscr{N}_f}_{t]}$.  Often, $L_t$ will be trivially adapted and correspond to an time-independent operator which acts as the identity on the entire $\mathscr{N}_f$ space. 

The approach we take in defining quantum stochastic integrals follows the one taken classically; we begin by defining the integral for simple processes and then look to define arbitrary integrals as a suitable limit of simple approximations.  First, recall that for $s < t$, any of the fundamental noise increments $M_t - M_s$ are affiliated to ${\mathscr{N}_f}_{[s,t]}$.  Given that $L_t$ is affiliated to ${\mathscr{N}_f}_{t]}$ by assumption, this means we may write $L_s \otimes (M_t - M_s) = L_s(M_t - M_s) = (M_t - M_s)L_s$, i.e.~the processes commute and there are no issues in multiplying these unbounded operators.  This is analogous to the non-anticipative property of the classical Wiener increment.  \emph{Simple processes} $L_t$ are those whose values change at the fixed sequence of times defined by the partition $\pi_n$ of $[0,t]$, e.g.
\begin{equation}
    L_t = \sum_{t_i \in \pi_n} L_{t_i}\chi_{[t_i,t_{i+1})}
\end{equation}
so that we may define the quantum stochastic integral for these simple processes as
\begin{equation}
    \int_0^t L_s dM_s = \sum_{t_i \in \pi_n} L_{t_i} (M_{t_{i+1}} - M_{t_i}) .
\end{equation}
As was the case classically, the difficulty now is to extend this definition to arbitrary $L_t$ in terms of an approximating sequence of simple processes $L_t^n$ whose stochastic integrals converge to give a unique integral for the initial process $L_t$.

More concretely, consider the set of adapted processed $(E,F,G,H)$ which admit the simple approximations $(E_t^n,F_t^n,G_t^n,H_t^n)$.  We want to define the integral
\begin{equation}
           I_t = \int_0^t E_s d\Lambda_s + F_s dA_s + G_s dA_s^{\dag} + H_s ds
\end{equation}
as a suitable limit of the corresponding simple approximations $I_t^n$ over the simple processes.  Recall that classically, we were able to use the It\^{o} isometry (Lem.~\ref{lem:ito_isometry}) to define this limit uniquely in $\mathcal{L}^2$.  Taking the same approach here is not quite so straightforward.  For example, let $\mathcal{H}_s = \mathbb{C}$ so that it may be ignored for the time being; then the seminorm is given by $\normsq{X}_\phi = \bra{\Phi}X\ket{\Phi}$.  Mean square convergence of $I_t^n \to I_t$ then corresponds to $\normsq{I_t -I_t^n}_\phi = \bra{\Phi}(I_t - I_t^n)^{\dag} (I_t - I_t^n)\ket{\Phi} \to 0$ as $n \to \infty$.  Such convergence is clearly sensitive to the particular state $\ket{\Phi}$.  Does this mean the domain of $I_t$ is only the vacuum?  That is, how does $I_t$ act on vectors orthogonal to the vacuum if it is only defined relative to convergence in the vacuum state?

There are many inequivalent ways out of this ambiguity and we follow the approach of Hudson and Parthasarthy \citep{Hudson:1984a}.  We fix the domain of $I_t$ to be $\mathcal{H}_s \otimes \mathcal{D}$ from the start, where $\mathcal{D}$ is the linear span of exponential vectors $\ket{e(f)}$.  $I_t$ is then the unique operator on this domain, such that $\bra{v}\otimes\bra{\psi} (I_t - I_t^n)^{\dag} (I_t - I_t^n)\ket{v}\otimes\ket{\psi}$ for every $\ket{v} \in \mathcal{H}_s, \ket{\psi} \in \mathcal{D}$.  This corresponds to a simultaneous mean square limit for \emph{all} states in our fixed domain.  Hudson and Parthasarthy show that this limit exists as long as $\int_0^t \normsq{(E_s - E_s^n)\ket{v}\otimes\ket{\psi}}ds \to 0$ as $n \to \infty\quad \forall \ket{v} \in \mathcal{H}_s, \ket{\psi} \in \mathcal{D}$ and likewise for $F,G,H$.  Additionally, they show that every square-integrable process, e.g. $\int_0^t \normsq{E_s\ket{v}\otimes\ket{\psi}} < \infty\quad \forall \ket{v} \in \mathcal{H}_s, \ket{\psi} \in \mathcal{D}$, admits a simple process approximation.  Thus on the fixed domain, we have essentially the same stochastic integral construction we did classically, in which square-integrable processes admit a unique simple approximation, the integrals of which admit a unique limit as long as these approximations converge independently of the choice of approximation for each of $E,F,G,H$ on all states in the domain.
\begin{definition}\label{def:quantum_ito_integral}
    The \emph{quantum It\^{o} integral} for the adapted and square-integrable processes $(E,F,G,H)$, written
    \begin{equation}
        I_t = \int_0^t E_s d\Lambda_s + F_s dA_s + G_s dA_s^{\dag} + H_s ds ,
    \end{equation}
    is uniquely defined on $\mathcal{H}_s \otimes \mathcal{D}$ as a limit of simple approximations.
\end{definition}
Note that for the vacuum reference state, we have the nice property that $\Lambda_t \ket{\Phi} = A_t\ket{\phi} = 0$, so that in the vacuum the $E_t,F_t$ terms go to zero.   Similarly, since $A_t^{\dag}$ acts to the left as $A_t$ does to the right, we further know that the $G_t$ term is zero in vacuum expectation, although $A_t^{\dag}\ket{\Phi} \neq 0$.  Therefore just like the classical It\^{o} integral, the quantum It\^{o} integral is entirely ``deterministic'', i.e. not driven by the quantum noises, in vacuum expectation. 

Rather than working with the quantum It\^{o} integral, we often write a corresponding \emph{quantum stochastic differential equation (QSDE)}
\begin{equation}
    dI_t = E_t d\Lambda_t + F_t dA_t + g_t dA_t^{\dag} + H_t dt
\end{equation}
which is really just a shorthand representation for the integral in Definition \ref{def:quantum_ito_integral}.  The differential notation is retained to remind us of the singular nature of the stochastic processes, which don't have a well-defined standard time derivative.  As the quantum generalization of the classical stochastic differential equation, we can also study the transformation rules of QSDEs.  Again, defining such properties are done relative to a particular domain, as it is not clear a priori that e.g.~the product of integrals $I_tJ_t$ is a well-defined operator on the domain $\mathcal{H}_s \otimes \mathcal{D}$.  The insight of Hudson and Parthsarathy is to use the fact that the adjoint\footnote{Taking the physicists perspective, I have been \emph{very} casual in using the adjoint $\dag$ in place of the Hilbert space adjoint independent of the domain of the operators.  It is not generally true that domain of the adjoint of an operator coincides with the domain of the operator itself.  Therefore the Hudson-Parthasarathy approach involves more care than I let on, but the details are not so relevant for this introduction.  The reader should consult the references \citep{Bouten:2007b,Parthasarathy:1992a} for more rigor.} $I_t^{\dag}$ is well-defined when restricted to our fixed basis, so that the expression for  $I_tJ_t$ is read off from examining the matrix elements $(\bra{v'}\otimes \bra{\psi'}I_t^{\dag}) (J_t\ket{v}\otimes\ket{\psi})$ for arbitrary states in the domain.  This gives rise to the quantum generalization of the It\^{o} rule (Thm.~\ref{thm:ito_rule}).
\begin{theorem}[\textbf{Quantum It\^{o} rule, Theorem 4.2 in \citep{Bouten:2007b}}]\label{thm:quantum_ito_rule}
    Let $(F^{qr},G^q,H^q,I)$, $(B^{qr},C^q,D^q,E)$ and $(B^{qr\dag},C^{q\dag},D^{q\dag},E^{\dag})$ be integrable stochastic processes where the latter are adjoint pairs.  Consider the stochastic integrals with QSDEs
    \begin{align}
        dX_t &= B_t^{qr}d\Lambda_t^{qr} + C_t^{q}dA_t^{q} + D_t^{q}dA_t^{q\dag} + E_t dt\\
        dY_t &= F_t^{qr}d\Lambda_t^{qr} + G_t^{q}dA_t^{q} + H_t^{q}dA_t^{q\dag} + I_t dt\\        
    \end{align}
where repeated polarization indices are summed over.  The stochastic process $Z_t = X_tY_t$ satisfies the QSDE
\begin{equation}
    dZ_t = X_t dY_t + dX_t Y_t + dX_t dY_t
\end{equation}
where the differential products are evaluated using the fundamental It\^{o} table
\begin{center}
	\begin{tabular}{c|c|c|c|c}
	    $dM_1 \backslash dM_2$ & $dA_t^{i\dag}$ & $d\Lambda_t^{ij}$ & $dA_t^{i}$ & $dt$\\
	    \hline
	    $dA_t^{k\dag}$ & $0$ & $0$ & $0$ & $0$\\
	    $d\Lambda_t^{kl}$ & $\delta_{li} dA_t^{k\dag}$ & $\delta_{li}d\Lambda_t^{kj}$ & $0$ & $0$\\
	    $dA_t^{k}$ & $\delta_{ki}dt$ & $\delta_{ki}dA_t^{j}$ & $0$ & $0$\\
	    $dt$ & $0$ & $0$ & $0$ & $0$
	\end{tabular}
\end{center}
\end{theorem}
\noindent Theorem \ref{thm:quantum_ito_rule} provides us with a simple set of algebraic rules for manipulating products of quantum stochastic differential equations and makes the transformation of complicated stochastic processes almost trivial.  It is worth noting that the Hudson-Parthasarathy construction was a particular choice which leads to a useful quantum stochastic calculus that describes many interesting physical setups (as we will soon see).  Nonetheless, there are open mathematical questions of how to generalize the approach or what alternate constructions may be useful.
\subsection{Quantum Stochastic Limit}
The main lingering question is whether the quantum white noise processes we have developed are actually useful in describing physical systems of interest.  After all, there is no utility gained in developing a quantum stochastic calculus if we can't use it in practice!  Obviously, we have been working with a particular physical model in mind, in which the quantum white noise processes are operators on the electromagnetic field.  We thus need to formally tie the physical quantum model one usually would write down for such a setup to the abstract mathematical model considered above.  Although there are several approaches one might consider, we follow that of \citet{Accardi:2002a} who layout a very general method for deriving quantum white noise approximations of a broad range of system-reservoir interactions.  That is, one considers a Hamiltonian of the form
\begin{equation}
    H = H_0 + \lambda H_I
\end{equation}
where $H_0 = H_S + H_R $ is the \emph{free system and reservoir Hamiltonian} and $H_I$ is the \emph{interaction Hamiltonian} modulated by the coupling parameter $\lambda$.  We are interested in regimes where the coupling is weak, $\lambda \to 0$, but when its affect builds up over long times, $t \to \infty$; essentially considering simultaneously the long-term regime of scattering theory and the weak-effect regime of perturbation theory.  Although Accardi et.~al refer to this regime as the quantum singular limit, it is also known as the quantum Markov limit, the van Hove limit, the quantum stochastic limit and the quantum central limit \citep{Gardiner:1985a,Accardi:1990a,Gough:2005a,Gough:1999a,vanHove:1955a}.  The quantum central limit name is particularly appealing since we are interested in the cumulative affect of many infinitesimal interactions, much as the classical central limit considers the accumulation of many infinitesimal random kicks, e.g. our construction of the Wiener process as the limit of a random walk in Eq.~\eqref{classical:eq:wiener_random_walk}.

In a physical sense, the quantum stochastic limit is related to a separation of timescales of an interacting system.  One timescale is the relaxation time $t_R$ which is the characteristic decay time of the correlation $\expect{H_I(0)H_I(t)}$, where $H_I(t)=e^{itH_0}H_Ie^{-itH_0}$ is the interaction Hamiltonian with respect to the free evolution.  The slow degrees of freedom have characteristic decay time $t_S$ with respect to the correlation $\expect{X(0)X(t)} - \expect{X(0)}\expect{X(t)}$ where $X(t) = e^{it H_0}Xe^{-itH_0}$ and $X$ is an arbitrary observable.  There is also the interaction time $t_{\text{int}}$ which again describes the decay of correlations of observables $X(t)$, but now where $X(t)$ is evolved under the total Hamiltonian $H$.  A typical scenario has $t_R \ll t_{\text{int}} \ll t_S$, so that the fast degrees of freedom, when considered relative to the slow degrees of freedom, look completely uncorrelated and are therefore well described as a white noise.  This should surely be the case for the quantum optics systems, where the vacuum fluctuations occur on a much faster timescale than atomic interactions.

The quantum stochastic limit then attempts to find the form of both $U_t$ and $H_I(t)$ in the following sense
\begin{equation} \label{quantum:eq:stochast_limit:general_form}
  \lim_{t\to\infty,\lambda\to 0}    
    \left[\partialD{U_t^{\lambda}}{t} = -i\lambda H_I(t)U_t^{(\lambda)} \right] \to
    \partialD{U_t}{t} = -i H_I(t) U_t .
\end{equation}
This is notably \emph{different} than the standard quantum Markov approximation taken in deriving a quantum master equation, as for e.g.~done in \citet[Chapter 6]{Walls:2008a}. Rather than finding effective dynamics for a \emph{reduced} system, i.e.~just the atoms, the quantum stochastic limit finds effective dynamics for the \emph{full} system, i.e.~both the atoms and field.  This is particularly useful for the quantum filter, in which we want to measure the electromagnetic field in order to perform inference on the atomic system; if it were eliminated in the weak limit, we would have nothing left to measure! 

Clearly, there must be some relationship between $t$ and $\lambda$ in this limit, since taking $\lambda \to 0$ independently would completely decouple the interaction so that $H_I(t) \to 0$ in the free Hamiltonian interaction picture.  The following lemma shows that the only sensible scaling is to set $t \mapsto t/\lambda^2$ and study just the $\lambda\to 0$ limit.
\begin{lemma}[Lemma 1.8.1 in \citep{Accardi:2002a}]\label{lem:lamda_scaling}
    Let $\expect{\cdot}$ represent expectation with respect to some fixed state and suppose that $H_I(t)$ as described above is mean zero, time-invariant and integrable:
    \begin{align}
        \expect{H_I(t)} & = 0\\
        \expect{H_I(t_1 + s)\cdots H_I(t_n+s)} & = \expect{H_I(t_1)\cdots H_I(t_n)}\\
        \int_{-\infty}^{\infty}\abs{\expect{H_I(0)H_I(t)}} &< \infty
    \end{align}
    Then the expectation value of the second-order term in a perturbative series of $U_t^{(\lambda)}$,
    \begin{equation}
        -\lambda^2\int_0^t dt_1 \int_0^{t_1}dt_2 \expect{H_I(t_1)H_I(t_2)} ,
    \end{equation}
    has a finite nonzero limit as $\lambda \to 0,t \to \infty$ if and only if
    \begin{equation}
        \lim_{\lambda\to 0, t \to\infty} \lambda^2t = \tau
    \end{equation}
    for some finite, non-zero constant $\tau$.   In this case, the limit is
    \begin{equation} \label{quantum:lambda_scaling_transport_limit}
        -\tau\int_{-\infty}^{0} ds \expect{H_I(0)H_I(s)}
    \end{equation}
    \begin{proof}
        By the time-translation invariance property, we may rewrite the second-order integral as
        \begin{equation}
            -\lambda^2 \int_0^{t} dt_1 \int_0^{t_1}dt_2\expect{H_I(0)H_I(t_2 - t_1)} .
        \end{equation}
        Setting $s_2 = t_2 - t_1$ this can further be rewritten as
        \begin{equation}
            -\lambda^2 \int_0^{t_1} dt_1 \int_{-t_1}^{0} ds_2 \expect{H_I(0)H_I(s_2)} .
        \end{equation}
        Now setting $s_1 = \lambda^2 t_1 $, we have
        \begin{equation}
            - \int_0^{\lambda^2 t} ds_1 \int_{-s_1/\lambda^2}^{0} ds_2 \expect{H_I(0)H_I(s_2)}
        \end{equation}
        Clearly $s_1 > 0$, so that as $\lambda \to 0$ the inner integral tends to
        \begin{equation}
            \int_{-\infty}^{0} ds \expect{H_I(0)H_I(s)}
        \end{equation}
        which is independent of $s_1$ since $s_1/\lambda^2 \to \infty$ independent of the value of $s_1$.  This decouples the two integrals and leaves only the outer one, which tends to zero as $\lambda \to 0$ unless the upper limit $\lambda^2 t \to \tau$ as given in the theorem, recovering Eq.~\eqref{quantum:lambda_scaling_transport_limit}.
    \end{proof}
\end{lemma}
We see that a non-trivial limit implies that for times of order $t/\lambda^2$, the interaction produces effects of order $\tau$ and thus $\lambda$ serves as a natural timescale for the problem. It is of note that this limit can be performed for all terms in the Dyson perturbation series, which may then be re-summed to give the effective stochastic propagator on the right hand side of Eq.~\eqref{quantum:eq:stochast_limit:general_form}.

As a prototypical example, we now study the stochastic limit for a single two-level atom coupled to the quantized electromagnetic field.  This will allow us to focus on the relevant details of the stochastic limit rather than issues involving representation of the reduced dipole operator for complicated atoms.  For a more general derivation, the reader is enthusiastically encouraged to consult Chapters 2-5 of \citet{Accardi:2002a}.  The general procedure is to first identify the free and interaction Hamiltonians in order to determine the interaction picture propagator in Eq.~\eqref{quantum:eq:stochast_limit:general_form}.  We then make the replacement $t \mapsto t/\lambda^2$ and study the \emph{time-correlations} of the suitably rescaled interaction picture field mode operators $\tilde{a}_t(\lambda),\tilde{a}_t^{\dag}(\lambda)$.  The hope is that in the $\lambda \to 0$ limit, the correlation $\expect{\tilde{a}_t(\lambda)\tilde{a}^{\dag}_s(\lambda)} \to \delta(t-s)$, allowing us to reconstruct the quantum Wiener processes in a fashion analogous to Eq.~\eqref{quantum:eq:field_mode_wiener_relation}.  I refer the reader to \citep{Walls:2008a} for more detail on developing the quantized electromagnetic field and the dipole Hamiltonian given below.

We begin by introducing the free atom Hamiltonian
\begin{equation}
    H_A = \frac{\hbar\omega_{eg}}{2}(\ketbra{e}{e} - \ketbra{g}{g}) 
        = \frac{\hbar\omega_{eg}}{2}\sigma_z
\end{equation}
where I have used the usual isomorphism between an arbitrary two-level system and the Pauli operators.  The free electromagnetic-field Hamiltonian is
\begin{equation}
    H_F = \int d^3\mathbf{k} \sum_{ q }( \hbar \omega_{\mathbf{k}} a^{\dag}_{\mathbf{k},q} a_{\mathbf{k},q} + \frac{1}{2})
\end{equation}
where $\omega_{\mathbf{k}} \geq 0$, $q$ is the polarization index, and the mode operators satisfy $[ a_{\mathbf{k},q},a^{\dag}_{\mathbf{k}'q'}] = \delta^{3}(\mathbf{k}-\mathbf{k}')\delta_{qq'}$.  The dipole interaction Hamiltonian is given by
\begin{equation}
    H_I = -\mathbf{d}\cdot\mathbf{E}(\mathbf{r})
\end{equation}
with the quantum electromagnetic field written as
\begin{equation}
    \mathbf{E}(\mathbf{r}) = i\mathbf{E}^{+}(\mathbf{r}) - i \mathbf{E}^{-}(\mathbf{r})
        =  i \sum_q \int d^3\mathbf{k} g(\mathbf{k})\sqrt{\frac{\hbar\omega_{\mathbf{k}}}{2\epsilon_0}}
            a_{\mathbf{k},q}\mathbf{e}_{\mathbf{k},q}e^{i\mathbf{k}\cdot\mathbf{r}}
            + \text{ h.c.}
\end{equation}
Here, $\mathbf{e}_{\mathbf{k},q}$ are polarization vectors and $g(\mathbf{k}) \geq 0$ is a mode function to account for the spatial variation of the optical field.  Although we leave it general, we assume it is integrable and infinitely differentiable.

In the dipole approximation, we take $\mathbf{r} = 0$ and write the dipole operator in terms of the atomic energy eigenstates: $\mathbf{d} = \bra{g}\mathbf{d}\ket{e}(\sigma_- + \sigma_+) = \mathbf{d}_{ge}(\sigma_- +\sigma_+)$ with $\sigma_- = \ketbra{g}{e}$, $\sigma_+ = \ketbra{e}{g}$.  We then rewrite the interaction Hamiltonian as
\begin{equation}
    H_I = -i\int d^3\mathbf{k}\sum_q \hbar g_{\mathbf{k},q} a_{\mathbf{k},q}(\sigma_- +\sigma_+) + h.c.
\end{equation}
with the newly defined coupling strength 
\begin{equation}
	g_{\mathbf{k},q} = \sqrt{\frac{\omega_{\mathbf{k}}}{2\hbar\epsilon_0}} g(\mathbf{k}) (e_{\mathbf{k},q}\cdot \mathbf{d}_{ge}) .
\end{equation}
Note that I have not taken the usual rotating-wave approximation, which drops non-energy conserving terms such as $a_{\mathbf{k},q}\sigma_{-}$.  These will end up dropping out as part of the stochastic limit and given our interest in a white noise process with infinite spectral bandwidth, it would be inconsistent to neglect these terms at the outset (even though we would get the same result).

We next use the fact that 
\begin{align}
    e^{i\frac{t}{\hbar}(H_A + H_F)}\sigma_{\pm}e^{-i\frac{t}{\hbar}(H_A + H_F)}
                &= \sigma_{\pm}e^{\pm i \omega_{eg} t}\\
    e^{i\frac{t}{\hbar}(H_A + H_F)}a_{\mathbf{k},q}e^{-i\frac{t}{\hbar}(H_A + H_F)}
                &=  a_{\mathbf{k,q}}e^{-i\omega_{\mathbf{k}}t}
\end{align}
to rewrite the interaction Hamiltonian in the interaction picture with respect to the free Hamiltonian $H_0 = H_A + H_F$:
\begin{equation}
    H_I(t) = -i\int d^3\mathbf{k}\sum_q \hbar g_{\mathbf{k},q} 
           (a_{\mathbf{k},q}\sigma_- e^{-i(\omega_{\mathbf{k}} + \omega_{eg})t} +
            a_{\mathbf{k},q}\sigma_+ e^{-i(\omega_{\mathbf{k}} - \omega_{eg})t}) + h.c.
\end{equation}
Plugging into Eq.~\eqref{quantum:eq:stochast_limit:general_form} and rescaling $t \mapsto t/\lambda^2$, we have
\begin{align}
    \partialD{U_t^{\lambda}}{t} &= -\frac{i}{\lambda} H_I(\frac{t}{\lambda^2})U_t^{(\lambda)}\\
      &= \left[-\hbar (
             \tilde{a}_{\lambda,-\omega_{eg}}(t) \sigma_- + 
             \tilde{a}_{\lambda,+\omega_{eg}}(t) \sigma_+) + h.c.\right]U_t^{(\lambda)}
\end{align}
where I have introduced the rescaled time-domain field operators
\begin{equation}
    \tilde{a}_{\lambda,\omega}(t)  = \int d^3\mathbf{k}\sum_q
                g_{\mathbf{k},q}
                \frac{1}{\lambda}a_{\mathbf{k},q}
                e^{-i(\omega_{\mathbf{k}} - \omega)t/\lambda^2} .
\end{equation}
In order to characterize the behavior of these operators in the limit $\lambda \to 0$, we study their correlation with respect to the vacuum field,
\begin{align}
    \lim_{\lambda \to 0}\bra{\Phi}\tilde{a}_{\lambda,\omega}(t) 
                \tilde{a}^{\dag}_{\lambda,\omega}(s)\ket{\Phi}
            &=  \lim_{\lambda \to 0} \int d^3\mathbf{k}\sum_q \modsq{g_{\mathbf{k},q}}
                \frac{1}{\lambda^2} e^{-i(\omega_{\mathbf{k}} - \omega)(t-s)/\lambda^2}
                \\
             &=  \delta(t-s) \int d^3\mathbf{k}\sum_q \modsq{g_{\mathbf{k},q}}2\pi
             \delta(\omega_{\mathbf{k}} - \omega)\\
             &= \kappa(\omega)\delta(t-s) 
\end{align}
where we used the identity $\lim_{\lambda\to 0}e^{-i\omega t/\lambda^2}/\lambda^2 = 2\pi\delta(t)\delta(\omega)$\footnote{Following Proposition 1.2.1 in \citet{Accardi:2002a}, we can easily demonstrate this with respect to two Schwartz functions, which are infinitely differentiable and whose derivatives decrease to zero faster than any polynomial.  Using the test functions $\psi(t),\phi(\omega)$, we have
\begin{equation}
    I = \frac{1}{\lambda^2}\int dt \psi(t) \int d\omega \phi(\omega) e^{-i\omega t/\lambda^2}
\end{equation}
Setting $t = \lambda^2 \tau$ this becomes
\begin{equation}
    I =\int d\tau \psi(\lambda^2 \tau) \int d\omega \phi(\omega) e^{-i\omega\tau}
      = \sqrt{2\pi} \int d\tau\psi(\lambda^2\tau) \hat{\phi}(\tau)
\end{equation} 
where $\hat{\phi}$ is the Fourier transform of $\phi$.  We then have
\begin{equation}
    \lim_{\lambda \to 0} \sqrt{2\pi} \int d\tau\psi(\lambda^2\tau) \hat{\phi}(\tau)
      = \sqrt{2\pi}\psi(0) \int d\tau \hat{\phi}(\tau)
      = 2\pi\psi(0)\phi(0)
\end{equation} 
which shows that this is equivalent to $\lim_{\lambda\to 0}e^{-i\omega t/\lambda^2}/\lambda^2 = 2\pi\delta(t)\delta(\omega)$ in the sense of distributions.}, which in turn allows the introduction of $\kappa(\omega) = 2\pi\sum_q \modsq{g_{w_{\mathbf{k}}^{-1}(\omega),q}}$, the mode function evaluated at $\omega_{\mathbf{k}} = \omega$.  This suggests that the limit of $\tilde{a}_{\lambda,+\omega_{eg}}(t)$ is a delta-correlated quantum Wiener process with strength $\sqrt{\kappa(\omega_{eg})}$.  However, looking at the limit for $\tilde{a}_{\lambda,-\omega_{eg}}(t)$ requires evaluating the coupling strength at $\omega_{\mathbf{k}} = -\omega_{eg}$, which is impossible since $\omega_{\mathbf{k}} \geq 0$ by definition.  In fact, one can show that these non-energy conserving terms go to zero in the weak coupling limit, thus recovering the rotating wave approximation.

Although we have demonstrated that the correlations of the rescaled field operators limit to those for a quantum Gaussian white noise (cf.~Eq.~\eqref{quantum:eq:field_mode_wiener_relation}), there is considerably more effort required to find the limit of the interaction picture propagator $U_t^{\lambda}$, which involves studying the convergence of each term in the Dyson series expansion.  Several chapters of \citet{Accardi:2002a} are devoted to this task, indicating it is certainly beyond the scope of this overview.  Instead,we simply quote the perhaps unsurprising result given by the following \emph{quantum Stratonovich propagator} 
\begin{equation}
    dU_t = \sqrt{\kappa} \left[dA_t^{\dag}\circ \sigma_{-}dU_t - dA_t \circ \sigma_+ U_t\right]
\end{equation}
where I have set $\kappa = \kappa(\omega_{eg})$.  The Stratonovich increments are defined as they were classically in Definition \ref{def:stratonovich}, except we now have to worry about operators not commuting.  Unlike the quantum It\^{o} formulation, the quantum Stratonovich noise increments \emph{do not commute} with adapted operators, e.g.~$ O\circ dA_t \neq dA_t \circ O$, although they do transform via the normal chain rule.  The fact that the stochastic limit of a standard quantum differential equation is a quantum Stratonovich equation is precisely the quantum generalization of the classical Wong-Zakai Theorem \eqref{thm:wong_zakai}.  As was the case then, we can still convert to the following \emph{quantum It\^{o} propagator} or \emph{stochastic propagator}
\begin{equation} \label{quantum:eq:ito_propagator_two_level_atom}
    dU_t = \left[\sqrt{\kappa} \sigma_{-} dA_t^{\dag} - \sqrt{\kappa}\sigma_+ dA_t
            - \frac{1}{2}\kappa \sigma_+\sigma_- dt \right]U_t .
\end{equation}
The upside of the It\^{o} form is that the quantum noise increments commute with adapted processes and are zero in vacuum expectation, although one now needs to use the quantum It\^{o} chain rule rather than the normal calculus chain rule.  
\subsection{Summary}
We have briefly developed the quantum generalizations of stochastic noise processes and stochastic differential equations discussed in Chapter \ref{chapter:classical}. Certainly, that is the take away message from this section---that in spirit, the quantum versions are really no different than their classical counterparts.  As such, I have not stressed many of the statistical features and intuitions we focused on classically because they are more or less the same.  The main differences come down to the non-commutativity of quantum mechanics, but as we saw when developing the quantum It\^{o} integral, the solutions amount to a careful extension of the classical approach.  The more novel discussion was devoted to developing a physical model of quantum white noise, since unlike was the case classically, we now have a very particular class of physical systems in mind.  We found that for such quantum optics systems, the quantum electromagnetic field serves as an excellent model of quantum white noise.  Moreover, the quantum stochastic description arises naturally from a weak coupling limit of standard physical models.  In the following section, we will use these techniques to carefully define and solve the quantum optics filtering problem.
\section{Quantum Filtering Theory}
By this point, it should come as no surprise that a broad class of models in quantum optics are well described by the quantum stochastic formalism.  One such setup is shown schematically in Figure \ref{quantum:fig:filter_setup_annotated}, in which an input field, described in terms of quantum Wiener processes $dA_t,dA_t^{\dag}$, interacts with a cloud of atoms.  Although the two systems are initially unentangled, a joint interaction such as the one in Eq.~\ref{quantum:eq:ito_propagator_two_level_atom} would correlate them.  Therefore, later measurements of the scattered light field should contain some information about the cloud of atoms, although they will also contain the inherent quantum noise fluctuations of the optical field.  The task of the quantum filter is to process this continuous measurement stream in order to best estimate the state of the atomic system.  In this section, we formalize this problem using our newly gained quantum probability and quantum stochastic skills and then derive the quantum filtering equation in analogy to the reference probability method we used to derive the classical filtering equation.
\begin{figure}[bt]
    \centering
        \includegraphics[scale=1]{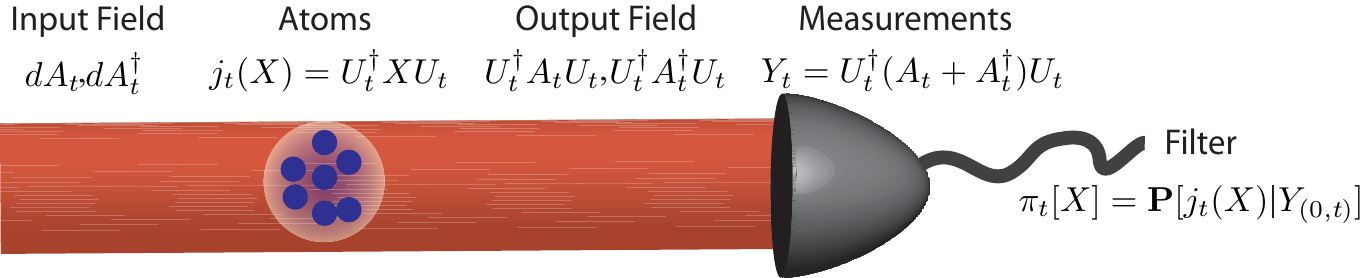}
        \caption[Schematic of continuous measurement in quantum optics]{Schematic of continuous measurement in quantum optics, in which light scattered
        by a cloud of atoms is continuously measured by a photodetector which is then filtered.  Based on Fig 5.1 in \citep{Bouten:2007b}. }
    \label{quantum:fig:filter_setup_annotated}
\end{figure}
\subsection{Statement of the Filtering Problem}
Following the classical approach, we would like to pose the quantum filtering problem in analogy to the systems-observations pair of Eqs.~\eqref{classical_filtering:eq:system} and \eqref{classical_filtering:eq:observations}, where now the system corresponds to the state of the atoms and the observations correspond to the measurements of the field.  Both are described by the quantum probability model considered in Subsection \ref{quantum:subsect:quantum_stochastic_calculus}.  The time evolution and corresponding QSDEs of both field and atom observables are determined by the quantum stochastic propagator for the experiment under consideration.  Based on Eq.~\ref{quantum:eq:ito_propagator_two_level_atom}, we will consider the generic stochastic propagator
\begin{equation} \label{quantum:eq:generic_propagator}
    dU_t = \left[ L dA_t^{\dag} - L^{\dag} dA_t - \frac{1}{2} L^{\dag}Ldt - iH dt\right]U_t
\end{equation}
where $L$ is some atomic operator that results in the weak coupling limit and $H$ is an arbitrary atomic Hamiltonian.  We also fix the quantum state as the product state $\rho \otimes \ketbra{\Phi}{\Phi}$, where $\rho$ is an arbitrary atomic state and $\ket{\Phi}$ is the usual vacuum state.  Note that we could equally well consider coherent field states by noting that $\ket{e(f)} = W_{f}\ket{\Phi}$, suggesting that we could explicitly include the Weyl operator in our dynamical equations and work with a vacuum reference state.  In particular, $U_t \mapsto U_tW_t$ and applying the quantum It\^{o} rules to $d(U_tW_t)$ would give a new propagator which describes the displaced dynamics.  Since this is not an essential part of the filtering problem, we will not dwell on it here.

Instead, we now focus on the atomic evolution given in terms of the \emph{quantum flow} or \emph{Heisenberg evolution}\footnote{It is a bit disingenuous to call this the Heisenberg evolution since $U_t$ is really the interaction picture propagator.  To return to the Heisenberg picture, we would need to rotate back by the free Hamiltonians.  Often the initial atomic state is an eigenstate of the free system Hamiltonian $H_S$, so that this return rotation introduces canceling phases, indicating $U_t$ already gives the Heisenberg evolution.  However, this is not the case generally and will depend on the specifics of the system at hand. }, written as $j_t(X) = U_t^{\dag}(X \otimes I) U_t$ for any atomic observable $X \in \mathscr{N}_s$.  Using the quantum It\^{o} rules and the fact that an explicitly time-independent observable satisfies $dX = 0$, $j_t(X)$ satisfies the following QSDE:
\begin{align}
    dj_t(X) &= dU_t^{\dag}(XU_t) + U_t^{\dag}XdU_t + dU_t^{\dag}XdU_t\\
            &= U_t^{\dag}\left[L^{\dag}X dA_t - LXdA_t^{\dag} - \frac{1}{2}L^{\dag}LX dt
                + iH X\right]U_t\nonumber\\
            &+ U_t^{\dag}\left[ XLdA_t^{\dag} - XL^{\dag}dA_t -\frac{1}{2}XL^{\dag}L
                - iXH\right]U_t\nonumber\\
            &+ U_t^{\dag}\left[ L^{\dag}XL dt\right]U_t\\
            &= j_t(\mathcal{L}_{L,H}[X])dt + j_t([L^{\dag},X])dA_t + j_t([X,L])dA_t^{\dag} \label{quantum:eq:jt_qsde}
\end{align}
where we have used the fact that the fundamental noises are non-anticipative to pull them out of the $j_t$ terms and introduced the familiar \emph{Lindblad generator} 
\begin{equation} \label{quantum:eq:lindblad_generator}
    \mathcal{L}_{L,H}[X] = L^{\dag}XL - \frac{1}{2}L^{\dag}LX - \frac{1}{2}XL^{\dag}L
                 + i[H,X] .
\end{equation}
The form of Eq.~\eqref{quantum:eq:jt_qsde} is pleasing, as it contains a deterministic part which depends on the familiar open quantum system Lindblad generator, plus quantum noise pieces which contain extra information regarding the field.  Thus, if we were to take a partial trace of the field system, we would recover the standard Heisenberg picture master equation.

Of course, if it were easy to measure the atomic system directly, we would be done at this point.  We have a dynamical equation which describes the exact evolution of the atoms and field in the weak coupling limit and $j_t(X)$ is the corresponding evolved atomic operator to be measured.  But lacking such direct atomitc measurements, we instead use the outgoing or scattered field as a probe of the atomic dynamics.  In order to perform inference, we therefore need to fix the probe observable we intend to measure.  The most common quantum optics measurements are photon counting, related to the $\Lambda_t$ process, and homodyne detection, related to the noise quadratures $e^{-i\phi}A_t + e^{i\phi}A_t^{\dag}$.  I will focus on the latter measurement and refer the reader to \citep{Bouten:2007b,Barchielli:2003} for more discussion on the topic. 

We still need to decide \emph{which} quadrature to measure and should so by looking at the form of $dU_t$.  Comparing to our two level atom example, $U_t$ appears to have come from a system-field coupling of the form $i(L+L^{\dag})(a - a^{\dag})$.  That is, the atoms appear to couple to the $p$ field quadrature.  We therefore would not want to measure that quadrature of the field, since it commutes with the coupling Hamiltonian and therefore remains unchanged under evolution by $U_t$.  Instead, we want to measure the orthogonal $x$ quadrature, which will evolve non-trivially under $U_t$ and carry off some information about its interaction with the atoms.  We thus write the measurements as $Y_t = U_t^{\dag}(A_t + A_t^{\dag})U_t$, which corresponds to the scattered $x$ quadrature.  We can again use the quantum It\^{o} rules and some patience to calculate\footnote{I don't list all the steps below because it is a useful exercise in the It\^{o} rules to calculate these terms by hand.  I will give the hint that all terms involving $(A_t + A_t^{\dag})$ are most easily treated together, as they simplify in a fashion similar to $d(U_t^{\dag}U_t)$, which we know is zero if $U_t^{\dag}U_t = I$.}
\begin{align}
    dY_t &= dU_t^{\dag}(A_t + A_t^{\dag})U_t + U_t^{\dag}d\left[(A_t + A_t^{\dag})U_t\right]
          + dU_t^{\dag}d\left[(A_t + A_t^{\dag})U_t\right]\\
         &= dU_t^{\dag}(A_t + A_t^{\dag})U_t + U_t^{\dag}(dA_t + dA_t^{\dag})U_t
          + U_t^{\dag}(A_t + A_t^{\dag})dU_t \nonumber\\
         & + U_t^{\dag}(dA_t + dA_t^{\dag})dU_t + dU_t^{\dag}(dA_t + dA_t^{\dag})U_t
           + dU_t^{\dag}(A_t + A_t^{\dag})dU_t \nonumber\\
         & + dU_t^{\dag}(dA_t + dA_t^{\dag})dU_t\\
         &= j_t(L + L^{\dag})dt + dA_t + dA_t^{\dag}
\end{align}
We see that the measurements are a noisy observation of $j_t(L+L^{\dag})$, albeit corrupted by the input $x$ quadrature noise $dA_t + dA_t^{\dag}$.  Recall that the $p$-quadrature drives the $j_t$ evolution, so that both non-commuting noises are somehow mixed up in the observations process.

We would like to define the filtering problem as $\mathbf{P}[j_t(X) | \mathscr{Y}_{(0,t)}]$, the conditional expectation of the atomic state given the entire measurement record.  However, as we saw in defining the quantum conditional expectation, we should be careful that this is actually a well-posed inference problem.  The first question is whether the entire observations process generates a commutative algebra.  If it did not, we would not be able to incrementally build up information over time, as future measurements could not be combined with past measurements.  Fortunately, it is straightforward to check that the observations are commutative which is equivalent to satisfying the \emph{self-nondemolition property}.  As a start, note that
\begin{equation}
	\begin{split}
	    U_t^{\dag}(A_s + A_s^{\dag})U_t &= U_s^{\dag}(A_s + A_s^{\dag})U_s
	                     + \int_s^{t}U_\tau^{\dag}\mathcal{L}_{L,H}(A+A^{\dag})U_\tau d\tau\\
	                     & + \int_s^{t}U_\tau^{\dag}[L^{\dag},(A+A^{\dag})]U_\tau d\tau
	                     + \int_s^{t}U_\tau^{\dag}[(A+A^{\dag}),L]U_\tau d\tau\\
	                     &= U_s^{\dag}(A_s + A_s^{\dag})U_s
	\end{split}
\end{equation}
indicating $Y_s =U_t^{\dag}(A_s + A_s^{\dag})U_t$ for $t \geq s$.  This is essentially a consequence of the Markov approximation implicit in the weak coupling limit.  We see that $A_s + A_s^{\dag}$, which corresponds to the free fields at time $s$ evolved under the free field Hamiltonian, only interacts with the atoms at time $s$ and is then moves on.  In other words, under the Markov approximation, the interaction occurs instantaneously, so that after $s < t$, $U_t$ does nothing to the interaction picture operators $A_s + A_s^{\dag}$.  Using this fact, we can readily verify
\begin{equation}
    [Y_t,Y_s] = [U^{\dag}_t(A_t + A_t^{\dag})U_t, U^{\dag}_t(A_s + A_s^{\dag})U_t]
              = U_t^{\dag}[A_t + A_t^{\dag}, A_s + A_s^{\dag}]U_t = 0 .
\end{equation}
Thus, $\mathscr{Y}_{(0,t)} = \vN{ Y_s : 0 \leq s \leq t}$ is a commutative von Neumann algebra which corresponds to a classical stochastic process via the spectral theorem.  This fixes the commutative algebra $\mathscr{A} = \mathscr{Y}_{(0,t)}$ used to define the quantum conditional expectation.  The only remaining check is that $j_t(X)$ is in its commutant, which is easily verified using the property just discussed,
\begin{equation}
    [j_t(X), Y_s] = [U_t^{\dag}XU_t, U_t^{\dag}(A_s + A_s^{\dag})U_t]
                  = 0 .
\end{equation} 

We therefore have a well-defined inference problem which is summarized in the following definition.
\begin{definition}\label{def:quantum:filtering_problem}
    The \emph{filtering problem} in quantum optics, defined on the quantum probability space $(\mathscr{N}_s\otimes\mathscr{N}_f)$ with state $\mathbf{P} = \mathbf{P}_s \otimes\mathbf{P}_\phi$, is to calculate
    \begin{equation}
        \pi_t[X] = \mathbf{P}(j_t(X) | \mathscr{Y}_{(0,t)})
    \end{equation}
for the system-observations pair
\begin{align}
    dj_t(X) & = j_t(\mathcal{L}_{L,H}[X])dt + j_t([L^{\dag},X])dA_t + j_t([X,L])dA_t^{\dag}\\
    dY_t &= j_t(L + L^{\dag})dt + dA_t + dA_t^{\dag} .
\end{align}
\end{definition}

Before solving this problem in the following subsection, let us pause and reflect on what makes this different than the classical filtering problem.  Classically, we considered the problem of estimating the state of a stochastically evolved system given observations corrupted by independent noise.  Thus, there was some underlying system state to find and only technical reasons limited our observation of that state.  In the quantum case, we have not added any extra corrupting noise; the limited observations of the system are a consequence of the fundamental uncertainties in quantum mechanics which arise from non-commuting observables, which here are the two field quadratures.  Furthermore, there is no hidden or underlying state independent of the observations process, since measurement back-action non-trivially changes the state of the system.  Fortunately, the structure of the filter is such that we can still estimate $j_t(X)$ using the observations process.  
\subsection{Quantum Filtering Equation}
Again we will follow the classical reference probability approach taken in solving the non-linear filtering problem.  Recall that the approach is to find a new measure under which the observations and state are statistically independent, so that the conditional expectation becomes trivial to evaluate.  For us, this means finding a new state under which the quantum conditional expectation is much simpler, after which an application of the quantum Bayes rule in Theorem \ref{thm:quantum_bayes_formula} allows us to relate this back to the original problem.   All of these steps were taken in Example \ref{quantum:example:stern_gerlach}, so you may refer to that for another demonstration of what follows.

Our first step is to find a quantum analogue of the Girsanov transformation, which here amounts to finding a state under which $Y_t$ is a Wiener process.  Instead, it actually is more convenient to work with the input quadrature $Z_t = A_t + A_t^{\dag}$ directly, suggesting we work under a new state
\begin{equation}
    \mathbf{Q}_t(X) = \mathbf{P}(U_t^{\dag}X U_t) .
\end{equation}
Thus, we move the time-evolution via $U_t$ onto the states and work with operators in the free field interaction picture.  From Example \ref{quantum:example:stern_gerlach}, we have $\mathbf{P}(U_t^{\dag}XU_t | U_t^{\dag}\mathscr{C}_t U_t) = U_t^{\dag}\mathbf{Q}_t(X | \mathscr{C}_t)U_t$, where $\mathscr{C}_t$ is the von Neumann algebra generated by $Z_t$ which is related to the original observations algebra via $\mathscr{Y}_{(0,t)} = U_t^{\dag}\mathscr{C}_tU_t$.  Thus, we have
\begin{equation} \label{quantum:eq:filtering_state_relation_temp}
    \mathbf{P}(j_t(X) | \mathscr{Y}_{(0,t)}) = U_t^{\dag}\mathbf{Q}_t(X | \mathscr{C}_t)U_t .
\end{equation}
The hope is that the conditional expectation may be calculated more easily under $\mathbf{Q}_t$, where the $U_t$ evolution is part of the state, after which we reapply $U_t$ to return to our original picture.  The Girsanov analogy comes from noting that $Z_t$ is precisely a classical Wiener process under the original state $\mathbf{P}$, so that applying the quantum Bayes rule to Eq.~\eqref{quantum:eq:filtering_state_relation_temp} would allow us to easily evaluate the $\mathbf{Q}_t$ conditional expectation in terms of $\mathbf{P}$ conditional expectations.  But just as in Example \ref{quantum:example:stern_gerlach} we have the problem that our change of measure operator $U_t$ is not in the commutant $\mathscr{C}_t^{\prime}$ since $Z_t$, which is the $x$ field quadrature, does not commute with the $p$ field quadrature which generates $U_t$. 

Fortunately, the vacuum reference state provides a nice means for finding a $V_t \in \mathscr{C}_t^{\prime}$ which nonetheless satisfies $\mathbf{P}(U_t^{\dag}X U_t) = \mathbf{P}(V_t^{\dag}X V_t)$ for all atomic operators $X$.  Such a $V_t$ is governed by the QSDE
\begin{equation}
    dV_t = \left[ L(dA_t + dA_t^{\dag}) - \frac{1}{2}L^{\dag}Ldt - iHdt\right]V_t
\end{equation}
which will give the same \emph{vacuum expectation} as $U_t$ since $dA_t\ket{\Phi} = 0$.  Clearly, $V_t \in \mathscr{C}_t^{\prime}$ since it is driven by the $x$ quadrature noise $Z_t = A_t + A_t^{\dag}$. 

An application of the quantum Bayes formula in Theorem \ref{thm:quantum_bayes_formula} gives the \emph{quantum Kallianpur-Striebel formula} 
\begin{equation}
    \pi_t[X] = \frac{U_t^{\dag}\mathbf{P}(V_t^{\dag}X V_t | \mathscr{C}_t)U_t}
                    {U_t^{\dag}\mathbf{P}(V_t^{\dag}V_t | \mathscr{C}_t)U_t}
             = \frac{\sigma_t(X)}{\sigma_t{(I)}} ,
\end{equation}
where all the condition is on $\mathscr{C}_t$, the algebra generated by the $\mathbf{P}$-Wiener process $A_t + A_t^{\dag}$.  In short, the whole point of introducing $\mathbf{Q}_t$ and $V_t$ was to make $\mathscr{C}_t$ the conditioned algebra, whose statistics we know given our understanding of the Wiener process. 

We now focus on deriving an SDE for $\sigma_t(X)$, which is done by explicit calculation.  From the quantum It\^{o} rules in integral form, we have
\begin{equation}
    V_t^{\dag}XV_t = X + \int_0^t V_s^{\dag}\mathcal{L}_{L,H}[X]V_s ds
                   + \int_0^t V_s^{\dag}(L^{\dag}X + XL) V_s d(A_s + A_s^{\dag}) .
\end{equation}
This is easily derived by noting that $j_t(X)$ in Eq.~\eqref{quantum:eq:jt_qsde} is identical save for changing from $U_t$ to $V_t$, which amounts to mapping $-L^{\dag}dA_t \mapsto +LdA_t$.  We next evaluate the conditional expectations of each term in this expression, using the fact that the conditional expectation may be pulled inside the integrals to find
\begin{multline}
    \mathbf{P}(V_t^{\dag}XV | \mathscr{C}_t) = \mathbf{P}(X)
                +  \int_0^t\mathbf{P}(V_s^{\dag}\mathcal{L}_{L,H}[X]V_s | \mathscr{C}_s)ds\\
                + \int_0^t \mathbf{P}(V_s^{\dag}(L^{\dag}X + XL) V_s | \mathscr{C}_s)dZ_s
\end{multline}
Finally, we apply the It\^{o} rules to $U_t^{\dag}\mathbf{P}(V_t^{\dag}XV | \mathscr{C}_t)U_t$ to find
\begin{equation}
    d\sigma_t(X) = \sigma_t(\mathcal{L}_{L,H}[X])dt + \sigma_t(L^{\dag}X + XL)dY_t
\end{equation}
where we have identified the observations process $Y_t = U_t^{\dag}Z_tU_t$.  In order to recover the normalized form, we again use the It\^{o} rules as we did in solving the classical \emph{Kushner-Stratonovich} equation, cf.~Eq.~\eqref{classical:eq:unnormalized_filter_ratio_sde}, to arrive at the filter given in the following theorem.
\begin{theorem}[\textbf{Quantum Filtering Equation}]\label{thm:quantum_filtering_equation}
    The solution to the quantum filtering problem satisfies the SDE
    \begin{equation} \label{quantum:eq:quantum_filter}
        d\pi_t[X] = \pi_t[\mathcal{L}_{L,H}[X]]dt
                  + \left(\pi_t[L^{\dag}X + XL] - \pi_t[L^{\dag} + L]\pi_t[X]\right)
                    \left(dY_t - \pi_t[L + L^{\dag}]dt\right)
    \end{equation}
    with $\pi_0[X] = \mathbf{P}_S(X) = \Tr{X\rho}$. 
\end{theorem}
\noindent  This is precisely a recursive formula which may integrated on a classical computer by processing the observations process $Y_t$.  Comparing this to the classical non-linear filter in Eq.~\ref{filtering:eq:kushner_stratonovich}, we see the \emph{innovations process} appearing as $dY_t - \pi_t[L + L^{\dag}]dt$, which is again precisely a classical Wiener process.  This provides a very nice interpretation of the resulting filter, in which the innovations drive the estimate by pulling out all the information from the measurement process which is not already in our estimate $\pi_t[X]$.  This information, which includes the quantum noise $dA_t + dA_t^{\dag}$ in addition to the true atomic state $j_t(X)$, is then used to condition our estimate of the atomic system in accordance with the expected back action which results from the field measurement.   

Just as we saw classically, the form of the filtering equation above is not always convenient, since it requires iterating $\pi_t[X],\pi_t[L + L^{\dag}],\pi_t[L^{\dag}X],\cdots$ until a closed system of equations is found.  Instead, one often works with a state representation in terms of the \emph{conditional density matrix} $\rho_t$ which satisfies $\pi_t[X] = \Tr{X\rho_t}$ for all atomic operators $X$.  Plugging this into Eq.~\eqref{quantum:eq:quantum_filter} gives the quantum filter in its adjoint form as
\begin{equation} \label{quantum:eq:quantum_filter_rho_form}
    d\rho_t = -i[H,\rho_t]dt + (L\rho_tL^{\dag} - \frac{1}{2}L^{\dag}L\rho_t 
                     -\frac{1}{2}\rho_tL^{\dag}L)dt
            + (L\rho_t + \rho_t L^{\dag} - \Tr{(L + L^{\dag})\rho_t}\rho_t)dW_t
\end{equation}
where we recognize the familiar Lindblad form for the deterministic pieces, characteristic of an open quantum system master equation (see \citep{Walls:2008a} for more detail).  The stochastic term, which is non-linear in $\rho_t$, performs the conditioning via the innovations process which I have written as the Wiener process $dW_t = dY_t - \Tr{(L+L^{\dag})\rho_t}dt$.  The adjoint form suggests a nice interpretation of the filter as a \emph{continuous measurement} of the observable $L+L^{\dag}$.  Indeed, one can show \citep{Adler:2001a} that if $H = 0$, the steady-state of $\rho_t$ is precisely an eigenstate of $L+L^{\dag}$ and occurs with probability $\Tr{(L + L^{\dag})\rho_0}$.  Thus, rather than considering an instantaneous projective measurement of $L+L^{\dag}$, the measurement is extended in time and appears as a deterministically driven Wiener process, opening the door for performing feedback \citep{Wiseman:1994a} using the current filtered estimate.  Such a possibility will be considered in Chapter \ref{chapter:error_correction}.

Before closing this section with an example, I do want to note that continuous measurement can be considered entirely within the generalized measurement and quantum operations framework of quantum information theory \citep{Jacobs:2006a}.  Although many of the mathematical subtleties are glossed over, I believe there is also an interpretational issue which arises.  Specifically, when working with the conditional density matrix formalism, the measurements are usually written as
\begin{equation}
    dY_t = \Tr{(L+L^{\dag})\rho_t}dt + dW_t
\end{equation}
where $\rho_t$ is the conditional density matrix and $dW_t$ is a Wiener process that arises from taking the central limit of many infinitesimal measurements.  But as we saw above, this is \emph{not} the measurement process, which actually contains the \emph{true} system state $j_t(L+L^{\dag})$ and the quantum noise $dA_t + dA_t^{\dag}$.  The innovations Wiener process $dW_t$ only arises by explicitly subtracting the current estimate from the measurements.  That is, if we want to \emph{learn} something about the system, the measurements better contain some information about it, rather than just our current estimate corrupted by noise.  Philosophically, this amounts to deciding whether $\rho_t$ is the true state of the atoms or whether it is simply our estimate of the true state.  I prefer the latter perspective, which allows for a careful consideration of the stability of the filter under incorrect initial state estimates \citep{vanHandel:2008a}.  But such a case is not uncommon, especially when the continuous measurement process is being used to measure an unknown initial state.  

\begin{example}[Qubit in a magnetic field]\label{quantum:example:qubit_filter}
    \begin{figure}[ht]
    	\centering
    		\includegraphics{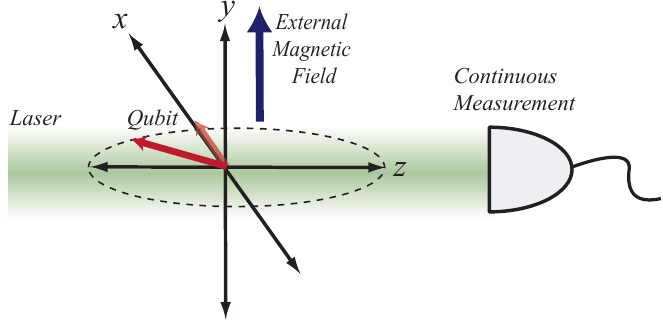}
    	\caption[Schematic for continuous qubit measurement in Example \ref{quantum:example:qubit_filter} ]{Continuous-measurement of single qubit precessing in an external magnetic field}
    	\label{quantum:example:qubit:fig:schematic}
    \end{figure}
    \begin{figure}[t]
    	\centering
    		\includegraphics[scale=1]{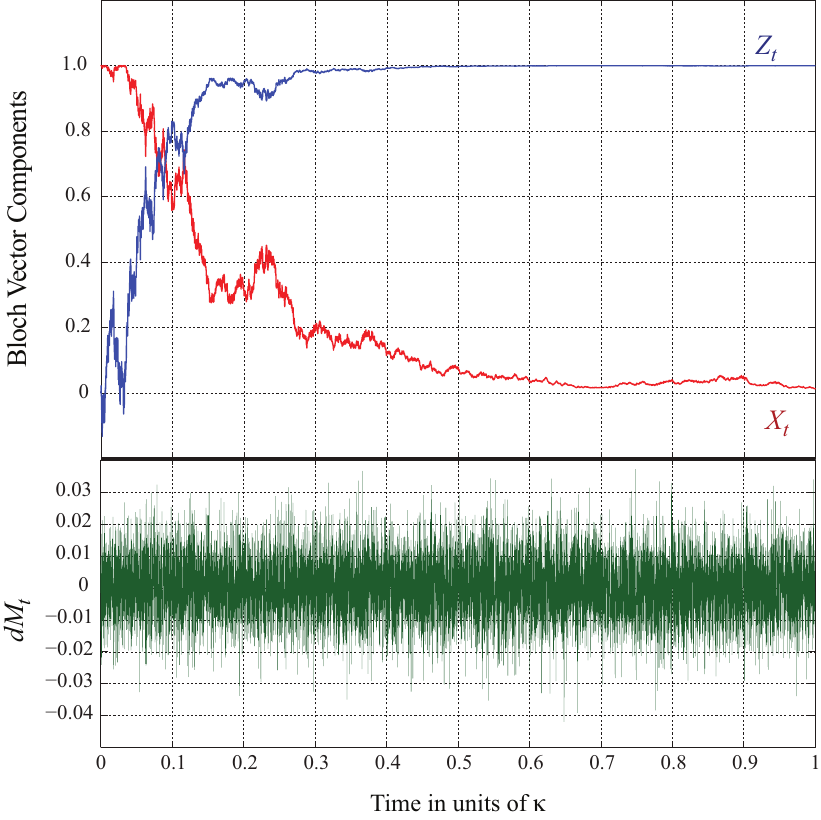}
    	\caption[Quantum filter simulation for qubit in Example \ref{quantum:example:qubit_filter}]{(Bottom) Simulated typical measurement trajectory for continuous $Z$ measurement, $\kappa = 1$, $B = 0$ (Top) Filtered values of $\pi_t[\sigma_x]$ and $\pi_t[\sigma_z]$ for simulated trajectory}
    	\label{quantum:example:qubit:fig:typical_trajectory}
    \end{figure}
    Consider the setup depicted in Figure \ref{quantum:example:qubit:fig:schematic}.  A qubit, initially in the pure state $\ket{+x}$, precesses about a magnetic field $B$ while undergoing a continuous measurement along $z$.  In terms of the general framework, $H = B \sigma_y$ and $L = \sqrt{\kappa}\sigma_z$, where $\sqrt{\kappa}$ is the continuous measurement strength in the weak coupling limit.  We will not dwell on the underlying physical mechanism which gives rise to the $\sigma_z$ measurement, though continuous polarimetry measurements could suffice \cite{Bouten:2007a}.  Plugging into Eq.~\eqref{quantum:eq:quantum_filter}, the quantum filter for the Bloch vector $n_t = (\pi_t[\sigma_x],\pi_t[\sigma_y],\pi_t[\sigma_z])$ is 
    \begin{align}
    	d\pi_t[\sigma_x] &= 2B\pi_t[\sigma_z] dt - 2\kappa \pi_t[\sigma_x] dt
    					  - 2\sqrt{\kappa}\pi_t[\sigma_x]\pi_t[\sigma_z] dW_t\\
    	d\pi_t[\sigma_y] &= -2M\pi_t[\sigma_y] dt - 2\sqrt{\kappa}\pi_t[\sigma_x]\pi_t[\sigma_y] dW_t\\
    	d\pi_t[\sigma_z] &= -2B\pi_t[\sigma_x]dt  + 2\sqrt{\kappa}( 1 -  \pi_t[\sigma_z]^2)dW_t
    \end{align}
    with innovations $dW_t = dM_t - 2\sqrt{\kappa}\pi_t[\sigma_z] dt$.  It is not difficult to verify that the quantum filter maintains pure states and that the initial state $n_0 = (1,0,0)$ remains on the Bloch circle in the $x$-$z$ plane.  Letting $\theta$ be the angle from the positive $x$-axis such that $\tan{\theta} = \pi_t[\sigma_z]/\pi_t[\sigma_x]$, we then simplify the filter to
    \begin{equation} \label{quantum:eq:known_qubit_filter}
    	d\theta_t = -2Bdt + \kappa \sin(2\theta_t)dt + 2\sqrt{\kappa}\cos(\theta_t)dW_t 
    \end{equation}
    where now $dW_t = dM_t - 2 \sqrt{\kappa}\sin\theta_t$.  Figure \ref{quantum:example:qubit:fig:typical_trajectory} shows a computer simulation of a typical measurement trajectory and filtered Bloch vector values when $B = 0$.  We see that the initial $+x$ state indeed collapses to a $+z$ eigenstate, which is then a fixed state of the continuous $\sigma_z$ measurement.
\end{example}   
\section{Summary}
Perhaps a yet unstated purpose of this chapter was to probe the distinction we tend to hold between what is quantum and what is classical in quantum information theory.  We found that a commutative set of operators is well-described by a classical probability model and that inference between commuting observables is readily performed in terms of the classical tools we developed in Chapter \ref{chapter:classical}.  Indeed, the quantum filter, which is capable of describing a continuous measurement of a quantum system as a stochastic process, is an entirely classical object.  Moreover, the filter did not require using the standard projection postulate but instead recovers it in the long time, strong measurement strength limit.  This perspective allows us to ``look inside'' the projective measurement, watch the wave function collapse and potentially modify it via feedback.  In short, the traditional weirdness of quantum back-action arises naturally from the interplay of classical conditioning and quantum dynamics.  The second goal was to again convince the reader that spending a little time familiarizing oneself with the mathematics of quantum probability theory, quantum stochastic calculus and quantum filtering provides a sophisticated yet simple approach to solving many problems in quantum optics.  In fact, most of the original research in this thesis leverages the filtering formalism in this chapter, solving the problem of continuous-time quantum parameter estimation and studying quantum error correction via continuous measurement and feedback.  In relating those results, I hope the reader will appreciate the groundwork they laid in this chapter and the ease with which classical stochastic control and estimation methods are trivially adapted to the quantum setting.  
\chapter{Quantum Parameter Estimation}
\label{chapter:quantum_parameter_estimation}
In this chapter, I extend the quantum filtering techniques of Chapter \ref{chapter:quantum} to allow for the estimation of unknown parameters which drive the evolution of the system undergoing continuous measurement.  By embedding parameter estimation in the standard quantum filtering formalism, we will find the optimal Bayesian filter for cases when the parameter takes on a finite range of values.  For cases when the parameter is continuous valued, I develop \emph{quantum particle filters} as a practical computational method for quantum parameter estimation.  The techniques developed within this chapter were published in \citep{Chase:2009a}.
\section{Introduction}
Determining unknown values of parameters from noisy measurements is a ubiquitous problem in physics and engineering.    In quantum mechanics, the single-parameter problem is posed as determining a coupling parameter $\xi$ that controls the evolution of a probe quantum system via a Hamiltonian of the form $H_{\xi} = \xi H_0$ \citep{Helstrom:1976a,Holevo:1982a,Braunstein:1994a,Braunstein:1995a,Giovannetti:2004a,Giovannetti:2006a,Boixo:2007a}.  Traditionally, an estimation procedure proceeds by (i) preparing an ensemble of probe systems, either independently or jointly; (ii) evolving the ensemble under $H_{\xi}$; (iii) measuring an appropriate observable in order to infer $\xi$.  The quantum Cram\`{e}r-Rao bound \citep{Cramer:1946a,Helstrom:1976a,Holevo:1982a,Braunstein:1994a,Braunstein:1995a} gives the optimal sensitivity for \emph{any} possible estimator and much research has focused on achieving this bound in practice, using entangled probe states and nonlinear probe Hamiltonians \citep{Nagata:2007a,Pezze:2007a,Woolley:2008a}. 

Yet, it is often technically difficult to prepare the exotic states and Hamiltonians needed for improved sensitivity.    Instead, an experiment is usually repeated many times to build up sufficient statistics for the estimator.  In contrast, the burgeoning field of continuous quantum measurement \citep{Bouten:2006a} provides an opportunity for on-line \emph{single-shot} parameter estimation, in which an estimate is provided in near real-time using a measurement trajectory from a single probe system.  Parameter estimation via continuous measurement has been previously studied in the context of force estimation \citep{Verstraete:2001a} and magnetometry \citep{Geremia:2003a}.  Although Verstraete et. al develop a general framework for quantum parameter estimation, both of \citep{Verstraete:2001a,Geremia:2003a} focus on the readily tractable case when the dynamical equations are linear and the quantum states have Gaussian statistics.  In this case, the optimal estimator is the quantum analog of the classical Kalman filter \citep{Belavkin:1999a,Kalman:1960a,Kalman:1961a}, seen in example \ref{filtering:example:parameter_estimation} in Chapter \ref{chapter:classical}.

In this chapter, I develop on-line estimators for continuous measurement when the dynamics and states are not restricted.  Rather than focusing on fundamental quantum limits (which is the topic of Chapter \ref{chapter:magnetometry}), I instead consider the more basic problem of developing an actual parameter filter for use with continuous quantum measurements.  By embedding parameter estimation in the standard quantum filtering formalism \citep{Bouten:2006a}, I construct the optimal Bayesian estimator for parameters drawn from a finite dimensional set.  The resulting filter is a generalized form of one derived by Jacobs for binary state discrimination \citep{Jacobs:2006a}.  Using recent stability results of van Handel \citep{vanHandel:2008a}, I give a simple check for whether the estimator can successfully track to the true parameter value in an asymptotic time limit.  For cases when the parameter is continuous valued, I develop \emph{quantum particle filters} as a practical computational method for quantum parameter estimation.  These are analogous to, and inspired by, particle filtering methods that have had much success in classical filtering theory \citep{Doucet:2001,Arulampalam:2002a}.  Although the quantum particle filter is necessarily sub-optimal, I present numerical simulations which suggest they perform well in practice.  Throughout, I demonstrate the technqiues using a single qubit magnetometer.  

\section{Estimation of a parameter from a finite set}
\label{parameter_estimation:sec:finite}
We begin by considering the case where the parameter takes on a known, finite set of values.  Using the quantum filtering techniques in Chapter \ref{chapter:quantum}, we know that continuous measurements of the probe\footnote{Note that the word ``probe'' is used in regard to the system which couples to the parameter $\xi$ and is then used to infer the value of $\xi$.  This is in addition to the idea of using an ancillary  system, such as the electromagnetic field, to perform continuous measurements on the probe system.} system that couples to the parameter $\xi$ are well-described by a quantum filter as in Eq.~\eqref{quantum:eq:quantum_filter} with Hamiltonian
\begin{equation} \label{parameter_estimation:eq:start_hamiltonian}
	H = \xi H_0 , \quad  H_0 \in \mathscr{N}_s .
\end{equation}
Recall that $\mathscr{N}_s$ is the space of system (atomic) operators as introduced in Chapter \ref{chapter:quantum}, which are distinct from $\mathscr{N}_p$, the set of operators on the ancillary system (field) used to perform the continuous measurement.  Although what follows applies for arbitrary systems which admit a continuous measurement description, we fix our language to that of atoms and fields for a more transparent discussion.

Supposing we knew the true value of the parameter, the quantum filtering equations would give us the best least-squares estimate of the atomic system conditioned on the measurements and the knowledge of dynamics induced by $\xi$ through $H$.  But given the optimality of the filter, we could equally well embed the parameter $\xi$ as a diagonal operator $\Xi$ acting on an auxiliary quantum space, after which the filter \emph{still} gives the best estimate of \emph{both} system and auxiliary space operators.  Finding the best estimate of $\xi$ conditioned on the measurements simply corresponds to integrating the equations for $\pi_t[\Xi]$.

More precisely, extend the atomic Hilbert space $\mathcal{H}_S \mapsto \mathcal{H}_{\xi}\otimes\mathcal{H}_S$ and the operator space $\mathscr{N}_s \mapsto \mathfrak{D}(\mathcal{H}_{\xi})\otimes\mathscr{N}_s$, where $\mathfrak{D}(\mathcal{H}_{\xi})$ is the set of diagonal operators on $\mathcal{H}_{\xi}$.  Assuming $\xi$ takes on $N$ possible values $\{\xi_1,\ldots, \xi_N\}$, $\dim{\mathfrak{D}(\mathcal{H}_{\xi})} = N$.   Introduce the diagonal operator 
\begin{equation}
	\mathfrak{D}(\mathcal{H}_{\xi}) \ni \Xi = \sum_{i=1}^{N} \xi_i \ketbra{\xi_i}{\xi_i}
\end{equation}
so that $\Xi\ket{\xi_i} = \xi_i\ket{\xi_i}$ with $\ket{\xi_i} \in \mathcal{H}_{\xi}$.  This allows one to generalize Eq.~\eqref{parameter_estimation:eq:start_hamiltonian} as 
\begin{equation}
	H \mapsto \Xi\otimes H_0 \in \mathfrak{D}(\mathcal{H}_{\xi})\otimes\mathscr{N}_s .
\end{equation}
Any remaining atomic operators $X_A \in \mathscr{N}_s$ act as the identity on the auxiliary space, i.e. $I\otimes X_A$.  
Given these definitions, the derivation of the quantum filtering equation remains essentially unchanged, so that the filter in either the operator form of Eq.~\eqref{quantum:eq:quantum_filter} or the adjoint form of Eq.~\eqref{quantum:eq:quantum_filter_rho_form} is simply updated with the extended forms of operators given in the last paragraph.  

Since $\xi$ is a classical parameter, we require that the reduced conditional density matrix $(\rho_{\xi})_t = \operatorname{Tr}_{\mathcal{H}_S}{(\rho_t)}$ be diagonal in the basis of $\Xi$.  Thus we can write
\begin{equation} \label{parameter_estimation:eq:reduced_state}
	(\rho_{\xi})_t = \sum_{i=1}^N p_t^{(i)} \ketbra{\xi_i}{\xi_i}
\end{equation}
where
\begin{multline}
	p_t^{(i)} \equiv \Tr{(\ketbra{\xi_i}{\xi_i} \otimes I)\rho_t} 
		\equiv \pi_t[\ketbra{\xi_i}{\xi_i}\otimes I]\\
		\equiv \mathbbm{E}[\ketbra{\xi_i}{\xi_i} \otimes I | M_{[0,t]}] 
		\equiv P(\xi = \xi_i | M_{[0,t]}) .
\end{multline}
Then $p_t^{(i)}$ is precisely the conditional probability for $\xi$ to have the value $\xi_i$ and the set $\{p_t^{(i)}\}$ gives the discrete conditional distribution of the random variable represented by $\Xi$.  Similarly, by requiring operators to be diagonal in $\mathcal{H}_{\xi}$, we ensure that they correspond to classical random variables.  In short, we have simply embedded filtering of a truly classical random variable in the quantum formalism.

The fact that both states and operators are diagonal in the auxiliary space suggests using an ensemble form for filtering.  As such, consider an ensemble consisting of a weighted set of $N$ conditional atomic states,  each state evolved under a different $\xi_i$.  Later, in section \ref{parameter_estimation:sec:infinite}, we will call each ensemble member a \emph{quantum particle}. For now, we explicitly write the conditional quantum state as
\begin{equation} \label{parameter_estimation:eq:finitedim:ensembleform}
	\rho_t^{E} = \sum_{i = 1}^{N} p^{(i)}_t \ketbra{\xi_i}{\xi_i} \otimes \rho^{(i)}_t 
\end{equation}
where $\rho^{(i)}_t$ is a density matrix on $\mathcal{H}_S$.  The reduced state, $\operatorname{Tr}_{\mathcal{H}_A}{(\rho_t^{E})}$, is clearly diagonal in the basis of $\Xi$.  Using the extended version of the adjoint quantum filter in Eq.~\eqref{quantum:eq:quantum_filter_rho_form}, one can derive the \emph{ensemble quantum filtering equations}
\begin{subequations} \label{parameter_estimation:eq:finitedim:ensemblefilter}
	\begin{align}
			d\rho_t^{(i)} &= -i[\xi_i H_0, \rho_t^{(i)}]dt   
						   + (L\rho_t^{(i)}L^{\dag} - \frac{1}{2}L^{\dag}L\rho_t^{(i)} - \frac{1}{2}\rho_t^{(i)}L^{\dag}L)dt
							\nonumber \\
					      & + \left(L\rho_t^{(i)} + \rho_t^{(i)}L^{\dag}
					 		- \Tr{(L+L^{\dag})\rho_t^{(i)}}\rho_t^{(i)}\right)dW_t \label{parameter_estimation:eq:finitedim:ensemblefilter:rho}\\
			dp_t^{(i)} &= \left(\Tr{(L+L^{\dag})\rho_t^{(i)}} - \Tr{I\otimes(L+ L^{\dag})\rho_t^{E}}\right)p_t^{(i)}dW_t
							\label{parameter_estimation:eq:finitedim:ensemblefilter:prob}\\
			dW_t &= dM_t - \Tr{I\otimes ( L +  L^{\dag})\rho_t^{E}}dt \label{parameter_estimation:eq:finitedim:ensemblefilter:innov}
	\end{align}
\end{subequations}
We see that each $\rho^{(i)}_t$ in the ensemble evolves under a quantum filter with $H = \xi_i H_0$ and is coupled to other ensemble members through the innovation factor $dW_t$, which depends on the ensemble expectation of the measurement observable.  Note that one can incorporate any prior knowledge of $\xi$ in the weights of the initial distribution $\{p_0^{(i)}\}$.

The reader should not be surprised that a similar approach would work for estimating more than one parameter at a time, such as three cartesian components of an applied magnetic field.  One would introduce an auxiliary space for each parameter and extend the operators in the obvious way.  The ensemble filter would then be for a joint distribution over the multi-dimensional parameter space.  Similarly, one could use this formalism to distinguish initial states, rather than parameters which couple via the Hamiltonian.  For example, in the case of state discrimination, one would introduce an auxiliary space which labels the possible input states, but does not play any role in the dynamics.  The filtered weights would then be the probabilities to have been given a particular initial state. In fact, using a slightly different derivation, Jacobs derived equations similar to Eq.~\eqref{parameter_estimation:eq:finitedim:ensemblefilter} for the case of binary state discrimination \citep{Jacobs:2006a}.  Yanagisawa recently studied the general problem of retrodiction or ``smoothing'' of quantum states \citep{Yanagisawa:2007a}.  In light of his work and results in the following section, the retrodictive capabilities of quantum filtering are very limited without significant prior knowledge or feedback.   
\subsection{Conditions for convergence}
Although introducing the auxiliary parameter space does not change the derivation of the quantum filter, it is not clear how the initial uncertainty in the parameter will impact the filter's ability to ultimately track to the correct value.  Indeed, outside of anecdotal numerical evidence (which I will presently add to), there has been little formal consideration of the sensitivity of the quantum filter to the initial state estimate.  Recently, van Handel presented a set of conditions which determine whether the quantum filter will asymptotically track to the correct state independently of the assumed initial state \citep{vanHandel:2008a}.  Since we have embedded parameter estimation in the state estimation framework, such stability then determines whether the quantum filter can asymptotically track to the true parameter, i.e. whether $\lim_{t \rightarrow \infty} p_{t}^{(j)} =  \delta_{ij}$ when $\xi = \xi_i$.  In this section, I present van Handel's results in the context of our parameter estimation formalism and present a simple check of asymptotic convergence of the parameter estimate.  We begin by reviewing the notions of absolute continuity and observability.

In the general stability problem, let $\rho_1$ be the true underlying state and $\rho_2$ be the initial filter estimate.  We say that $\rho_1$ is \emph{absolutely continuous} with respect to $\rho_2$, written $\rho_1 \ll \rho_2$, if and only if $\operatorname{ker} \rho_1 \supset \operatorname{ker} \rho_2$.  In the context of parameter estimation, we assume that we know the initial atomic state exactly, so that $\rho_1 \ll \rho_2$ as long as the reduced states satisfy $\rho_1^{E} \ll \rho_2^{E}$.  Since these reduced states are simply discrete probability distributions, $\{(p_t^{i})_1\}$ and $\{(p_t^{i})_2\}$, this is just the standard definition of absolute continuity in classical probability theory as we saw in Chapter \ref{chapter:classical} when studying the Radon-Nikodym theorem \ref{thm:radon_nikodym}.  In our case, the true state has ${(p_{t=0}^{(j)})}_{1} = \delta_{ij}$ if the parameter has value $\xi_i$.  Thus, as long as our estimate has non-zero weight on the $i$-th component, $\rho_1 \ll \rho_2$.  This is trivially satisfied if ${(p_{t=0}^{(j)})}_{2} \neq 0$ for all $j$.

The other condition for asymptotic convergence is that of observability.  A system is \emph{observable} if one can determine the exact initial atomic state given the entire measurement record over the infinite time interval.  Observability is then akin to the ability to distinguish any pair of initial states on the basis of the measurement statistics alone.  Recall the definition of the Lindblad generator in Eq.~\eqref{quantum:eq:lindblad_generator} and further define the operator $\mathcal{K}[X_A] = L^{\dag}X_a + X_aL$.  Then according to Proposition 5.7 in \citep{vanHandel:2008a}, the observable space $\mathcal{O}$ is defined as the smallest linear subspace of $\mathscr{N}_S$ containing the identity and which is invariant under the action of $\mathcal{L}$ and $\mathcal{K}$.  The filter is observable if and only if $\mathscr{N}_s = \mathcal{O}$, or equivalently $\dim{\mathscr{N}_s} = \dim{\mathcal{O}}$.  

In the finite-dimensional case, van Handel presents an iterative procedure for constructing the observable space.  Define the linear spaces $\mathcal{Z}_n \subset \mathscr{N}_s$ as
\begin{equation}
	\begin{split}
		\mathcal{Z}_0 &= \operatorname{span}\{I\}\\
		 \mathcal{Z}_n &= \operatorname{span}\{\mathcal{Z}_{n-1},
								\mathcal{L}[\mathcal{Z}_{n-1}], \mathcal{K}[\mathcal{Z}_{n-1}]\},\quad n > 0
	\end{split}
\end{equation}
The procedure terminates when $\mathcal{Z}_n = \mathcal{Z}_{n+1}$, which is guaranteed for some finite $n = m$, as the dimension of $\mathcal{Z}_n$ cannot exceed the dimension of the ambient space $\mathscr{N}_s$.  Moreover, the terminal $\mathcal{Z}_m = \mathcal{O}$, so that using a Gram-Schmidt procedure, one can iteratively find a basis for $\mathcal{O}$ and easily compute its dimension.  Note that for operators $A$ and $B$, the inner-product $\langle A, B \rangle$ is the Hilbert-Schmidt inner product $\Tr{A^{\dag} B}$.

Given these definitions, one has the following theorem for filter convergence and corollary for parameter estimation.

\begin{theorem} (Theorem 2.5 in \citep{vanHandel:2008a}) Let $\pi_t^{\rho_i}(X_A)$ be the evolved filter estimate, initialized under state $\rho_i$.  If the system is observable and $\rho_1 \ll \rho_2$, the quantum filter is asymptotically stable in the sense that
\begin{equation}
	\abs{\pi_t^{\rho_1}(X_A) - \pi_t^{\rho_2}(X_A)}_{M^{\rho_1}_{[0,t]}} \stackrel{t \rightarrow\infty}
						{\longrightarrow} 0 \quad \forall X_a \in \mathscr{N}_s
\end{equation} 
where the convergence is under the observations generated by $\rho_1$.
\end{theorem}

One could use this theorem to directly check the stability of the quantum filter for parameter estimation, using the extended forms of operators in $\mathcal{L}$ and $\mathcal{K}$ and being careful that the observability condition is now $\dim{\mathcal{O}} = \dim{\mathcal{D}{(\mathcal{H}_\xi)}\otimes\mathscr{N}_s}$.  However, the following corollary relates the observability of the parameter filter to the observability of the related filter for a known parameter.  Combined with the discussion of extending the absolute continuity condition, this then gives a simple check for the stability of the parameter filter.  

\begin{corollary}
Consider a parameter $\xi$ which takes on one of $N$ distinct positive real values $\{\xi_i\}$.  If the quantum filter with known parameter is observable, then the corresponding extended filter for estimation of $\xi$ is observable.
\end{corollary}
\begin{proof}
In order to satisfy the observability condition, we require $\dim{\mathcal{O}} = Nr$, where we have set $\dim{\mathscr{N}_s} = r$ and used the fact that $\dim{\mathfrak{D}(\mathcal{H}_\xi)} = N$.  Given that the filter for a known parameter is observable, its observable space coincides with $\mathscr{N}_s$ and has an orthogonal operator basis $\{A_i\}$, where we take $A_0 = I$.

Similarly, consider the $N$-dimensional operator space $\mathfrak{D}(\mathcal{H}_\xi)$.  If $\{\xi_i\}$ are distinct, any set of the form
\begin{equation}
	\{\Xi^{k_1}, \Xi^{k_2}, \ldots, \Xi^{k_N}\}, k_i \in \mathbbm{N}, k_i \neq k_j \text{ if } i \neq j
\end{equation}
is linearly independent, since the corresponding generalized Vandermonde matrix
\begin{equation}
	V_\xi = \begin{pmatrix}
		\xi_1^{k_1} & \xi_1^{k_2} & \ldots &  \xi_1^{k_N} \\
		\vdots & \vdots & \ddots & \vdots  \\
		\xi_N^{k_1} & \xi_N^{k_2} &\ldots & \xi_N^{k_N}
	\end{pmatrix}
\end{equation}
has linearly independent columns \citep{Gantmakher:2000a}.  

Following the iterative procedure, we construct the observable space for the parameter estimation filter starting with $I \otimes A_0$, which is the identity in the extended space. We then iteratively apply $\mathcal{L}$ and $\mathcal{K}$ until we have an invariant linear span of operators.  The only non-trivial operator on the auxiliary space comes from the Hamiltonian part of the Lindblad generator, which introduces higher and higher powers of the diagonal matrix $\Xi$.  Since $\dim{\mathfrak{D}{(\mathcal{H}_\xi)}\otimes\mathscr{N}_s}$ is finite, this procedure must terminate. The resulting observable space can be decomposed into subspaces
\begin{equation}
\mathcal{O}_{i} = \{\Xi^{k_i^j} \otimes A_i\}, \quad i = 1,\ldots, r \quad k_i^j \in \mathbbm{N}  
\end{equation}
where $k_i^j$ is some increasing sequence of non-negative integers which correspond to the powers of $\Xi$ that are introduced via the Hamiltonian.  Note that the specific values of $k_i^j$ depend on the commutator algebra of $H_0$ and the atomic-space operator basis $\{A_i\}$.  Regardless, since the Hamiltonian in $\mathcal{L}$ can always add more powers of $\Xi$, the procedure will not terminate until $\mathcal{O}_i$ is composed of a largest linearly independent set of powers of $\Xi$.  This set has at most $N$ distinct powers of $\Xi$, since it cannot exceed the dimension of the auxiliary space.  Given that any collection of $N$ powers of $\Xi$ is linearly independent, this means once we reach a set of $N$ powers $k_i^j$, the procedure terminates and $\dim{\mathcal{O}_i} = N$.   Since $\mathcal{O}$ has $r$ subspaces $\mathcal{O}_i$, each of dimension $N$, $\dim{\mathcal{O}} =  Nr$ as desired and the observability condition is satisfied.
\end{proof}

Although these conditions provide a simple check, I would like to stress that they do not determine how quickly the convergence occurs, which will depend on the specifics of the problem at hand.  Additionally, as posed, the question of observability is a binary one.  One might expect that some unobservable systems are nonetheless ``more observable'' than others or simply that unobservable systems might still be useful for parameter estimation.  Given the corollary above, one can see that this may occur if a single parameter $\xi_j = 0$.  Then $V_{\xi}$ has a row of all zeros, so that the maximal dimension of a set of linearly independent powers of $\Xi$ is $N - 1$.  Similarly, if one allows both positive and negative real-valued parameters, the properties of $V_{\xi}$ are not as obvious, though in many circumstances, having both $\xi_i$ and $-\xi_i$ renders the system unobservable.  We explore these nuances in numerical simulations presented in the following section.  
\subsubsection*{Qubit Example}
Consider using the single qubit from Example \ref{quantum:example:qubit_filter} in Chapter \ref{chapter:quantum} as a probe for the magnetic field $B$.  Since the initial state is restricted to the $x$-$z$ plane, the $y$ component of the Bloch vector is always zero and thus is not a relevant part of the atomic observable space, which is spanned by $\{I,\sigma_x,\sigma_z\}$.  In other words, the filter with known $B$ is trivially observable, since we assume the initial state is known precisely.  

When $B$ is unknown, the ensemble parameter filter is given by
\begin{subequations} \label{parameter_estimation:eq:qubit_ensemble}
	\begin{align} 
			d\theta_t^{(i)} &= -2B_idt + \kappa \cos(\theta_t^{(i)})
									(\sin(\theta_t^{(i)}) - 2\expect{\sigma_z}^{E})dt \nonumber\\
							&	+ 2\sqrt{\kappa}\cos(\theta_t^{(i)})dW_t \\
			dp^{(i)}_t &= 2\sqrt{\kappa}(\sin(\theta_t^{(i)}) - \expect{\sigma_z}^{E}) p^{(i)}_t dW_t
	\end{align}
\end{subequations}
where $dW_t = dM_t -2 \sqrt{\kappa} \expect{\sigma_z}^{(E)}$ and $\expect{\sigma_z}^{E} = \sum_{i} p_t^{(i)}\sin(\theta_t^{(i)})$.  We simulated this filter by numerically integrating the quantum filter in Eq.~\eqref{quantum:eq:known_qubit_filter} using a value for $B$ uniformly chosen from the given ensemble of potential $B$ values.  This generates a measurement current $dM_t$, which is then fed into the ensemble filter of Eq.~\eqref{parameter_estimation:eq:qubit_ensemble}.  For all simulations, I set $\kappa = 1$ and used a simple It\^{o}-Euler integrator as described in Appendix \ref{appendix:numerical_methods_for_stochastic_differential_equations} with a step-size $dt = 10^{-5}$.

\begin{figure}[bt]
	\centering
		\includegraphics[scale=0.9]{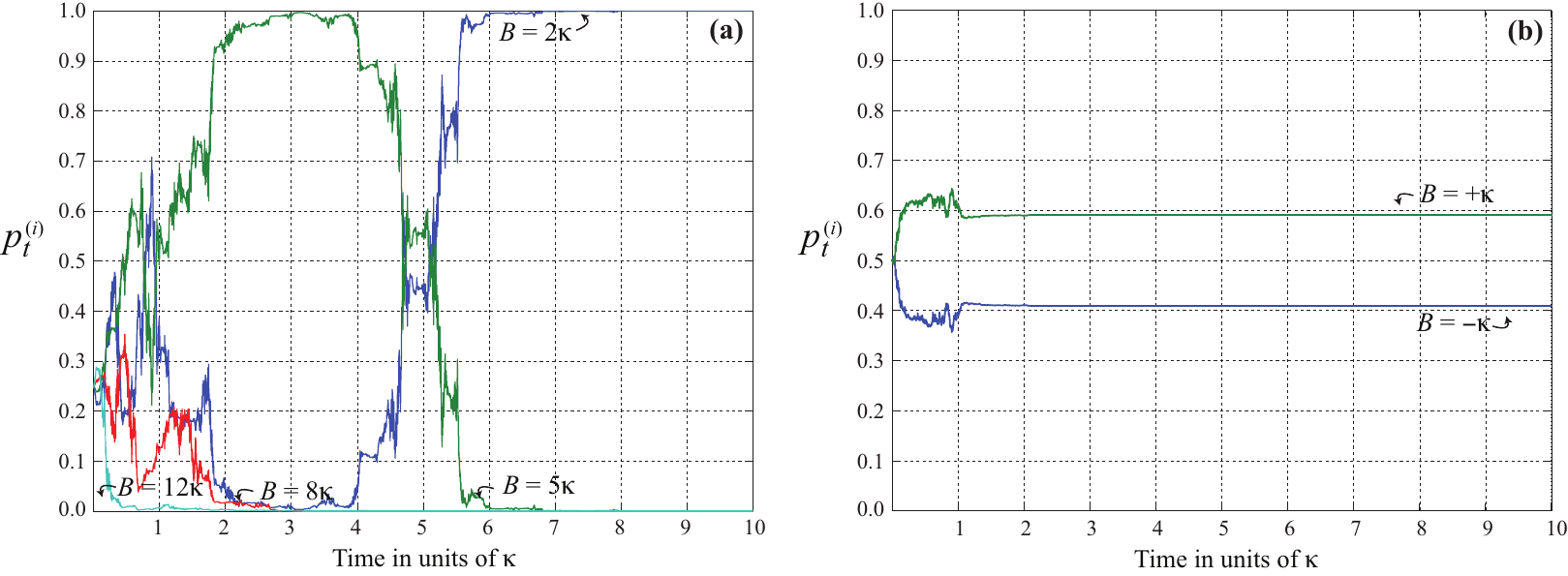}
	\caption[Qubit magnetometer filter performance]{(a) Filtered $p_t^{(i)}$ for $B \in \{ 2\kappa, 5\kappa, 8\kappa, 12\kappa\}$.  The filter tracks to the true underlying value of $B = 2\kappa$ (b) Filtered $p_t^{(i)}$ for $B \in \{ -\kappa, +\kappa\}$.  The filter does not track to $B = +\kappa$ with probability one, though it is the most probable parameter value.}
	\label{parameter_estimation:fig:discrete:combinedruns}
\end{figure}

Figure \ref{parameter_estimation:fig:discrete:combinedruns}(a) shows a simulation of a filter for the case $B \in \{ 2\kappa, 5\kappa, 8\kappa, 12\kappa\}$.  The filter was initialized with a uniform distribution, $p^{(i)}_0 = 1/4$.  For the particular trajectory shown, the true value of $B$ was $2\kappa$ and we see that the filter successfully tracks to the correct $B$ value.  This is not surprising, given that the potential values of $B$ are positive and distinct, thus satisfying the convergence corollary.  It is also interesting to note that the filter quickly discounts the probabilities for $8\kappa,12\kappa$, which are far from the true value.  Conversely, the filter initially favors the incorrect $B = 5\kappa$ value before honing in on the correct parameter value.

In Figure \ref{parameter_estimation:fig:discrete:combinedruns}(b), we see a simulation for the case of $B \in \{+\kappa, -\kappa\}$, which does not satisfy the convergence corollary.  In fact, using the iterative procedure, one finds the observable space is spanned by $\{I\otimes I, I \otimes \sigma_z, B \otimes \sigma_x, B^2 \otimes I, B^2 \otimes \sigma_z, B^3 \otimes \sigma_x\}$.  But since $B = \left(\begin{smallmatrix} \kappa & 0\\ 0 & -\kappa \end{smallmatrix}\right)$, $B^2 = \kappa^2 I$ so that only 3 of the 6 operators are linearly independent.  Although the filter does not converge to the true underlying value of $B = +\kappa$, it does reach a steady-state that weights the true value of $B$ more heavily.  Simulating 100 different trajectories for the filter, there were 81 trials for which the final probabilities were weighted more heavily towards the true value of $B$.  This confirms our intuition that the binary question of observability does not entirely characterize the performance of the parameter filter.  

\begin{figure}[hb]
	\centering
		\includegraphics[scale=1]{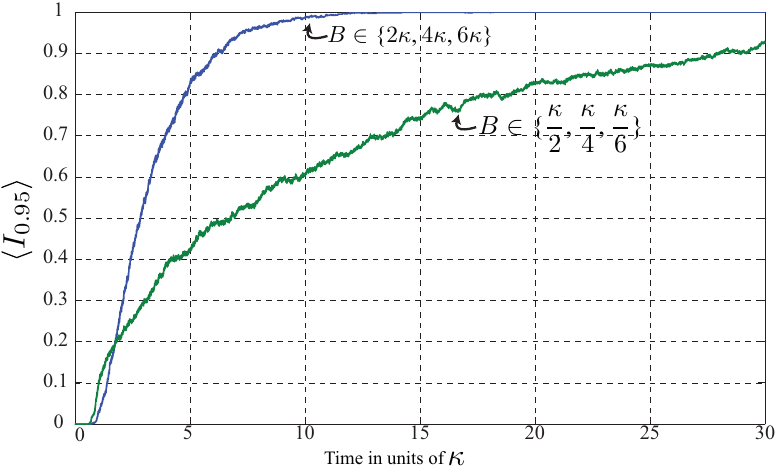}
	\caption[Qubit magnetometer filter convergence]{Rate of convergence ($I_{0.95}$), averaged over 1000 trajectories.  The filters are for cases when possible $B$ values are either all larger or all smaller than the measurement strength $\kappa$.  }
	\label{parameter_estimation:fig:discrete:convergence}
\end{figure}
Figure \ref{parameter_estimation:fig:discrete:convergence} shows the rate of convergence of filters meant to distinguish different sets of $B$.  The rate of convergence is defined as the ensemble average of the random variable
\begin{equation}
	I_{\alpha} = \begin{cases}
		1, &\text{ if $p_t^{(i)} > \alpha $ for any $i$}\\
		0, &\text{ otherwise }
	\end{cases} .
\end{equation}
Although any individual run might fluctuate before converging to the underlying $B$ value, the average of $I_\alpha$ over many runs should give some sense of the rate at which these fluctuations die down.  For the simulation shown, I set $\alpha = 0.95$ and averaged $I_{0.95}$ over 1000 runs for two different cases---either all possible $B$ values are greater than $\kappa$ or all are less than $\kappa$.  As shown in the plot, the former case shows faster convergence since the $B$ field drives the dynamics more strongly than the measurement process, which in turn makes the trajectories of different ensemble members more distinct.  Of course, one cannot make the measurement strength too weak since we need to learn about the system evolution.  Therefore care must be taken to tune the signal-to-noise ratio in the problem at hand, relative to the timescales relevant for the parameter values of interest.
\section{Quantum Particle Filter}
\label{parameter_estimation:sec:infinite}
Abstractly, developing a parameter estimator in the continuous case is not very different than in the finite dimensional case.  One can still introduce an auxiliary space $\mathcal{H}_{\xi}$, which is now infinite dimensional.  In this space, we embed the operator version of $\xi$ as 
\begin{equation}
	\mathfrak{D}(\mathcal{H}_{\xi}) \ni \Xi = \int d\xi \xi \ketbra{\xi}{\xi} ,
\end{equation}
where $\Xi\ket{\xi} = \xi\ket{\xi}$ and $\braket{\xi}{\xi'} = \delta(\xi - \xi')$.  Again, by extending operators appropriately, the filters in Eq.~\eqref{quantum:eq:quantum_filter} and Eq.~\eqref{quantum:eq:quantum_filter_rho_form} become optimal parameter estimation filters.  We generalize the conditional ensemble state of Eq.~\eqref{parameter_estimation:eq:finitedim:ensembleform} to
\begin{equation} \label{parameter_estimation:eq:infinite:rho}
		\rho_t^{E} = \int d\xi p_t(\xi) \ketbra{\xi}{\xi} \otimes \rho^{(\xi)}_t ,
\end{equation}
where $p_t(\xi) \equiv  P(\xi | M_{[0,t]})$ is the continuous conditional probability density.  Although the quantum filter provides an exact formula for the evolution of this density, calculating it is impractical, as one cannot exactly represent the continuous distribution on a computer.  The obvious approximation is to discretize the space of parameter values and then use the ensemble filter determined by Eq.~\eqref{parameter_estimation:eq:finitedim:ensemblefilter}; indeed such an approach is very common in classical filtering theory and encompasses a broad set of Monte Carlo methods called particle filters \citep{Doucet:2001,Arulampalam:2002a}.

The inspiration for particle filtering comes from noting that any distribution can be approximated by a weighted set of point masses or \emph{particles}.  In the quantum case, we introduce a \emph{quantum particle} approximation of the conditional density in Eq.~\eqref{parameter_estimation:eq:infinite:rho} as
\begin{equation}
	p_t(\xi) \approx \sum_{i = 1}^{N} p^{(i)}_t \delta(\xi - \xi_i) .
\end{equation}
The approximation can be made arbitrarily accurate in the limit of $N \rightarrow \infty$.  Plugging this into Eq.~\eqref{parameter_estimation:eq:infinite:rho}, we recover precisely the form for the discrete conditional state given in Eq.~\eqref{parameter_estimation:eq:finitedim:ensembleform}.  Accordingly, the quantum particle filtering equations are identical to those of the ensemble filter given in Eq.~\eqref{parameter_estimation:eq:finitedim:ensemblefilter}.  The only distinction here is in the initial approximation of the space of parameter values.  Thus the basic quantum particle filter simply involves discretizing the parameter space,  then integrating the filter according to the ensemble filtering equations.  

The basic particle filter suffers from a degeneracy problem, in that all but a few particles may end up with negligible weights $p^{(i)}_t$.  This problem is even more relevant when performing parameter estimation, since the set of possible values for $\xi$ are fixed at the outset by the choice of discretization.  Even if a region in parameter space has low weights, its particles take up computational resources, but contribute little to the estimate of $\xi$.  More importantly, the ultimate precision of the parameter estimate is inherently limited by the initial discretization; we can never have a particle whose parameter value $\xi_i$ is any closer to the true value $\xi$ than the closest initial discretized value.

In order to circumvent these issues, we can adopt the kernel resampling techniques of Liu and West \citep{Liu:2001a}.  The idea is to replace low weight particles with new ones concentrated in high weight regions of parameter space.  One first samples a source particle from the discrete distribution given by the weights $\{p^{(i)}_t\}$, ensuring new particles come from more probable regions of parameter space.  Given a source particle, we then create a child particle by sampling from a Gaussian kernel centered near the source particle.  By repeating this procedure $N$ times, we create a new set of particles which populate more probable regions of parameter space.  Over time, this adaptive procedure allows the filter to move away from unimportant regions of parameter space and more finely explore the most probable parameter values.

The details of the adaptive filter lie in parameterizing and sampling from the Gaussian kernel.  Essentially, we are given a source particle, characterized by $\ketbra{\xi_i}{\xi_i}$ and $\rho_t^{(i)}$, and using the kernel, create a child particle, characterized by $\ketbra{\tilde{\xi}_i}{\tilde{\xi}_i}$ and $\tilde{\rho}_t^{(i)}$.  One could attempt to sample from a multi-dimensional Gaussian over both the parameter and atomic state components, but ensuring that the sampled $\tilde{\rho}_t^{(i)}$ is a valid atomic state would be non-trivial in general.  There will be some cases, including the qubit example in the following section, where the atomic state is conveniently parameterized for Gaussian resampling.  But for clarity in presenting the general filter, we will create a child particle with the same atomic state as the parent particle.  

Under this assumption, the Gaussian kernel for parent particle $i$ is characterized by a mean $\mu^{(i)}$ and variance ${\sigma^2}^{(i)}$, both defined over the one dimensional parameter space.  Rather than setting the mean of this kernel to the parameter value of the parent, Liu and West suggest setting
\begin{equation} \label{parameter_estimation:eq:kernel:mean}
	\mu^{(i)} = a \xi_i + (1 - a) \bar{\xi} , \quad a \in [0,1]
\end{equation}
where $\bar{\xi} = \sum_i p_t^{(i)}\xi_i$ is the ensemble mean.  The parameter $a$ is generally taken to be close to one and serves as a mean reverting factor.  This is important because simply resampling from Gaussians centered at $\xi_i$ results in an overly dispersed ensemble relative to the parent ensemble.  The kernel variance is set to
\begin{equation} \label{parameter_estimation:eq:kernel:var}
	{\sigma^2}^{(i)} = h^2 V_t , \quad h \in [0,1]
\end{equation}
where $V_t = \sum_i p_t^{(i)}( \xi_i - \bar{\xi})^2$ is the ensemble variance and $h$ is the smoothing parameter.  It is generally a small number chosen to scale with $N$, so as to control how much kernel sampling explores parameter space.  While $a$ and $h$ can be chosen independently,  Liu and West relate them by $h^2 = 1 - a^2$, so that the new sample does not have an increased variance.
  
Of course, it would be computationally inefficient to perform this resampling strategy at every timestep, especially since there will be many steps where most particles have non-negligible contributions to the parameter estimate.  Instead, we should only resample if some undesired level of degeneracy is reached.  As discussed by Arulampalam et al. \citep{Arulampalam:2002a}, one measure of degeneracy is the effective sample size
\begin{equation}
	N_{\text{eff}} = \frac{1}{\sum_{i=1}^N (p_t^{(i)})^2} .
\end{equation}
At each timestep, we then resample if the ratio $N_{\text{eff}}/N$ is below some given threshold.  We are not aware of an optimal threshold to chose in general, but the literature suggest $2/3$ as a rule of thumb \citep{Doucet:2001}.

Altogether, the \emph{resampling quantum particle filter} algorithm proceeds as follows:
\begin{description}
	\item[Initialization] for $i = 1,\ldots, N$:
	 \begin{enumerate}
	 	\item Sample $\xi_i$ from the prior parameter distribution.
		\item Create a quantum particle with weight $p_t^{(i)} = 1/N$, parameter state $\ketbra{\xi_i}{\xi_i}$ and atomic state $\rho_0^{(i)} = \rho_0$, where $\rho_0$ is the known initial atomic state.
	 \end{enumerate}
	\item[Repeat] for all time:
		\begin{enumerate}
			\item Update the particle ensemble by integrating a timestep of the filter given in Eq.~\eqref{parameter_estimation:eq:finitedim:ensemblefilter}.
			 \item If $N_{\text{eff}}/N$ is less than the target threshold, create a new particle ensemble:
			  \begin{description}
			  	\item[Resample] for $i = 1,\ldots, N$:
				 \begin{enumerate}
				 	\item Sample an index $i$ from the discrete density $\{p_t^{(i)}\}$.
					\item Sample a new parameter value $\tilde{\xi}_i$ from the Gaussian kernel with mean $\mu^{(i)}$ and variance ${\sigma^2}^{(i)}$ given by Eq.~\eqref{parameter_estimation:eq:kernel:mean} and Eq.~\eqref{parameter_estimation:eq:kernel:var}. 
					\item Add a quantum particle to the new ensemble with weight $p_t^{(i)} = 1/N$, parameter state $\ketbra{\tilde{\xi}_i}{\tilde{\xi}_i}$ and atomic state $\rho_t^{(i)} = \rho_t^{(i)}$
				 \end{enumerate}
			  \end{description}
		\end{enumerate}
\end{description}

Unfortunately, checking asymptotic convergence of the filter is more involved in the continuous-valued case, as the observability and absolute continuity conditions require extra care in infinite dimensions.  However, given that the quantum particle filter actually works on a discretized space, in practice we can simply use the results we had for the finite-dimensional case.  As before, we note that one can generalize the quantum particle filter to multidimensional parameters by using a multi-dimensional Gaussian kernel.  One might also consider using alternate kernel forms, such as a regular grid which has increasingly finer resolution with each resampling stage.  We will not consider such extensions here.  
\subsubsection*{Qubit Example}
We now consider a resampling quantum particle filter for the qubit magnetometer introduced earlier in this chapter.  As hinted at in the previous section, since the qubit state is parameterized by the continuous variable $\theta_t$, we can easily resample both the magnetic field $B_i$ and state $\theta^{(i)}$ using a two-dimensional Gaussian kernel for $(\tilde{B}_i,\tilde{\theta}^{(i)})$, with mean vector and covariance matrix given by generalizations of Eq.~\eqref{parameter_estimation:eq:kernel:mean} and Eq.~\eqref{parameter_estimation:eq:kernel:var}.  Since different values of $B$ result in different state evolutions, resampling both the state and magnetic field values should result in child particles that are closer to the true evolved state. 

\begin{figure}[t]
	\centering
		\includegraphics[scale=0.5]{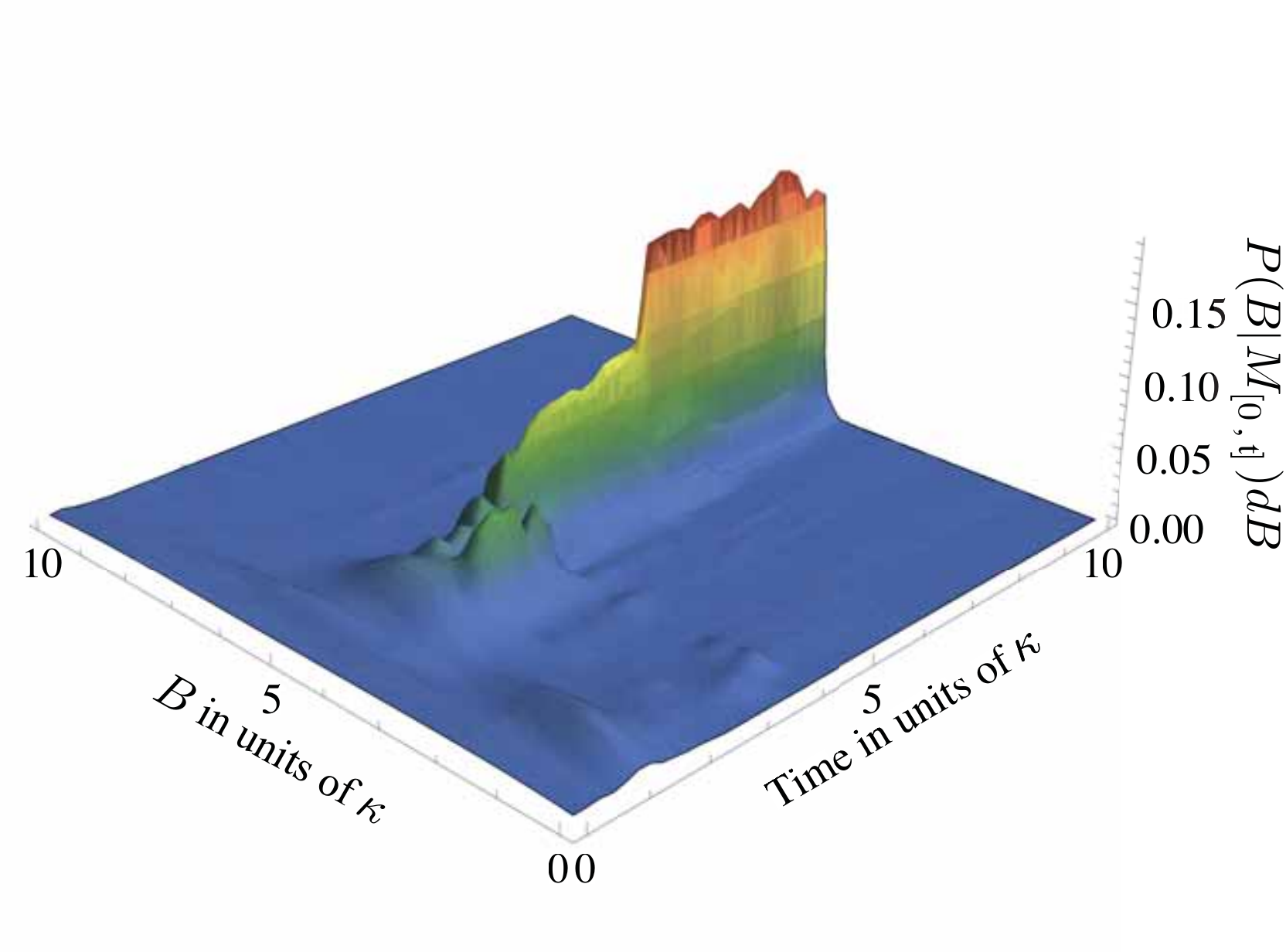}
	\caption[Kernel density reconstruction of qubit particle filter performance]{Kernel density reconstruction of $p_t(B)dB =  P(B | M_{[0,t]})dB$ for $N = 1000$ particle filter set with $dB = 10\kappa/150$, $a = 0.98$, $h = 10^{-3}$ and resampling threshold of $2/3$.  The true magnetic field was $B = 5\kappa$.}
	\label{parameter_estimation:fig:densityPlot}
\end{figure}

Figure \ref{parameter_estimation:fig:densityPlot} shows a typical run of the quantum particle filter for $N = 1000$ particles.  The true $B$ value was $5\kappa$ and the prior distribution over $B$ was taken to be uniform over the interval $[0,10\kappa]$.  As before, I used an It\^{o}-Euler integrator with a step-size of $dt = 10^{-5}\kappa$.  Note that both the timespan of integration and the potential values of $B$ range from $0$ to $10\kappa$ in our units.  The resampling parameters were $a = 0.98$, $h = 10^{-3}$ and resampling threshold $2/3$.  Note that I chose not to use Liu and West's relation between $a$ and $h$.   

In order to generate the figure, each particle's weight and parameter values were stored at 50 equally spaced times over the integration timespan.  Using Matlab's \texttt{ksdensity} function, these samples were then used to reconstruct $p_t(B)$ via a Gaussian kernel density estimate of the distribution.  The resulting kernel density estimate was then evaluated at 150 equally spaced $B$ values in the range $[0,10\kappa]$, which I plotted as $p_t(B)dB$ with $dB = 10\kappa/150$.  As is seen in the figure, after some initial multi-modal distributions over parameter space, the filter hones in on the true value of $B = 5\kappa$.  For the simulation shown, the final estimate was $\hat{B} = 5.03\kappa$ with uncertainty $\sigma_{\hat{B}} = 0.18\kappa$.  The filter resampled 7 times over the course of integration.    
\section{Summary}
\label{parameter_estimation:sec:conclude}
I have presented practical methods for single-shot parameter estimation via continuous quantum measurement.  By embedding the parameter estimation problem in the standard quantum filtering problem, the optimal parameter filter is given by an extended form of the standard quantum filtering equation.  For parameters taking values in a finite set, I gave conditions for determining whether the parameter filter will asymptotically converge to the correct value.  For parameters taking values from an infinite set, I introduced the quantum particle filter as a computational tool for suboptimal estimation.  Throughout, I presented numerical simulations of the methods using a single qubit magnetometer.  

These techniques should generalize straightforwardly for estimating time-dependent parameters and to a lesser extent, estimating initial state parameters.  The binary state discrimination problem studied by \citep{Jacobs:2006a} is one such example and his approach is essentially a special case of our ensemble parameter filter.  Future extensions of this work include exploring alternate resampling techniques for the quantum particle filter, considering alternative discretization schemes beyond the delta function particle basis and developing feedback strategies for improving the parameter estimate. 

\chapter{Precision Magnetometry}
\label{chapter:magnetometry}
In this chapter, I review the application of the parameter estimation techniques developed in Chapter \ref{chapter:quantum_parameter_estimation} for a proposed experimental demonstration of precision magnetometry.  By double-passing an optical field through an atomic system, one hopes to create effective nonlinear interactions which offer improved sensitivity to the strength of an external magnetic field.  Using quantum stochastic formalsim of Chapter \ref{chapter:quantum}, I review simulations of quantum information theoretic bounds on the optimal estimator performance which suggest magnetic field uncertainty scalings better than that of traditional atomic magnetometers, which is further supported by simulations of corresponding quantum particle filed parameter estimators.  The research in this chapter appears in \citep{Chase:2009b,Chase:2009c,Chase:2009d}.

\section{Introduction}
It is well-appreciated in physics that the properties of a field must often be determined indirectly, such as by observing the effect of the field on a test particle.  Take magnetometry for example:  the strength of a magnetic field might be inferred by observing Larmor precession in a spin-polarized atomic sample \cite{Budker:2002a} and estimating the field strength $B$ from the precession rate.  Inherent in this process is the fact that the atomic spin must be measured to determine the extent of the magnetically-induced dynamics.  For very precise measurements, uncertainty $\delta \tilde{B}$ in the estimated value $\tilde{B}$ of the field is dominated by quantum fluctuations in the observations performed on the atomic sample.   The results presented here fall under the umbrella of quantum parameter estimation theory \cite{Helstrom:1976a,Braunstein:1994a}, where the objective is to work within the rules of quantum mechanics to minimize, as much as possible, the propagation of this quantum uncertainty into the determination of metrological quantities, like $B$.

Given, for instance, a $y$-axis magnetic field $\mathbf{B} = B\,  \vec{\mathbf{y}}$, an atomic sample couples to $B$ via the magnetic dipole Hamiltonian
\begin{equation} \label{magnetometry:Equation::HLarmor}
	\hat{H} = - \hbar \gamma B \Fy,
\end{equation}
where $\gamma$ is the atomic gyromagnetic ratio and $\hat{F}_i =\sum_{j=1}^N \hat{f}_i^{(j)}$ ($i=\mathrm{x}, \mathrm{y}, \mathrm{z}$) are the collective spin operators obtained from a symmetric sum over $N$ identical spin-$f$ atoms.  If the atoms are initially polarized along the $x$-axis, the Larmor dynamics and thus $B$ can be inferred by observing the $z$-component of the atomic spin $F_{\mathrm{z}}$ \cite{Budker:2002a,Romalis:2003a,Geremia:2003a}.

Through the quantum Cram\'{e}r-Rao inequality \cite{Helstrom:1976a,Holevo:1982a,Braunstein:1994a,Braunstein:1995a}, it is possible to place an information-theoretic lower bound on the units-corrected mean-square deviation of the estimate $\tilde{B}$ from $B$,
\begin{equation} \label{magnetometry:Equation::CRError}
    \delta \tilde{B} = \left\langle\left(\frac{\tilde{B}}
    {\abs{d\expect{\tilde{B}}/dB}} - B\right)^2 \right\rangle^{1/2} .
\end{equation}
The behavior of the estimator uncertainty with the number of atoms $N$ depends on the characteristics (e.g., separable, entangled, etc.) of the quantum states used to compute the expectation value in Eq.\ (\ref{magnetometry:Equation::CRError}) as well as the nature of the induced dynamics \cite{Boixo:2007a}.  If one does not permit quantum entanglement between the different atoms in the probe, it can be shown that the optimal parameter resolution obtained from Eq.\ (\ref{magnetometry:Equation::HLarmor}) is given by the so-called shotnoise uncertainty \cite{Budker:2002a,Geremia:2003a}
\begin{equation} \label{magnetometry:Equation::DeltaBSN}
	\delta \tilde{B}_{\mathrm{SN}}(t) = \frac{1} { \gamma t \sqrt{2 F }} ,
\end{equation}
whose characteristic $1/\sqrt{N}$ scaling is a byproduct of the projection noise $\langle \Delta \Fz\rangle =\sqrt{F/2}$ for a spin coherent state \cite{Wineland:1994a} (here $F = f N$ for a sample of $N$ atoms each with total spin quantum number $f$).  It was believed for some time that the fundamental limit to parameter estimation, even when exploiting arbitrary entanglement between atoms in the probe, offers only a quadratic improvement
\begin{equation} \label{magnetometry:Equation::DeltaBHL}
	\delta \tilde{B}_{\mathrm{HL}}( t ) = \frac{\alpha} { \gamma t F} ,
\end{equation}
up to an implementation-dependent constant $\alpha$.  Eq.\ (\ref{magnetometry:Equation::DeltaBHL}) has traditionally been called the Heisenberg uncertainty scaling, and it can be achieved in principle for various spin resonance metrology problems \cite{Wineland:1994a}, including magnetometry \cite{Geremia:2003a}.  For an ensemble of $N$ spin-1/2 particles prepared into the initial cat-state $( \ket{\uparrow\uparrow \cdots \uparrow} + \ket{\downarrow\downarrow\cdots\downarrow})/\sqrt{2}$, the uncertainty scaling is given by $1/\gamma t N$ and is sometimes called the \textit{Heisenberg Limit}. 

Recently, however, it was shown that $1/N$ scaling can be surpassed \cite{Boixo:2007a} by extending the linear coupling that underlies Eq.\ (\ref{magnetometry:Equation::HLarmor}) to allow for multi-body collective interactions \cite{Boixo:2007a,Rey:2007a}.  Were one to engineer a probe Hamiltonian where $B$ multiplies $k$-body probe operators, such as $\Fy^k$, then the quantum Cramer-Rao bound \cite{Braunstein:1994a} indicates that the optimal estimation uncertainty would scale more favorably as $\Delta B_k \sim 1/N^k$ \cite{Boixo:2007a}.   Unfortunately, metrological coupling Hamiltonians are rarely up to us--- they come from nature, like the Zeeman interaction--- suggesting that one is stuck with a given uncertainty scaling without changing the fundamental structure of Eq.\ (\ref{magnetometry:Equation::HLarmor}).  Furthermore, it was shown in Ref.\ \cite{Boixo:2007a} that the addition of an auxiliary parameter-independent Hamiltonian $\hat{H}_1(t)$ such that
\begin{equation} \label{magnetometry:Equation::Haux}
	\hat{H} = - \hbar \gamma B \hat{F}_{\mathrm{y}} + \hat{H}_1(t)
\end{equation}
does not change the scaling of the parameter uncertainty for any choice of $\hat{H}_1(t)$.

At the same time, however, it should be well-appreciated that the dynamics one encounters in any actual physical setting are \textit{effective dynamics}.  Indeed, even the hyperfine Zeeman Hamiltonian Eq.\ (\ref{magnetometry:Equation::HLarmor}) is an effective description at some level.  This begs the question as to whether one can utilize an auxiliary system to induce effective dynamics that improve the uncertainty scaling in quantum parameter estimation by going outside the structure of Eq.\ (\ref{magnetometry:Equation::Haux}).  The purpose of this paper is to provide some direct evidence that doing so is possible.

In particular, we will study effective nonlinear couplings generated by double-passing an optical field through an atomic sample (q.v. Figure \ref{magnetometry:Figure::Schematic}) \cite{Sherson:2006a,Sarma:2008a}.  Continuous measurement of the scattered field then allows for the estimation of $\Fz$ and by extension, the magnetic field.  Building on the quantum stochastic calculus approach in \cite{Sarma:2008a}, I present the quantum filtering equations for estimating the state of the atomic sample.  Although the effective dynamics are no longer described by a Hamiltonian, numerical calculations of the quantum Fisher information can be used to obtain a theoretical lower bound on the uncertainty scaling of an optimal magnetic field estimator \cite{Braunstein:1994a}.  Such simulations suggest that for certain parameter regimes, the double-pass system's sensitivity to magnetic fields scales better than that of a comparable single-pass system and what would be computed by applying the methods of Ref.\ \cite{Boixo:2007a} to Eq.\ (\ref{magnetometry:Equation::Haux}).  Other simulations suggest that the quantum Heisenberg limit may be attained without generating any appreciable entanglement.  I also review direct simulations of magnetic field estimation for the system using quantum particle filters as further evidence for the improved uncertainty scaling provided by our proposed magnetometer.

Unfortunately the results are somewhat muted by the fact that despite our best efforts, we have not found a parameter estimator whose uncertainty scaling can be shown analytically to outperform the conventional Heisenberg limit.  In particular, I show that improved scaling is not achieved by a quantum Kalman filter \cite{Belavkin:1999a,Kalman:1960a,Kalman:1961a}, as such a filter is only suitable for estimating magnetic fields in the linear small-angle regime and where the state is Gaussian and the dynamics are well approximated by a low order Holstein-Primakoff expansion \cite{Holstein:1940a,Geremia:2003a}.   Although Kalman filters have had success in describing the single-pass system \cite{Geremia:2003a}, simulations suggest the Gaussian and small-angle approximations break down precisely when exact simulations of the double-pass system show improved sensitivity.  For pedagogical purposes, I detail the derivation of such linear-Gaussian filters using the method of projection filtering \cite{vanHandel:2005b,Mabuchi:2008a}.  Doing so allows us to observe directly the limitations that arise when imposing the small-angle and Gaussian assumptions, and it also provides a framework for the future development of more sophisticated filters.  


\begin{figure}[t]
\begin{center}
\includegraphics{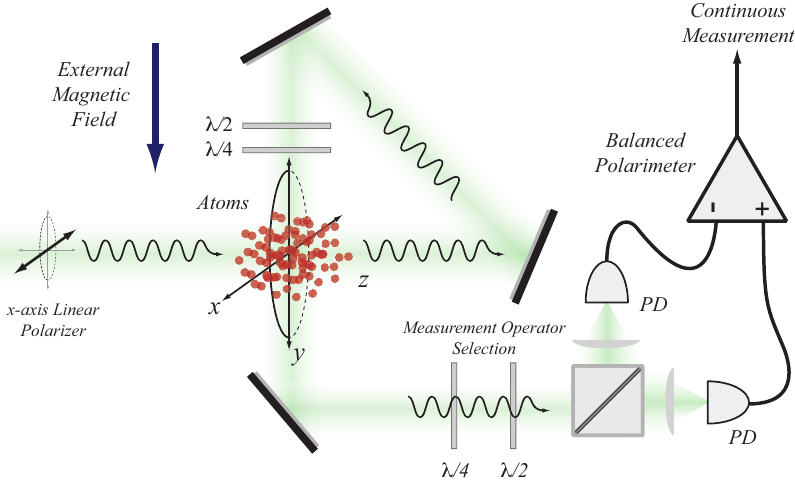}
\end{center}
\vspace{-4mm}
\caption[Schematic of double-pass atomic magnetometer]{ Schematic of a broadband atomic magnetometer based on continuous observation of a polarized optical probe field double-passed through the atomic sample.}
\label{magnetometry:Figure::Schematic}
\end{figure}

\section{Continuous measurement of the double-pass system} 
\label{magnetometry:sec:continuous_measurement_of_double_pass_system}

Consider the schematic in Fig.\ \ref{magnetometry:Figure::Schematic}.  The objective of this apparatus is to estimate the strength of a magnetic field oriented along the laboratory $y$-axis by observing the effect of that field on the spin state of the atomic sample.  Like most atomic magnetometer configurations, our procedure relies upon Larmor precession and uses a far-detuned laser probe to observe the spin angular momentum of the atomic sample.  Unlike conventional atomic magnetometer configurations, however, the probe laser is routed in such a way that it passes through the atomic sample twice prior to detection \cite{Sherson:2006a,Sarma:2008a}. 

Qualitatively, the magnetometer operates as follows.  The incoming probe field propagates initially along the atomic $z$-axis and is linearly polarized.  As a result of the atomic polarizability of the atoms, the probe laser polarization acquires a Faraday rotation proportional to the $z$-component of the collective atomic spin.  Two folding mirrors are then used to direct the forward scattered probe field to pass through the atomic sample a second time, now propagating along the atomic $y$-axis.  Prior to its second interaction with the atoms, polarization optics convert the initial Faraday rotation into ellipticity.  Thus on the second pass, the atoms perceive the optical helicity as a fictitious $y$-axis magnetic field acting in addition to the real field $B$, providing a positive feedback effect modulated by the strength of $B$.  The twice forward-scattered optical field is then measured in such a way that is sensitive only to the Faraday rotation induced by the first pass atom-field interaction.

\subsection{Quantum Stochastic Model}

When the collective spin angular momentum of a multilevel atomic system interacts dispersively with a traveling wave laser field with wavevector $\mathbf{k}$, the atomic spin couples to the two polarization modes of the electromagnetic field transverse to $\mathbf{k}$.   These polarization modes can can be viewed as a Schwinger-Bose field that when quantized in terms of a plane-wave mode decomposition yields the familiar Stokes operators:
\begin{eqnarray}
 	\sHOmegaZero & = &  +\frac{1}{2} \left(
				\aDOmegaX \aHOmegaX + \aDOmegaY \aHOmegaY 
			\right) \\ 
			& = & + \frac{1}{2} \left(
				\aDOmegaP \aHOmegaP + \aDOmegaM \aHOmegaM
			\right) \nonumber \\
	 \sHOmegaX & = & +\frac{1}{2} \left(
			\aDOmegaY \aHOmegaY - \aDOmegaX \aHOmegaX	 
	 	\right) \\
				& = & + \frac{1}{2} \left(
				\aDOmegaP \aHOmegaM + \aDOmegaM \aHOmegaP 	
			\right) \nonumber \\
	 \sHOmegaY & = & - \frac{1}{2} \left(
	 		\aDOmegaY \aHOmegaX + \aDOmegaX \aHOmegaY
	 	\right) \\
		& = & - \frac{i}{2} \left(
			\aDOmegaP \aHOmegaM -  \aDOmegaM \aHOmegaP
		\right) \nonumber \\
	 \sHOmegaZ & = & + \frac{i}{2} \left( 
	 		\aDOmegaY \aHOmegaX - \aDOmegaX \aHOmegaY
	 	\right) \\
		& = & + \frac{1}{2} \left(
			\aDOmegaP \aHOmegaP - \aDOmegaM \aHOmegaM
		\right). \nonumber		
 \end{eqnarray} 
Here, we have expressed the Stokes operators in terms of the Schr\"{o}dinger-picture field annihilation operators, $\aHOmegaX$ and $\aHOmegaY$, for the plane-wave modes with frequency $\omega$ and linear polarization along the x- and y-axes, respectively, as well as their corresponding transformations into the spherical polarization basis.  

In developing a physical model for the atom-field interaction in Fig.\ \ref{magnetometry:Figure::Schematic}, it is convenient to transform from a plane-wave mode decomposition of the electromagetic field to operators that are labeled by time.  Towards this end, we define the time-domain Schwinger boson annihilation operator as the operator distribution
\begin{equation}
	\hat{s}_t = \frac{1}{2} \int_{-\infty}^{+\infty}  g(\omega) 
		\, \aDOmegaX \aHOmegaY e^{i \omega t} d \omega,
\end{equation}
where $g(\omega)$ is a form factor.  This definition permits us to express the Stokes operators as
\begin{equation}
	\hat{s}_{\mathrm{z},t} =  i 
		\left(  \hat{s}^\dagger_t - \hat{s}^{\phantom{\dagger}}_t \right) \quad\text{and}\quad
	\hat{s}_{\mathrm{y},t}  =  - \left( \hat{s}^{\phantom{\dagger}}_t + \hat{s}^\dagger_t  \right),
\end{equation}
which should be reminiscent of quadrature operators and also places the field operators in a form that is directly in line with the standard nomenclature adopted in the field of quantum stochastic calculus. 

With a suitable orientation of the polarization optics ($\lambda/2$ and $\lambda/4$) in Fig.\ \ref{magnetometry:Figure::Schematic}, the interaction Hamiltonians for each pass of the probe light through the sample are then
\begin{eqnarray} 
		\hat{H}_t^{(1)}  & = & + \hbar \mu \Fz \hat{s}_{\mathrm{z},t}
			=  + i\hbar\mu\Fz \left( \hat{s}_t^{\dagger} - \hat{s}_t^{\phantom{\dagger}} \right)
		\label{magnetometry:Equation::H1} \\
		\hat{H}_t^{(2)}  &= &  +\hbar \kappa \Fy \hat{s}_{\mathrm{y},t} =  -\hbar\kappa \hat{F}_{\mathrm{y}} \left( \hat{s}_t^{\phantom{\dagger}} + \hat{s}_t^\dagger \right),
		\label{magnetometry:Equation::H2}
\end{eqnarray}
respectively.  Note that in developing these Hamiltonians, which are of the standard atomic polarizability form, it was assumed that rank-two spherical tensor interactions \cite{Jessen:2004a,Geremia:2006b} can be neglected.  In practice, the validity of such an assumption can depend heavily on the choice of atomic level structure and experimental parameters such as the intensity and detuning of the probe laser field.

In addition to specifying the Hamiltonians for the two atom-field interactions, it is also necessary to stipulate the measurement to be performed on the probe laser.  Since we expect that the amount of Larmor precession (possibly augmented by the addition of the double-passed probe field) will cary information about the magnetic field strength $B$, we must choose the measured field operator $\hat{z}_t$ appropriately.   Since the magnetic field drives rotations about the atomic $y$-axis, it is the $z$-component of the atomic spin that indicate such a rotation.  From the form of the first-pass interaction Hamiltonian $\hat{H}^{(1)}_t$, we see that the $z$-component of the atomic spin couples to dynamics generated by the field operator $\hat{s}_{\mathrm{z},t} = i ( \hat{s}^\dagger_t - \hat{s}^{\phantom{\dagger}}_t)$.  The affect of such a coupling is then observed by measuring the orthogonal quadrature, indicating that the appropriate polarization measurement should be $\hat{z}_t  =  \hat{s}_{\mathrm{y,t}}$.

\subsubsection{The Stochastic Propagator and Quantum Filter}

Analyzing the two individual interactions $\hat{H}_1$ and $\hat{H}_2$ via the stochastic limit studied in Chapter \ref{chapter:quantum} gives rise to the following quantum stochastic differential equations (QSDE) for the interaction-picture propagators
\begin{align}
	d\hat{U}_t^{(1)} &= \left\{\sqrt{m}\Fz (dS_t^\dagger - dS_t) 
			- \frac{1}{2}m\Fz^2dt 
			- \frac{i}{\hbar} \hat{H} dt \right\}U_t^{(1)} \label{magnetometry:Equation::dU1}\\
	d\hat{U}_t^{(2)} &= \left\{i\sqrt{k}\Fy(dS_t + dS^{\dagger}_t) 
			- \frac{1}{2}k\Fy^2dt
			- \frac{i}{\hbar} \hat{H}dt \right\}U_t^{(2)} \label{magnetometry:Equation::dU2}
\end{align}
where $m$ and $k$ are the weak-coupling interaction strengths obtained from the rates $\mu$ and $\kappa$, $\hat{H}$ is an arbitrary atomic Hamiltonian and $d\St^\dagger$ and $d \St$ are delta-correlated noise operators derived from the quantum Brownian motion
\begin{equation}
	\St = \int_0^t \hat{s}_u du.
\end{equation}
The noise terms satisfy the quantum It\^o rules: $d \St d \St^\dagger = dt$ and $d \St^\dagger d \St = d \St^2 = (d \St^\dagger)^2 = 0$, and can be viewed heuristically as a consequence of vacuum fluctuations in the probe field.

To obtain a single weak-coupling limit for the double-pass interaction, we combine the separate equations of motion for the two propagators into a single weak-couping limit as follows.  First, write the two single-pass evolutions in terms of the generators of the dynamics
\begin{equation}
	d\hat{U}_t^{(1)} = \hat{a}_t U_t^{(1)}, \quad \mathrm{and},
	\quad
	d\hat{U}_t^{(2)} = \hat{b}_t U_t^{(1)}
\end{equation}
and then expand the differential $d\hat{U}_t$ of the combined propagator
\begin{eqnarray}
	d \hat{U}_{t+\delta t}  & = & ( \hat{1} +\hat{b}_t ) ( \hat{1} + \hat{a}_t ) \hat{U}_t \\
	& = &  \hat{U}_t + \left( \hat{a} + \hat{b} + \hat{b} \hat{a} \right) \hat{U}_t
\end{eqnarray}
such that the combined propagator $d \hat{U}_t = \hat{U}_{t+\delta t} - \hat{U}_t$ then satisfies
\begin{equation}
	d\hat{U}_t = \left( \hat{a} + \hat{b} + \hat{b} \hat{a} \right) \hat{U}_t \,.
\end{equation}
After evaluating the combined evolution for the propagators in Eqs.\ (\ref{magnetometry:Equation::dU1}) and (\ref{magnetometry:Equation::dU2}) in light of the quantum It\^{o} rules, we find that the single weak-coupling limit propagator satisfies
\begin{multline}
	d\hat{U}_t=   \left[i\sqrt{km}\Fy\Fz dt 
		- \frac{1}{2}m\Fz^2 dt - \frac{1}{2}k\Fy^2dt  - \frac{2 i}{\hbar} \hat{H} dt \right. \\
		 \left. 
			+ \sqrt{m}\Fz(d\St^{\dagger} - d\St)
			+ i \sqrt{k}\Fy(d\St^{\dagger} + d\St) \right]\hat{U}_t. 
\end{multline}
Observe that as a result of the manner in which the combined weak-coupling limit was taken, the Hamiltonian term has the property that rates which appear in it differ by a factor of two from those that would be expected from a single weak-coupling limit.  This factor of two is essentially the rescaling of time units that arises from aggregating two sequential weak-coupling limits as a single differential process.  To retain consistency with the original definition of the frequencies that appear in the parameter-coupling Hamiltonian, it is essential to rescale time units such that frequencies in the parameter-coupling Hamiltonian are as expected.  Doing so is accomplished by reversing the effective $2 dt  \rightarrow dt$ transformation that occurred in the derivation, and thus dividing all rates by two to give
\begin{multline} \label{magnetometry:eq:double_pass_prop}
	d\hat{U}_t =  \left[i\sqrt{KM}\Fy\Fz dt 
		- \frac{1}{2}M\Fz^2 dt - \frac{1}{2}K\Fy^2dt  - \frac{i}{\hbar} \hat{H} dt \right. \\
		 \left. 
			+ \sqrt{M}\Fz(d\St^{\dagger} - d\St)
			+ i \sqrt{K}\Fy(d\St^{\dagger} + d\St) \right]\hat{U}_t
\end{multline}
where $M = m / 2$ and $K= k/2$.  I note that this final result agrees with the propagator obtained by Sarma~et.~al \cite{Sarma:2008a} who also derived the quantum stochastic propagator of this system in order to characterize the generation of polarization and spin squeezing as suggested by Sherson and M{\o}lmer \cite{Sherson:2006a}.

Following the derivation of the quantum filter in Chapter \ref{chapter:quantum}, we recognize the dipole operator $\hat{L} = \sqrt{M}\Fz + i\sqrt{K}\Fy$ and Hamiltonian $H = -\gamma B\Fy - \sqrt{KM}(\Fz\Fy + \Fy\Fz)/2$ in comparing the double pass propagator of Eq.~\eqref{magnetometry:eq:double_pass_prop} to the general form of Eq.~\eqref{quantum:eq:generic_propagator}.  Plugging in these forms into the adjoint filter of Eq.~\eqref{quantum:eq:quantum_filter_rho_form} yields the \emph{double-pass quantum filter}
\begin{eqnarray} 
    d\rho_t & = & i\gamma B[\Fy,\rho_t]dt + i\sqrt{KM}[\Fy,\{\Fz,\rho_t\}]dt \nonumber \\
	&&    + M \mathcal{D}[\Fz]\rho_t dt 
             + K\mathcal{D}[\Fy]\rho_t dt \label{magnetometry:eq:adjoint_quantum_filter} \\
         &&
             + \left(\sqrt{M}\mathcal{M}[\Fz]\rho_t 
             + i\sqrt{K}[\Fy,\rho_t]\right)dW_t \nonumber
\end{eqnarray}
where the innovations process 
\begin{equation}
	dW_t = dZ_t - 2\sqrt{M}\Tr{\Fz\rho_t}dt
\end{equation}
is a Wiener process, i.e. $\mathbbm{E}[dW_t] = 0, dW_t^2 = dt$.  The various superoperators are defined as
\begin{align}
    \mathcal{D}[\Fk]\rho_t &= \Fk\rho_t\Fk^{\dag} - \frac{1}{2}\Fk^{\dag}\Fk\rho_t - \frac{1}{2}\rho_t\Fk^{\dag}\Fk\\
    \mathcal{M}[\Fz]\rho_t &= \Fz\rho_t + \rho_t\Fz - 2\Tr{\Fz\rho_t}\rho_t\\
    \{\Fz,\rho_t\} &= \Fz\rho_t + \rho_t\Fz
\end{align}

One other form which is useful when the quantum state remains pure is the stochastic Schr\"{odinger} equation (SSE).  As developed in Appendix \ref{magnetometry:app:SSE}, the SSE for the double-pass quantum filter is
\begin{eqnarray} \label{magnetometry:eq:double_pass_SSE}
	d\ket{\psi}_t & = & \left(
		i\gamma B\Fy -\frac{M}{2}(\Fz-\expect{\Fz}_t)^2 \right. \\
	&& \left. + i\sqrt{KM}\Fy(\Fz + \expect{\Fz}_t)
	 					- \frac{K}{2}\Fy^2\right)\ket{\psi}_t dt \nonumber \\
	&&
				  + \left( \sqrt{M}(\Fz - \expect{\Fz}_t) + i\sqrt{K}\Fy\right)\ket{\psi}_t dW_t . \nonumber
\end{eqnarray}


\section{The Quantum Cram\'{e}r-Rao Inequality }
In order to characterize the performance of the magnetometer, we may consider quantum information theoretic bounds on the units-corrected mean-square deviation of the magnetic field estimate $\tilde{B}$ of the true magnetic field $B$ \cite{Braunstein:1994a,Braunstein:1995a},  given in Eq.\ (\ref{magnetometry:Equation::CRError}).  The quantum Cram\'{e}r-Rao bound  \cite{Helstrom:1976a,Holevo:1982a,Braunstein:1994a,Braunstein:1995a} states that the deviation of \emph{any} estimator is constrained by  
\begin{equation}
    \delta \tilde{B} \geq \frac{1}{\sqrt{\mathcal{I}_B(t)}}, \quad \mathcal{I}_B(t) = \Tr{\rho_B(t)\mathfrak{L}_B^2(t)} ,
\end{equation}
where the quantum Fisher information $\mathcal{I}_B(t)$ is the expectation of the square of the symmetric logarithmic derivative operator, defined implicitly as
\begin{equation}
    \frac{\partial \rho_B(t)}{\partial B} = \frac{1}{2}(\mathfrak{L}_B(t)\rho_B(t) + \rho_B(t)\mathfrak{L}_B(t)) .
\end{equation}
For pure states, $\rho_B^2 = \rho_B$, so that
\begin{equation}
    \mathfrak{L}_B(t) = 2\frac{\partial \rho_B(t)}{\partial B}
\end{equation}
which indicates
\begin{equation} \label{magnetometry:eq:cramer_rao}
    \delta \tilde{B} \geq \frac{1}{2} \left\langle\left(\frac{\partial \rho_B(t)}{\partial B}\right)^2\right\rangle^{-\frac{1}{2}} .
\end{equation}
In this form, we see that the lower bound is related to the sensitivity of the evolved state to the magnetic field parameter.  That is, any estimator's performance is constrained by how well the dynamics transform differences in the value of $B$ into differences in Hilbert space.  

As discussed by Boixo et.~al in \cite{Boixo:2007a}, for Hamiltonian evolution, the quantum Cram\'{e}r-Rao bound may be expressed in terms of the operator semi-norm, which is the difference between the largest and smallest (non-degenerate) eigenvalues of the probe Hamiltonian.  For the magnetic dipole Hamiltonian in Eq.\ (\ref{magnetometry:Equation::HLarmor}), this bound is simply the Heisenberg limit in Eq.\ (\ref{magnetometry:Equation::DeltaBHL}).  More generally, the authors show that a probe Hamiltonian which involves $k$-body operators gives rise to an uncertainty scaling of $1/tF^k$.  They further argue that no ancillary quantum systems or auxiliary Hamiltonians contribute to this bound; it is determined solely by the Hamiltonian that directly involves the parameter of interest.

Such analysis suggests the double-pass quantum system, whose only direct magnetic field coupling is in the magnetic dipole Hamiltonian, should show no more sensitivity than a single pass system.  There are several reasons why one might believe there is more to the story.  Firstly, the unitary evolution of the joint atom-field system in Eq.~(\ref{magnetometry:eq:double_pass_prop}) involves an auxiliary system of infinite dimension.  As such, it is not clear that the arguments leading to the operator semi-norm are valid, in particular due to the fact that the white noise terms $d\St,d\St^{\dag}$ are singular.  Additionally, the double-pass limit is a Markov one, in which the interaction the light field mediates between atoms is essentially instantaneous relative to other time-scales in the problem.  The effective interaction is therefore fundamentally different than one in which measurements of a finite dimensional ancilla system are used to modulate the evolution of the probe atoms.  Thus the conditioned system, given in terms of the quantum filter of Eq.~(\ref{magnetometry:eq:double_pass_SSE}), does not correspond to unitary dynamics.  Indeed, looking at Eq.\ (\ref{magnetometry:eq:double_pass_SSE}), we see that the local generator of dynamics is \emph{path-dependent}, given in terms of the expectation of $\Fz$.  Therefore, as the magnetic field directly impacts the state through the magnetic dipole term, it also non-trivially modulates future dynamics through a state-dependent generator.

\subsection{Numerical Analysis of the Quantum Fisher Information}
Unfortunately, it is not clear how to fold the quantum stochastic or quantum filtered dynamics analytically into the semi-norm bound considered in \cite{Boixo:2007a}.  Nonetheless, the quantum Cram\'{e}r-Rao bound in Eq.\ (\ref{magnetometry:eq:cramer_rao}) is excellent fodder for computer simulation.  By numerically integrating the stochastic Schr\"{o}dinger form of the quantum filter in Eq.\ (\ref{magnetometry:eq:double_pass_SSE}), a finite difference approximation of $\partial \rho_B(t)/\partial B$ may be evaluated for different collective spin sizes $F$.  That is, for a given choice of $F$, a finite difference approximation of the quantum Fisher information near $B = 0$ can be constructed by co-evolving three trajectories, $\rho_0(\tau)$, $\rho_{\delta B}(\tau)$, and $\rho_{-\delta B}(\tau)$ (seeded by the same noise realization),  and calculating
\begin{equation}
    \mathcal{I}_B | Z_{(0,t)} \left< \left(\frac{\partial \rho_B(\tau)}{\partial B}\right)^2 \right>
        \approx \Tr{\left(\frac{\rho_{\delta B}(\tau) 
                - \rho_{-\delta B}(\tau)}{2\delta B}\right)^2 \rho_{0}(\tau)}.
                \label{magnetometry:eqn:cr_fd}
\end{equation}   
As is suggestively written, the Fisher information calculated on the particular measurement realization that generated $\hat{\rho}_t$ and must be averaged over many realizations to obtain the unconditional quantum Fisher information $\mathcal{I}_t = \mathbbm{E}[ \mathcal{I}_t | Z_t ]$.   The lower bound $\delta\tilde{B}_\tau$ can then be obtained from Eq.\ (\ref{magnetometry:eq:cramer_rao}) with statistical errorbars given by $\sigma (\delta\tilde{B}_\tau) = \mathcal{I}_\tau^{-3/2} \sigma[ \mathcal{I}_\tau | Z_t ] / 2$.

\subsubsection{Simulation Results}
\begin{figure}[t]
    \centering
        \includegraphics[scale=1.2]{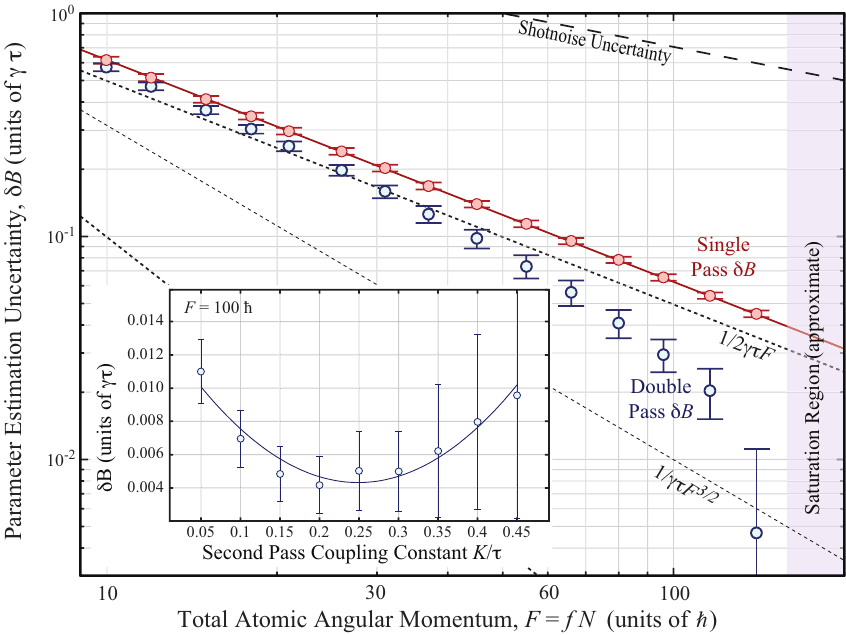}
    \caption[Numerical Fisher Information scaling beyond $1/N$]{ Comparison of the estimation uncertainty $\Delta \tilde{B}$ as a function of the total atomic angular momentum (proportional to $N$) for double-pass and single-pass atomic magnetometers determined by calculating the quantum Fisher Information with $M=1$ (in units of $1/\tau$) and $K=1\times10^{-4}$ chosen to be optimal for $F=140 \hbar$. }
    \label{magnetometry:fig:cramer_rao}
\end{figure}

We calculated $\mathcal{I}_t(B)$ over a range of spin quantum quantum numbers $F=Nf$ spanning more than an order of magnitude to determine a lower bound on the magnetic field estimation uncertainty using  Eq.\ (\ref{magnetometry:eq:cramer_rao}).   The results indicate that the Fisher information depends heavily upon the choice of the coupling strengths $M$ and $K$, which is not surprising since the measurement strength $M$ determines how much spin-squeezing is generated and $K$ determines the strength of the effective nonlinearity.   Like any measurement procedure that involves amplification, both the signal and noise are affected, and optimal performance requires choosing the correct gain.

If one choses $M = 1 / \tau$, to obtain an optimal spin-squeezed state at the final time $t=\tau$ \cite{Geremia:2003a}, then it is straightforward to optimize over the nonlinearity $K$, as illustrated in the inset of Fig.\ (\ref{magnetometry:fig:cramer_rao}) for $F=100 \hbar$.  We found that the optimal value $K^*$ depends upon the number of atoms, and that the Fisher information saturates and then decreases if the number of atoms exceeds the value of $N=F_{\mathrm{sat}}/f$ used to compute $K^*(F)$.    Figure \ref{magnetometry:fig:cramer_rao}) shows the behavior of $\delta \tilde{B}_\tau$ as a function of $F$ up to the saturation point $F < F_{\mathrm{sat}}\sim 150$.  The largest value of $F$ prior to saturation yields a $\delta\tilde{B}_\tau$ that is slightly below the bound $1/\tau \gamma F^{3/2}$ that would be obtained for a two-body coupling Hamiltonian and an initially separable state $\hat{\rho}_0$ \cite{Boixo:2008a}.  Despite this saturation of the quantum Fisher information for $F > F_{\mathrm{sat}}$ at a given choice of $K$, one can choose the value of $K^*$ such that saturation occurs only for $F_{\mathrm{sat}} > F_{\mathrm{max}}$ over any specified finite range $F \le F_{\mathrm{max}}$.  An improvement beyond $1/N$ scaling can be achieved over any physically realistic number of particles.

The saturation effect can be understood in light of the quantum stochastic model of the previous subsection.  In considering the general stochastic propagator of Eq.\ (\ref{magnetometry:eq:double_pass_prop}), we identified the coupling operator $\hat{L} = \sqrt{M}\Fz + i\sqrt{K}\Fy$, which if $M = K$, is essentially the angular momentum lowering operator along $x$---$\Fmx$.  If $M,K \gg \gamma B$, a continuous measurement of this operator very quickly moves the $+x$-polarized initial state onto the $-x$-polarized state, which is an attractive fixed point of $\Fmx$.  Once this state is reached, the dynamics become relatively insensitive to the magnetic field value and result in a poor uncertainty lower bound.  On the other hand, if $M,K$ are much smaller than $\gamma B$, the positive feedback from the $i\sqrt{K}\Fy$ term is washed out by Larmor precession.  Given that we are interested in detection limits, i.e. $B \approx 0$, we do not focus on the regime where Larmor precession dominates.

A second approach to avoiding saturation of the Fisher information for large $F$ is to scale the parameters $M$ and $K$ as a decreasing function of $F$.   For practical considerations, it is also desirable to set $M=K$ as these parameters are determined by the atom-field coupling strengths on the first and second pass interactions, thus quantities such as the laser intensity and detuning not easily changed between the two passes.  We have found that scaling $M$ and $K$ according to the functional form
\begin{equation} \label{magnetometry:Equation::MKScaling}
	M = K = c / \tau F^\alpha, 
\end{equation}
where $c$ and $\alpha$ are constants, leads to a power-law scaling for the uncertainty bound $\delta \tilde{B}_\tau \sim 1/N^k$.  The inset plot in Fig.\ (\ref{magnetometry:Figure::WeakFScaling}) shows the slope of a linear fit of $\log_{10}\delta\tilde{B}_\tau$ to $\log_{10}F$ (i.e., a slope of $k=-1$ corresponds to the Heisenberg uncertainty scaling) as a function of $\alpha$ (with $c$ chosen so as to avoid the saturation behavior described above).  As demonstrated by the data points in Fig.\ (\ref{magnetometry:Figure::WeakFScaling}), it is possible to achieve $1/N$ scaling (to within a small prefactor offset) with $\alpha = 0.77$ and $c=0.589$.   The distribution of conditional uncertainties $\delta\tilde{B}_\tau | Z_t$ for the statistical ensemble of measurement realizations [dots in Fig.\ (\ref{magnetometry:Figure::WeakFScaling})] is depicted for the different values of $F$.  The mean and uncertainty of this distribution are denoted by the circles and errorbars, and a fit to this data gives $\delta\tilde{B}_\tau \sim F^{-0.97}$.
\begin{figure}[tb]
\begin{center}
\includegraphics{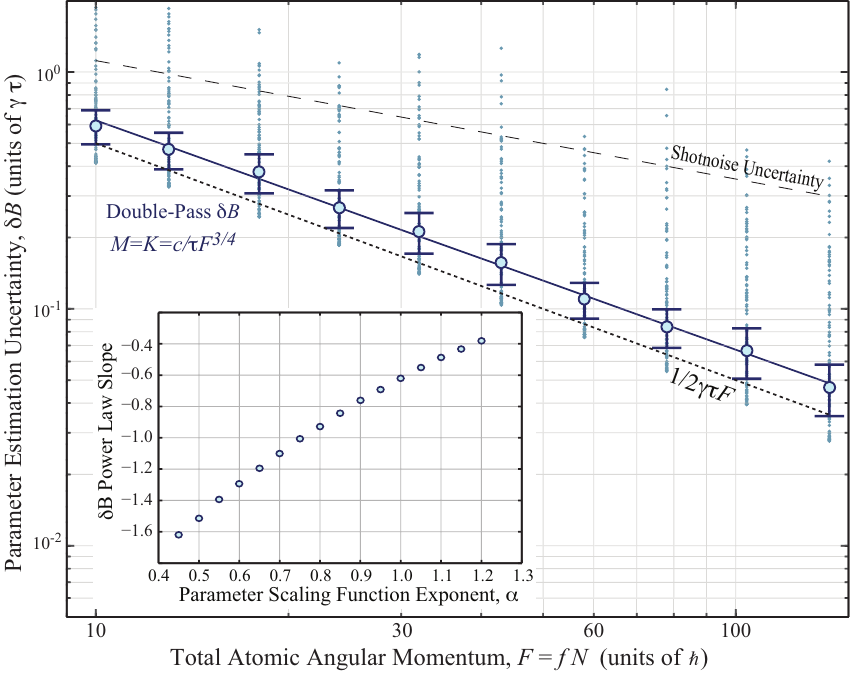}
\end{center}
\vspace{-5mm}
\caption[Numerical Fisher Information scaling as $1/N$ with minimal squeezing]{Evidence that the field estimation uncertainty $\Delta \tilde{B}$ can be made to scale as a power law $\delta\tilde{B}_\tau \sim 1/N^k$ by decreasing the parameters $M$ and $K$ as a function of the total angular momentum $F$ according to Eq.\ (\ref{magnetometry:Equation::MKScaling}) with $\alpha \approx 3/4$.  The power-law fit (solid line) has a slope of $-0.97$. \label{magnetometry:Figure::WeakFScaling}} 
\end{figure}

In short, Figures \ref{magnetometry:Figure::WeakFScaling} and \ref{magnetometry:fig:cramer_rao} suggest that there are some parameter values, appropriate for some range of $F$, which show an estimator uncertainty lower bound scaling at and or the Heisenberg limit.  In practice, it seems that one would need to fine tune the coupling strengths $M$ and $K$ in order to be in a regime with such scaling.  It may be that such coupling strengths are inaccessible in an experimental setting.  While this is an important consideration, there is a more pressing theoretical question---does a practical estimator exist which saturates the quantum Cram\'{e}r-Rao bound?  I summarize our search for such an estimator in the following section.

\section{Magnetic Field Estimators} 
\label{magnetometry:sec:magnetic_field_estimators}

While studying the properties of lower bounds on estimator performance is important for developing an understanding of the capabilities of a given parameter coupling scheme, any actual procedure for implementing quantum parameter estimation must also develop a constructive procedure for doing the estimation.

In this section, we consider two methods for estimating the strength of the magnetic field $B$ based on the stochastic measurement record $Z_{(0,t)}$.  In both cases, we extend the quantum filters developed in the previous section to account for our uncertainty in $B$, which in turn results in new filters capable of estimating $B$.  

\subsection{Quantum Particle Filter} 
\label{magnetometry:sec:quantum_particle_filter}
The technique of quantum particle filtering, as developed in Chapter \ref{chapter:quantum_parameter_estimation} and reviewed below, leverages the fact that the quantum filtering equations already provide a means for estimating the \emph{state} of a quantum system conditioned on the measurement record.  If we place the magnetic field parameter on the same footing as the quantum state, we can simply apply the quantum filtering results we already derived.  Indeed, by embedding the magnetic field parameter as a diagonal operator in an auxiliary Hilbert space, the quantum filter \emph{still} gives the best estimate of \emph{both} system and auxiliary space operators.  We accomplish this by promoting the magnetic field parameter to the diagonal operator
\begin{equation}
    B \mapsto \hat{B} = \int B \ketbra{B}{B} dB \in \mathcal{H}_B, 
\end{equation}
where $\mathcal{H}_B$ is the new auxiliary Hilbert space with basis states satisfying $\hat{B}\ket{B} = B\ket{B}$ and $\braket{B}{B'} = \delta(B-B')$.  All atomic operators and states, which are associated with the atomic Hilbert space $\mathcal{H}_A$, act as the identity on $\mathcal{H}_B$, e.g. $\Fz \mapsto I \otimes \Fz$.  The only operator which joins the two spaces is the magnetic dipole Hamiltonian, which is now given by
\begin{equation}
    \hat{H} \mapsto -\hbar\gamma \hat{B}\otimes\Fy
\end{equation}
The derivation of the quantum filtering equation is essentially unchanged, provided one replaces atomic operators with the appropriate forms for the joint space $\mathcal{H}_B \otimes \mathcal{H}_A$.  

For parameter estimation, the adjoint form is the more convenient version of the quantum filter.  Since $\hat{B}$ corresponds to a classical parameter, we require the marginal density matrix $(\rho_B)_t = \operatorname{Tr}_{\mathcal{H_A}}[\rho_t]$ be diagonal in the basis of $\hat{B}$, so that it corresponds to a classical probability distribution.  This suggests we write the total conditional density matrix in the ensemble form
\begin{equation} \label{magnetometry:eq:ensemble_continuous_form}
    \rho_t^E = \int dB p_t(B) \ketbra{B}{B} \otimes \rho_t^{(B)} 
\end{equation}  
where $p_t(B) = P(B | Z_{(0,t)})$ is precisely the conditional probability density for $B$.

While one could attempt to update this state via the quantum filter, doing so is entirely impractical, as one can not represent an arbitrary distribution for $p_t(B)$ with finite resources.  Instead, one approximates the distribution with a weighted set of point masses or \emph{particles}:
\begin{equation} \label{magnetometry:eq:approximate_density}
    p_t(B) \approx \sum_{i=1}^N p_t^{(i)} \delta(B - B_i) .
\end{equation}
The approximation can be made arbitrarily accurate in the limit of $N \rightarrow \infty$.  Plugging this distribution into the ensemble density matrix form of Eq. \ref{magnetometry:eq:ensemble_continuous_form} gives
\begin{equation} \label{magnetometry:eq:discrete_ensemble_rho}
    \rho_t^E = \sum_{i = 1}^N p_t^{(i)} \ketbra{B_i}{B_i} \otimes \rho_t^{(B_i)} 
\end{equation}
Each of the $N$ triples $\{p_t^{(i)}, B_i, \rho_t^{(B_i)} \}$ is called a \emph{quantum particle}.  Intuitively, the particle filter works by discretizing the parameter space and then evolving an ensemble of quantum systems according to the exact dynamics for each parameter value.  The filtering equations below perform Bayesian inference on this ensemble, updating the relative probabilities of particular parameter values given the measurement record.

The quantum particle filter for the double-pass system with \emph{unknown} $B$ is found by plugging the discretized ensemble $\rho_t^E$ into the extended double-pass filter.  After a little manipulation, one finds
\begin{subequations}\label{magnetometry:eq:ensemble_filtering_eqs}
	\begin{eqnarray} 
		    dp_t^{(i)} &= & 2\sqrt{M}(\Tr{\Fz\rho_t^{(B_i)}} \nonumber \\
		    &&
		        - \sum_{j=1}^N p_t^{(j)}\Tr{\Fz\rho_t^{(B_j)}})p_t^{(i)}dW_t 
		          \\
		    d\rho_t^{(B_i)} &= & i\gamma B_i[\Fy,\rho_t^{(B_i)}]dt
		     + \sqrt{KM}[\Fy,\{\Fz,\rho_t^{(B_i)}\}]dt \nonumber\\
		    && + M \mathcal{D}[\Fz]\rho_t^{(B_i)} dt 
		             + K\mathcal{D}[\Fz]\rho_t^{(B_i)} dt \\
		             && + \left(\sqrt{M}\mathcal{M}[\Fz]\rho_t^{(B_i)} 
		             + i\sqrt{K}[\Fy,\rho_t^{(B_i)}]\right)dW_t \nonumber \\
		    dW_t &=& dZ_t 
		        - 2\sqrt{M}\sum_{i=1}^Np_t^{(i)}  \Tr{\Fz\rho_t^{(B_i)}}dt
	\end{eqnarray}
\end{subequations}
where the prior distribution $p_0(B)$ is used to determine the initial parameter weights, $p_0^{(i)}$, and values, $B_i$.  All initial quantum states, $\rho_0^{(B_i)}$, are taken to be the spin coherent state pointing along $+x$.

An estimate of the magnetic field strength is then constructed from the approximate density in Eq.\ (\ref{magnetometry:eq:approximate_density}), either taking the most probable $B$ value, corresponding to the largest $p_t^{(i)}$ or calculating the expected value of $\hat{B}$
\begin{equation}
    \tilde{B}_{pf} = \expect{\hat{B}} = \sum_{i=1}^{N} p_t^{(i)} B_i .
\end{equation}
For the latter estimate, the uncertainty is given by
\begin{eqnarray}
    \Delta \tilde{B}_{pf} & = & \left(\expect{\hat{B}^2} - \tilde{B}_{pf}^2\right)^{1/2}  \nonumber \\
    & = & \left(\sum_{i=1}^{N} p_t^{(i)} B_i^2 - \tilde{B}_{pf}^2\right)^{1/2}  .
\end{eqnarray}

\subsection{Quantum Kalman Filter} 
\label{magnetometry:sec:quantum_kalman_filter}
Rather than constructing a magnetic field estimator from the exact quantum dynamics, one could instead first focus on deriving an approximate filter for the atomic state, which is then a starting point for the magnetic field estimator.  Indeed, previous work in precision magnetometry via continuous measurement \cite{Geremia:2003a} has taken this route by constructing a \emph{quantum Kalman filter} to describe the atomic dynamics.  Such a filter leverages the fact that for an initially spin polarized state of many atoms (say along $+x$), a first order Holstein-Primakoff expansion \cite{Holstein:1940a} linearizes the small-angle dynamics in terms of a Gaussian state characterized by the means $\pi_t[\Fz], \pi_t[\Fy]$ and the covariances $\Delta\Fz^2,\Delta\Fy^2,\Delta\Fz\Fy$.  Just as we saw in developing the Kalman-Bucy filter of Theorem \ref{thm:kalman_bucy}, the conditional state for a linear system with Gaussian noise is itself described by a Gaussian distribution and therefore only requires filtering equations for the means and a deterministic equation for the variances \cite{Kalman:1960a,Kalman:1961a}.  For the case of magnetometry, the number of these parameters is independent of the number of atoms in the atomic ensemble.  We will also find that within this approximation, we can again embed $B$ as an unknown state parameter and find a corresponding Kalman filter appropriate for estimating its value.  

However, applying the small-angle and Gaussian approximations in the quantum case is usually done in an ad-hoc fashion, especially in light of the recent introduction of \emph{projection filtering} into the quantum filtering setting \cite{vanHandel:2005b,Mabuchi:2008a}.  In this framework, one selects a convenient manifold of states whose parameterization reflects the approximations to enforce.  At each point in this manifold, the exact differential dynamics induced on these states is orthogonally projected back into the chosen family.  For our purposes, this means projecting the filter in Eq.\ (\ref{magnetometry:eq:double_pass_SSE}) onto a manifold of Gaussian spin states. Although the resulting equations are not substantively different than those derived less carefully, we believe the potential application of projection filtering in deriving other approximate filters and master equations warrants the following exposition.  

\subsubsection{Projection Filter Overview}
Abstractly, projection filtering proceeds as follows.  We assume we already have a dynamical equation, such as Eq.\ (\ref{magnetometry:eq:double_pass_SSE}), for a given manifold of states, such as pure states.  For convenience, let these dynamics be represented as
\begin{equation}
    d\ket{\psi}_t = \mathcal{N}[\ket{\psi}_t] ,
\end{equation}
where $\mathcal{N}$ is the generator of dynamics.  Now select the desired family of ``approximating'' states which are a submanifold of the exact states.  We assume this family is parameterized by a finite number of parameters $x_1,x_2,\ldots,x_n$ and we denote states in this family as $\ket{x_1,x_2,\ldots,x_n}$.  At every point in this manifold, the tangent space is spanned by the tangent vectors
\begin{equation}
    v_i = \partialD{\ket{x_1,x_2,\ldots, x_n}}{x_i} .
\end{equation}
Loosely speaking, these tangent vectors tell us how differential changes in the parameters move us through the corresponding submanifold of $\ket{x_1,x_2,\ldots,x_n}$ states in the space of pure states.  This is particularly useful, as the action of the generator $\mathcal{N}[\ket{x_1,x_2,\ldots,x_n}]$ does not necessarily result in a state within the family.  But by projecting the dynamics onto the tangent space, we can find a filter, called the projection filter, which constrains evolution within the chosen submanifold.  Explicitly, this projection is written as
\begin{eqnarray} 
    T & = & \Pi_{\text{span}\{v_i\}}[d\ket{x_1,x_2,\ldots,x_n}] \nonumber \\
        & = & \sum_i \frac{\langle v_i,
            \mathcal{N}[\ket{x_1,x_2,\ldots,x_n}] \rangle}
                {\langle v_i, v_i \rangle}v_i, \label{magnetometry:eq:general_projection}
\end{eqnarray}
where in this pure state formulation, the inner product is the standard Hilbert space inner product.  

\subsubsection{Gaussian State Family and Tangent Vectors}
For our double-pass magnetometer, we begin by introducing the two-parameter family of Gaussian states
\begin{align} \label{magnetometry:eq:gaussian_state_parameterization}
        \ket{\theta_t,\xi_t} &= e^{-i\theta_t\Fy}e^{-2i\xi_t(\Fz\Fy + \Fy\Fz)}\ket{F,+F_x} \nonumber\\
        &= \Yt\Sqt\ket{F,+F_x}, 
\end{align}
where $\ket{F,+F_x}$ is the spin coherent state pointing along $+x$, $\Sqt$ is a spin squeezing operator \cite{Kitagawa:1993a} with squeezing parameter $\xi_t$ and $\Yt$ is a rotation about the $y$-axis by angle $\theta_t$.  Intuitively, the squeezing along $z$ generated by $\Sqt$ corresponds to the squeezing induced by measuring $\Fz$.  The rotation via $\Yt$ then accounts for both the random evolution due to the measurement as well as any rotation induced by the magnetic field.  The tangent vectors for these states are
\begin{eqnarray}
    v_{\theta_t} & = & \partialD{\ket{\theta_t,\xi_t}}{\theta_t} \nonumber \\
    & = &   -i\Fy \Yt\Sqt\ket{F,+F_x}\\
    v_{\xi_t} & = & \partialD{\ket{\theta_t,\xi_t}}{\xi_t} \nonumber \\
    & = &
        \Yt\Sqt
        (-2i(\Fz\Fy + \Fy\Fz))\ket{F,+F_x} .
\end{eqnarray}

In calculating the normalization of these tangent vectors, we encounter terms such as
\begin{equation}
    \langle v_{\theta_t}, v_{\theta_t} \rangle
     = \bra{F,+F_x}\Sqt^{\dag}\Fy^2\Sqt\ket{F,+F_x}.
\end{equation}
More generally, almost all inner-products needed for the projection filter will be of the form
\begin{equation}
    \bra{F,+F_x} \Sqt^{\dag}g(\Fx,\Fy,\Fz)\Yt^{\dag} f(\Fx,\Fy,\Fz)
            \Yt\Sqt\ket{F,+F_x} . \nonumber
\end{equation}
Here, $g$ and $f$ are polynomial functions of their arguments.  Since $\Yt$ is a rotation, we can exactly evaluate
\begin{equation}
    \Yt^{\dag}f(\Fx,\Fy,\Fz) \Yt 
       = f( \Yt^{\dag}\Fx \Yt, 
           \Yt^{\dag}\Fy\Yt, \Yt^{\dag}\Fz\Yt), 
\end{equation}
where
\begin{align}
    \Yt^{\dag}\Fx\Yt &=  \Fx(\theta_t) = \Fx \cos{\theta_t} + \Fz\sin{\theta_t} \\
    \Yt^{\dag}\Fy\Yt &=  \Fy\\
    \Yt^{\dag}\Fz\Yt &=  \Fz(\theta_t) = \Fz \cos{\theta_t} - \Fx\sin{\theta_t} . 
\end{align}
This leaves us with expectations of the form
\begin{equation} \label{magnetometry:eq:hp_expectations}
    \bra{F,+F_x}
    \Sqt^{\dag}g(\Fx,\Fy,\Fz)
            f(\Fx(\theta_t),\Fy,\Fz(\theta_t))
            \Sqt \ket{F,+F_x}
\end{equation}
where $g\times f$ will just be linear combinations of powers and products of $\Fx,\Fy,\Fz$.  Unfortunately, we cannot evaluate this expectation for arbitrary $\xi_t$.  However, for small $\xi_t$, the state which we are taking expectations with respect to is the ``squeezed vacuum'' in our preferred basis, e.g. it is the state $\ket{F,+F_x}$ pointing in the same direction, but with squeezed uncertainty in $\Fz$ and increased uncertainty in $\Fy$.  

For large $F$, angular momentum expectations of such a state are extremely well described by the Holstein-Primakoff approximation to lowest order \cite{Holstein:1940a}
\begin{equation} \label{magnetometry:eq:HP_firstorder}
	\begin{split} 
		\Fpx &\approx \sqrt{2F}a\\
		\Fmx &\approx \sqrt{2F}a^{\dag}\\
		\Fx &\approx F ,
	\end{split}
\end{equation}
where $\hat{F}_{\pm,x} = \Fy \pm i\Fz$, and $a,a^{\dag}$ are bosonic creation and annihilation operators.  We then write our state as $\ket{F,+F_x} = \ket{0}$, which is the vacuum in the Holstein-Primakoff representation.  Under this approximation, we can use the relations
\begin{align}
    \Sqt^{\dag}\Fx\Sqt &= F\\
    \Sqt^{\dag}\Fy\Sqt &= \frac{\sqrt{2F}}{2}e^{4F\xi_t}(a + a^{\dag})\\
    \Sqt^{\dag}\Fz\Sqt &= -i\frac{\sqrt{2F}}{2}e^{-4F\xi_t}(a - a^{\dag})   
\end{align}
to evaluate the expectation in Eq.\ (\ref{magnetometry:eq:hp_expectations}).  In light of this approximation, the tangent vector overlaps are readily shown to be
\begin{align}
     \expect{v_{\theta_t},v_{\theta_t}} &= \frac{Fe^{8F\xi_t}}{2}\\
     \expect{v_{\xi_t},v_{\xi_t}} &= 8F^2\\
     \expect{v_{\xi_t},v_{\theta_t}} &= 0 ,
\end{align}
where the last result indicates the tangent vectors are orthogonal as desired.

\subsubsection{Orthogonal Projection of Double-pass Filter}
Before performing orthogonal projection of the dynamics onto the tangent space, we must first convert the filtering equation from It\^{o} to Stratonovich form.  As is discussed in Ref.\ \cite{vanHandel:2005a}, the It\^{o} chain rule is incompatible with the differential geometry picture of projecting onto the tangent space.  Fortunately, Stratonovich stochastic integrals follow the standard chain rule and are thus amenable to projection filtering methods.  Following the derivation in Appendix \ref{magnetometry:sec:app:ito_to_stratonovich}, we find that the Stratonovich SSE is given by
\begin{equation} \label{magnetometry:eq:double_pass_SSE_stratonovich}
	\begin{split} 
	    d\ket{\psi}_t &= \left[ -i\gamma B\Fy 
	        -M\left[ (\Fz - \expect{\Fz}_t)^2 - \expect{\Delta\Fz^2}_t\right]
		        -\frac{\sqrt{KM}}{2}\Fx \right.\\
		     & \left.
		        +2i\sqrt{KM}\expect{\Fz}_t\Fy
		        + i\sqrt{KM} \expect{\Fz\Fy}_t
		        \right]\ket{\psi}_t dt\\
		 &+ \left( \sqrt{M}(\Fz - \expect{\Fz}_t) + i\sqrt{K}\Fy\right)\ket{\psi}_t \circ dW_t ,
	\end{split}  
\end{equation}
where $\expect{\Delta\Fz^2}_t = \expect{\Fz^2} - \expect{\Fz}^2$.

In order to find the projection filter, we compare the general projection formula in Eq.\ (\ref{magnetometry:eq:general_projection}) to the general dynamical equation for states in our chosen family, given by
\begin{equation} \label{magnetometry:eq:projlhs}
    d\ket{\xi_t,\theta_t} = v_{\xi_t}d\xi_t + v_{\theta_t}d{\theta_t} .
\end{equation}
Using the orthogonality of the tangent vectors, the general forms for $d\xi_t$ and  $d\theta_t$ are
\begin{align}
    d\theta_t &= \frac{2e^{-8F\xi_t}}{F} \expect{v_{\theta_t}, d\ket{\psi_t}[\xi_t,\theta_t]}\\
    d\xi_t &= \frac{1}{8F^2}\expect{v_{\xi_t}, d\ket{\psi_t}[\xi_t,\theta_t]} ,
\end{align}
where $d\ket{\psi_t}[\xi_t,\theta_t]$ is the evolution of $\ket{\xi_t,\theta_t}$ under the Stratonovich filter of Eq.\ (\ref{magnetometry:eq:double_pass_SSE_stratonovich}).  

As an example calculation using these methods, consider projecting the dynamics generated by the magnetic field term.  Its contribution $c_\theta$ to the $\theta_t$ dynamics is given by
\begin{eqnarray}
   c_\theta & = &\frac{2e^{-8F\xi_t}}{F}\langle v_{\theta_t}, -i\gamma B \Fy \ket{\theta_t,\xi_t}dt\rangle
   	\nonumber \\
    & =& \frac{2\gamma Be^{-8F\xi_t}}{F}
        \bra{0}\Sqt^{\dag}\Yt^{\dag}\Fy^2\Yt\Sqt\ket{0}dt\\
    & =& \gamma B \bra{0}(a+a^{\dag})^2\ket{0} dt \nonumber \\
    & =& \gamma B dt \nonumber . 
\end{eqnarray}
Similarly, the contribution to $\xi_t$ is
\begin{eqnarray}
  c_\xi & = &\frac{ 1 }{8F^2}  \langle v_{\xi_t}, -i\gamma B \Fy 
        \ket{\theta_t,\xi_t}dt\rangle \nonumber \\
    & = & \frac{\gamma B}{4F^2}\bra{0} \Sqt{^\dag}(\Fz\Fy + \Fy\Fz)\Yt^{\dag}
    \Fy\Yt\Sqt\ket{0}dt \nonumber \\
   &  \propto& \bra{0} a^3 + a^2a^{\dag} - {a^{\dag}}^2a - {a^{\dag}}^3\ket{0}dt\\
    & =& 0  \nonumber.
\end{eqnarray}
Chugging through the remaining terms in a similar fashion, we arrive at the full projection filter equations
\begin{eqnarray}
    d\theta_t & = & \gamma B dt  
        + \frac{\sqrt{KM}}{2}e^{-8F\xi_t}\sin{\theta_t}dt
        + 2F\sqrt{KM}\sin{\theta_t} dt \nonumber\\
   &&     -\left[\sqrt{M}e^{-8F\xi_t}\cos{\theta_t}
         +\sqrt{K}\right] \circ dW_t
\end{eqnarray}
and 
\begin{equation}
    d\xi_t = \frac{M}{4}e^{-8F\xi_t}\cos^2{\theta_t}dt .
\end{equation}
Converting back to It\^{o} form using Eq.\ (\ref{probability:eq:ito_to_stratonovich}), we have
\begin{subequations}\label{magnetometry:eq:projected_filter_ito}
	\begin{eqnarray} 
		    d\theta_t &= &\left[B \gamma 
		               - \frac{M}{4}e^{-16F\xi_t}\sin(2\theta_t)
		               + 2F\sqrt{KM}\sin{\theta_t}\right]dt \nonumber \\
		              && -\left[\sqrt{M}e^{-8F\xi_t}\cos{\theta_t} + \sqrt{K}\right]dW_t   \\
		    d\xi_t &= & \frac{M}{4}e^{-8F\xi_t}\cos^2{\theta_t}dt, 
	\end{eqnarray}
\end{subequations}
where the innovations are now in terms of the approximation of $\expect{\Fz}_t$ within the Gaussian family:
\begin{eqnarray}
	dW_t & = & dZ_t - 2\sqrt{M}\expect{\Fz}_tdt \nonumber \\
	& = & dZ_t + 2F\sqrt{M}\sin{\theta_t}dt .
\end{eqnarray}
\subsubsection{Small-angle Kalman Filter}
We see that the projected filter in Eq.\ (\ref{magnetometry:eq:projected_filter_ito}) is actually more general than the filters usually derived for the magnetometry problem, which do not distinguish the Gaussian and small-angle approximations.  That is, the family of states in Eq.\ (\ref{magnetometry:eq:gaussian_state_parameterization}) and the approximations considered in the above derivation only enforce the Gaussian state assumption through the Holstein-Primakoff approximation.  We can separately apply the small-angle approximation to recover an equation appropriate for the Kalman filter.  In this limit, the equation for $\xi_t$ completely decouples and has a closed form solution
\begin{equation}
    \xi_t = \frac{1}{8F}\ln\left[1+2FMt\right] .
\end{equation}
Taking the small-angle approximation for $\theta_t$ and plugging in the explicit form of $\xi_t$ gives
\begin{multline}
	   d\theta_t = \left[B\gamma + \left(2F\sqrt{KM} 
                - \frac{M}{2(1+2FMt)^2}\right)\theta_t\right]dt\\
                - \left[\frac{\sqrt{M}}{1+2FMt} + \sqrt{K}\right]dW_t, 
\end{multline}
which is linear in the remaining state parameter $\theta_t$.

While we could consider the Kalman filter for the quantum state alone, we can just as easily account for our uncertainty in $B$ at the same time.  That is, if we now embed $B$ as a state variable, setting $X_t = [ \theta_t\ B]^T$, the dynamics can be written in a linear form as
\begin{eqnarray}
    dX_t &=& A X_t dt + B dW_t \label{magnetometry:eq:linear_system}\\
    dZ_t &=& C X_t dt  + D dW_t \label{magnetometry:eq:linear_observations}\\
    A &=& \begin{pmatrix}
         2F\sqrt{KM} - \frac{M}{2(1+2FMt)^2} & \gamma\\
         0 & 0
    \end{pmatrix}\\
    B &=& \begin{pmatrix}
           -\frac{\sqrt{M}}{1+2FMt} - \sqrt{K}\\0
    \end{pmatrix}\\
    C &=& \begin{pmatrix}
        -2\sqrt{M}F & 0
    \end{pmatrix}\\
    D &=& 1.
\end{eqnarray}
Equations (\ref{magnetometry:eq:linear_system}) and (\ref{magnetometry:eq:linear_observations}) are precisely a classical linear system/observation pair, in which the same white noise process (the innovations) drives \emph{both} the system and observation processes.  The estimate $\tilde{X}_t = \mathbbm{E}[X_t | Z_{(0,t)}]$ admits a Kalman filter solution \cite{Liptser:1977}, given by
\begin{eqnarray}
    d\tilde{X}_t &= &A\tilde{X}_t dt + (B + VC^{\dag})d\tilde{W}_t\\
    \dot{V}    &= & AV +VA^{\dag} + BB^{\dag} - (B+VC^{\dag})(B+VC^{\dag})^{\dag} \nonumber
\end{eqnarray}
where $V$ is the covariance matrix
\begin{align}
    V &= \mathbbm{E}[(\tilde{X} - \mathbbm{E}[\tilde{X}])(\tilde{X} - \mathbbm{E}[\tilde{X}])^T]\\
      &= \begin{pmatrix}
        \Delta \tilde{\theta}_t^2 & \Delta\tilde{\theta}_t\tilde{B}_{kf}\\
        \Delta\tilde{\theta}_t\tilde{B}_{kf} & \Delta \tilde{B}_{kf}^2
      \end{pmatrix}
\end{align}
and
\begin{equation}
	d\tilde{W}_t = dZ_t + 2F\sqrt{M}\tilde{\theta}_t dt
\end{equation}
is the innovations constructed from the current $\theta_t$ estimate in the small-angle approximation.

Looking at the explicit system of equations for the variances, which unfortunately do not admit a straightforward analytic Riccati solution as discussed in Appendix \ref{appendix:riccati}, we have
\begin{eqnarray} \label{magnetometry:eq:kalmanvar}
	   \frac{d(\Delta \tilde{\theta}_t^2)}{dt} &= & - M\Delta\tilde{\theta}_t^2
	     \left(\frac{1+4 F+8 F^2 M t}{(1+2 F M t)^2} 
	        + 4 F^2 \Delta\tilde{\theta}_t^2\right)\nonumber\\
	     & &+2 \gamma  \Delta\tilde{\theta}_t\tilde{B}_{kf}\\
	    \frac{d(\Delta \tilde{B}_{kf}^2)}{dt}&= &
	            -4F^2M (\Delta\tilde{\theta}_t\tilde{B}_{kf})^2\\
	        \frac{d(\Delta\tilde{\theta}_t\tilde{B}_{kf})}{dt} &= &
	           \gamma\Delta \tilde{B}_{kf}^2 -\frac{M}{2 (1+2 F M t)^2} 
	            \nonumber\\
	           & &  \left(1+4 F+8 F^2 M t+ \right. \\
	           && \left.  8 F^2 (1+2 F M t)^2
	               \Delta \tilde{\theta}_t^2\right)
	                \Delta\tilde{\theta}_t\tilde{B}_{kf} \nonumber
\end{eqnarray}
which are completely \emph{independent} of the second-pass coupling strength $K$.  That is, within the small-angle and Gaussian approximations, the double-pass system has no improvement in sensitivity and gives rise to the same $F^{-1}$ uncertainty scaling found previously for single-pass systems \cite{Geremia:2003a}. Perhaps this is unsurprising, as we attempted to find a linear description of an essentially non-linear affect.  Indeed, the numeric simulations in the next section suggest the single-mode Gaussian approximation breaks down just as the double-pass filter begins to show improved sensitivity to the magnetic field parameter.  Finally, not that I have also derived a filtering equation which retains the next term in the Holstein-Primakoff expansion, but whose $K$ dependence nonetheless shows a negligible change relative to the lowest order expansion.

\section{Simulations} 
\label{magnetometry:sec:simulations}
Given the absence of an analytic improvement in the sensitivity of the quantum Kalman filter, we turn to numerical simulations of the quantum particle filter in order to gauge the potential of the double-pass system for magnetometry.  First recall how the filter would be used in an actual experiment.  Continuous measurements of the atomic cloud Larmor precessing under a particular, albeit unknown, magnetic field $B$ would give rise to the observations process $Z_{(0,t)}$.  This would then be fed into a classical computer to propagate the quantum particle filtering equations given in (\ref{magnetometry:eq:ensemble_filtering_eqs}).  The computer would then use the quantum particle set to provide the estimate $\tilde{B}_{pf}$ and uncertainty $\Delta \tilde{B}_{pf}^2$.  

In order to simulate such an experiment, we can generate the stochastic measurement record $Z_{(0,t)}$ using the quantum filter for the double-pass system given in Eq.\ (\ref{magnetometry:eq:double_pass_SSE}), evolved with a known magnetic field $B$.  Since the system is driven by the white noise process $dW_t$, the filtering equations may be integrated by the same integrator previously used to approximate the quantum Cram\'{e}r-Rao bound.  The measurements generated by these trajectories are equivalent to what the quantum particle filter would receive in an experiment, which means they can then be fed into the same particle filtering code to simulate an estimate of $B$.  In order to compare performance, we actually simulate two systems in parallel, one representing the double-pass system and the other, with $K=0$, representing a single-pass system.  Both utilize the same noise realizations on an individual trajectory.  

As is common when considering detection limits, we focus on the case of $B = 0$.  Although an unbiased estimator would assume no prior knowledge of the magnetic field value, such an approach is impractical for the particle filter, which would fail in approximating such large uncertainty with a finite number of particles.  As such, we take the initial distribution of $B$ values for the quantum particle set to be Gaussian
\begin{equation}
    p_0(B;\mu_B,\sigma_B) = \frac{1}{\sqrt{2\pi\sigma_B^2}}
    \exp(-\frac{(B-\mu_B)^2}{2\sigma_B^2})
\end{equation}
with mean $\mu_B = 0$ and variance $\sigma_B^2 = 10{\tau^{-1}}^2$, where we again set $\gamma = 1$ and again define all parameters in units of ${\tau^{-1}}$.  For a set of $N$ particles, the particle magnetic field values $\{B_i\}$ are drawn from the initial distribution, with weights $p_0^{(i)} = 1/N$.  The initial quantum state for all particles is set to the spin-coherent state along $+x$,i.e. $\ket{F,+F_x}$.  
\begin{figure}[tb]
    \centering
        \includegraphics[scale=0.8]{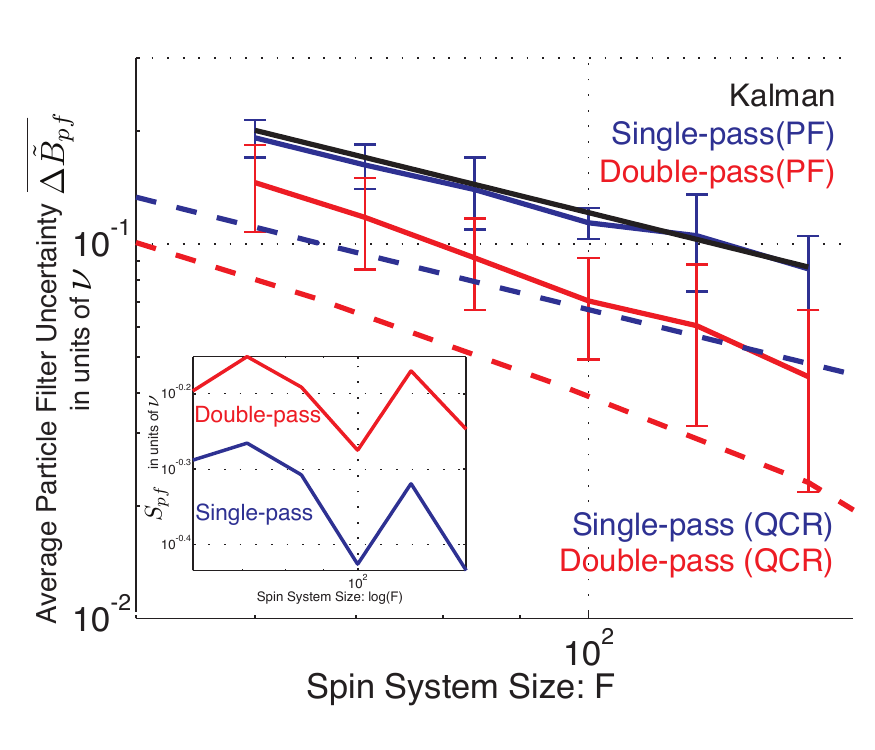}
    \caption[Particle filter performance for double-passed magnetometer]{ Estimator uncertainties as a function of $F$ averaged over 100 trajectories with $M = 10{\tau^{-1}}$, $K=0.0006{\tau^{-1}}$, $B=0$ and $\tau = 0.1$.  The initial $N = 1000$ particle set was drawn from a Gaussian distribution with mean zero and variance $10{\tau^{-1}}^2$, which was also the same initial uncertainty in the Kalman filter $\Delta \tilde{B}_{kf}$.  A power-law fit to the particle filter (PF) scalings shows a single-pass scaling of $F^{-0.93}$ and a double-pass scaling of $F^{-1.39}$.  Also shown are the quantum Cram\'{e}r-Rao (QCR) bounds previously simulated for Figure \ref{magnetometry:fig:cramer_rao}. The inset shows the sample estimator deviation $S_{pf}$ for the same simulations.}
    \label{magnetometry:fig:particleFilterSimulations}
\end{figure}

Figure \ref{magnetometry:fig:particleFilterSimulations}(a) shows, with solid lines, the average particle filter uncertainty $\overline{\Delta \tilde{B}_{pf}}$ as a function of $F$, averaged over 100 measurement realizations using $N = 1000$ particles in each run of the filter. The error bars represent the the deviation in the simulated uncertainties over the 100 runs. As was the case for the Fisher information calculations, we observe an improved sensitivity scaling for the double-pass system, albeit with increased fluctuations in the individual run uncertainty $\Delta \tilde{B}_{pf}$.  Power-law least-squares fits of the average give a single-pass uncertainty scaling $F^{-0.93}$ and a double-pass scaling of $F^{-1.39}$ which are consistent with the quantum Cram\'{e}r-Rao scalings in figure \ref{magnetometry:fig:cramer_rao}.  Also shown is the analytic single-pass uncertainty scaling given by numerical integration of the Kalman covariance matrix via Eq.\ (\ref{magnetometry:eq:kalmanvar}).  We see that this agrees very well with the single-pass particle filter scaling and since it is consistent with previous Kalman filters used for magnetometry \cite{Geremia:2003a}, suggests the double-pass scaling does indicate improved sensitivity.  

Of course, these statements are not without caveats.  The dashed lines in the plot correspond to the numerically computed quantum Cram\'{e}r-Rao bound, which is clearly below the estimates of all the filters.  This might mean that the continuous-measurement which gives rise to the numerical bound is simply not saturated by the corresponding estimator for that continuous-measurement.  Unfortunately, the above data took a week to generate on a quad-core workstation, indicating the technical challenges already present in simulating an $N=1000$ quantum particle set for the depicted range of $F$ limits the quality of the statistics.  As previously mentioned, the particle filter approximation is inherently biased, with the variance of estimates converging as $N^{-1}$.  
The inset in figure \ref{magnetometry:fig:particleFilterSimulations} shows the sample estimator deviation $S_{pf}$, which is the deviation in the actual performance error of the particle filter on each individual run, i.e. $\tilde{B}_{pf} -B $ where the true $B = 0$.  In other words, $\Delta \tilde{B}_{pf}$ is the uncertainty calculated for an individual trajectory from the particle distribution $\{p_t^{(i)}\}$, which is averaged over many trajectories to get $\overline{\Delta \tilde{B}_{pf}}$.  However, an individual run of the particle filter also gives an estimate $\tilde{B}_{pf}$ of the true magnetic field $B$.  Since we know that the measurements were generated from a system evolved with $B = 0$, we can calculate the deviation in the actual estimates $\tilde{B}_{pf}$.  If the particle filter were unbiased, we would expect this sample deviation to equal the average particle filter deviation, i.e. $S_{pf} = \overline{\Delta \tilde{B}_{pf}}$.  Instead, the sample deviation dwarfs the average estimator uncertainty, indicating that the particle filter bias dominates.  As discussed in \citep{Chase:2009d}, this bias seems to be due to the prior distribution considered for $B$.  Ideally, we would want this distribution to have infinite variance in order to be truly unbiased, but that is not practical for the particle filter simulations.  Instead, future work will need to consider alternate strategies for eliminating this bias in practice.


\begin{figure}[tb]
\begin{center}
        \includegraphics[scale=0.4]{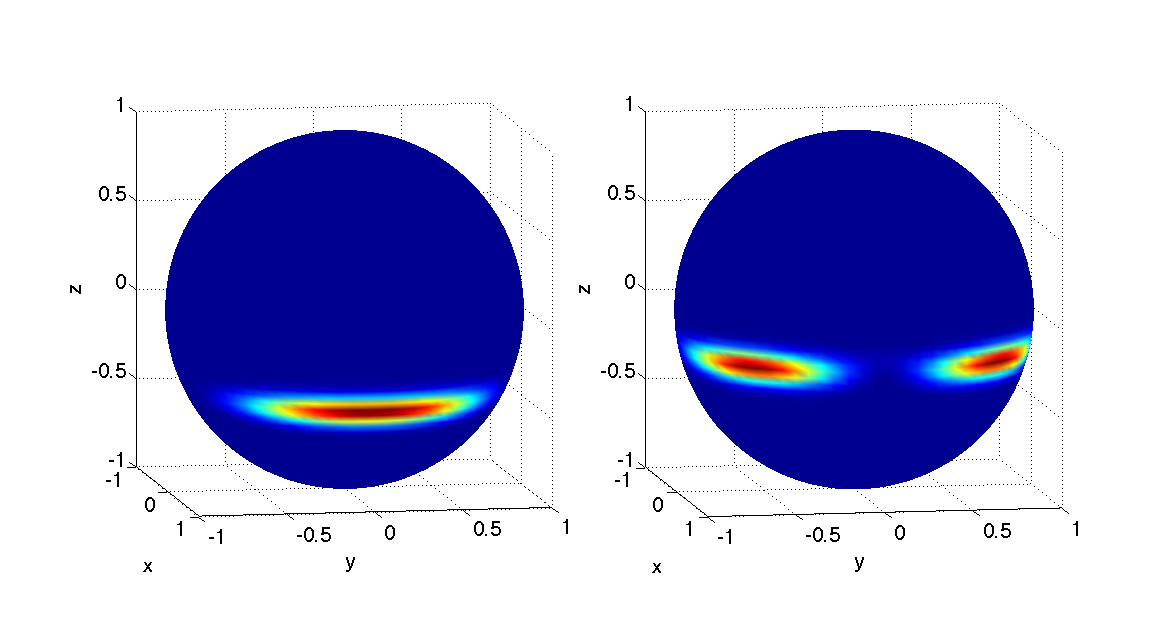}
\end{center}
\vspace{-10mm}
    \caption[Quasi-Probability distribution showing bimodal Gaussian distribution]{ Quasi-Probability distributions $Q(\theta,\phi,t)$ for two different trajectories at time $\tau = 0.1{\tau^{-1}}^{-1}$ for $M = 10{\tau^{-1}}, K = 0.0006{\tau^{-1}}, B = 0$ and $F=140$.}
    \label{magnetometry:fig:splitting_example}
\end{figure}

Numerical simulation also provides insight into how the Gaussian state assumption of the Kalman filter applies in the double-pass case.  Figure \ref{magnetometry:fig:splitting_example} shows quantum states evolved under two different noise realizations with $B=0,M=10{\tau^{-1}},K=0.0006{\tau^{-1}}$.  Both states were initially spin-polarized along $+x$ and evolved under the full double-pass SSE in Eq.\ (\ref{magnetometry:eq:double_pass_SSE}).  The Q-function shown is defined as
\begin{equation}
    Q(\theta,\phi,t) = \modsq{\braket{\theta,\phi}{\psi_t}}
\end{equation}
where the spin-coherent state $\ket{\theta,\phi}$ is the $+F$ eigenstate of the spin-operator 
\begin{equation}
	\Fx \sin\theta \cos\phi + \Fy \sin\theta \sin\phi + \Fz \cos\theta.
\end{equation}
Although one example shows a Gaussian squeezed spin state, the other shows a state with a bimodal Gaussian distribution.  Such a state is poorly described by the Gaussian family in Eq.\ (\ref{magnetometry:eq:gaussian_state_parameterization}) and helps explain why the Kalman filter fails to find a difference between the single and double-pass setup.  These plots suggests a family of bimodal Gaussian states might result in a useful projection filter.  I have been unable to find a parameterization of such a family which admits an analytic derivation of a projection filter.

\section{Summary} 
\label{magnetometry:sec:conclusion}
In this chapter, we have explored the use of double-pass continuous measurement for precision magnetometry.  The primary result involves numerical simulations of the quantum Cram\'{e}r-Rao bound which indicate that a double-pass system shows an improved magnetic field uncertainty scaling with atom number over a comparable single-pass system, albeit only for particular choices of coupling strengths relative to the collective spin size.  This is in contrast to quantum information theoretic bounds which suggest that the Heisenberg limit bounds the uncertainty scaling for both a single and double-pass system.  Clearly, future work aimed at reconciling these results is necessary, particularly deriving analytic quantum Cram\'{e}r-Rao bounds for unbounded ancilla systems.  We have also explored estimators intended to achieve the uncertainty scaling seen in numerical simulations.  Taking a brute force approach, quantum particle filters show evidence of the improved double-pass scaling, although the results suffer from limited statistics which can not be significantly improved with current computational power.  More practical quantum Kalman filters show no improved sensitivity, which are consistent with an observed breakdown in the Gaussian state assumption used to derive them.  However, the general projecting filtering technique used in the Kalman filter derivation provides an avenue for deriving more appropriate filters which might prove more tractable for practical magnetic field estimation.  More generally, similar effective nonlinear interactions may prove an important tool in precision measurement.

\begin{subappendices}
\section{Converting between It\^{o} and Stratonovich SDE} 
\label{magnetometry:sec:app:ito_to_stratonovich}
For the double-pass It\^{o} SSE in Eq.\ (\ref{magnetometry:eq:double_pass_SSE}), we begin the conversion by noting that states with entirely real amplitudes form an invariant set and therefore write $\ket{\psi}_t = \sum_{m=-F}^F x_t^{m} \ket{m}$.  The stochastic coefficient is then 
\begin{equation}
	\begin{split}
	    b(t,x_t) &= \sqrt{M}\sum_{m = -F}^F m {x_t}^m\ket{m}
	                      - \sqrt{M}\sum_{m,n = -F}^F  n({x_t}^n)^2 {x_t}^m\ket{m}\\
	                      &+ \frac{1}{2}\sqrt{K}\sum_{m = -F}^F
	                        \left[ \sqrt{(F-m)(F+m+1)}{x_t}^m\ket{m+1} \right.\\
	                       & \left. - \sqrt{(F+m)(F-m+1)}{x_t}^m\ket{m-1}\right]
	\end{split}
\end{equation}
which has as its $j$-th entry
\begin{equation}
    \begin{split}
	    b^j(t,x_t) &= \sqrt{M} (j  - \sum_{n=-F}^{F}
	                                n({x_t}^n)^2){x_t}^j \\
	                &+\frac{\sqrt{K}}{2}
	                \left[\sqrt{(F-j+1)(F+j)}{x_t}^{j-1} \right.\\
	                & \left.
	                    - \sqrt{(F+j+1)(F-j)}{x_t}^{j+1}\right]
    \end{split}
\end{equation}
The derivative with respect to ${x_t}^k$ is then
\begin{equation}
	\begin{split}
	    \partialD{b^j(t,x_t)}{x_k} &= \sqrt{M} (j  - \sum_{n=-F}^{F}
	                                n({x_t}^n)^2)\delta_{jk}
	                                -\sqrt{M}2k{x_t}^k{x_t}^j \nonumber\\
	             &+\frac{\sqrt{K}}{2}
	                \left[\sqrt{(F-j+1)(F+j)}\delta_{(j-1),k} \right.\\
	             &\left. - \sqrt{(F+j+1)(F-j)}\delta_{(j+1),k}\right]
	\end{split}
\end{equation}
so that the sum in Eq.\ (\ref{probability:eq:ito_to_stratonovich}) is

\begin{multline}
	    \sum_{k = -F}^Fb^k(t,x_t) \partialD{b^j(t,x_t)}{{x_t}^k} = 
	             \sqrt{M} (j  - \sum_{n=-F}^{F}n({x_t}^n)^2)b^j(t,x_t)
	             - 2\sqrt{M}\sum_{k} k{x_t}^kb^k(t,x_t) {x_t}^j \\
	             +\frac{\sqrt{K}}{2}
	                 \left[\sqrt{(F-j+1)(F+j)}b^{j-1}(t,x_t)
	                 - \sqrt{(F+j+1)(F-j)}b^{j+1}(t,x_t)\right] 
\end{multline}
This suggests an equivalent operator form
\begin{multline}    
	\left[\sqrt{M}(\Fz - \expect{\Fz}_t) + i\sqrt{K}\Fy\right]^2 
	- 2\sqrt{M}\expect{\Fz(\sqrt{M}(\Fz - \expect{\Fz}) + i\sqrt{K}\Fy)}_t
     = \\ \left[\sqrt{M}(\Fz - \expect{\Fz}_t)+ i\sqrt{K}\Fy\right]^2 
     - 2M\expect{\Delta\Fz^2}_t -2i\sqrt{KM}\expect{\Fz\Fy}_t,
\end{multline}
where $\expect{\Delta\Fz^2}_t = \expect{\Fz^2} - \expect{\Fz}^2$, so the Stratonovich SSE is
\begin{multline}
    d\ket{\psi}_t = 
    \left[ -i\gamma B\Fy -M\left[ (\Fz - \expect{\Fz}_t)^2 - \expect{\Delta\Fz^2}_t\right]  
    		-\frac{\sqrt{KM}}{2}\Fx \right.\\
    	   \shoveright{\left.
				        +2i\sqrt{KM}\expect{\Fz}_t\Fy
				        + i\sqrt{KM} \expect{\Fz\Fy}_t
				        \right]\ket{\psi}_t dt}\\
		 + \left( \sqrt{M}(\Fz - \expect{\Fz})_t + i\sqrt{K}\Fy\right)\ket{\psi}_t \circ dW_t . \label{magnetometry:app:eq:SSE_stratonovich} 
\end{multline}
\end{subappendices}


\chapter{Feedback controllers for quantum error correction}
\label{chapter:error_correction}
In this chapter, I review quantum feedback protocols for performing continuous-time quantum error correction.  After studying the structure of the quantum filter, I describe a low-dimensional representation which although inexact, gives rise to the same feedback performance of the exact quantum filter.  The work presented here is published in \citep{Chase:2008a} and I refer the reader to \citep{Nielsen:2000a,Gottesman:1997a} for a thorough introduction to quantum error correction.
\section{Introduction}

Quantum error correction is inherently a feedback process where the error
syndrome of encoded qubits is measured and used to apply conditional
recovery operations \citep{Gottesman:1997a}.  Most
formulations of quantum error correction treat this feedback process as a
sequence of discrete steps.  Syndrome measurements and recovery operations
are performed periodically, separated by a time-interval chosen small enough
to avoid excessive accumulation of errors but still comparable to the time
required to implement quantum logic gates
\citep{Gottesman:1997a,Nielsen:2000a}.  There is, however, mounting evidence
from the field of real-time quantum feedback control
\citep{Wiseman:1994a,Armen:2002a,Bouten:2006a} that continuous
observation processes offer new, sometimes technologically advantageous,
opportunities for quantum information processing.

Toward this end, Ahn, Doherty and Landahl (ADL) \citep{Ahn:2002a} devised a
scheme to implement general stabilizer quantum error correction
\citep{Gottesman:1997a} using continuous measurement and feedback.
Unfortunately an exact implementation of the ADL scheme is computationally
demanding.    For an $n$-qubit code, the procedure requires one to time-evolve 
a $2^n$-dimensional density matrix for the logical qubit alongside the
quantum computation \citep{Ahn:2002a}.    This classical
information-processing overhead must be performed to interpret the continuous-time
error syndrome measurement data and determine how recovery operations, in
the form of a time-dependent feedback Hamiltonian, should be applied.  
While $n$ is a constant for any particular choice of code, even modest codes such as the
five-qubit code \citep{Bennett:1996a,Laflamme:1996a} and the seven-qubit Steane code
\citep{Steane:1996a} push classical computers to their limits.  Despite
state-of-the art experimental capabilities, it would be extremely
difficult to implement the ADL bit-flip code in practice.  Consequently, Ahn and others have devised alternate feedback protocols which are less demanding \citep{Sarovar:2004a, Ahn:2004a}, but perform worse than the the original ADL scheme.

Recently, van Handel and Mabuchi addressed the computational overhead of
continuous-time error syndrome detection \citep{vanHandel:2005c} using
techniques from quantum filtering theory presented in Chapter \ref{chapter:quantum}.  They developed an exact,
low-dimensional model for continuous-time error syndrome measurements, but did not
go on to treat continuous-time recovery.  The complication
is that any feedback Hamiltonian suitable for correcting errors during the
syndrome measurements violates the dynamical symmetries that
were exploited to obtain the low-dimensional filter in Ref.
\citep{vanHandel:2005c}.  While one might address this complication by simply
postponing error recovery operations until a point where the 
measurements can be stopped, there may be scenarios where it would be preferable to perform
error recovery in real-time.  For example, if the recovery operation is not instantaneous, responding to errors as they occur might outperform
protocols where there are periods without any error correction.  

In this chapter, I extend the quantum filtering approach developed by van
Handel and Mabuchi to include recovery operations.  I further consider an
error-correcting feedback Hamiltonian of the form devised by Ahn, Doherty
and Landahl, but the approach readily extends to other forms for the
feedback.  While an exact low-dimensional model for continuous-time
stabilizer generator measurements in the presence of feedback does not
appear to exist, I present an approximate filter that is still
low-dimensional, yet sufficiently accurate such that high-quality error
correction is possible.  

\section{Continuous-Time Quantum Error Correction}

For our purposes, a quantum error correcting code is a triple $(E,
\mathcal{G}, R)$.  The quantum operation $E:\mathbb{C}^{2k} \mapsto
\mathbb{C}^{2n}$ \emph{encodes} $k$ logical qubits in $n$ physical qubits.
$\mathcal{G}$ is a set of $l=n-k$ stabilizer generator observables with
outcomes $\pm 1$ that define the \emph{error syndrome}.   $R:\{\pm
1\}^{\otimes l}\mapsto \mathbb{C}^{2n\times 2n}$ is the
\textit{recovery operation}, which specifies what correction should be
applied to the physical qubits in response to the syndrome measurement
outcomes. 

The particular choice of code $(E, \mathcal{G},R)$ is usually made with
consideration for the nature of the decoherence affecting the physical
qubits \citep{Knill:2000a}.  For example, the bit-flip code (considered by
both ADL and van Handel and Mabuchi) improves protection against an error
channel that applies the Pauli $\sigma_x$ operator to single qubits at a 
rate $\gamma$.  Here, we adopt the notation that  $X_n$ represents the Pauli
$\sigma_x$ operator on qubit $n$, and similarly for $Y_n$ and
$Z_n$.  In the bit-flip code,  $E$ encodes $k=1$ qubits in $n=3$ qubits by
the map $\alpha\ket{0}+\beta\ket{1} \mapsto \alpha\ket{000} +
\beta\ket{111}$.  The $l=2$ stabilizer generators are $g_1 = ZZI :=
\sigma_z\otimes\sigma_z\otimes I$ and $g_2 = IZZ :=  I \otimes \sigma_z
\otimes \sigma_z$; each extracts the parity of different qubit pairs.  The
recovery $R$, given the outcomes of measuring $(g_1,g_2)$, is defined by
$(+1,+1)\mapsto I $, $(+1,-1)\mapsto X_3$, $(-1,+1)\mapsto X_1$ and
$(-1,-1)\mapsto X_2$.  

In this chapter, we focus primarily on the five-qubit-code ($n = 5, k = 1$)
that increases protection against general separable channels, and in
particular the continuous-time symmetric depolarizing channel that applies
all three Pauli operators to each of the physical qubits at the same rate
$\gamma$.  The five-qubit code has $l=4$ stabilizer generators
$\{XZZXI,IXZZX,XIXZZ,ZXIXZ\}$.  It is also a \emph{perfect} code in that all
16 distinct syndrome outcomes indicate distinct errors: one corresponding to
the no-error condition, and one syndrome for each of the three Pauli errors
on each of the five qubits.  I defer to \citep{Nielsen:2000a,
Gottesman:1997a} for the encoding and recovery procedures for this code.

\subsection{Stabilizer Generator Measurements}

Quantum error correction can be extended to continuous time by replacing
discrete measurements of the stabilizer generators $g_1, \ldots, g_l$ with a
set of $l$ continuous observation processes $dQ_t^{(i)}$ \citep{Ahn:2002a}.  We do not consider here
how one might implement the set of $l$ simultaneous stabilizer generator
observations other than to comment that doing so in an AMO technology would
likely involve coupling the $n$ physical qubits to a set of electromagnetic
field modes and then performing continuous photodetection on the scattered
fields.  While this model is rather general, we take the same measurement strength
$\kappa$ for each qubit, implying symmetric coupling of the qubits.

Following the techniques in Chapter \ref{chapter:quantum}, one arrives at the
following form of the quantum filter for the conditional density matrix $\rho_t$.
\begin{eqnarray}   \label{error_correction:eq:BelavkinFilter}
	d\rho_t & = & 
		\gamma  \sum_{m=1}^n \sum_{j}\mathcal{D}[\sigma_{j}^{(m)}]\rho_t dt 
		+ \kappa \sum_{i=1}^{l}\mathcal{D}[g_i]\rho_t dt \nonumber \\
	& & + \sqrt{\kappa}\sum_{i=1}^l \mathcal{H}[g_i]\rho_t \left( dQ_t^{(i)}
		 - 2 \sqrt{\kappa}\, \mathrm{Tr}[ g_i \rho_t ] dt \right) \nonumber \\
	& & 
		- i [ H_t, \rho_t ] dt\, ,
\end{eqnarray}    
where $j\in \{ x, y, z\}$ and the superoperators are defined as:
$\mathcal{D}[\sigma]\rho = \sigma\rho \sigma-\rho$ and $\mathcal{H}[g_l]\rho
= g_l\rho  + \rho g_l - 2\Tr{g_l\rho }\rho $.  The innovations
    \begin{equation}
        dW_t^{(i)} =  dQ_t^{(i)} - 2\sqrt{\kappa} \Tr{g_i \rho_t}dt
    \end{equation} obtained from the measurements $dQ_t^{(i)}$
     are independent Wiener processes, each with $\mathbbm{E}[dW_t] = 0$ and $dW_t^2 = dt$.
The first term in the
filtering equation accounts for the action of the continuous-time symmetric
depolarizing channel.  The time evolution $\rho_t$ generated by a particular noise realization is generally called a \emph{trajectory}.  

The final term in Eq.\ (\ref{error_correction:eq:BelavkinFilter}) describes the action of the time-dependent feedback Hamiltonian used to implement error recovery.   Following Ahn, Doherty and Landahl, we choose the feedback Hamiltonian to be of the form
\begin{equation} \label{error_correction:eq:Hamiltonian}
	H_t = \sum_{m=1}^{n} \sum_j \lambda_{j,t}^{(m)} \sigma_{j}^{(m)},
\end{equation}
which corresponds to applying Pauli operators $\sigma^{(m)}_j$ to each qubit
with a controllable strength $\lambda_{j,t}^{(m)}$.   The policy for determining the feedback strengths $\lambda_{j,t}^{(m)}$ at each point in time should be chosen optimally.  Ahn, Doherty, and Landahl obtained their feedback policy by defining the \emph{codespace projector} $\Pi_0$ onto the no error states (states which are $+1$ eigenvectors of all stabilizers) and then maximizing the \emph{codespace fidelity} $\Tr{\Pi_0\rho_t}$.   Assuming a maximum feedback strength  $\lambda_{\text{max}}$, the resulting feedback policy is given by setting
\begin{equation} \label{error_correction:eq:FeedbackPolicy}
	\lambda_{j,t}^{(m)} = \lambda_{\text{max}} \sgn\bigl(\Tr{-i[\Pi_0,\sigma_{j}^{(m)}] \rho_t}\bigr) \, .
\end{equation}

\subsubsection{Computational Expense}

Because this is a closed-loop strategy, the feedback controller must determine each $\lambda_{j,t}^{(m)}$ from the evolving measurement in real time.    The utility of feedback in any real setting then relies greatly upon the controller's ability to integrate the filtering equation rapidly enough to maintain pace with the quantum dynamics of the qubits.  For the five-qubit code, $1024-1$ real parameters are needed to represent the density matrix.   We found that stable numerical integration via the techniques in Appendix~\ref{appendix:numerical_methods_for_stochastic_differential_equations} for even a single trajectory required approximately 36 seconds on a 2.1 GHz desktop computer ($\gamma dt \approx 10^{-5}$ over a timespan $[0,0.25\gamma]$).   This is far from adequate for use in an actual feedback controller even in state-of-the-art experiments.

Moreover, Eq.\ (\ref{error_correction:eq:BelavkinFilter}) is a nonlinear filter, and for such
filters it is rarely possible to evaluate even qualitative properties
analytically.  One must then average over an appreciable number of
trajectories to find the expected behavior of quantities such as the
codespace fidelity as a function of time.  For the five-qubit code, our
integrator requires approximately 10 hours to simulate 1000 trajectories.

\subsection{Reduced-Dimensional Filters}

Considering that the syndrome measurements yield information about
correlations between qubits and not information about the individual states
of the qubits, one can imagine that  propagating the full density matrix is
excessive.  Indeed, the ADL scheme only makes use of the projection of
$\rho_t$ onto the codespace, generating the same feedback policy
regardless of which state $\rho_0$ in the codespace is initially chosen.  It
is reasonable to expect that a lower dimensional model could track solely
the information extracted from the syndrome measurements.  This is exactly
the premise used by van Handel and Mabuchi to obtain a low-dimensional model
of continuous-time stabilizer generator measurements (in the absence of
feedback) \citep{vanHandel:2005c}.   They formulate the problem as a graph
whose vertices correspond to syndromes and whose edges reflect the action of
the error model.  The filtering problem is then reduced to tracking the node
probabilities, i.e., the likelihoods for the qubit to be described by each
of the various syndrome conditions.  Dynamical transitions occur between the
syndromes due to the error channel, and the filter works to discern these
transitions from the stabilizer measurement data.

For an $(E, \mathcal{G}, R)$ code, van Handel and Mabuchi define a set of projectors onto the distinct syndrome spaces.  For the five-qubit code, there are 16 such projectors; $\Pi_0$ is the codespace projector as before and $\Pi_{j}^{(m)}=\sigma_j^{(m)}\Pi_0\sigma_j^{(m)}$ are projectors onto states with a syndrome consistent with a $\sigma_j$ error on qubit $m$.  Forming the probabilities 
\begin{equation}
	p_{j,t}^{(m)}=\Tr{\Pi_{j}^{(m)} \rho_t}
\end{equation}
into a vector $\mathbf{p}_t$ and computing $dp_{j,t}^{(m)}$ from the full dynamics leads to the reduced filter
\begin{equation} \label{error_correction:eq:BasicWonham}
	d\mathbf{p}_t = \Lambda\mathbf{p}_t \, dt 
		+ 2\sqrt{\kappa}\sum_{k=1}^{l}(H_l - \mathbf{h_l}^T\mathbf{p}_t \, I) \mathbf{p}_t \, 
			dW_t
\end{equation}
with $\Lambda_{rs} = \gamma(1-16\delta_{rs})$, $h_l^{j,m}$ the outcome of
measuring $g_l$ on $\Pi_{j}^{(m)}$ and $H_l =
\operatorname{diag}\mathbf{h}_l$ (Eq.\ (4) in Ref.\ \citep{vanHandel:2005c}).  The
equations for $p_{j,t}^{(m)}$ are closed and encapsulate all the information
that is gathered from measuring the stabilizer generators.  Equation
(\ref{error_correction:eq:BasicWonham}) is an example of a \emph{Wonham filter}, which is the
classical optimal filter for a continuous-time finite-state Markov chain with an observation process driven by white noise \citep{Wonham:1965a}.    Further discussion of the Wonham filter and its use
in conjunction with discrete-time error correction can be found e.g., in Ref. \citep{vanHandel:2005c}.

\section{Error Correction with Feedback}

We now extend Eq.\ (\ref{error_correction:eq:BasicWonham}) to include a feedback Hamiltonian suitable for error recovery.  Following van Handel and Mabuchi's lead, we see that Eq.\ (\ref{error_correction:eq:BasicWonham}) was derived by taking $dp_{j,t}^{(m)} = \Tr{\Pi_{j}^{(m)} d\rho_t}$ for a basis which closed under the dynamics of the continuous syndrome measurement.  One hope is that simply adding the feedback term in by calculating $\Tr{-i\lambda_{k,t}^{(r)}\Pi_{j}^{(m)}[\sigma_k^{(r)},\rho_t]}$ also results in a set of closed equations.  However, that is not the case when using the basis of the sixteen syndrome space projectors $\Pi_{j}^{(m)}$.  Specifically,  $[\Pi_{j}^{(m)},\sigma_k^{(r)}]$ cannot be written as a linear combination of syndrome space projectors.  This is not surprising as the feedback Hamiltonian term under consideration is the only term which generates unitary dynamics.

Inspired by the form of the commutator between the feedback and the syndrome space projectors, we define feedback coefficient operators
\begin{equation} \label{error_correction:eq:FeedbackCoefficient}
\Pi_{j,c}^{(m)} = (+i\text{ or } +1)\sigma^{\otimes 5}\Pi_{j}^{(m)}\sigma^{\otimes 5}\, ,
\end{equation}
where $c$ is an arbitrarily chosen index used to distinguish the $i$ or 1 prefactor and combination of Pauli matrices which sandwich the syndrome space projector $\Pi_{j}^{(m)}$.   For the five-qubit code, the syndrome projectors are simply those operators which have the 1 prefactor and 10 identity matrices.  The corresponding feedback coefficient is $p_{j,c}^{(m)}=\Tr{\Pi_{j,c}^{(m)}\rho_t}$.  If we then iterate the dynamics of the filter (\ref{error_correction:eq:BelavkinFilter}) by calculating $p_{j,c}^{(m)}$ starting from the syndrome space projectors, we find that each feedback Hamiltonian term generates pairs of feedback coefficient terms.  For example, calculating the dynamics due to feedback $X_1$ on $\Pi_0$ generates two feedback coefficient operators: $\Pi_{0,0}=i\Pi_0X_1$ and $\Pi_{0,1}=iX_1\Pi_0$.  We must then determine the dynamics for these first level feedback coefficients.  This will include calculating the $Y_5$ feedback on $\Pi_{0,1}$, which generates second level feedback coefficients $\Pi_{0,2}=X_1Y_5\Pi_0$ and $\Pi_{0,3}=X_1\Pi_0Y_5$.  Continuing to iterate feedback coefficient terms, we find that an additional 1008 distinct $p_{j,c}^{(m)}$ terms are needed to close the dynamics and form a complete basis.  Adding in the initial 16 syndrome space projectors gives a 1024 dimensional basis---clearly no better than propagating the full density matrix.  However, it is now relatively easy to calculate the feedback strengths, which depend only on pairs of first-level feedback coefficients.  For example, from Eq.\ (\ref{error_correction:eq:FeedbackPolicy}) we find that $\lambda_{0,t}^{(1)}=\lambda_{\text{max}}\sgn\left(-p_{0,0}+p_{0,1}\right)$, where $p_{0,0}=\Tr{\Pi_{0,0}\rho_t}$ and $p_{0,1}=\Tr{\Pi_{0,1}\rho_t}$ are first-level coefficients developed earlier in the paragraph.

\subsection{Approximate Filter for the Five-Qubit Code}
\begin{figure}[bt]
\begin{center}
	\includegraphics[scale=0.8]{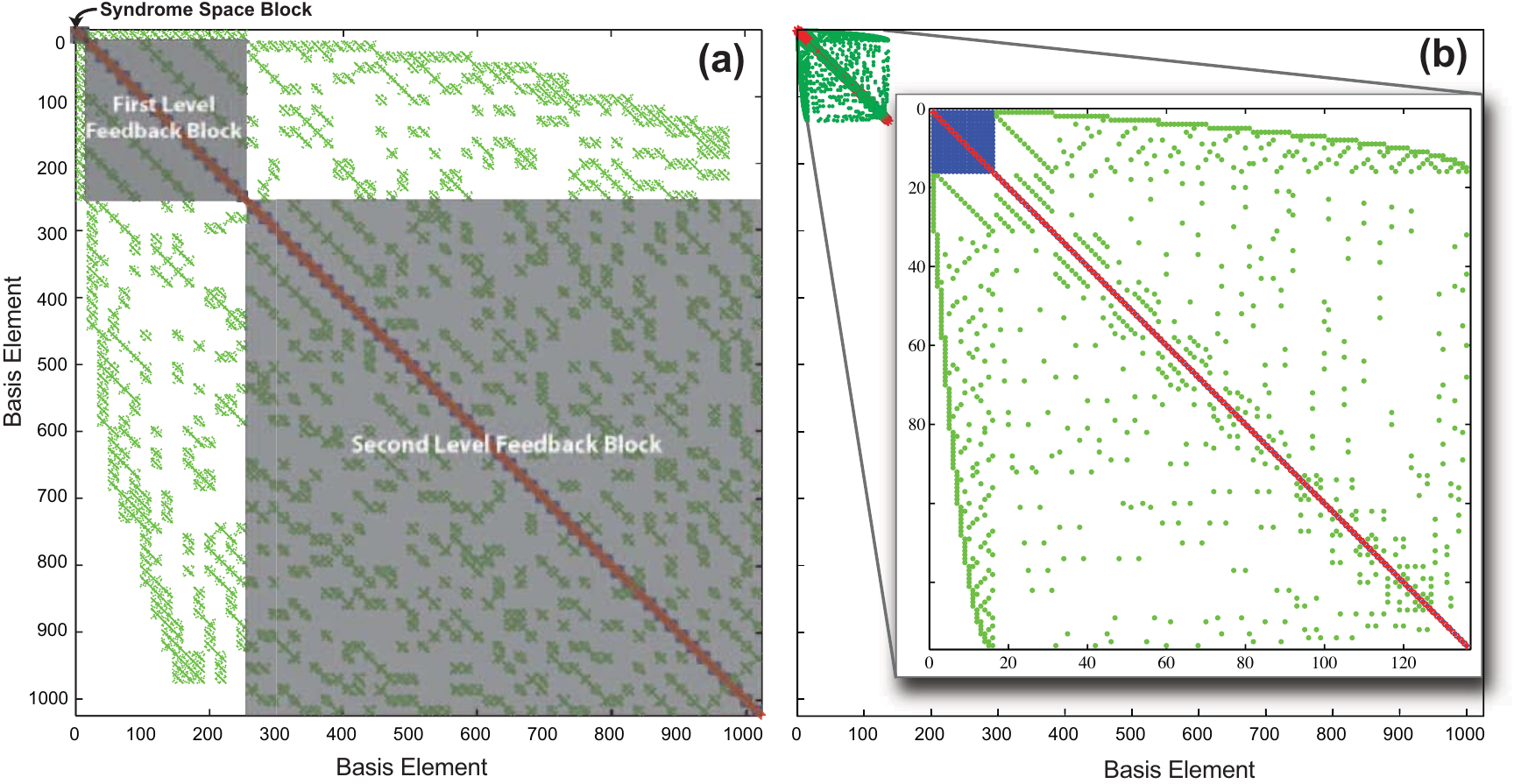}
\end{center}
\vspace{-5mm}
\caption[Graphical representation of non-zero matrix elements of error correction filter.]{Non-zero matrix elements of $(a)$ untruncated and $(b)$ truncated
filter.  Blue squares correspond to decoherence terms, red crosses
correspond to measurement terms and green dots correspond to feedback terms.
Note the difference in dimension of the matrices.} \label{error_correction:fig:annotatedMatrix}
\end{figure}

Although the dimension of the alternate basis is no smaller than the dimension of the full density matrix, the structure of the filter represented in the alternate basis provides a manner for interpreting the relative importance of the $p_{j,c}^{(m)}$ feedback coefficients.  This is best seen graphically in Fig.\ \ref{error_correction:fig:annotatedMatrix}(a), which superimposes the non-zero matrix elements coming from the noise, measurement and feedback terms.  Both measurement and noise are block diagonal as expected; it is the feedback that couples blocks together in a hierarchical fashion.  This hierarchy can be parameterized by the number of ``feedback transitions'' which connect a given feedback coefficient to the syndrome space block.  For example, the upper left block, which corresponds to the syndrome space projectors, is connected via feedback terms to the first level feedback block, whose feedback coefficients are each one feedback transition away from the syndrome space block.  In turn, the first level block is then connected to a second level feedback block, whose feedback coefficients are two feedback transitions away from the syndrome block.

Given that the initial state starts within the codespace and given that
feedback is always on, the feedback coefficients that are more than one
feedback transition from the syndrome space block should be vanishingly
small.  Limiting consideration to these first two blocks, we also find that
pairs of feedback coefficients couple identically to the syndrome space
block.  For example, we find that $-iX_1\Pi_0$ and $i\Pi_0X_1$ couple to syndrome space projectors identically.  This is not surprising, as these two terms comprise the commutator that results from the $X_1$ feedback Hamiltonian.   However, outside the first level of feedback transitions, the matrix elements of these feedback coefficients differ.  Additionally, feedback coefficients involving feedback Hamiltonians which correspond to a syndrome error on the codespace projector are related as
\begin{equation} \label{error_correction:eq:PauliRelation}
-i\sigma_j^{(m)}\Pi_0+i\Pi_0\sigma_j^{(m)} = -i\Pi_{j}^{(m)}\sigma_j^{(m)}+i\sigma_j^{(m)}\Pi_j^{(m)}\, .
\end{equation}
For the feedback coefficient examples just mentioned, this relation is $-iX_1\Pi_1^{(1)} + i\Pi_1^{(1)}X_1 = -i\Pi_0X_1 + iX_1\Pi_0$.  Truncating the dynamics to include only the first level of feedback and combining distinct feedback coefficients which act identically within this block results in the matrix of Fig.\ \ref{error_correction:fig:annotatedMatrix}(b) over only 136 basis elements.  Note that the controller now only needs to reference a single basis element for calculating a given feedback strength $\lambda_{j,t}^{(m)}$.

\subsection{Approximate Filter for General Codes}

The truncation scheme generalizes for reducing the dimensionality of the
quantum filter for an arbitrary $(E,\mathcal{G},R)$ code.  Such a filter for
an $[\![n,k]\!]$ quantum error-correcting code \citep{Nielsen:2000a} has the same form as
Eq.\ (\ref{error_correction:eq:BelavkinFilter}), but involves $n$ physical qubits and $l = n-k$
continuous-time stabilizer generator measurements.  In the following, we assume the continuous-time symmetric depolarizing channel, though it should be straightforward to extend to other noise models.  For a non-perfect, non-degenerate code, there are a total of $2^{n-l}$ stabilizer generator measurement outcomes, but only $3n+1$ will be observed for the given noise channel.  For a perfect, non-degenerate code ($2^{n-l}=3n+1$), all possible syndrome outcomes are observed.  In either case, given the observable syndrome outcomes, we can define $3n+1$ syndrome space projectors and $3n$ feedback parameters needed for recovery.  Degenerate codes require fewer than $3n$ recovery operations, as distinct actions of the noise channel give rise to identical errors and recovery operations.  The degeneracy depends greatly on the particular code, so we merely note that degenerate codes will require \emph{fewer} syndrome space projectors and feedback parameters than their non-degenerate relatives. 

Once we determine the syndrome space projectors and feedback parameters for the code, we can introduce feedback coefficient operators of the
form of (\ref{error_correction:eq:FeedbackCoefficient}) but over $n$ qubits.  A truncated filter is constructed as follows.
\begin{enumerate}
\item Close the dynamics of the syndrome space projectors by introducing first-level feedback terms.
\item Close the dynamics of the first-level feedback terms by truncating to a basis of syndrome space and first-level feedback terms, i.e. throw out potential second-level feedback terms.
\item Each of the $3n+1$ syndrome space projectors in this truncated form have $3n$ feedback coefficients, with pairs of these terms comprising each feedback Hamiltonian commutator.  Moreover, there is a factor of degeneracy between syndrome space projectors and feedback coefficients which involve the same Pauli matrix [c.f., Eq.\ (\ref{error_correction:eq:PauliRelation})].  A similarity transform is used to combine these pairs leaving $(3n+1)+(3n+1)3n/2=\frac{1}{2}\left(2+9n(n+1)\right)$ basis elements in the fully truncated filter.
\end{enumerate}
The truncated filter requires only $\mathcal{O}(n^2)$ basis elements, as compared to the $4^n$ parameters for the full density matrix.  Additionally, the feedback strengths in Eq.\ (\ref{error_correction:eq:FeedbackPolicy}) are readily calculated from the combined first-level feedback coefficients.  The truncation process is depicted schematically in the left half of Fig.\ \ref{error_correction:fig:hierarchy}.  The right half of the figure gives examples of a few of the 1024 terms involved in the truncation procedure for the five-qubit code.
\begin{figure}[t!]
\vspace{2mm}
\begin{center} 
\includegraphics{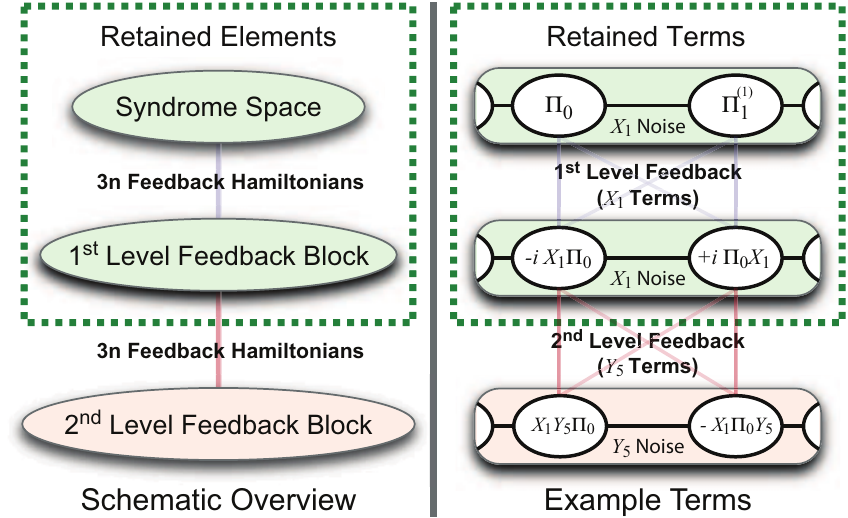}
\end{center}
\vspace{-3mm}
\caption[Schematic of filter truncation.]{On the left, a schematic diagram of truncating the filter to only syndrome space and first level feedback blocks.  On the right, just a few of the 1024 feedback coefficients of the five-qubit code representing the different feedback block levels.} \label{error_correction:fig:hierarchy}
\end{figure}

\begin{figure}[t!]
\begin{center} 
\includegraphics[scale=0.9]{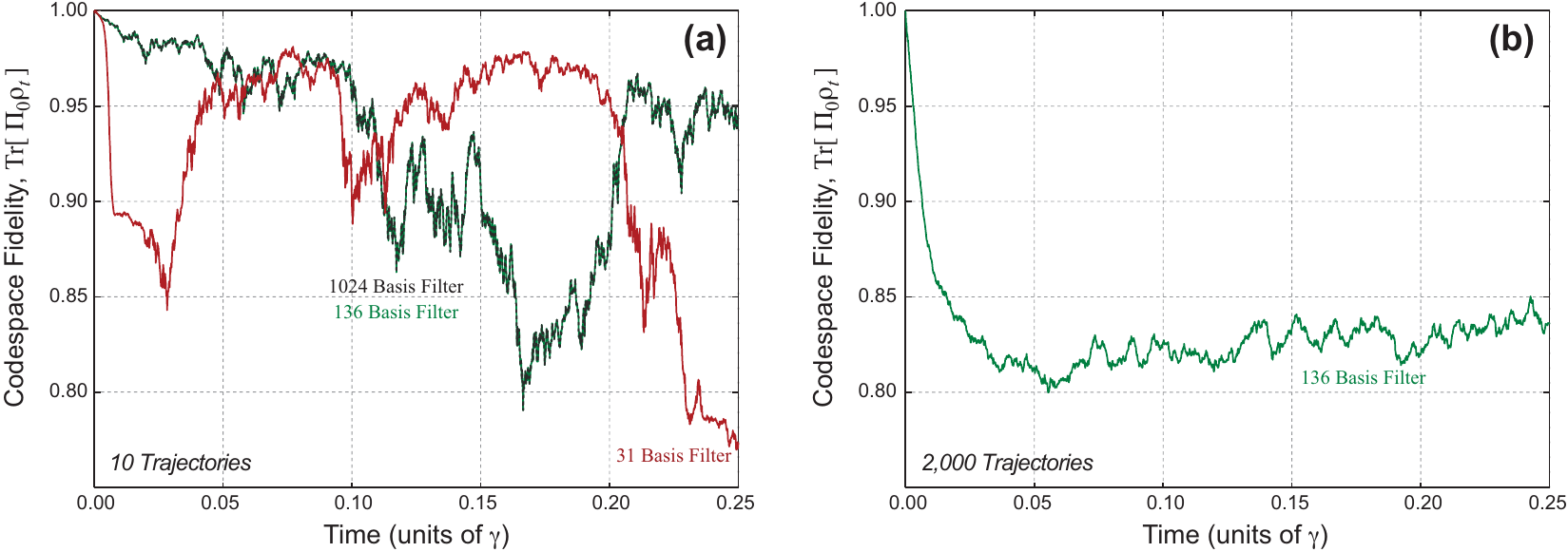}
\end{center}
\vspace{-5mm}
\caption[Numerical simulations of the five qubit code to assess the average
code space fidelity]{Numerical simulations of the five qubit code to assess the average
code space fidelity.  Plot (a) compares the codespace fidelity (averaged
over 10 trajectories) for filters with different levels of truncation: the
full (1024-dimensional) and first-level truncated (136-dimensional) filters
are essentially identical.  Plot (b) shows the codespace fidelity averaged
over 2,000 trajectories using the truncated 136-dimensional filter for error
correction.  (Simulation parameters: $\lambda_{\mathrm{max}} = 200 \gamma$ and
$\kappa = 100 \gamma$.)  (Color online.) \label{error_correction:fig:codewordPlots}}
\end{figure}

\subsection{Numerical Simulation}

Since the truncated filter is also nonlinear, it is difficult to provide analytic bounds on possible degradation in performance. However, we can easily compare numerical simulation between feedback controllers which use the full or truncated filter.  In fact, the dynamics should be close for the same noise realizations, indicating that they should be close per trajectory.  

In order to analyze the feedback controller's performance, the full filter Eq.\ (\ref{error_correction:eq:BelavkinFilter}) is used to represent the underlying physical system.  The feedback controller was modeled by simultaneously integrating the truncated filter, driven by the measurement current from the full filter.  The feedback controller then calculated the feedback strengths which were fed back into the full filter.  The dynamics described by the full filter were then used to compute the codespace fidelity.  Using a predictor-corrector SDE integrator discussed in Appendix \ref{appendix:numerical_methods_for_stochastic_differential_equations} and varying $\kappa$ and $\lambda_{\text{max}}$ over a wide range, we find essentially indistinguishable performance between the full and truncated filters.  Using $\kappa=100\gamma$ and $\lambda_{\text{max}}=200\gamma$ as representative parameters, Figure \ref{error_correction:fig:codewordPlots}(a) demonstrates this general behavior by comparing the average codespace fidelity of a handful of trajectories using the different filters. Integrating an individual trajectory takes approximately 39.5 seconds using a 2.1 GHz PowerPC processor.  Integrating the full filter alone takes approximately 36 seconds, while integrating the  truncated filter alone takes approximately 3.5 seconds.

In addition to showing the identical performance of the full and truncated filters, Fig.\ \ref{error_correction:fig:codewordPlots}(a) also shows the loss in performance if one were to truncate further.  The 31 dimensional filter is comprised of the 16 syndrome projectors and the 15 feedback coefficients which have non-zero feedback matrix elements with the codespace $\Pi_0$.  These are the only elements explicitly needed to calculate the feedback strengths in Eq. \ref{error_correction:eq:FeedbackPolicy}.  This filter fails because it tacitly assumes the action of feedback on the codespace is more ``important'' than on the other 15 syndrome spaces.  Since feedback impacts all syndrome spaces equally, we need to retain those terms in order to properly maintain syndrome space probabilities. Intuitively, this suggests that the 136 dimensional filter is the best we can do using this heuristic truncation strategy.  For reference, Fig.\ \ref{error_correction:fig:codewordPlots}(b) shows the average codespace fidelity of 2000 trajectories when using the truncated filter. 
\subsubsection{Comparison with Discrete Error Correction}
Given the success of the truncation scheme, we now compare the performance
of feedback-assisted error correction to that of discrete-time error correction for the five-qubit code.  The discrete model considers qubits exposed to the depolarizing channel 
\begin{equation}
d\rho_{\text{discrete}} = \gamma\sum_{j=x,y,z}\sum_{m=1}^{n=5}\mathcal{D}[\sigma_j^{(m)}]\rho_{\text{discrete}}dt
\end{equation}
up to a time $t$, after which discrete-time error correction is performed.  The
solution of this master equation can be explicitly calculated using the ansatz
\begin{equation}
\rho_{\text{discrete}}(t) = \sum_{e=0}^5\sum_{P; pw(P) = e}a_e(t)P\rho_0P ,
\end{equation}
where $P$ is a tensor product of Pauli matrices and the identity.  The
function $pw(P)$ gives the Pauli weight of a matrix, defined as the number
of $\sigma_x, \sigma_y,$ and $\sigma_z$ terms in the tensor representation.
Thus, $a_0(t)$ is the coefficient of $\rho_0$ and similarly $a_1(t)$ is the
coefficient of all single qubit errors from the initial state, e.g.,
$XIIII(\rho_0)XIIII, IIZII(\rho_0)IIZII$.  

The codespace fidelity considered earlier is not a useful metric for
comparison, as discrete-time error correction is guaranteed to restore the
state to the codespace. Following Ahn, Doherty and Landahl, we instead use the \emph{codeword
fidelity} $F_{cw}(t) := \Tr{\rho_0\rho(t)}$, which is a measure relevant for
a quantum memory.  Since error correction is independent of the encoded
state, we choose the encoded $\ket{0}$ state as a fiducial initial state.
Given that the five-qubit code protects against only single qubit errors, we
find that after error correction at time $t$, the codeword fidelity for
discrete-time error correction is
\begin{multline}
F_{cw}^{\text{discrete}} = a_0(t) + a_1(t) = \frac{1}{256} e^{-20 t \gamma } \left(3+e^{4 t \gamma }\right)^4 \left(-3+4 e^{4 t \gamma }\right)\, ,
\end{multline}
which asymptotes to $1/64$.  This limit arises because prior to the stabilizer generator measurements, the noise pushes the state to the maximally mixed state, which is predominately composed of the $a_2(t)$ through $a_5(t)$ terms. 

The feedback codeword fidelity $F_{cw}^{\text{feedback}}$ was calculated
by integrating both the full quantum filter (\ref{error_correction:eq:BelavkinFilter}),
representing the underlying system of qubits, and the truncated filter,
representing the feedback controller.  Again, we chose $\kappa = 100\gamma$,
$\lambda_{\text{max}}=200\gamma$ and $dt = 10^{-5}\gamma$ and used the same
SDE integrator described above. Figure \ref{error_correction:fig:fiveQubitPerformance} shows
the average of $F_{cw}^{\text{feedback}}$ over 2000 trajectories,
demonstrating that there are regimes where feedback-assisted error
correction can significantly outperform discrete-time error correction.  Feedback-assisted error correction appears to approach an asymptotic codeword fidelity greater than what would be obtained by decoherence followed by discrete-time error correction.  Due to the nonlinear feedback, it is difficult to calculate an analytic asymptotic expression for the continuous-time strategy.
\begin{figure}[t]
\begin{center} 
\includegraphics{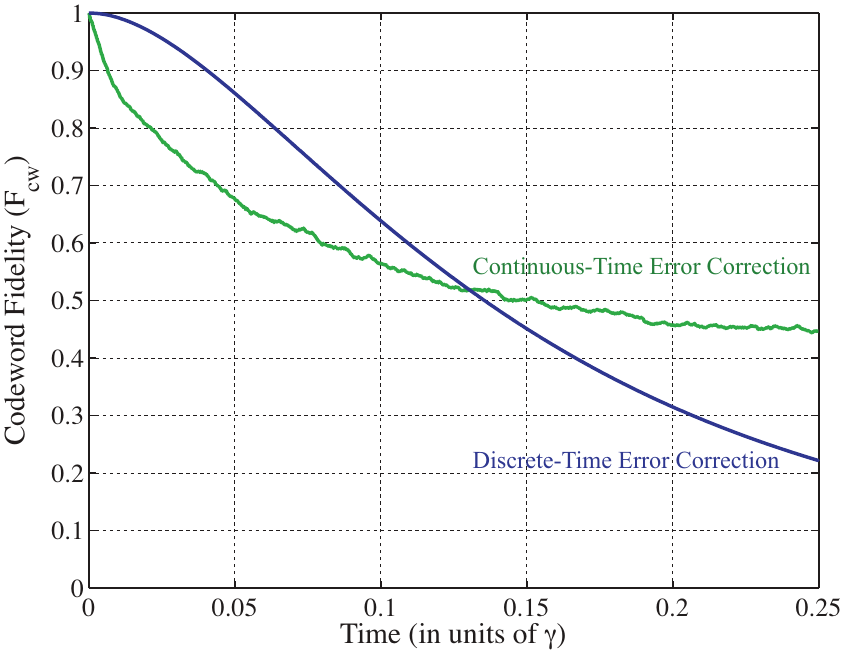}
\end{center}
\vspace{-5mm}
\caption[Comparison between continuous-time and discrete-time error correction
for the five-qubit code.]{Comparison between continuous-time and discrete-time error correction
for the five-qubit code.  For the  continuous-time error correction
simulations, the codeword fidelity was averaged over 2,000 trajectories with
$\kappa = 100 \gamma$ and $\lambda_{\text{max}}=200 \gamma$. (Color online.) \label{error_correction:fig:fiveQubitPerformance}}
\end{figure}
Nonetheless, the improved performance for the timespan considered suggests that better quantum memory is possible using the feedback scheme.  

\section{Summary}

Extending control theory techniques introduced by van Handel and Mabuchi \citep{vanHandel:2005c}, I have developed a computationally efficient feedback controller for continuous-time quantum error correction.  For the truncation scheme, the dimension of the filtering equations grows as $\mathcal{O}(n^2)$ in the number of physical qubits $n$, rather than $\mathcal{O}(4^n)$ for the original Ahn, Doherty and Landahl procedure \citep{Ahn:2002a}.  By numerical simulation of the five-qubit code, we have seen the viability of such a filter for a quantum memory protecting against a depolarizing noise channel.  Moreover, in all simulations, this performance is indistinguishable from that of the computationally more demanding filter of the ADL style.  

In systems where recovery operations are not instantaneous relative to decoherence, consideration suggests that it is desirable to perform syndrome measurement, recovery, and logic gates simultaneously.  However, it is not immediately clear how gates impact the feedback controller.  Indeed, if a Hamiltonian is in the code's normalizer, the continuous-time
feedback protocol and its performance are unchanged.  Though a universal set
of such Hamiltonians can be found, it might be desirable to find universal
gates which have physically simple interactions.  Future work involves
finding such gate sets and developing a framework for universal quantum
computing.  Additional issues of fault-tolerance and robustness could then
be explored within such a universal setup.  Exploring feedback error
correction in the context of specific physical models will provide
opportunities to tailor feedback strategies to available control parameters
and salient noise channels.  Such systems might allow the calculation of
globally optimal feedback control strategies.

\chapter{Model Reduction of Collective Qubit Dynamics}
\label{chapter:collective}
We saw in the previous chapter that the ability to find a low-dimensional model of a collective quantum systems allows one to efficiently simulate complex dynamics and in turn, design a practical feedback controller.  In this chapter, I focus on a problem outside the quantum feedback and control realm and present an exact, but nonetheless computationally appealing, description of arbitrary collective processes on open qubit systems.  The work presented here was published in \citep{Chase:2008c}.
\section{Introduction}
The ability to model the open system dynamics of large spin ensembles is crucial to experiments that make use of many atoms, as is often the case in precision metrology \citep{Wineland:1993a,Romalis:2003a}, quantum information science \citep{Polzik:2001a,Kuzmich:2003a,Jessen:2007a} and quantum optical simulations of condensed matter phenomena \citep{Bloch:2002a,StamperKurn:2006a,Parkins:2008a}.  Unfortunately, the mathematical description of large atomic spin systems is complicated by the fact that the dimension of the Hilbert space $\mathscr{H}_N$ grows exponentially in the number of atoms $N$.  Realistic simulations of experiments quickly become intractable even for atom numbers smaller than $N \sim 10$.  Current experiments, however, often work with atom numbers of more than $N \sim 10^{10}$, meaning that direct simulation of these systems is well beyond feasible.   Moreover, simulations over a range $N \sim 1-10$ are far from adequate to discern even the qualitative behavior that would be expected in the $N \gg 1$ limit.  Fortunately, it is often the case that experiments involving large spin ensembles respect one or more dynamical symmetries that can be exploited to reduce the effective dimension of the ensemble's Hilbert space. One can then hope to achieve a sufficiently realistic model of experiments without an exponentially large description of the system.   

In particular, previous work has focused on the \emph{symmetric collective states} $\ket{\psi_S}$, which are invariant under the permutation of particle labels: $\hat{\Pi}_{ij}\ket{\psi_S} = \ket{\psi_S}$.  These states span the subspace $\mathscr{H}_S \subset\mathscr{H}_N$, which grows linearly with the number of particles, $\operatorname{dim}(\mathscr{H}_S) = Nj + 1$.  However, in order for $\mathscr{H}_S$ to be an invariant subspace, the dynamics of the system must be expressible solely in terms of \emph{symmetric processes}, which are particle permutation invariant, and \emph{collective operators}, which respect the irreducible representation structure of rotations on the spin ensemble.  Fortunately, even within this restrictive class, a wide variety of phenomenon may be observed, including spin-squeezing \citep{Kitagawa:1993a,Polzik:1999a} and zero-temperature phase transitions \citep{Parkins:2008a}.

In practice, symmetric atomic dynamics are achieved by ensuring that there is identical coupling between all the atoms in the ensemble and the electromagnetic fields (optical, magnetic, microwave, etc.) used to both drive and observe the system \citep{Stockton:2003a}.  This approximation can be quite good for all of the coherent dynamics, because with sufficient laboratory effort, electromagnetic intensities can be made homogeneous, ensuring that interactions do not distinguish between different atoms in the ensemble.  However, incoherent dynamics are often beyond the experimenter's control.  Although most types of decoherence are symmetric, they are not generally written using collective operators.  Instead they are expressed as identical Lindblad operators for each spin, i.e.
\begin{equation} \label{collective:Intro::SymmetricLindblad}
		\mathcal{L}[\hat{s}]\hat{\rho} = \sum_{n=1}^{N}\biggr[ \hat{s}^{(n)}\hat{\rho}(\hat{s}^{(n)})^{\dagger} - 
					\frac{1}{2} (\hat{s}^{(n)})^{\dagger}\hat{s}^{(n)} \hat{\rho} -
					\frac{1}{2} \hat{\rho} (\hat{s}^{(n)})^{\dagger}\hat{s}^{(n)} \biggr] \ .
\end{equation}

The fact that decoherence does not preserve $\mathscr{H}_S$ has been well appreciated and the standard practice in experiments that address the collective state of atomic ensembles has been either: (i) to model such experiments only in a very short-time limit where decoherence can be approximately ignored; or (ii) to use decoherence models that do respect the particle symmetry, but which are written using only collective operators, even when doing so is not necessarily  physically justified.  In atomic spin ensembles, for example, a typical source of decoherence comes from spontaneous emission, yet collective radiative processes only occur under specific conditions such as superradiance from highly confined atoms \citep{Dicke:1954a} and some cavity-QED or spin-grating settings \citep{Black:2005a}.

In this chapter, I generalize the collective states of an ensemble of spin-1/2 particles (qubits) to include states that are preserved under symmetric--- but not necessarily collective--- transformations.      Specifically, I generalize from the strict condition of complete permutation invariance to the broader class of states that are indistinguishable across degenerate irreducible representations (irreps) of the rotation group. While the representation theory of the rotation group has been utilized in a wide variety of contexts, such as to protect quantum information from decoherence by encoding it into degenerate irreps with the same total angular momentum \citep{Bacon:2001a, Lidar:1998a}, I utilize relevant aspects of the representation theory to obtain a reduced-dimensional description of quantum maps that act locally but identically on every member of an ensemble of qubits.   

The main result, presented in Eq.\ (\ref{collective:Equation::generalFormula}), enables  us to represent arbitrary symmetric Lindblad operators in the collective state basis.  We find that the dimension of the Hilbert space $\mathscr{H}_{\mathrm{C}}$ spanned by these generalized collective states scales favorably, $\dim(\mathscr{H}_{\mathrm{C}}) \sim N^2$.  This allows for efficient simulation of a broader class of collective spin dynamics and in particular, allows one to consider the effects of decoherence on previous simulations of symmetric collective spin states.  We note that dynamical symmetries for spin-1/2 particles have been studied in the context of decoherence-free quantum information processing \citep{Lidar:1998a, Bacon:2001a}.  Unlike the work in this chapter, which uses symmetries to find a reduced description of a quantum system, these works seek to protect quantum information from decoherence by encoding within the degeneracies introduced by dynamical symmetries.

The remainder of this chapter is organized as follows.  Section \ref{collective:Section::GeneralStates} reviews the representation theory of the rotation group, which plays an important role in defining the symmetries related to $\mathscr{H}_S$ and $\mathscr{H}_C$.  Section \ref{collective:Section::CollectiveStates} introduces collective states and Section \ref{collective:Section::CollectiveProcesses} defines collective processes over these states. Section \ref{collective:Section::MainResult} gives an identity for expressing arbitrary symmetric superoperators, e.g. Eq. \ref{collective:Intro::SymmetricLindblad}, over the collective states.  Section \ref{collective:Section::Examples} leverages this formalism to compare the effect of different decoherence models in non-classical atomic ensemble states.  Section \ref{collective:Section::Conclusion} concludes.
\section{General states of the ensemble} \label{collective:Section::GeneralStates}
Consider an ensemble of $N$ spin-1/2 particles, with the $n^{th}$ spin characterized by its angular momentum $\hat{\mathbf{j}}^{(n)} = \{ \hat{j}_x^{(n)}, \hat{j}_y^{(n)}, \hat{j}_z^{(n)} \}$.  States of the spin ensemble are elements of the composite Hilbert space
\begin{equation}
	\mathscr{H}_N = \mathscr{H}^{(1)} \otimes \mathscr{H}^{(2)} \otimes \cdots \otimes \mathscr{H}^{(N)} 
\end{equation}
with $\operatorname{dim}(\mathscr{H}_N) = 2^N$.  Pure states of the ensemble, $\ket{\psi} \in \mathscr{H}_N$, are written as
\begin{equation}
	\ket{\psi} = \sum_{m_1,m_2,\ldots, m_N} c_{m_1,m_2,\ldots, m_N} \ket{m_1, m_2, \ldots, m_N}
\end{equation}
with $m_n = \pm \frac{1}{2}$ and where
\begin{equation}
	\ket{m_1,m_2,\ldots, m_N} = \ket{\frac{1}{2}, m_1}_1 \otimes \ket{\frac{1}{2}, m_2}_2 \otimes \cdots \otimes \ket{\frac{1}{2}, m_N}_N
\end{equation}
satisfies
\begin{equation} \label{collective:Equation::ProductStateKet}
	\hat{j}_z^{(n)}\ket{m_1,m_2,\ldots,m_N} = \hbar m_n\ket{m_1,m_2,\ldots,m_N} .
\end{equation}
When studying the open-system dynamics of the spin ensemble, one must generally consider the density operator
\begin{equation} \label{collective:Equation::ProductStateRho}
	\hat{\rho}  =  \sum_{\substack{m_1,m_2,\ldots, m_N\\ m'_1,m'_2,\ldots,m'_N}} 
		 \rho_{m_1,m_2, \cdots, m_N ;  m_1^\prime, m_2^\prime \cdots ,m_N^\prime}
		  \times \ket{m_1,m_2,\ldots,m_N}
		 	\dual{m_1^\prime,m_2^\prime,\ldots,m_N^\prime}
\end{equation}
States expanded as in Eqs. \ref{collective:Equation::ProductStateKet} and \ref{collective:Equation::ProductStateRho} are said to be written in the \emph{product basis}.
\subsection{Representations of the Rotation Group}
For a single spin-1/2 particle, a spatial rotation through the Euler angles $R=(\alpha,\beta,\gamma)$ is described by the rotation operator
\begin{equation}
	\hat{R}(\alpha,\beta,\gamma) = 
	e^{ - i \alpha \hat{j}_{\mathrm{z}} }
		e^{ - i \beta \hat{j}_{\mathrm{y}}}
	e^{ - i \gamma \hat{j}_{\mathrm{z}} }
\end{equation}
The basis kets $\ket{\frac{1}{2},m}$ for this particle therefore transform under the rotation $R$ according to
\begin{equation}
	\hat{R} \ket{\frac{1}{2},m'} = \sum_m \mathscr{D}^{\frac{1}{2}}_{m',m} (R) \ket{\frac{1}{2},m}
\end{equation}
where the matrices $\mathscr{D}^{\frac{1}{2}}(R)$ have the elements
\begin{equation}
	\mathscr{D}^{\frac{1}{2}}_{m',m} = \dual{\frac{1}{2},m'}  \hat{R}(\alpha,\beta,\gamma) \ket{\frac{1}{2},m} .
\end{equation}
The rotation matrices $\mathscr{D}^{\frac{1}{2}}(R)$ form a $2-$dimensional representation of the rotation group.

For the ensemble of $N$ spin-1/2 particles, each component of the ket $\ket{\psi} = \ket{m_1,m_2,\ldots,m_N}$ transforms separately under a rotation so that an arbitrary state transforms as
\begin{equation}
	\ket{\psi'} =  [\mathscr{D}^{\frac{1}{2}}(R)]^{\otimes N} \ket{\psi}.
\end{equation}
The rotation matrices $\mathscr{D}(R) = [\mathscr{D}^{\frac{1}{2}}(R)]^{\otimes N}$ provide a reducible representation for the rotation group but can be decomposed into irreducible representations (irreps) as  
\begin{equation} \label{collective:Equation::IrreducibleRotations}
	\mathscr{D}(R) = \bigoplus_{J=J_\text{min}}^{J_\text{max}}\bigoplus_{i=1}^{d^J_N}\mathscr{D}^{J,i}(R) \ .
\end{equation}
The quantum number $i(J) = 1, 2, \ldots, d^J_N$ is used to distinguish between the
\begin{equation}
	d^J_N =\frac{N!(2J+1)}{(\frac{N}{2}-J)!(\frac{N}{2}+J+1)!} \ ,\ \  J_\text{min} \leq J \leq J_\text{max}\\
\end{equation}
degenerate irreps with total angular momentum $J$ \citep{Mihailov:1977a}.  That is to say, $d^J_N$ is the number of ways one can combine $N$ spin-1/2 particles to obtain total angular momentum $J$.  The matrix elements of a given irrep $\mathscr{D}^{J,i}(R)$
\begin{equation}
	\mathscr{D}^{J,i}_{M,M'} (R) = \dual{J,M,i}  \mathscr{D}^{\frac{1}{2}}(R)^{\otimes N} \ket{J,M',i} 
\end{equation}
are written in terms of the total angular momentum eigenstates
\begin{eqnarray}
	\hat{\mathbf{J}}^2 \ket{J,M,i} & = & J(J+1) \ket{J,M,i} \\
	\hat{J}_{\mathrm{z}} \ket{J,M,i} & = & M \ket{J,M,i} 
\end{eqnarray}
with $\hat{J}_z = \sum_{n = 1}^N \hat{j}^{(n)}_z$, $J_{\text{max}} = \frac{N}{2}$ and
	\begin{equation}
		J_{\text{min}} = \begin{cases}
						  \frac{1}{2} & \text{$N$ odd}\\
						  0 & \text{$N$ even} \ .					
					  	  \end{cases}
	\end{equation}
It is important to note that degenerate irreps have \emph{identical} matrix elements, i.e. 
\begin{equation}
	\dual{J,M,i}  \mathscr{D}^{\frac{1}{2}}(R)^{\otimes N} \ket{J,M',i}
	  = \dual{J,M,i'}  \mathscr{D}^{\frac{1}{2}}(R)^{\otimes N} \ket{J,M',i'}
\end{equation}
for all $i,i'$.

In this representation, pure states are written as
\begin{equation} \label{collective:Equation::IrrepKet}
	\ket{\psi} = \sum_{J = J_\text{min}}^{J_\text{max}}\sum_{M = -J}^{J}\sum_{i=1}^{d^J_N} c_{J,M,i}\ket{J,M,i}
\end{equation}
and mixed states as
\begin{equation} \label{collective:Equation::IrrepRho}
	\hat{\rho} = \sum_{J,J' = J_\text{min}}^{J_\text{max}}
	\sum_{M,M' = -J,J'}^{J,J'}
	\sum_{i,i'=1}^{d^J_N,d^{J'}_N} 
			\rho_{J,M,i;J',M',i'} \ketbra{J,M,i}{J',M',i'}
\end{equation}
States written in the form of Eqs. \ref{collective:Equation::IrrepKet} or \ref{collective:Equation::IrrepRho} are said to be written in the \emph{irrep basis}.  We stress that both the product and irrep bases can describe any arbitrary state in $\mathscr{H}_N$.
\section{Collective States} \label{collective:Section::CollectiveStates}
While the representations in Section \ref{collective:Section::GeneralStates} allow us to express any state of the ensemble of spin-1/2 particles, the irrep basis suggests a scenario in which we could restrict attention to a much smaller subspace of $\mathscr{H}_N$.  In particular, the irrep structure of the rotation group, as expressed in Eq. \ref{collective:Equation::IrreducibleRotations}, indicates that rotations on the ensemble do not mix irreps and that degenerate irreps transform identically under a rotation.

Following this line of reasoning, we introduce the \emph{collective states}, $\ket{\psi_C}$, which span the sub-Hilbert space $\mathscr{H}_C \subset \mathscr{H}_N$.  Collective states have the property that degenerate irreps are identical; for pure states, $c_{J,M,i} = c_{J,M,i'}$ for all $i$ and $i'$.  We note that the symmetric collective states mentioned in the introduction are the collective states with $c_{J,M,i} = 0 $ unless $J=\frac{N}{2}$ and thus correspond to the largest $J$ value irrep.  We also note that
\begin{align}
	\operatorname{dim}{\mathscr{H}_C} &= 
	   \sum_{J=J_{\text{min}}}^{J_{\text{max}}}(2J+1) \nonumber\\
	    & =
		\begin{cases}
			\frac{1}{4}(N+3)(N+1), &\text{ if $N$ odd}\\
			\frac{1}{4}(N+2)^2, &\text{ if $N$ even}		
		\end{cases} \ .
\end{align}  

Physically, the collective states reflect an inability to address different degenerate irreps of the same total $J$.  This new symmetry allows us to effectively ignore the quantum number $i$ and write
\begin{align}
	\ket{\psi_C} &= \sum_{J = J_{\text{min}}}^{J_{\text{max}}}
					   \sum_{M = -J}^{J}
					   \sum_{i=1}^{d^J_N} c_{J,M,i} \ket{J,M,i}\nonumber\\
			&= \sum_{J = J_{\text{min}}}^{J_{\text{max}}}\sum_{M = -J}^{J}
					\sqrt{d^J_N} c_{J,M} \ket{J,M}
\end{align}
where I have defined effective basis kets 
\begin{equation} \label{collective:Equation::EffectiveKet}
	\ket{J,M} = \frac{1}{\sqrt{d^J_N}}\sum_{i=1}^{d^J_N} \ket{J,M,i}
\end{equation}
with effective amplitude $c_{J,M} = c_{J,M,i}$ for all $i$ (since the $c_{J,M,i}$ are equal for collective states).  

The factor of $\sqrt{d^J_N}$ serves as normalization, so that we can apply standard spin-$J$ operators to the effective kets without explicitly referencing their constituent degenerate irrep kets $\ket{J,M,i}$.  In other words, $\ket{J,M}$ actually represents $d^J_N$ degenerate kets, each with identical probability amplitude coefficients.  But since the matrix elements of a spin-$J$ operator are identical for irreps, we need not evaluate them individually.  

As an example, consider a rotation operator $\hat{R}$ which necessarily respects the irrep structure of the rotation group.  Calculating the expectation value of $\hat{R}$ by expanding the collective state $\ket{\psi_C}$ in the full irrep basis, we have
\begin{align}
	\bra{\psi_C} \hat{R} \ket{\psi_C}  
	&=  \sum_{J,J'}\sum_{M,M'}\sum_{i,i'} 
					c^{*}_{J,M,i}c_{J',M',i'}\bra{J,M,i} \hat{R}\ket{J',M',i'} 
						\label{collective:Equation::CollectiveKetExpecation::First}\\
						              & = \sum_{J}\sum_{M,M'}\sum_{i} c^{*}_{J,M,i}c_{J,M',i}
																\bra{J,M,i} \hat{R}\ket{J,M',i}
						\label{collective:Equation::CollectiveKetExpecation::Second}\\
							& = \sum_{J}\sum_{M,M'} d^J_N c^{*}_{J,M} c_{J,M'}\bra{J,M} \hat{R} \ket{J,M'}
						\label{collective:Equation::CollectiveKetExpecation::Third}
\end{align}
where in going from Eq. \ref{collective:Equation::CollectiveKetExpecation::First} to \ref{collective:Equation::CollectiveKetExpecation::Second}, we set $J=J'$ and $i=i'$ since rotation group elements do not mix irreps.  In reaching Eq. \ref{collective:Equation::CollectiveKetExpecation::Third}, I have further used the collective state property that $c_{J,M,i} = c_{J,M,i'} \forall i,i'$ and the rotation irrep property that $\bra{J,M,i} \hat{R} \ket{J,M',i} = \bra{J,M,i^{\prime}} \hat{R} \ket{J,M',i^{\prime}} \forall i,i'$ to drop the index $i$.

Equivalently we can evaluate the expectation using the effective basis kets $\ket{J,M}$ directly: 
\begin{align}
    \bra{\psi_C} \hat{R} \ket{\psi_C}
	& = \sum_{J,J'}\sum_{M,M'}\sqrt{d^J_N}\sqrt{d^{J'}_N} c^{*}_{J,M}c_{J',M'} 
					\bra{J,M} \hat{R} \ket{J',M'}\\
					& = \sum_{J}\sum_{M,M'} d^J_N 
							c^{*}_{J,M}c_{J,M'} \bra{J,M} \hat{R} \ket{J,M'} \ . \label{collective:Equation::CollectiveKetExpectation:Proper}
\end{align}
Comparing this to Eq. \ref{collective:Equation::CollectiveKetExpecation::Third} and recalling that $c_{J,M} = c_{J,M,i}$ for all $i$, we see that the effective calculation gives the same result.  

We can similarly define collective state density operators, $\hat{\rho}_C$, which have the properties that (i) there are no coherences between different irrep blocks and (ii) degenerate irrep blocks have identical density matrix elements.  The second assumption again means we can effectively drop the index $i$, since $\rho_{J,M,i;J,M',i} = \rho_{J,M,i';J,M',i'}$ for any $i$ and $i'$.  This allows us to write
\begin{equation}
	\hat{\rho}_C = \sum_{J = J_\text{min}}^{J_\text{max}}
					\sum_{M,M'= -J}^J
					\rho_{J,M;J,M'} \Eketbra{J,M}{J,M'}
\end{equation}
where the effective density matrix elements, written using an overlined outer product, are related to the irrep matrix elements via
\begin{equation} \label{collective:Equation::EffectiveDensityMatrixElement}
	\rho_{J,M;J,M'}\Eketbra{J,M}{J,M'} := 
		\frac{1}{d^J_N}
		\sum_{i=1}^{d^J_n} \rho_{J,M,i;J,M',i} \op{J,M,i}{J,M',i} \ .
\end{equation}
Just as for the effective kets, the normalization factor of $d^J_N$ ensures expectations are correctly calculated using the standard spin-$J$ operators.  The density matrix has $\sum_{J=J_{\text{min}}}^{J_{\text{max}}}(2J+1)^2=\frac{1}{6}(N+3)(N+2)(N+1)$ elements.

We stress that the overlined outer product notation is \emph{different} than naively taking the outer product of the effective kets defined in Eq. \ref{collective:Equation::EffectiveKet}.  Such an approach would involve outer products of kets between different, although degenerate, irreps.  Such terms are strictly forbidden by the first property of collective state density operators.  Instead, one should consider the effective density operator as a representation of $d^J_N$ identical copies of a spin-$J$ particle.  The overline notation is meant to remind the reader that the outer product beneath should only be interpreted using Eq. \ref{collective:Equation::EffectiveDensityMatrixElement} to relate back to the irrep basis. 

\section{Collective Processes} \label{collective:Section::CollectiveProcesses}
We are now interested in describing quantum processes, $\mathcal{L}$, which preserve collective states, $\hat{\rho}'_C = \mathcal{L}\hat{\rho}_C$.  Writing this explicitly, we must have
\begin{multline}
	 \sum_{J_1}\sum_{M_1,M'_1} d^{J_1}_N 
				\rho'_{J_1,M_1;J_1,M'_1} \Eketbra{J_1,M_1}{J_1,M'_1} \\
	 =			\sum_{J_2}\sum_{M_2,M'_2} d^{J_2}_N 
					\rho_{J_2,M_2;J_2,M'_2} 
					\mathcal{L}	\Eketbra{J_2,M_2}{J_2,M'_2} \label{collective:Equation::CollectivePreservingDynamics} .
\end{multline}
If we define the action of $\mathcal{L}$ on collective density matrix elements as
\begin{equation}
	f^{J,M,M'} = \mathcal{L}	\Eketbra{J,M}{J,M'}
\end{equation}
we immediately see that this action must be expressible as
\begin{equation} \label{collective:Equation::SymmetricProcess}
	f^{J,M,M'} = \sum_{J_1}\sum_{M_1,M_1'} \lambda^{J,M,M'}_{J_1,M_1,M_1'} \Eketbra{J_1,M_1}{J_1,M_1'}
\end{equation}
in order for the equality in Eq. \ref{collective:Equation::CollectivePreservingDynamics} to be met.  Here $\lambda^{J,M,M'}_{J_1,M_1,M_1'}$ is an arbitrary function of its indices.  Any process which preserves collective states by satisfying Eq. \ref{collective:Equation::SymmetricProcess} is a \emph{collective process}.

Examples of collective processes are those involving \emph{collective angular momentum operators} $\{\hat{J}_x,\hat{J}_y, \ldots\}$ and more generally, arbitrary \emph{collective operators} $\hat{C} = \sum_{n=1}^N \hat{c}^{(n)}$.  Since collective operators correspond to precisely the rotations considered when defining the irrep structure of the rotation group, they can all be written as 
\begin{equation} \label{collective:Equation::CollectiveOperator}
	\hat{C} = \sum_{J}\sum_{M,M'} c_{J,M,M'}\Eketbra{J,M}{J,M'} \ , 
\end{equation}
which cannot couple effective matrix elements with different $J$.  

However, the collective operators define a more restrictive class than an arbitrary collective process, which \emph{can} couple different $J$ blocks, so long as it does not create coherences between them.  In fact, if all operators are collective, then the symmetric collective states ($\ket{\psi_S}$) span an invariant subspace of the map.  This holds even when considering Lindblad operators that are written in terms of collective operators,
\begin{equation} \label{collective:eq:collectiveLindblad}
	\mathcal{L}[\hat{S}]\hat{\rho} = 
				\biggl[\hat{S}\hat{\rho}\hat{S}_{\dagger} -
				\frac{1}{2}\hat{S}^{\dagger}\hat{S}\hat{\rho}
				-\frac{1}{2}\hat{\rho}\hat{S}^{\dagger}\hat{S}\biggr]\ ,
\end{equation} 
where $\hat{S} = \sum_n \hat{s}^{(n)}$.

In the following section, I demonstrate that a process of the form
\begin{equation} \label{collective:Equation::GeneralSymmetric}
	f^{J,M,M'} = \sum_{n = 1}^{N} \hat{s}^{(n)} \Eketbra{J,M}{J,M'} (\hat{t}^{(n)})^{\dagger} \ ,
\end{equation}
which cannot be written solely in terms of collective operators, is nonetheless a collective process.  Moreover, if we expand the operators in the spherical Pauli basis via $\hat{s} = \vec{s}\cdot\vec{\sigma}$ and $\hat{t}^{\dagger} = \vec{t}\cdot\vec{\sigma}^{\dagger}$, we find 
\begin{equation}
	f^{J,M,M'} = \vec{s}\cdot\mathbf{g}(J,M,M',N)\cdot\vec{t}
\end{equation}
with the tensor $\mathbf{g}(J,M,M',N)$ defined as
\begin{equation} \label{collective:Equation::TensorDefinition}
   \mathbf{g}_{qr}(J,M,M',N) = \sum_{n=1}^{N} \hat{\sigma}_{q}^{(n)}\Eketbra{J,M}{J,M'}(\hat{\sigma}_r^{(n)})^{\dagger} .
\end{equation}
The tensor is written as a function of $N$ to coincide with the notation in the following section.

Before deriving a closed form expression for $\mathbf{g}(J,M,M',N)$, we would like to relate it to modeling \emph{symmetric decoherence processes},  which take the form
\begin{equation} \label{collective:Equation::SymmetricLindbladForm}
		\mathcal{L}[\hat{s}]\hat{\rho} = \sum_{n=1}^{N}\biggr[ \hat{s}^{(n)}\hat{\rho}(\hat{s}^{(n)})^{\dagger} - 
					\frac{1}{2} (\hat{s}^{(n)})^{\dagger}\hat{s}^{(n)} \hat{\rho} -
					\frac{1}{2} \hat{\rho} (\hat{s}^{(n)})^{\dagger}\hat{s}^{(n)} \biggr] \,.
\end{equation}
In order to relate $\mathcal{L}[\hat{s}]$ to Eqs. \ref{collective:Equation::GeneralSymmetric} and \ref{collective:Equation::TensorDefinition}, set $\hat{t}=\hat{s}$ and expand the single spin operator $\hat{s}$ in the spherical Pauli basis
\begin{equation}
	\hat{s} = s_I\hat{I} + \sum_{q} s_q\hat{\sigma}_q = s_{I}\hat{I} + s_{+}\hat{\sigma}_{+} + s_{-}\hat{\sigma}_{-} + s_{z}\hat{\sigma}_{z} 
\end{equation}
with the convention $\hbar = 1$, 
$\hat{\sigma}_{+} = \bigl(\begin{smallmatrix}
		0 & 1\\
		0 & 0
	   \end{smallmatrix}\bigr)$,
$\hat{\sigma}_{-} = \bigl(\begin{smallmatrix}
		0 & 0\\
		1 & 0
	   \end{smallmatrix}\bigr)$	and
$\hat{\sigma}_{z} = \bigl(\begin{smallmatrix}
		1 & 0\\
		0 & -1
	   \end{smallmatrix}\bigr)$.   The symmetric Lindblad of Eq. \ref{collective:Equation::SymmetricLindbladForm} can be expanded as
\begin{equation} \label{collective:Equation::SymmetricLindblad}
	\begin{split}
		\mathcal{L}[\hat{s}]\hat{\rho} &= \sum_{n=1}^{N}\biggr[ \hat{s}^{(n)}\hat{\rho}(\hat{s}^{(n)})^{\dagger} - 
					\frac{1}{2} (\hat{s}^{(n)})^{\dagger}\hat{s}^{(n)} \hat{\rho} -
					\frac{1}{2} \hat{\rho} (\hat{s}^{(n)})^{\dagger}\hat{s}^{(n)} \biggr]\\
					& = \sum_{n=1}^N \biggl[\hat{s}^{(n)}\hat{\rho}(\hat{s}^{(n)})^{\dagger}\biggr] -
							  \frac{1}{2} \hat{S}_N \hat{\rho}  - \frac{1}{2} \hat{\rho} \hat{S}_N
	\end{split}
\end{equation}
with the collective operator $\hat{S}_N$ given by
\begin{equation}
	\begin{split}
		\hat{S}_N &= \sum_{n=1}^N (\hat{s}^{(n)})^{\dagger}\hat{s}^{(n)} \\
		          &= \bigl(\frac{1}{2} \modsq{s_{-}} + 
						   \frac{1}{2} \modsq{s_{+}} +
						   \modsq{s_{I}}+ \modsq{s_z}\bigr)N\hat{I}\\
		          &+ \bigl( s_{-}^*s_I - s_{-}^{*}s_z + s_I^{*}s_{+} + s_z^{*}s_{+}\bigr)\hat{J}_{+}\\						
		 		  &+ \bigl( s_I^{*}s_{-} + s_{+}^{*}s_I + s_{+}^{*}s_Z - s_{z}^{*}s_{-} \bigr)\hat{J}_{-} \\
		          &+ \bigl( \frac{1}{2}\modsq{s_{-}} 
							- \frac{1}{2}\modsq{s_{+}} 
							+ s_I^{*}s_z + s_z^{*}s_I\bigr)\hat{J}_{z} \ .
	\end{split}
\end{equation} 
and  $\hat{J}_q = \sum_{n=1}^N \hat{\sigma}^{(n)}_q$ a collective spin operator.

In this form, it is clear that only the first term of the symmetric Lindbladian is not written using collective operators.  In fact, if we again expand $\hat{s}^{(n)}$ in the spherical basis, we observe that the only terms which involve non-collective operators are those which do not involve the identity operator, 
\begin{equation} \label{collective:Equation::NonCollectiveLindbladTerms}
	\begin{split}
		\sum_{n=1}^N \biggl[\hat{s}^{(n)}\hat{\rho}(\hat{s}^{(n)})^{\dagger}\biggr] &= \modsq{s_I}N\hat{\rho} + 
		\sum_{q}\bigl(s_qs_I^{*} \hat{J}_q\hat{\rho} + s_Is_q^{*}\hat{\rho} \hat{J}_q^{\dagger}\bigr) \\
		& + \sum_{n=1}^N\biggl[\sum_{q,r}s_qs_r^{*} \hat{\sigma}^{(n)}_q \hat{\rho} (\hat{\sigma}^{(n)}_q)^{\dagger} \biggr]
	\end{split}
\end{equation}
The last term here is precisely the tensor evaluation of $\vec{s}\cdot\mathbf{g}(J,M,M',N)\cdot\vec{s}^{*}$.  We now proceed to give an identity for the tensor elements.  

\section{Identity} \label{collective:Section::MainResult}
\begin{id} \label{collective:Id::MainResult}
	Given a collective density matrix element for $N$ spin-1/2 particles, $\Eketbra{J,M}{J,M'}$, we have
	\begin{align}
		&\mathbf{g}_{qr}(J,M,M',N)\\
		=&\sum_{n=1}^N \hat{\sigma}_{q}^{(n)}\Eketbra{J,M}{J,M'}(\hat{\sigma}_{r}^{(n)})^{\dagger}
				\label{collective:Equation::BasicSum}\\
		=& \frac{1}{2J}\biggl[1+\frac{\alpha^{J+1}_N}{d^J_N}\frac{2J+1}{J+1}\biggr]
		A_q^{J,M}\Eketbra{J,M_q}{J,M'_r} A_r^{J,M'} \nonumber \\
		+&\frac{\alpha^J_N}{d^J_N2J}B_q^{J,M}\Eketbra{J-1,M_q}{J-1,M'_r} B_r^{J,M'} \nonumber\\
		+&\frac{\alpha^{J+1}_N}{d^{J}_N2(J+1)}D_q^{J,M}
		  \Eketbra{J+1,M_q}{J+1,M'_r}D_r^{J,M'} \label{collective:Equation::generalFormula}
	\end{align}
	where $q,r \in \{+,-,z\}$, $M_+ = M + 1$, $M_- = M - 1$ and $M_z = M$,
	\begin{equation}
		\alpha^J_N = \sum_{J'=J}^{\frac{N}{2}}d^{J'}_N = \frac{N!}{\left(\frac{N}{2}-J\right)! \left(\frac{N}{2}+J\right)!}
	\end{equation}

	and
	\begin{subequations}
		\begin{align}
			A_+^{J,M} &= \sqrt{(J-M)(J+M+1)}\\
			A_-^{J,M} &= \sqrt{(J+M)(J-M+1)}\\
			A_z^{J,M} &= M
		\end{align}
	\end{subequations}
	
	and
	\begin{subequations}
		\begin{align}
			B_+^{J,M} &=\sqrt{(J-M)(J-M-1)}\\
			B_-^{J,M} &=-\sqrt{(J+M)(J+M-1)} \\	
			B_z^{J,M} &= \sqrt{(J+M)(J-M)} 
		\end{align}
	\end{subequations}
	
	and lastly
	\begin{subequations}
		\begin{align}
			D_+^{J,M} &= -\sqrt{(J+M+1)(J+M+2)} \\
			D_-^{J,M} &=\sqrt{(J-M+1)(J-M+2)} \\	
			D_z^{J,M} &= \sqrt{(J+M+1)(J-M+1)} \ .
		\end{align}
	\end{subequations}
	
\end{id}
\noindent Note that $\alpha^J_N$ and $d^J_N$ are zero if $J$ is negative or $J = N/2$, ensuring that only valid density matrix elements are involved. 

In the following subsections, we prove Identity \ref{collective:Id::MainResult} inductively. The motivation for the inductive proof comes from the simple recursive structure of adding spin-1/2 particles.  As seen in Fig. \ref{collective:fig:MomentumTree}, the $d^J_N$ irreps which correspond to a total spin $J$ particle composed of $N$ spin-1/2 particles can be split into two groups, depending on how angular momentum was added to reach them.  By expressing the $N$ particle states in terms of bipartite states of a single spin-1/2 particle and a spin-$(N-1)$ particle, we can then evaluate the dynamics independently on either half by assuming Identity \ref{collective:Id::MainResult} holds.  Returning the resulting state to the $N$ particle basis should then confirm the Identity.  By inspection, the base case of $N=1$ holds, as the $A_q^{J,M}$ terms reduce to the single spin-1/2 matrix elements.  We now proceed to the inductive case.
\begin{figure}[h]
	\centering
	\includegraphics{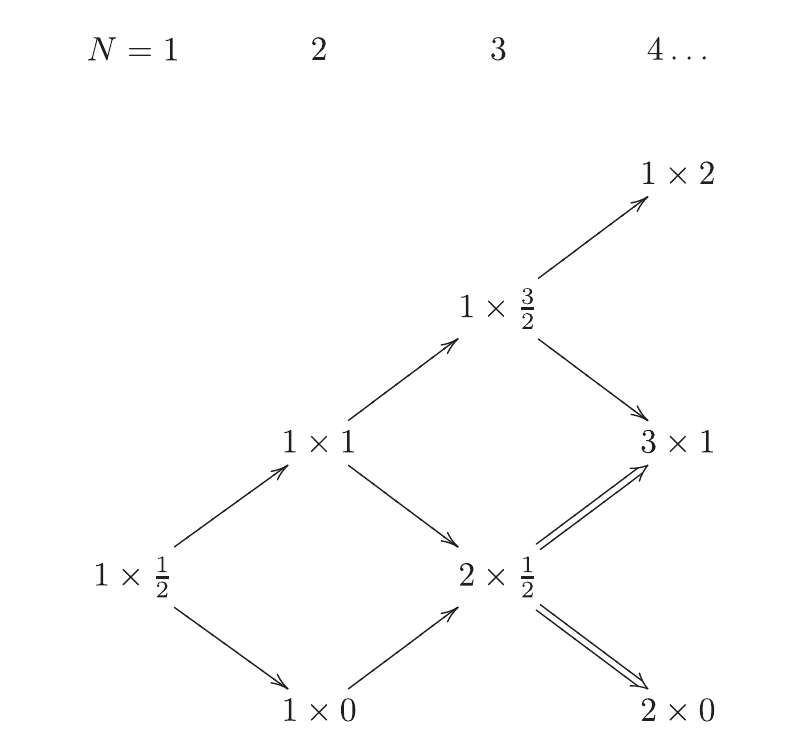}
	\caption[Degeneracy structure from adding spin-1/2 particles]{Degeneracy structure from adding spin-1/2 particles, labeled as $d^J_N \times J$} \label{collective:fig:MomentumTree} 
\end{figure}
\subsection{Recursive state structure}
In order to apply the inductive hypothesis, we need to express an $N$ particle state in terms of $N-1$ particle states.  This recursive structure is best seen by examining Fig. \ref{collective:fig:MomentumTree}, which illustrates the branching structure for adding spin-1/2 particles.  For example, the three-fold degenerate $N=4$ spin-1 irreps arise from two different spin additions---adding a single spin-1/2 particle to the non-degenerate $J=\frac{3}{2}$,$N=3$ irrep and adding to the 2-fold degenerate $J=\frac{1}{2}$, $N=3$ irreps. Since we are always adding a spin-1/2 particle, the tree is at most binary. This allows us to recursively decompose the degenerate irreps for a given $J$ in terms of adding a single spin-1/2 particle to the two related $N-1$ degenerate irreps.  

Recall that for the collective states, we defined effective density matrix elements which group degenerate irreps (Eq. \ref{collective:Equation::EffectiveDensityMatrixElement}).  In order to make the relationship between states of different $N$ clear, in this section we will add the index $N$ to all effective density matrix elements---$\Eketbra{J,M,N}{J,M',N}$.  Similarly, when expressing the collective state in the irrep basis, we will also use kets with the index $N$, i.e. $\ket{J,M,N,i}$.  Here, the $N$ and $i$ indices indicate the state is from $i$-th degenerate total spin-$J$ irrep that comes from adding $N$ spin-1/2 particles. So that we can leverage the binary branching structure seen in Fig. \ref{collective:fig:MomentumTree}, we also need to relate the $N$ particle irrep states to the $N-1$ particle irrep states. Accordingly, we define $\ket{J,M;\frac{1}{2},J\pm\frac{1}{2},N-1,i_1}$, where the last four entries indicate that the overall $N$ spin state can be viewed as combining a single spin-1/2 particle with a spin $J\pm \frac{1}{2}$ particle.  The spin $J\pm \frac{1}{2}$ particle is from the $i_1$-st such irrep for $N-1$ spin-1/2 particles.  With these definitions, we can now relate the $N$ particle states to the $N-1$ particle states by explicitly tensoring out a single spin-1/2 particle:
	\begin{align} 
			 &\Eketbra{J,M,N}{J,M',N} \nonumber\\
			 =& \frac{1}{d^J_N}\sum_{i=1}^{d^J_N}\ketbra{J,M,N,i}{J,M',N,i}\\
			 =&\frac{1}{d^J_N}\sum_{i_1=1}^{d^{J+\frac{1}{2}}_{N-1}}
									\ketbra{J,M;\frac{1}{2},J+\frac{1}{2},N-1,i_1}
										   {J,M';\frac{1}{2},J+\frac{1}{2}, N-1,i_1} \nonumber\\
			+ & \frac{1}{d^J_N}\sum_{i_2=1}^{d^{J-\frac{1}{2}}_{N-1}}
			\ketbra{J,M;\frac{1}{2},J-\frac{1}{2},N-1,i_2}
			    {J,M';\frac{1}{2},J-\frac{1}{2},N-1,i_2}\displaybreak[1]\\
			= & \frac{d^{J+\frac{1}{2}}_{N-1}}{d^J_N}
			 \sum_{m_1,m_1'}
			  \cg{\frac{1}{2},m_1}{J+\frac{1}{2},M-m_1}{J,M}
			\ketbra{\frac{1}{2},m_1}{\frac{1}{2},m_1'} \nonumber\\
			 & \qquad\qquad \otimes
			\Eketbra{J+\frac{1}{2},M-m_1,N-1}{J+\frac{1}{2},M'-m_1',N-1}
			 \cg{\frac{1}{2},m_1'}{J+\frac{1}{2},M'-m_1'}{J,M'} \nonumber \\ 
			+ & \frac{d^{J-\frac{1}{2}}_{n-1}}{d^J_N}
			    \sum_{m_2,m_2'}
			 \cg{\frac{1}{2},m_2}{J-\frac{1}{2},M-m_2}{J,M}
		    \ketbra{\frac{1}{2},m_2}{\frac{1}{2},m_2'} \nonumber\\
		    & \qquad\qquad  \otimes
		    \Eketbra{J-\frac{1}{2},M-m_2,N-1}{J-\frac{1}{2},M'-m_2',N-1}
		     \cg{\frac{1}{2},m_2'}{J-\frac{1}{2},M'-m_2'}{J,M'} \label{collective:eq:recursiveStateDefinition}
	\end{align}

with Clebsch-Gordan coefficients $\cg{j_1,m_1}{j_2,m_2}{J,M} = \braket{J,M; j_1, j_2}{j_1,m_1;j_2,m_2}$ and the $m_i,m_i'$ sums over single spin projection values $\pm \frac{1}{2}$.  In reaching Eq. \ref{collective:eq:recursiveStateDefinition}, we made use of the definition of the effective density matrix element for $N-1$ spins given in Eq. \ref{collective:Equation::EffectiveDensityMatrixElement}.  With this recursive state definition, we can now start the inductive step of the proof.
\subsection{Applying inductive hypothesis}
In order to prove the Identity, we must be able to apply the inductive hypothesis to  Eq. \ref{collective:Equation::BasicSum}.  Ignoring the Clesbsch-Gordan coefficients for the moment, consider an arbitrary term from Eq. \ref{collective:eq:recursiveStateDefinition}.  The dynamics distribute as

		\begin{align}
				& \sum_{n=1}^N  \sigma_{q}^{(n)} 
				\biggl[ 
				\ketbra{\frac{1}{2},m_i}
				       {\frac{1}{2},m_i'} \otimes \Eketbra{J\pm \frac{1}{2},M-m_i,N-1}
				      {J \pm \frac{1}{2},M-m_i',N-1}\biggr]\sigma_{r}^{(n)} \nonumber \displaybreak[1]\\
				& = \mathbf{g}_{qr}(\frac{1}{2}, m_i,m_i',1)  \otimes \Eketbra{J \pm \frac{1}{2},M-m_i,N-1}
				       			{J \pm \frac{1}{2},M-m_i',N-1}\\
				& + \ketbra{\frac{1}{2},m_i}
				   		 {\frac{1}{2},m_i'}\otimes \mathbf{g}_{qr}(J  \pm \frac{1}{2},M-m_i,M'-m_i',N-1) \label{collective:eq:recursiveDynamics} \ .
		\end{align}

By extension, all terms in Eq. \ref{collective:eq:recursiveStateDefinition} split the dynamics in this manner, which allows us to apply the inductive hypothesis to evaluate $\mathbf{g}_{qr}(\frac{1}{2}, m_i,m_i',1)$ and $\mathbf{g}_{qr}(J \pm \frac{1}{2},M-m_i,M'-m_i',N-1)$.  This means evaluating the $\mathbf{g}_{qr}$ terms according to the hypothesis in Eq. \ref{collective:Equation::generalFormula}, after which we rewrite the bipartite states in the $N$ spin basis. 

We have the $\mathbf{g}_{qr}(\frac{1}{2}, m_1,m_1',1)$ terms
\begin{equation}
	\begin{split}
		\frac{1}{d^J_N} \sum_{i_1=1}^{d^{J+\frac{1}{2}}_{N-1}} & \sum_{J_1=J}^{J+1} \sum_{m_1}
			\biggl[ 
			A_q^{\frac{1}{2},m_1}
			\cg{\frac{1}{2}, {m_1}_q}
			   {J + \frac{1}{2}, M-m_1}
			   {J_1, M_q}
			\cg{\frac{1}{2},m_1}
			   {J+\frac{1}{2},M-m_1}
			   {J,M}
			\ket{J_1, M_q; \frac{1}{2}, J + \frac{1}{2}, N-1, i_1}\biggr] \\
			\times &\sum_{J'_1=J}^{J+1}\sum_{m_1'}
			\biggl[ \bra{J'_1, M'_r; \frac{1}{2}, J + \frac{1}{2}, N-1, i_1}
			\cg{\frac{1}{2},m'_1}
			   {J+\frac{1}{2},M'-m'_1}
			   {J,M'}		
			\cg{\frac{1}{2}, {m'_1}_r}
			   {J + \frac{1}{2}, M'-m'_1}
			   {J'_1, M'_r}
			A_r^{\frac{1}{2},m'_1} \biggr]  \label{collective:Equation::Unsimplified::Jplus12::Basic}
   \end{split}
\end{equation}
and the $\mathbf{g}_{qr}(\frac{1}{2}, m_2,m_2',1)$ terms
\begin{equation}
	\begin{split}
		\frac{1}{d^J_N}& \sum_{i_2=1}^{d^{J-\frac{1}{2}}_{N-1}}\sum_{J_2=J-1}^{J} \sum_{m_2}
				\biggl[ 
				A_q^{\frac{1}{2},m_2}
				\cg{\frac{1}{2}, {m_2}_q}
				   {J - \frac{1}{2}, M-m_2}
				   {J_2, M_q}
				\cg{\frac{1}{2},m_2}
				   {J-\frac{1}{2},M-m_2}
				   {J,M}
				\ket{J_2, M_q; \frac{1}{2}, J - \frac{1}{2}, N-1, i_2}\biggr]\\
		\times & \sum_{J'_2=J-1}^{J}\sum_{m_2'}
		 \biggl[\bra{J'_2, M'_r; \frac{1}{2}, J - \frac{1}{2}, N-1, i_2}
			\cg{\frac{1}{2},m'_2}
			   {J-\frac{1}{2},M'-m'_2}
			   {J,M'}
			\cg{\frac{1}{2}, {m'_2}_r}
			   {J - \frac{1}{2}, M'-m'_2}
			   {J'_2, M'_r}
			A_r^{\frac{1}{2},m_2'}
		\biggr] \ . \label{collective:Equation::Unsimplified::Jminus12::Basic}
	\end{split}
\end{equation}
The $\mathbf{g}_{qr}(J + \frac{1}{2},M-m_1,M'-m_1',N-1)$ terms are
   \begin{multline} \label{collective:Equation::Unsimplified::Jplus12::A}
   			\shoveleft{\frac{1}{d^J_N(2J+1)}
   			\biggl[1 + 
			\frac{\alpha^{J+\frac{3}{2}}_{N-1}}
			     {d^{J+\frac{1}{2}}_{N-1}}
			\frac{2J+2}
			     {J+\frac{3}{2}}\biggr] \sum_{i_1 = 1}^{d^{J+\frac{1}{2}}_{N-1}}
			\sum_{J_1=J}^{J+1}
			\sum_{m_1}
					\biggl[
					A_q^{J+\frac{1}{2},M-m_1}
					\cg{\frac{1}{2},  m_1}
					   {J + \frac{1}{2}, M_q-m_1}
					   {J_1, M_q}\biggr.}\\
			\shoveright{\biggl.
					\cg{\frac{1}{2},m_1}
					   {J+\frac{1}{2}, M-m_1}
					   {J,M}
					\ket{J_1, M_q; \frac{1}{2}, J+\frac{1}{2}, N-1,i_1}
					\biggr]}\\
			\shoveleft{\times  \sum_{J'_1=J}^{J+1}\sum_{m_1'}
					\biggl[
					\bra{J'_1, M_r'; \frac{1}{2}, J+\frac{1}{2}, N-1,i_1}
					\cg{\frac{1}{2},m'_1}
					   {J+\frac{1}{2}, M'-m'_1}
					   {J,M'}\biggr.}\\
			\biggl.
					\cg{\frac{1}{2},  m'_1}
					   {J + \frac{1}{2}, M'_r-m'_1}
					   {J'_1, M'_r}
					A_r^{J+\frac{1}{2},M'-m'_1}
					\biggr]  
	\end{multline}
     \begin{multline} \label{collective:Equation::Unsimplified::Jplus12::B}
	       +\frac{\alpha^{J+\frac{1}{2}}_{N-1}}
				{d^J_N d^{J-\frac{1}{2}}_{N-1}2(J+\frac{1}{2})} 
			\sum_{i_1 = 1}^{d^{J-\frac{1}{2}}_{N-1}}
			\sum_{J_1=J-1}^{J}
			\sum_{m_1}
				\biggl[
				B_q^{J+\frac{1}{2},M-m_1}
				\cg{\frac{1}{2},  m_1}
				   {J - \frac{1}{2}, M_q - m_1}
				   {J_1, M_q}\biggr.\\
		    \shoveright{\biggl.
				\cg{\frac{1}{2},m_1}
				   {J+\frac{1}{2}, M-m_1}
				   {J,M}
				\ket{J_1, M_q; \frac{1}{2}, J-\frac{1}{2}, N-1,i_1}
				\biggr] }\\
			\times  \sum_{J'_1=J-1}^{J}\sum_{m_1'}
				     \biggl[
					\ket{J'_1, M'_r; \frac{1}{2}, J-\frac{1}{2}, N-1,i_1}					
					\cg{\frac{1}{2},m'_1}
					   {J+\frac{1}{2}, M'-m'_1}
					   {J,M'}\biggr.\\
					\biggl.\cg{\frac{1}{2},  m'_1}
					   {J - \frac{1}{2}, M'_r - m'_1}
					   {J'_1, M'_r}
					B_r^{J+\frac{1}{2},M'-m'_1}
					\biggr] 
		\end{multline}
	    \begin{multline} \label{collective:Equation::Unsimplified::Jplus12::C}
			+ \frac{\alpha^{J+\frac{3}{2}}_{N-1}}
					{d^J_N d^{J+\frac{3}{2}}_{N-1}2(J+\frac{3}{2})}
		      \sum_{i_1 = 1}^{d^{J+\frac{3}{2}}_{N-1}}
			  \sum_{J_1=J+1}^{J+2}
			  \sum_{m_1}
				\biggl[
				D_q^{J+\frac{1}{2},M-m_1}
				\cg{\frac{1}{2},  m_1}
				   {J + \frac{3}{2}, M_q - m_1}
				   {J_1, M_q}\biggr.\\
				\shoveright{\biggl.
				\cg{\frac{1}{2},m_1}
				   {J+\frac{1}{2}, M-m_1}
				   {J,M}
				\ket{J_1, M_q; \frac{1}{2}, J+\frac{3}{2}, N-1,i_1}
				\biggr]} \\
			\times  \sum_{J'_1=J+1}^{J+2}\sum_{m_1'}
					\biggl[
					\ket{J'_1, M'_r; \frac{1}{2}, J+\frac{3}{2}, N-1,i_1}					
					\cg{\frac{1}{2},m'_1}
					   {J+\frac{1}{2}, M'-m'_1}
					   {J,M'}\biggr.\\
					\biggl.
					\cg{\frac{1}{2},  m'_1}
					   {J + \frac{3}{2}, M'_r - m'_1}
					   {J'_1, M'_r}
					D_r^{J+\frac{1}{2},M'-m'_1}
					\biggr] 
   		\end{multline}
and lastly, the $\mathbf{g}_{qr}(J-\frac{1}{2},M - m_2, M' - m_2',N-1)$ terms are
	\begin{multline}  \label{collective:Equation::Unsimplified::Jminus12::A}
		\frac{1}{d^J_N2(J-\frac{1}{2})}
		\bigl[ 1 +
		       \frac{\alpha^{J + \frac{1}{2}}_{N-1}}
		            {d^{J-\frac{1}{2}}_{N-1}}
		       \frac{2J}
		            { J + \frac{1}{2}}\bigr]
		\sum_{i_2 = 1}^{d^{J-\frac{1}{2}}_{N-1}}
		\sum_{J_2=J-1}^{J}
		\sum_{m_2}
			\biggl[
			A_q^{J-\frac{1}{2},M-m_2}
			\cg{\frac{1}{2},  m_2}
			   {J - \frac{1}{2}, M_q - m_2}
			   {J_2, M_q}\biggr.\\
		\shoveright{	\biggl.
			\cg{\frac{1}{2},m_2}
			   {J-\frac{1}{2}, M-m_2}
			   {J,M}
			\ket{J_2, M_q; \frac{1}{2}, J-\frac{1}{2}, N-1,i_2}
			\biggr]} \\
	     \times \sum_{J_2=J-1}^{J}	\sum_{m'_2}	\biggl[
				\ket{J'_2, M'_r; \frac{1}{2}, J-\frac{1}{2}, N-1,i_2}
				\cg{\frac{1}{2},m'_2}
				   {J-\frac{1}{2}, M'-m'_2}
				   {J,M'}\biggr.\\
				\biggl.\cg{\frac{1}{2},  m'_2}
				   {J - \frac{1}{2}, M'_r - m'_2}
				   {J'_2, M'_r}
				A_r^{J-\frac{1}{2},M'-m'_2}
				\biggr]
	    \end{multline}
	    \begin{multline} \label{collective:Equation::Unsimplified::Jminus12::B}
		+ \frac{\alpha^{J-\frac{1}{2}}_{N-1}}
					{d^J_N d^{J-\frac{3}{2}}_{N-1}2(J - \frac{1}{2})}
		\sum_{i_2=1}^{d^{J-\frac{3}{2}}_{N-1}}
			\sum_{J_2=J-2}^{J-1}
		    \sum_{m_2}
			\biggl[
			B_q^{J-\frac{1}{2},M-m_2}
			\cg{\frac{1}{2},  m_2}
			   {J - \frac{3}{2}, M_q - m_2}
			   {J_2, M_q}\biggr.\\
		   \shoveright{\biggl.
			\cg{\frac{1}{2},m_2}
			   {J-\frac{1}{2}, M-m_2}
			   {J,M}
			\ket{J_2, M_q; \frac{1}{2}, J-\frac{3}{2}, N-1,i_2}
			\biggr]} \\
		  \times \sum_{J_2=J-2}^{J-1}\sum_{m_2'}
			\biggl[
			\ket{J'_2, M'_r; \frac{1}{2}, J-\frac{3}{2}, N-1,i_2}			
			\cg{\frac{1}{2},m'_2}
			   {J-\frac{1}{2}, M'-m'_2}
			   {J,M'}\biggr.\\
			 \biggl.
			\cg{\frac{1}{2},  m'_2}
			   {J - \frac{3}{2}, M'_r - m'_2}
			   {J'_2, M'_r}
			B_r^{J-\frac{1}{2},M'-m'_2}
			\biggr] 
	    \end{multline}
		\begin{multline}\label{collective:Equation::Unsimplified::Jminus12::C}
		+ \frac{\alpha^{J+\frac{1}{2}}_{N-1}}
					{d^J_N d^{J+\frac{1}{2}}_{N-1}(2J+1)}
			\sum_{i_2=1}^{d^{J + \frac{1}{2}}_{N-1}}
			\sum_{J_2=J}^{J+1}
			 \sum_{m_2}
			\biggl[
			D_q^{J-\frac{1}{2},M-m_2}
			\cg{\frac{1}{2},  m_2}
			   {J + \frac{1}{2}, M_q - m_2}
			   {J_2, M_q}\biggr.\\
             \shoveright{\biggl.
			\cg{\frac{1}{2},m_2}
			   {J-\frac{1}{2}, M-m_2}
			   {J,M}
			\ket{J_2, M_q; \frac{1}{2}, J + \frac{1}{2}, N-1,i_2}
			\biggr]}\\
		\times  \sum_{J'_2=J}^{J+1}	\sum_{m'_2}
			\biggl[
			\ket{J'_2, M'_r; \frac{1}{2}, J + \frac{1}{2}, N-1,i_2}
			\cg{\frac{1}{2},m'_2}
			   {J-\frac{1}{2}, M'-m'_2}
			   {J,M'}\biggr.\\
			 \biggl.
			\cg{\frac{1}{2},  m'_2}
			   {J + \frac{1}{2}, M'_r - m'_2}
			   {J'_2, M'_r}
			D_r^{J-\frac{1}{2},M-m'_2} 					
			\biggr] \ . 
	\end{multline}

\subsection{Evaluate sums} \label{collective:SubSubSec::EvaluateSums}

We are now tasked with showing that Eqs. \ref{collective:Equation::Unsimplified::Jplus12::Basic}-\ref{collective:Equation::Unsimplified::Jminus12::C} sum to $\mathbf{g}_{qr}(J,M,M',N)$ as written in Eq. \ref{collective:Equation::generalFormula}.  Before doing so, we observe that the $J_i,m_i$ and $J'_i,m'_i$ sums factor in all the equations above.  Moreover, if one replaces primed quantities with unprimed ones, the Clebsch-Gordan and $A,B,D$ coefficients of the kets in a given $J_i,m_i$ sum are identical to those of the bras in the related $J'_i,m'_i$ sum.  Therefore, we focus on simplifying the unprimed sums and then apply those results to the primed sums in order to simplify Eqs. \ref{collective:Equation::Unsimplified::Jplus12::Basic}-\ref{collective:Equation::Unsimplified::Jminus12::C}.  In Appendix \ref{collective:Appendix::SimplifySums}, we explicitly calculate two representative sums from these equations.  The calculations involve manipulating products of Clebsch-Gordan and $A,B,D$ coefficients.  Although tedious, the interested and pertinacious reader should have no trouble evaluating them for all relevant sums, finding in particular that the $J\pm 2$ terms vanish.  We forego detailing all those manipulations here and simply use the results in both the primed and unprimed terms of the equations above, which then simplify	Eq. \ref{collective:Equation::Unsimplified::Jplus12::Basic} to
	\begin{equation}  \label{collective:Equation::Simplified::JJplus1::First}
		\begin{split}
			&\frac{1}{d^J_N(2J+2)^2}\times \\
			\sum_{i_1 = 1}^{d^{J+\frac{1}{2}}_{N-1}}&
			D_q^{J,M}\ketbra{J+1, M_q; \frac{1}{2}, J + \frac{1}{2},N-1, i_1}
			                {J+1, M'_r;\frac{1}{2}, J + \frac{1}{2}, N-1, i_1}D_r^{J,M'}\\
			- & A_q^{J,M}\ketbra{J, M_q; \frac{1}{2}, J + \frac{1}{2}, N-1, i_1}
			                    {J+1, M'_r; \frac{1}{2}, J + \frac{1}{2}, N-1, i_1}D_r^{J,M'}\\
			- & D_q^{J,M}\ketbra{J+1, M_q; \frac{1}{2}, J + \frac{1}{2}, N-1, i_1}
								{J, M'_r; \frac{1}{2}, J + \frac{1}{2}, N-1, i_1}A_r^{J,M'}\\
			+ & A_q^{J,M}\ketbra{J, M_q; \frac{1}{2}, J + \frac{1}{2}, N-1, i_1}
								{J, M'_r; \frac{1}{2}, J + \frac{1}{2}, N-1, i_1}A_r^{J,M'}\ ,
		\end{split}
	\end{equation}	
	Eq. \ref{collective:Equation::Unsimplified::Jminus12::Basic} to 
	\begin{equation} \label{collective:Equation::Simplified::JJminus1::First}
		\begin{split}
			\frac{1}{d^J_N 4J^2}\sum_{i_1=1}^{d^{J-\frac{1}{2}}_{N-1}} &
			  B_q^{J,M} \ketbra{J-1,M_q; \frac{1}{2}, J - \frac{1}{2}; N-1, i_1}
				              {J-1,M_r'; \frac{1}{2}, J - \frac{1}{2}; N-1, i_1} B_r^{J,M'}\\
			 +&A_q^{J,M} \ketbra{J,M_q; \frac{1}{2}, J - \frac{1}{2}; N-1, i_1}
					              {J-1,M_r'; \frac{1}{2}, J - \frac{1}{2}; N-1, i_1} B_r^{J,M'}\\
			 +&B_q^{J,M} \ketbra{J - 1,M_q; \frac{1}{2}, J - \frac{1}{2}; N-1, i_1}
					              {J,M_r'; \frac{1}{2}, J - \frac{1}{2}; N-1, i_1} A_r^{J,M'}\\					
			 +&A_q^{J,M} \ketbra{J,M_q; \frac{1}{2}, J - \frac{1}{2}; N-1, i_1}
					              {J,M_r'; \frac{1}{2}, J - \frac{1}{2}; N-1, i_1} A_r^{J,M'} \ ,
		\end{split}
	\end{equation}
	Eq. \ref{collective:Equation::Unsimplified::Jplus12::A} to 
	\begin{equation}  \label{collective:Equation::Simplified::JJplus1::Second}
		\begin{split}
			&\frac{1}{d^J_N(2J+1)} 
			\bigl[1 + 
			\frac{\alpha^{J+\frac{3}{2}}_{N-1}}
			     {d^{J+\frac{1}{2}}_{N-1}} 
			\frac{2J+2}
			     {J+\frac{3}{2}} \bigr] \sum_{i_1=1}^{d^{J+\frac{1}{2}}_{N-1}}\\
			&\frac{1}{(2J+2)^2}D_q^{J,M} \ketbra{J+1, M_q;\frac{1}{2}, J + \frac{1}{2}, N-1, i_1}
												{J+1, M'_r;\frac{1}{2}, J + \frac{1}{2}, N-1, i_1}D_r^{J,M'} \\
		   -&\frac{2(J + \frac{3}{2})}{(2J+2)^2} A_q^{J,M} \ketbra{J, M_q;\frac{1}{2}, J + \frac{1}{2}, N-1, i_1}
										{J+1, M'_r;\frac{1}{2}, J + \frac{1}{2}, N-1, i_1}D_r^{J,M'} \\
		   -&\frac{2(J + \frac{3}{2})}{(2J+2)^2} D_q^{J,M}\ketbra{J + 1, M_q;\frac{1}{2}, J + \frac{1}{2}, N-1, i_1}
												{J, M'_r;\frac{1}{2}, J + \frac{1}{2}, N-1, i_1}A_r^{J,M'} \\
		   +&\frac{(J + \frac{3}{2})^2}{(J+1)^2} A_q^{J,M}\ketbra{J, M_q;\frac{1}{2}, J + \frac{1}{2}, N-1, i_1}
													 {J, M'_r;\frac{1}{2}, J + \frac{1}{2}, N-1, i_1}A_r^{J,M'} \ ,
		\end{split}
	\end{equation}
	Eq. \ref{collective:Equation::Unsimplified::Jplus12::B} to
	\begin{equation} \label{collective:Equation::Simplified::JJminus1::Second}
		\begin{split}
			&\frac{\alpha^{J+\frac{1}{2}}_{N-1}}
			{d^J_Nd^{J-\frac{1}{2}}_{N-1} 2J(J+\frac{1}{2})} \times\\
			\sum_{i_1=1}^{d^{J-\frac{1}{2}}_{N-1}} &
			   (J+1) B_q^{J,M} \ketbra{J-1,M_q; \frac{1}{2}, J - \frac{1}{2}; N-1, i_1}
				              {J-1,M_r'; \frac{1}{2}, J - \frac{1}{2}; N-1, i_1} B_r^{J,M'}\\
			 +&A_q^{J,M} \ketbra{J,M_q; \frac{1}{2}, J - \frac{1}{2}; N-1, i_1}
					              {J-1,M_r'; \frac{1}{2}, J - \frac{1}{2}; N-1, i_1} B_r^{J,M'}\\
			 +&B_q^{J,M} \ketbra{J - 1,M_q; \frac{1}{2}, J - \frac{1}{2}; N-1, i_1}
					              {J,M_r'; \frac{1}{2}, J - \frac{1}{2}; N-1, i_1} A_r^{J,M'}\\					
			 +&\frac{1}{ J + 1 }A_q^{J,M} \ketbra{J,M_q; \frac{1}{2}, J - \frac{1}{2}; N-1, i_1}
					              {J,M_r'; \frac{1}{2}, J - \frac{1}{2}; N-1, i_1} A_r^{J,M'} \ , 
		\end{split}
	\end{equation}
	Eq. \ref{collective:Equation::Unsimplified::Jplus12::C} to (since $J + 2$ terms vanish) 
	\begin{multline} \label{collective:Equation::Simplified::Jplus1::First}
		\frac{\alpha^{J+1}_N}{d^J_N2(J+1)}
		\frac{1}{d^{J+1}_N}
		\sum_{i_1=1}^{d^{J+\frac{3}{2}}_{N-1}} 
			D_q^{J,M}
			\ket{J+1,M_q;\frac{1}{2}, J+\frac{3}{2},N-1,i_1}\\
			      \bra{J+1,M'_r;\frac{1}{2}, J+\frac{3}{2},N-1,i_1} 
		    D_r^{J,M'} \ ,
	\end{multline}
	Eq. \ref{collective:Equation::Unsimplified::Jminus12::A} to
	\begin{equation} \label{collective:Equation::Simplified::JJminus1::Third}
		\begin{split}
			&\frac{1}{d^J_N4J^22(J-\frac{1}{2})}
			\bigl[ 1 +
			       \frac{\alpha^{J + \frac{1}{2}}_{N-1}}
			            {d^{J-\frac{1}{2}}_{N-1}}
			       \frac{2J}
			            { J + \frac{1}{2}}\bigr] \times\\
			&\sum_{i_1=1}^{d^{J-\frac{1}{2}}_{N-1}} 
			   B_q^{J,M} \ketbra{J-1,M_q; \frac{1}{2}, J - \frac{1}{2}; N-1, i_1}
				              {J-1,M_r'; \frac{1}{2}, J - \frac{1}{2}; N-1, i_1} B_r^{J,M'}\\
			 -&2(J-\frac{1}{2}) A_q^{J,M} \ketbra{J,M_q; \frac{1}{2}, J - \frac{1}{2}; N-1, i_1}
					              {J-1,M_r'; \frac{1}{2}, J - \frac{1}{2}; N-1, i_1} B_r^{J,M'}\\
			 -&2(J-\frac{1}{2})B_q^{J,M} \ketbra{J - 1,M_q; \frac{1}{2}, J - \frac{1}{2}; N-1, i_1}
					              {J,M_r'; \frac{1}{2}, J - \frac{1}{2}; N-1, i_1} A_r^{J,M'}\\					
			 +&4(J-\frac{1}{2})^2A_q^{J,M} \ketbra{J,M_q; \frac{1}{2}, J - \frac{1}{2}; N-1, i_1}
					              {J,M_r'; \frac{1}{2}, J - \frac{1}{2}; N-1, i_1} A_r^{J,M'}\ ,
		\end{split}
	\end{equation}
	Eq. \ref{collective:Equation::Unsimplified::Jminus12::B} to (since $J - 2$ terms vanish) 
	\begin{equation} \label{collective:Equation::Simplified::Jminus1::First}
		\frac{\alpha^J_N}{d^J_N2J}
		\frac{1}{d^{J-1}_N}
		\sum_{i_1=1}^{d^{J-\frac{3}{2}}_{N-1}}
		B_q^{J,M}
		\ketbra{J-1,M_q;\frac{1}{2},J-\frac{3}{2},N-1,i_1}
		{J-1,M'_r;\frac{1}{2},J-\frac{3}{2},N-1,i_1} 
		B_r^{J,M'} \ ,
	\end{equation}
	and Eq. \ref{collective:Equation::Unsimplified::Jminus12::C} to
	\begin{equation}  \label{collective:Equation::Simplified::JJplus1::Third}
		\begin{split}
			&\frac{\alpha^{J+\frac{1}{2}}_{N-1}}
						{d^J_N d^{J+\frac{1}{2}}_{N-1}(2J+1)} \times\\
		    \sum_{i_1=1}^{d^{j+\frac{1}{2}}_{N-1}} 
			&\frac{J}{J+1}D_q^{J,M}\ketbra{J+1, M_q;\frac{1}{2}, J + \frac{1}{2}, N-1, i_1}
					{J+1, M'_r;\frac{1}{2}, J + \frac{1}{2}, N-1, i_1}D_r^{J,M'} \\
		   +&\frac{1}{J+1}A_q^{J,M}\ketbra{J, M_q;\frac{1}{2}, J + \frac{1}{2}, N-1, i_1}
						{J+1, M'_r;\frac{1}{2}, J + \frac{1}{2}, N-1, i_1}D_r^{J,M'} \\			
		   +&\frac{1}{J+1}D_q^{J,M}\ketbra{J+1, M_q;\frac{1}{2}, J + \frac{1}{2}, N-1, i_1}
									{J, M'_r;\frac{1}{2}, J + \frac{1}{2}, N-1, i_1}A_r^{J,M'} \\					
		   +&\frac{1}{J(J+1)}A_q^{J,M}\ketbra{J, M_q;\frac{1}{2}, J + \frac{1}{2}, N-1, i_1}
						{J, M'_r;\frac{1}{2}, J + \frac{1}{2}, N-1, i_1}A_r^{J,M'} \ .
		\end{split}
	\end{equation}

\subsection{Recover $\mathbf{g}_{qr}(J,M,M',N)$}
We now combine the equations from the previous subsection to recover the Identity in Eq. \ref{collective:Equation::generalFormula}.  Given that density operators in the collective state representation lack coherences between different $J$ irreps, we expect $\ketbra{J\pm 1}{J}$ and $\ketbra{J}{J \pm 1}$ terms to vanish.  Since both the $\ketbra{J}{J \pm 1}$ and $\ketbra{J \pm 1}{J}$ terms have the same coefficients, we need only explicitly deal with one of the two.  Starting with $\ketbra{J+1}{J}$ coefficients from Eqs. \ref{collective:Equation::Simplified::JJplus1::First}, \ref{collective:Equation::Simplified::JJplus1::Second} and \ref{collective:Equation::Simplified::JJplus1::Third}, we find
		\begin{align} 
				\frac{1}{d^J_N}&\biggl(
				-\frac{1}{(2J+2)^2} -\frac{(2J+3)}{(2J+1)(2J+2)^2}\bigl[1 + 
				           \frac{\alpha^{J+\frac{3}{2}}_{N-1}}
				                {d^{J+\frac{1}{2}}_{N-1}}
				    \frac{2J+2}{J+\frac{3}{2}}\bigr] \nonumber\\
				&+\frac{1}{(J+1)(2J+1)}
				 \frac{\alpha^{J+\frac{1}{2}}_{N-1}}
				      {d^{J+\frac{1}{2}}_{N-1}}\biggr)\nonumber\displaybreak[1]\\
				= &\frac{1}{d^J_N(2J+2)^2}\biggl(-1-\frac{(N+1)}{2J+1}+\frac{N+2J+2}{2J+1}\biggr) \\
			    = &\ 0
		\end{align}
Similarly, for $\ketbra{J-1}{J}$ coefficients in Eqs. \ref{collective:Equation::Simplified::JJminus1::First}, \ref{collective:Equation::Simplified::JJminus1::Second} and \ref{collective:Equation::Simplified::JJminus1::Third}, we have
\begin{equation}
   \frac{1}{d^J_N4J^2}\biggl[ 1 + 
					\frac{\alpha^{J+\frac{1}{2}}_{N-1} 2J}
					{d^{J-\frac{1}{2}}_{N-1}(J+\frac{1}{2})}
					- \bigl[ 1 +
					       \frac{\alpha^{J + \frac{1}{2}}_{N-1}}
					            {d^{J-\frac{1}{2}}_{N-1}}
					       \frac{2J}
					            { J + \frac{1}{2}}\bigr]\biggr] = 0
\end{equation}

Turning to $J+1$ terms from Eqs. \ref{collective:Equation::Simplified::JJplus1::First}, \ref{collective:Equation::Simplified::JJplus1::Second} and \ref{collective:Equation::Simplified::JJplus1::Third}, the coefficients sum to
\begin{align}
		\frac{1}{d^J_N}  &
		\biggl(\frac{1}{(2J+2)^2} + 
		      \frac{1}{(2J+1)(2J+2)^2} 
		         \bigl[1 + 
		           \frac{\alpha^{J+\frac{3}{2}}_{N-1}}
		                {d^{J+\frac{1}{2}}_{N-1}}
		           \frac{2J+2}{J+\frac{3}{2}}\bigr] \\
		       & + \frac{J}{(J+1)(2J+1)}
		         \frac{\alpha^{J+\frac{1}{2}}_{N-1}}
		          {d^{J+\frac{1}{2}}_{N-1}}\biggr)\nonumber\displaybreak[1]\\
		=&     \frac{1}{d^J_N(2J+2)^2} 
		     \biggl(1+\frac{N+1}{(2J+3)(2J+1)} 
		         + \frac{J(N+2J+2)}{2J+1}\biggr)\nonumber\displaybreak[1]\\
		=&\frac{1}{d^J_N} 
		 \frac{2 J+N+4}{8 J^2+20 J+12}\nonumber\\
		=& \frac{1}{d^J_N2(J+1)}
		    \frac{\alpha^{J+1}_N}{d^{J+1}_N}
\end{align}
which gives overall
	\begin{multline}  \label{collective:Equation::Simplified::Jplus1::Second}
			\frac{\alpha^{J+1}_N}{d^J_N2(J+1)}
			\frac{1}{d^{J+1}_N}
			\sum_{i_1=1}^{d^{J+\frac{1}{2}}_N}
			D_q^{J,M}
			\ket{J+1,M_q;\frac{1}{2},J+\frac{1}{2},N-1,i_1}\\
			       \times\bra{J+1,M'_r;\frac{1}{2},J+\frac{1}{2},N-1,i_1}
			D_r^{J,M'} \ .
	\end{multline}
	
The $J$ terms from Eqs. \ref{collective:Equation::Simplified::JJplus1::First}, \ref{collective:Equation::Simplified::JJplus1::Second} and \ref{collective:Equation::Simplified::JJplus1::Third} have coefficients
	\begin{align}
			\frac{1}{d^J_N} &
			\biggl(\frac{1}{(2J+2)^2} 
			 + \frac{(2J+3)^2}{(2J+1)(2J+2)^2}
			    \bigl[1 + 
			          \frac{\alpha^{J+\frac{3}{2}}_{N-1}}
			               {d^{J+\frac{1}{2}}_{N-1}}
			          \frac{2J+2}{J+\frac{3}{2}}\bigr] \nonumber\\
			 & + \frac{1}{J(J+1)(2J+1)}
			    \frac{\alpha^{J+\frac{1}{2}}_{N-1}}
			      {d^{J+\frac{1}{2}}_{N-1}}\biggr) \nonumber \displaybreak[1]\\
			=&\frac{1}{d^J_N(2J+2)^2}
			    \biggl(1+\frac{(N+1)(2J+3)}{2J+1}
			   +\frac{N+2J+2}{J(2J+1)}\biggr) \\ 
		   =& \frac{1}{d^J_N2J} \biggl[ 1 + 
			    \frac{\alpha^{J+1}_N}{d^J_N}
			     \frac{2J+1}{J+1}\biggr]
	\end{align}
which gives overall
	\begin{multline} \label{collective:Equation::Simplified::J::First}
	  \frac{1}{2J}
	  \biggl[1+
	         \frac{\alpha^{J+1}_N}{d^J_N}
	         \frac{2J+1}{J+1}\biggr] \times\\
	         \frac{1}{d^J_N}\sum_{i_2=1}^{d^{J+\frac{1}{2}}_{N-1}}
			 A_q^{J,M}
	         \ket{J,M_q;\frac{1}{2},J+\frac{1}{2},N-1,i_1}\\
	            \times  \bra{J,M'_r;\frac{1}{2},J+\frac{1}{2},N-1,i_1}
	         A_r^{J,M'} \ .
	\end{multline}

Similarly, the $J$ terms from Eqs. \ref{collective:Equation::Simplified::JJminus1::First}, \ref{collective:Equation::Simplified::JJminus1::Second} and \ref{collective:Equation::Simplified::JJminus1::Third} have coefficients
\begin{align}
		\frac{1}{d^J_N4J^2}&\biggl[ 1 + 
							\frac{\alpha^{J+\frac{1}{2}}_{N-1} 2J}
							{d^{J-\frac{1}{2}}_{N-1}(J+\frac{1}{2})(J+1)} \nonumber\\
							&+ 2(J-\frac{1}{2})\bigl[ 1 +
							       \frac{\alpha^{J + \frac{1}{2}}_{N-1}}
							            {d^{J-\frac{1}{2}}_{N-1}}
							       \frac{2J}
							            { J + \frac{1}{2}}\bigr]\biggr] \nonumber\displaybreak[1]\\
		= & \frac{1}{d^J_N4J^2}\biggl[ 2J + \frac{N - 2J}{2J+1}\bigl(\frac{1}{J+1} + 2J -1\bigr)\biggr]\nonumber\\
		=& \frac{1}{d^J_N2J}\biggl[ 1 + \frac{\alpha^{J+1}_{N}}
											 {d^J_{N}}
										\frac{2J+1}
										     {J+1}\biggr]
\end{align}
which gives 
\begin{multline} \label{collective:Equation::Simplified::J::Second}
	\frac{1}{2J}\biggl[ 1 + \frac{\alpha^{J+1}_{N}}
										 {d^J_{N}}
									\frac{2J+1}
									     {J+1}\biggr] \times\\
	\frac{1}{d^J_N}\sum_{i_1 = 1}^{d^{J - \frac{1}{2}}_{N-1}}
	A_q^{J,M}
	\ket{J,M_q;\frac{1}{2},J-\frac{1}{2},N-1,i_1}\\
	\times	 \bra{J,M'_r;\frac{1}{2},J-\frac{1}{2},N-1,i_1}
	A_r^{J,M'} \ .
\end{multline}

And finally, the $J-1$ sums from Eqs. \ref{collective:Equation::Simplified::JJminus1::First}, \ref{collective:Equation::Simplified::JJminus1::Second} and \ref{collective:Equation::Simplified::JJminus1::Third} have coefficients
	\begin{align} 
			\frac{1}{d^J_N4J^2}&\biggl[ 1 + 
								\frac{\alpha^{J+\frac{1}{2}}_{N-1} 2J(J+1)}
								{d^{J-\frac{1}{2}}_{N-1}(J+\frac{1}{2})} \nonumber\\
								&+ \frac{1}{2(J-\frac{1}{2})}\bigl[ 1 +
								       \frac{\alpha^{J + \frac{1}{2}}_{N-1}}
								            {d^{J-\frac{1}{2}}_{N-1}}
								       \frac{2J}
								            { J + \frac{1}{2}}\bigr]\biggr] \nonumber\displaybreak[1]\\
		  &=\frac{1}{d^J_N4J^2}\biggl[ 1 + \frac{1}{2J-1}  + 
		                               \frac{N-2J}{2J+1}\bigl( J + 1 + \frac{1}{2J-1}\bigr) \biggr] \nonumber\\
		  &=\frac{1}{d^J_N2J}\frac{\alpha^J_N}{d^{J-1}_N}		
	\end{align}
which gives
	\begin{multline} \label{collective:Equation::Simplified::Jminus1::Second}
		\frac{\alpha^J_N}{d^J_N2J}\frac{1}{d^{J-1}_N}
			\sum_{i_1=1}^{d^{J-\frac{1}{2}}_{N-1}}
			B_q^{J,M}
			\ket{J - 1,M_q;\frac{1}{2},J-\frac{1}{2},N-1,i_1}\\
			\times	   \bra{J - 1,M'_r;\frac{1}{2},J-\frac{1}{2},N-1,i_1}
			B_r^{J,M'} \ .
	\end{multline}
	
From the definition of $\ketbra{J,M,N}{J,M',N}$ given in Eq. \ref{collective:Equation::EffectiveDensityMatrixElement}, we see that Eqs. \ref{collective:Equation::Simplified::J::First} and \ref{collective:Equation::Simplified::J::Second} correspond to the $\Eketbra{J,M,N}{J,M',N}$ terms in Eq. \ref{collective:Equation::generalFormula}.  A similar combination of Eqs. \ref{collective:Equation::Simplified::Jminus1::First} and \ref{collective:Equation::Simplified::Jminus1::Second} corresponds to the $J-1$ term and the combination of Eqs. \ref{collective:Equation::Simplified::Jplus1::First} and \ref{collective:Equation::Simplified::Jplus1::Second} corresponds to the $J$ term.  We have thus shown inductively that Identity \ref{collective:Id::MainResult} holds. \qed
\section{Examples} \label{collective:Section::Examples}
As discussed in the introduction, realistic decoherence models for an ensemble of spin particles are often described most aptly by a symmetric sum over local channels.  Consider, for example, the open system dynamics governed by the master equation
\begin{equation} \label{collective:Equation::master}
	\frac{d \hat{\rho}(t)}{dt}  = - i [ \hat{H}, \hat{\rho}(t) ] + \Gamma \mathcal{L}[ \hat{s} ]
		\hat{\rho}(t),
\end{equation}
where $\hat{H}$ (and any measurements performed) are described by collective operators, but the decoherence involves the symmetric Linblad superoperator
$\mathcal{L}[\hat{s}]$ of the form in Eq.\ (\ref{collective:Equation::SymmetricLindbladForm}).  As this decoherence model does not preserve symmetric states, it has been common practice to consider instead the associated collective process $\mathcal{L}[\hat{S}]$ given in Eq.\ (\ref{collective:eq:collectiveLindblad}) with $\hat{S} = \sum_n \hat{s}^{(n)}$.

To illustrate the difference between symmetric and collective decoherence models, consider the open system dynamics of two representative problems.  First, compare the dynamics generated by the symmetric-local $\mathcal{L}[\hat{s}]$ versus collective $\mathcal{L}[\hat{S}]$ Linblad master equations applied to an initial superposition (cat) state $\ket{\psi(0)} = \left( \ket{\frac{N}{2},+ \frac{N}{2}} 
	                    + \ket{\frac{N}{2},-\frac{N}{2}} \right) / \sqrt{2}$. 
Figure \ref{collective:Figure::catStateDecoherence}(a-b) depicts the fidelity $\mathcal{F}(t) =  \dual{\psi(0)} \hat{\rho}(t)\ket{\psi(0}$ evolved under Eq.\ (\ref{collective:Equation::master}) (with $\hat{H} = 0$) for two different types of decoherence channels: Fig.\ \ref{collective:Figure::catStateDecoherence}(a1-a2) compares the collective versus symmetric master equations with $\hat{s} = \hat{\sigma_{-}}$ for $N=10$ and $N=100$ particles, respectively; and Fig.\ \ref{collective:Figure::catStateDecoherence}(b1-b2) makes a similar comparison for $\hat{s} = \hat{\sigma}_{\mathrm{z}}$.  The examples considered (including some not reported here) suggest symmetric local decoherence models can generate dynamics that are appreciably different from their collective analogs.  This is perhaps not too surprising: for an initially symmetric state, collective decoherence models $\mathcal{L}[ \hat{S}]$ confine the dynamics to only maximum-$J$ irrep; symmetric local models $\mathcal{L}[\hat{s}]$ do not necessary preserve the irrep decomposition of the initial state.  Fig. \ref{collective:Figure::catStateDecoherence}(c) depicts the norm of each total-$J$ irrep block of the density operator $N_J = \mathrm{tr}[ \hat{P}_J \hat{\rho}(t)]$ as a function of time for $\mathcal{L}[ \hat{\sigma}_-]$ ($\hat{P}_{\mathrm{J}} = \sum_M \Eketbra{J,M}{J,M}$).   The observation that small-$J$ irreps are only minimally populated suggests that further model reduction by truncating the Hilbert space to only the largest $J$ blocks could be beneficial.

\begin{figure}[tb]
\begin{center} \includegraphics[scale=0.9]{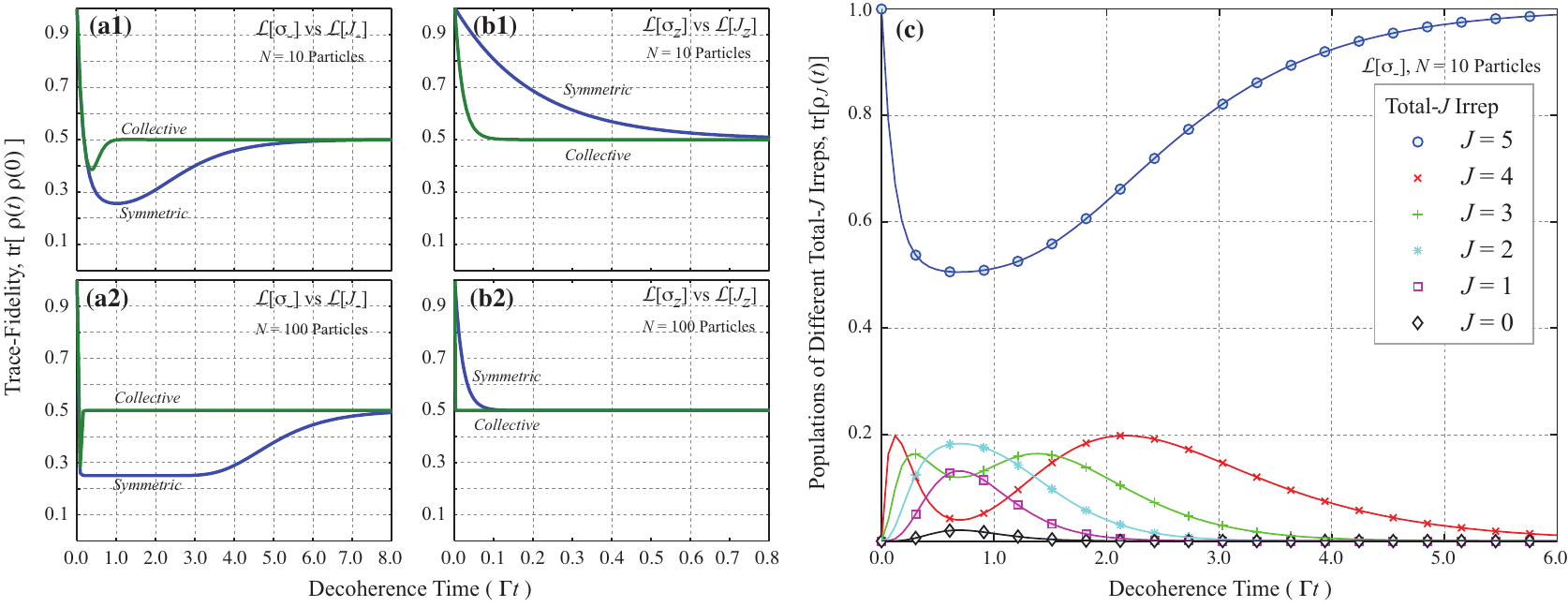} \end{center}
\vspace{-5mm}
\caption[Cat state undergoing collective decoherence]{Decoherence of the initial superposition state $\ket{\psi(0)} = ( \ket{+\frac{N}{2}} + \ket{-\frac{N}{2}}) / \sqrt{2}$: (a1-a2) time-dependent fidelity with the initial state for both symmetric local $\mathcal{L}[\hat{\sigma}_-]$ and collective $\mathcal{L}[\hat{J}_-]$ decoherence for different numbers of particles; (b1-b2) similar comparison for $\mathcal{L}[\hat{\sigma}_{\mathrm{z}}]$ versus $\mathcal{L}[\hat{J}_{\mathrm{z}}]$; (c) time-dependent populations of different total-$J$ irreps for $\mathcal{L}[\hat{\sigma}_-]$. \label{collective:Figure::catStateDecoherence}}
\end{figure}

\begin{figure}[th]
\begin{center}	\includegraphics{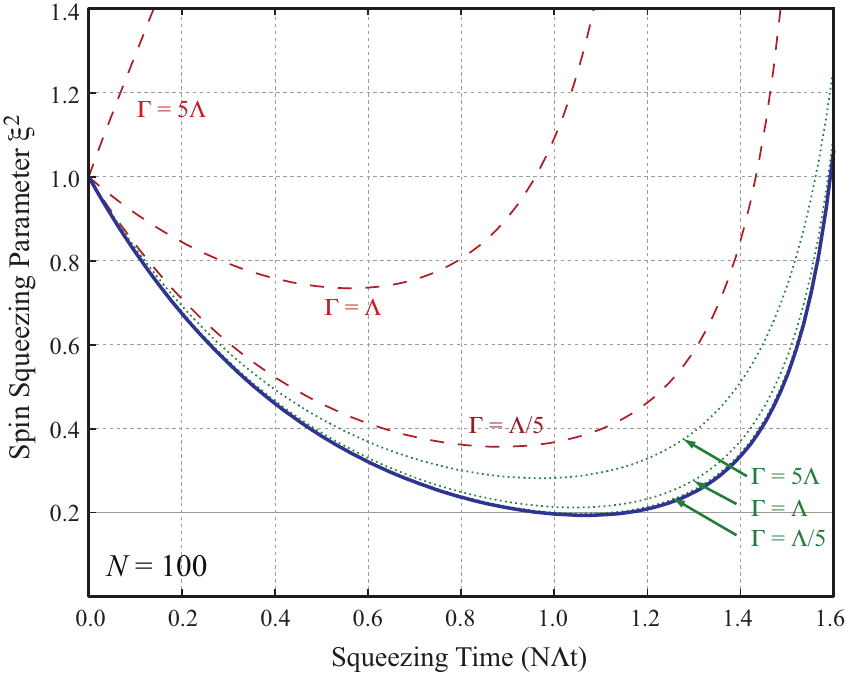} \end{center}
\vspace{-5mm}
\caption[Collective spin squeezing with collective decoherence]{Time-evolution of the squeezing parameter $\xi^2$ for a spin ensemble driven by $\hat{H} = -i  \Lambda ( \hat{J}_+^2 - \hat{J}_-^2)$ subject to $\mathcal{L}[ \hat{\sigma}_-]$ (dotted lines) and $\mathcal{L}[ \hat{J}_-]$ (dashed lines) with relative decoherence rates $\Gamma = \Lambda/ 5, \Lambda, 5 \Lambda$.  For comparison, the solid line denotes decoherence-free squeezing.
	\label{collective:Figure::Squeezing}
}
\end{figure}

As a second example, consider comparing symmetric-local versus collective decoherence models applied to dynamically-generated spin squeezing under the counter-twisting Hamiltonian $\hat{H} = -i \Lambda (\hat{J}_{+}^2 - \hat{J}_{-}^2)$ \citep{Kitagawa:1993a}.  I performed simulations by time-evolving Eq.\ (\ref{collective:Equation::master}) from the initial spin-coherent state $\ket{\frac{N}{2}, \frac{N}{2}}$ for $N=100$ with $\mathcal{L}[\hat{\sigma}_-]$ and $\mathcal{L}[\hat{J}_-]$.  Figure \ref{collective:Figure::Squeezing} depicts the time-dependent squeezing parameter $\xi^2 = N \langle \Delta \hat{J}_y^2 \rangle / \langle \hat{J}_z \rangle ^2$, each for $\Gamma = \Lambda/5, \Lambda, 5 \Lambda$.   Under the conditions considered, symmetric local decoherence wave evidently less destructive to the squeezing dynamics than collective models.  As observed for the cat-state dynamics, to large extent the main effect of symmetric-local decoherence is leakage from the maximum $J$ irrep.  But since the driving Hamiltonian $\hat{H}$ involves only collective spin operators, the coherent dynamics decouple for different total $J$: the population in each irrep block then undergoes its own squeezing, evidently making the dynamics more resistant to symmetric local decoherence than collective processes.
\section{Summary} \label{collective:Section::Conclusion}
I have presented an exact formula for efficiently expressing symmetric processes of an ensemble of spin-1/2 particles.  The efficiency is achieved by generalizing the notion of collective spin states to be any such state which does not distinguish degenerate irreps.  For a collection of $N$ spin-1/2 particles, the effective Hilbert space dimension grows as $N^2$, a drastic reduction from the full Hilbert space scaling of $2^N$.  The collective representation is used in Identity \ref{collective:Id::MainResult}, which gives a closed-form expression for evaluating non-collective terms from symmetric Lindblad operators.  Simulations confirm that symmetric local decoherence models can be drastically different than collective decoherence models.  Unfortunately, due to the complicated structure of adding spin-$J>\frac{1}{2}$ particles \citep{Mihailov:1977a}, these results do not appear to generalize.  Nonetheless, I believe that this approach will become a useful tool in analyzing collective spin phenomenon and in particular, accurately considering the role of decoherence in collective spin experiments.

\begin{subappendices}
\section{Explicit Simplification of Typical Sums} \label{collective:Appendix::SimplifySums}

	In \ref{collective:SubSubSec::EvaluateSums}, we simplify the sums in Eqs. \ref{collective:Equation::Unsimplified::Jplus12::Basic}-\ref{collective:Equation::Unsimplified::Jminus12::C} but do not go through the detailed algebra.  The work involves manipulating products of Clebsch-Gordan and $A,B,D$ coefficients.  In this appendix, I explicitly calculate two representative sums from this set and invite the reader to calculate the remainder in a similar fashion.  
	
	First, consider the sums over $J_1$ and $m_1$ in Eq. \ref{collective:Equation::Unsimplified::Jplus12::Basic}, which is representative of sums in  Eqs. \ref{collective:Equation::Unsimplified::Jplus12::Basic} and \ref{collective:Equation::Unsimplified::Jminus12::Basic}.  For $J_1 = J+1$
	\begin{align}
		&A_q^{\frac{1}{2},\frac{1}{2}}
		\cg{\frac{1}{2}, {\frac{1}{2}}_q}
		   {J + \frac{1}{2}, M-\frac{1}{2}}
		   {J+1, M_q}
		\cg{\frac{1}{2},\frac{1}{2}}
		   {J+\frac{1}{2},M-\frac{1}{2}}
		   {J,M}\nonumber\\
	    &+
		A_q^{\frac{1}{2},-\frac{1}{2}}
		\cg{\frac{1}{2}, {-\frac{1}{2}}_q}
		   {J + \frac{1}{2}, M + \frac{1}{2}}
		   {J+1, M_q}
		\cg{\frac{1}{2},-\frac{1}{2}}
		   {J+\frac{1}{2},M + \frac{1}{2}}
		   {J,M}\nonumber \displaybreak[1]\\
		=&\frac{1}{2(J+1)} \begin{cases}
			    -\sqrt{(J+M+2)(J+M+1)} & q = +\\
			     \sqrt{(J-M+2)(J-M+1)} & q = -\\
			     \sqrt{(J-M+1)(J+M+1)} & q = z
	   	    \end{cases}\\
		=&  \frac{1}{2J+2} D_q^{J,M}
	\end{align}
and for $J_1 = J$
	\begin{align}
		&A_q^{\frac{1}{2},\frac{1}{2}}
		\cg{\frac{1}{2}, {\frac{1}{2}}_q}
		   {J + \frac{1}{2}, M-\frac{1}{2}}
		   {J, M_q}
		\cg{\frac{1}{2},\frac{1}{2}}
		   {J+\frac{1}{2},M-\frac{1}{2}}
		   {J,M}\nonumber\\
	    &+ 
		A_q^{\frac{1}{2},-\frac{1}{2}}
		\cg{\frac{1}{2}, {-\frac{1}{2}}_q}
		   {J + \frac{1}{2}, M + \frac{1}{2}}
		   {J, M_q}
		\cg{\frac{1}{2},-\frac{1}{2}}
		   {J+\frac{1}{2},M + \frac{1}{2}}
		   {J,M}\nonumber \displaybreak[1]\\
		=&-\frac{1}{2(J+1)}
			\begin{cases}
				\sqrt{(J-M)(J+M+1)} & q = +\\
				\sqrt{(J+M)(J-M+1)} & q = -\\
				M & q = z
			\end{cases} \\
		= & -\frac{1}{2J+2} A_q^{J,M}
	\end{align}
	where $A_{+}^{\frac{1}{2},\frac{1}{2}} = A_{-}^{\frac{1}{2},-\frac{1}{2}} = 0$, $A_{+}^{\frac{1}{2},-\frac{1}{2}} = A_{-}^{\frac{1}{2},\frac{1}{2}} = 1$ and $A_{z}^{\frac{1}{2},\pm \frac{1}{2}} = \pm \frac{1}{2}$. 
	
Similarly, consider the sums over $J_1$ and $m_1$ in Eq. \ref{collective:Equation::Unsimplified::Jplus12::B}, which is representative of Eqs. \ref{collective:Equation::Unsimplified::Jplus12::A}-\ref{collective:Equation::Unsimplified::Jminus12::C}.  For $J_1 = J-1$, we have
	\begin{align}
		B_q^{J+\frac{1}{2},M-\frac{1}{2}}&
		\cg{\frac{1}{2},  \frac{1}{2}}
		   {J - \frac{1}{2}, M_q - \frac{1}{2}}
		   {J-1, M_q}
		\cg{\frac{1}{2},\frac{1}{2}}
		   {J+\frac{1}{2}, M-\frac{1}{2}}
		   {J,M} \nonumber\\
	   &+B_q^{J+\frac{1}{2},M + \frac{1}{2}}
		\cg{\frac{1}{2},  -\frac{1}{2}}
		   {J - \frac{1}{2}, M_q + \frac{1}{2}}
		   {J-1, M_q}
		\cg{\frac{1}{2},-\frac{1}{2}}
		   {J+\frac{1}{2}, M + \frac{1}{2}}
		   {J,M} \nonumber \displaybreak[1]\\
	 = & \sqrt{\frac{(J-M_q)(J-M+1)}
			        {4J(J+1)}}
			B_q^{J+\frac{1}{2},M-\frac{1}{2}}\nonumber\\
		&\times \biggl[ 1 + 	\sqrt{\frac{J+M_q}
					   {J-M_q}}
			\sqrt{\frac{J+M+1}
					   {J-M+1}}
			\frac{B_q^{J+\frac{1}{2},M + \frac{1}{2}}}
			     {B_q^{J+\frac{1}{2},M-\frac{1}{2}}}\biggr]\nonumber\displaybreak[1]\\
	 = & \sqrt{\frac{(J-M_q)(J-M+1)}
			        {4J(J+1)}}
			B_q^{J+\frac{1}{2},M-\frac{1}{2}}\frac{2(J+1)}{J-M+1}\nonumber\displaybreak[1]\\
	 = & \sqrt{\frac{J+1}{J}}
		\begin{cases}
			\sqrt{(J-M)(J-M-1)} & q = +\\
		   -\sqrt{(J+M)(J+M-1)} & q = -\\
		    \sqrt{(J+M)(J-M)} &   q = z
	 	 \end{cases} \nonumber\\
	 = & \sqrt{\frac{J+1}{J}} B_q^{J,M} \ .
	\end{align}
Similarly, for $J_1 = J$, we have
	\begin{align}
		B_q^{J+\frac{1}{2},M-\frac{1}{2}}&
		\cg{\frac{1}{2},  \frac{1}{2}}
		   {J - \frac{1}{2}, M_q - \frac{1}{2}}
		   {J, M_q}
		\cg{\frac{1}{2},\frac{1}{2}}
		   {J+\frac{1}{2}, M-\frac{1}{2}}
		   {J,M} \nonumber \\
	   &+ B_q{J+\frac{1}{2},M + \frac{1}{2}}
		\cg{\frac{1}{2},  -\frac{1}{2}}
		   {J - \frac{1}{2}, M_q + \frac{1}{2}}
		   {J, M_q}
		\cg{\frac{1}{2},-\frac{1}{2}}
		   {J+\frac{1}{2}, M + \frac{1}{2}}
		   {J,M} \nonumber\displaybreak[1]\\
		= & \sqrt{\frac{(J + M_q)(J-M+1)}
				        {4J(J+1)}}
				B_q^{J+\frac{1}{2},M-\frac{1}{2}} \nonumber\\
		  &\times		\biggl[ 1 -
				\sqrt{\frac{J-M_q}
						   {J+M_q}}
				\sqrt{\frac{J+M+1}
						   {J-M+1}}
				\frac{B_q^{J+\frac{1}{2},M + \frac{1}{2}}}
				     {B_q^{J+\frac{1}{2},M-\frac{1}{2}}}\biggr]\nonumber\displaybreak[1]\\
	    = & \sqrt{\frac{(J + M_q)(J-M+1)}
				        {4J(J+1)}}
				B_q^{J+\frac{1}{2},M-\frac{1}{2}} \nonumber\\
		  & \times 
			2\begin{cases}
				\frac{1}{J-M+1} & q = +\\
			   -\frac{1}{J + M + 1} & q = -\\
			    \frac{M}{(J+M)(J-M+1)} & q = z
			\end{cases} \nonumber \displaybreak[1]\\
		= & \sqrt{\frac{1}{J(J+1)}}
			\begin{cases}
				\sqrt{(J-M)(J+M+1)} & q = +\\
				\sqrt{(J+M)(J-M+1)} & q = -\\				
				M & q = z				
			\end{cases} \nonumber \displaybreak[1]\\
		= & \sqrt{\frac{1}{J(J+1)}} A_q^{J,M} \ .
	\end{align}
\end{subappendices}
\appendix
\chapter{Riccati Equations}
\label{appendix:riccati}
The ability to solve matrix Riccati equations is an important tool when using the Kalman filter given in Theorem \ref{thm:kalman_bucy}.  Following \citep{Stockton:2004a,Reid:1972a}, I review a technique for reducing the nonlinear system into a set of linear differential equations.  Consider the matrix $Z(t)$ which satisfies the Riccati equation
\begin{equation}
	\frac{dZ}{dt} = A(t)Z - ZD(t) - ZC(t)Z + B(t) .
\end{equation}
Instead of solving this directly introduce the decomposition $Z(t) = X(t)Y^{-1}(t)$ and solve the equivalent linear system
\begin{equation}
	\begin{bmatrix}
		\frac{dX(t)}{dt}\\
		\frac{dY(t)}{dt}\\		
	\end{bmatrix}
	= 
	\begin{bmatrix}
		A(t) & B(t)\\
		C(t) & D(t)
	\end{bmatrix}
	\begin{bmatrix}
		X(t)\\
		Y(t)
	\end{bmatrix}
\end{equation}
To check that this decomposition satisfies the original Riccati equation, we simply calculate
\begin{align}
	\frac{dZ}{dt} &= X\frac{dY^{-1}}{dt} + \frac{dX}{dt}Y^{-1}\\
	              &= X(-Y^{-1}\frac{dY}{dt}Y^{-1}) + \frac{dX}{dt}Y^{-1}\\
	              &= -XY^{-1}\bigl(C(t)X + D(t)Y\bigr)Y^{-1} + (A(t)X + B(t)Y)Y^{-1}\\
				  &=  -ZC(t)Z -ZD(t) + A(t)Z + B(t)
\end{align}
\chapter{Numerical Methods for Stochastic Differential Equations} 
\label{appendix:numerical_methods_for_stochastic_differential_equations}
Many of the filters and SDEs in this thesis do not admit an analytic solution.  As such, it is useful to have methods for numerically simulating or integrating a stochastic system.  An excellent resource for such methods is the text by \citet{Kloeden:1992a}, in which the following two integrators are discussed in more detail.  In the following, I consider the $n$-dimensional stochastic process $X_t$ and the $m$-dimensional Wiener process $dW_t$ related via the SDE
\begin{equation}
    dX_t = a(t,X_t)dt + \sum_{j=1}^{m} b^{j}(t,X_t)dW_t^j.
\end{equation}

The first integrator is the \emph{Euler} or \emph{Euler-Maruyama} scheme and is the trivial extension of the standard Euler method for integrating ordinary differential equations.  We begin by discretizing the time-domain in terms of a step-size $\Delta t$.  The integrator then estimates the state at times $t_n = t_0 + n\Delta_t$ by stepping the state forward via the SDE.  The Euler approximation for the $k$-th entry of $X_t$ at time-step $t_{n+1}$, given the state at timestep $t_n$, is then given by
\begin{equation}
    \tilde{X}_{t_{n+1}}^k = \tilde{X}_{t_n}^k + a^k(t_n,\tilde{X}_{t_n})\Delta t 
                  + \sum_{j=1}^{m} b^{k,j}(t_n,\tilde{X}_{t_{n+1}})\Delta W^j
\end{equation}
where $\Delta W^j$ is a pseudo-random number with mean zero and variance $\Delta t$.  This simplicity of this approach is a clear advantage, but its order of convergence is 0.5, meaning $\norm{X_{t_n} - \tilde{X}_{t_n}}_{2,\mathbb{P}} \leq \alpha (\Delta t)^{1/2}$ where $\alpha$ is some constant independent of $\Delta t$.  Note that this is worse than the Euler method for ODEs, which is order 1.0.

The other method used significantly in simulations for this thesis is an order 2.0 weak predictor-corrector method, which offers improved stability and convergence at the cost of more computational complexity.  
For simplicity, restrict consideration to a single noise term $m = 1$ and time-independent $a$ and $b$; see \citet[Chapt. 15]{Kloeden:1992a} for the multiple noise version.  The estimate is then given by
\begin{equation}
    \tilde{X}_{t_{n+1}} = \tilde{X}_{t_n} + \frac{1}{2}( a(\bar{X}_{t_{n+1}}) + a( \tilde{X}_{t_n}))\Delta t
                 + \phi_{t_n}
\end{equation}
where
\begin{multline}
    \phi_{t_n} = \frac{1}{4}\left[
                        b(\bar{\Upsilon}^+) + b(\bar{\Upsilon}^-) + 2b(\tilde{X}_{t_n})
                        \right]\Delta W \\
               + \frac{1}{4}\left[b(\bar{\Upsilon}^+) - b(\bar{\Upsilon}^-)\right]
                            \left[ (\Delta W)^2 - \Delta t\right] (\Delta t)^{-1/2}
\end{multline}
where the supporting values are given by
\begin{equation}
    \bar{\Upsilon}^{\pm} = \tilde{X}_{t_n} + a(\tilde{X}_{t_n})\Delta t 
                           \pm b(\tilde{X}_{t_n}) \sqrt{\Delta t}
\end{equation}
and with predictor
\begin{equation}
    \bar{X}_{t_{n+1}} = \tilde{X}_{t_n} + \frac{1}{2}( a(\bar{\Upsilon}) + a( \tilde{X}_{t_n}))\Delta t
                 + \phi_{t_n}
\end{equation}
with supporting value
\begin{equation}
    \bar{\Upsilon} = \tilde{X}_{t_{n}} + a(\tilde{X}_{t_n})\Delta t + b(\tilde{X}_{t_n}) \Delta W .
\end{equation}
\chapter{Stochastic Schr\"{o}dinger Equation} \label{magnetometry:app:SSE}
Lacking any extra sources of decoherence, pure states remain pure under the dynamics described by the quantum filtering equation.  As such, it is often convenient for analysis and simulation to have a pure state description of the dynamics in terms of a \emph{stochastic Schr\"{o}dinger equation} (SSE).  In this appendix, I briefly derive the SSE for the general adjoint filter
\begin{align}
    d\rho_t &= -i[H,\rho_t]dt + \left(\hat{L}\rho_t\hat{L}^{\dag} - \frac{1}{2} \hat{L}^{\dag}\hat{L}\rho_t - \frac{1}{2}\rho_t\hat{L}^{\dag}\hat{L}\right)dt \nonumber\\
           & + \left( \hat{L}\rho_t + \rho_t \hat{L}^{\dag}
            - \Tr{(\hat{L} + \hat{L}^{\dag})\rho_t}\rho_t\right)dW_t .
\end{align}
We begin by writing
\begin{align}
	d\ket{\psi}_t &= A\ket{\psi}_tdt + B\ket{\psi}_tdW_t\\
	d\bra{\psi}_t &= \bra{\psi}_tA^{\dagger}dt + \bra{\psi}_tB^{\dagger}dW_t
\end{align}
From the It\^{o} rules, we have
\begin{eqnarray}
	d(\rho_t) &= & d(\ketbra{\psi}{\psi}_t) \nonumber \\
			 &= & \ket{\psi}_td(\bra{\psi}_t) + d(\ket{\psi}_t)\bra{\psi}_t + d(\ket{\psi_t})d(\bra{\psi}_t) \\
			  &= & (A\rho_t + \rho_tA^{\dagger})dt + (B\rho_t + \rho_tB^{\dagger})dW_t + B\rho_t B^{\dagger}dt \nonumber
\end{eqnarray}
Comparing the coefficients to the quantum filtering equation, we read off
\begin{align}
	B & = L - \expect{L}\\
	B^{\dagger} & = L^{\dagger} - \expect{L^{\dagger}}	
\end{align}
so that
\begin{equation}
	B\rho_tB^{\dagger} = L\rho_tL^{\dagger} - \expect{L^{\dagger}}L\rho_t
						 - \expect{L}\rho_tL^{\dagger}
						 + \expect{L}\expect{L^{\dagger}}\rho_t
\end{equation}
We try setting
\begin{equation}
	A = -\frac{1}{2}\left(L^{\dagger}L - 2\expect{L^{\dagger}}L 
				+ \expect{L}\expect{L^{\dagger}}\right)
\end{equation}
which means that
\begin{align}
	A \rho_t + \rho_t A^{\dagger} &=
					-\frac{1}{2}\left(L^{\dagger}L - 2\expect{L^{\dagger}}L
								+ \expect{L}\expect{L^{\dagger}}\right)\rho_t \nonumber\\
			&	  - \rho_t	\frac{1}{2}\left(L^{\dagger}L - 2\expect{L}L^{\dagger}
								+ \expect{L}\expect{L^{\dagger}}\right)\\
		&= -\frac{1}{2}L^{\dagger}L\rho_t -\frac{1}{2}\rho_t L^{\dagger}L
			+\expect{L^{\dagger}}L\rho_t \nonumber\\
		& + \rho_tL^{\dagger}\expect{L}
			- \expect{L}\expect{L^{\dagger}}\rho_t
\end{align}
and therefore
\begin{align}
	A \rho_t + \rho_t A^{\dagger} + B\rho_tB^{\dagger} 
			&= -\frac{1}{2}L^{\dagger}L\rho_t -\frac{1}{2}\rho_t L^{\dagger}L \nonumber\\
			&	+\expect{L^{\dagger}}L\rho_t + \rho_tL^{\dagger}\expect{L}\nonumber\\
			&	- \expect{L}\expect{L^{\dagger}}\rho_t
				+	L\rho_tL^{\dagger} \nonumber\\
			& - \expect{L^{\dagger}}L\rho_t
									 - \expect{L}\rho_tL^{\dagger}
									 + \expect{L}\expect{L^{\dagger}}\rho_t\\
			&= L \rho_t L^{\dagger} -\frac{1}{2}L^{\dagger}L\rho_t -\frac{1}{2}\rho_t L^{\dagger}L
\end{align}
which is the deterministic part of the quantum filtering equation as desired.

\end{document}